\documentclass[a4paper,11pt]{article}
\pdfoutput=1 
\usepackage{jheppub}
\usepackage[T1]{fontenc}
\usepackage{amssymb,amsmath,amsfonts,nccmath,amsthm}
\usepackage{hyperref}
\usepackage{IEEEtrantools}
\usepackage{float}
\usepackage{graphicx}
\usepackage{xcolor}
\usepackage{tensor}
\usepackage[utf8]{inputenc}
\usepackage{subfig}
\usepackage{bm,bbm,bbold}	
\usepackage{scalerel}[2016/12/29]
\usepackage{enumitem}
\usepackage{multirow}
\usepackage{tabularx}
\usepackage{booktabs}
\usepackage{array, makecell}
\usepackage{mathtools}
\usepackage{algorithm}
\usepackage[noend]{algpseudocode}

\usepackage{soul}
\usepackage{nccmath}
\usepackage{amsmath, blkarray}

\newcommand{\red}[1]{{\color{red} #1}}
\newcommand{\blue}[1]{{\color{blue} #1}}
\definecolor{myGreen}{RGB}{83,113,55}

\def\dd{\mathrm{d}}
\newcommand{\pd}{\partial}

\newcommand{\E}{\mathcal{E}}

\newcommand{\M}{\mathcal{M}}
\renewcommand{\O}{\mathcal{O}}
\newcommand{\R}{\mathcal{R}}
\renewcommand{\S}{\mathcal{S}}

\newcommand{\bbE}{\mathbb{E}}
\newcommand{\bbN}{\mathbb{N}}
\newcommand{\bbQ}{\mathbb{Q}}
\newcommand{\bbR}{\mathbb{R}}

\newcommand{\bfm}{\boldsymbol{m}}

\newcommand{\ud}[2]{^{#1}{}_{#2}}
\newcommand{\du}[2]{_{#1}{}^{#2}}

\newcommand{\ce}{\coloneqq}

\DeclareMathOperator{\rank}{rank}
\DeclareMathOperator{\class}{class}
\DeclareMathOperator{\csize}{size}
\newcommand{\lie}{\mathcal{L}}

\newcommand{\PD}[2]{\ensuremath{\frac{\partial #1}{\partial #2}}}

\makeatletter
\renewcommand{\boxed}[1]{\text{\fboxsep=.3em\fbox{\m@th$\displaystyle#1$}}}
\makeatother

\usepackage{eqparbox}
\newcommand{\Boxed}[2][A]{\fbox{\eqmakebox[#1][c]{$\mathstrut#2$}}}

\usepackage{blkarray}

\title{Counting Degrees of Freedom:\\ A Method Applicable from Scalars to $f(\bbQ)$ Gravity and Beyond}
\author{Lavinia Heisenberg}
\affiliation{Institut f\"{u}r Theoretische Physik, Philosophenweg 16, 69120 Heidelberg, Germany}

\emailAdd{L.Heisenberg@thphys.uni-heidelberg.de}

\abstract{We present a clear, step-by-step method for counting degrees of freedom and identifying constraints in general field theories. This approach, grounded in the works of Einstein, Hilbert, Cartan, Kuranishi, and, more recently, Seiler, is neither Lagrangian nor Hamiltonian in nature. Instead, it applies directly to the field equations. We offer a transparent physical interpretation of the method, establish new results, and explore a broad set of examples to demonstrate its power and generality. Notably, we apply the method to $f(\bbQ)$ gravity, where traditional techniques such as the Dirac-Bergmann algorithm are ineffective, and obtain seven propagating degrees of freedom.}


\begin{document}
	\allowdisplaybreaks[1]
	\maketitle
	\flushbottom

\newtheorem{definition}{Definition}[section]
\newtheorem{example}{Example}[section]
\newtheorem{theorem}{Theorem}[section]
\newtheorem{corollary}{Corollary}[section]
\newtheorem{lemma}{Lemma}[section]

\section{Introduction}\label{sec:Introduction}
The concept of degrees of freedom is fundamental to our understanding of dynamical systems, encompassing areas such as classical mechanics, statistical physics, quantum field theory, effective field theory, and modified gravity. In essence, the degrees of freedom correspond to the number of independent variables needed to fully characterize the state of a dynamical system. The presence of constraints and gauge symmetries complicates this task, necessitating a systematic and mathematically rigorous approach. Misidentifying the propagating degrees of freedom can lead to inconsistencies, including the emergence of unphysical modes or the misclassification of interactions. These issues have far-reaching consequences in high-energy physics, cosmology, and condensed matter theory, where the accurate enumeration of degrees of freedom directly impacts renormalization procedures, anomaly cancellations, and the theoretical consistency of proposed models. This point is especially critical: errors in counting can lead to mistaken conclusions about the presence of ghosts, the stability of solutions, or the internal consistency of a new theoretical framework.

In this work, we present a mathematically robust methodology for systematically determining the number of independent degrees of freedom in field theories. Our approach is rooted in first principles and is broadly applicable—from classical field models to quantum gauge theories. By providing a clear and unambiguous framework for counting, we aim to sharpen our understanding of fundamental physics and offer a practical tool for ongoing research in high-energy physics, gravitation, and beyond.

Several established methods for counting degrees of freedom are widely used in the literature. Chief among them are the Dirac–Bergmann algorithm~\cite{Dirac:1950, Bergmann:1951, DiracBook, Wipf:1993} and the covariant phase space method~\cite{Crnkovic:1987, Lee:1990, HenneauxBook}. However, both approaches face technical challenges and can be difficult to implement in practice. In particular, the Dirac–Bergmann algorithm is known to break down in certain contexts~\cite{SundermeyerBook, Seiler:1995, Seiler:2000, Blagojevic:2020}.

A notable example of this failure arises in teleparallel theories of gravity. When applied to $f(\bbQ)$ gravity, the Dirac–Bergmann method incorrectly predicted eight propagating degrees of freedom~\cite{Hu:2022, Tomonari:2023}. This error was promptly identified and corrected in~\cite{DAmbrosio:2023}, where it was shown that the mistake stemmed from a flawed assumption: namely, that the time evolution of primary constraints yields a system of \emph{linear equations} for the Lagrange multipliers. In fact, in gravity theories involving torsion~\cite{BeltranJimenez:2019, Blixt:2018, Blixt:2019, Blixt:2020} or non-metricity~\cite{BeltranJimenez:2019, DAmbrosio:2020a, Dambrosio:2020b}—as well as in other contexts~\cite{SundermeyerBook, Wipf:1993}—this assumption fails. Instead, the time evolution produces partial differential equations for the Lagrange multipliers, obstructing completion of the Dirac–Bergmann procedure. Entire classes of teleparallel theories suffer from this issue.

In~\cite{DAmbrosio:2023}, an upper bound of seven degrees of freedom was established for $f(\bbQ)$ gravity, correcting the earlier result. This was later confirmed through independent, though labor-intensive, perturbative methods~\cite{Heisenberg:2023b}, reinforcing the conclusion that the theory indeed propagates seven degrees of freedom.

This work is motivated by the need for a reliable, broadly applicable method that avoids the limitations of the Dirac–Bergmann approach while remaining rooted in solid mathematical foundations.

The method we develop here is neither Lagrangian nor Hamiltonian in nature. Instead, it deals directly with the field equations. Its origins trace back to the work of Albert Einstein~\cite{EinsteinBook}, \'Elie Cartan~\cite{Cartan:1930, CartanBook}, Masatake Kuranishi~\cite{Kuranishi:1957}, and, more recently Werner Seiler~\cite{Seiler:1995, Seiler:1995b, Seiler:2000, SeilerBook}. The core idea---due to Einstein---is to use formal power series to analyze how strongly field equations constrain the fields. While Einstein’s original formulation is too cumbersome for practical use, refinements by Cartan and Kuranishi made the method tractable. Further modern developments, particularly by Seiler~\cite{Seiler:1995, Seiler:1995b, Seiler:2000, SeilerBook}, have rendered it into a powerful analytical framework.

The structure of this paper is as follows: In Section~\ref{sec:BasicIdeas}, we review existing methods for counting degrees of freedom and discuss their limitations. Section~\ref{sec:JetBundleApproach} introduces the new framework and illustrates its utility through explicit examples. In Section~\ref{sec:FormalPSSExtnsionToGT}, we examine its implications for various classes of field theories. We conclude in Section~\ref{sec:Conclusion} with a summary and outlook on future directions.

\section{Basic Ideas on Counting Degrees of Freedom}\label{sec:BasicIdeas}
What is a degree of freedom? This question typically arises in classical mechanics, where it is explored in the context of point particles and rigid bodies. However, a deeper examination of this concept---especially its adaptation in gauge and field theories---provides valuable insight, particularly for understanding the method we illustrate here. In the following subsections, we will develop the notion of degrees of freedom step by step, introducing progressively more refined definitions. The fundamental ideas presented here will later reappear in different mathematical forms in Sections~\ref{sec:JetBundleApproach} and~\ref{sec:FormalPSSExtnsionToGT}.

\subsection{Degrees of Freedom in Classical Mechanics}\label{ssec:DOFinCM}
Let us begin with an informal definition in the context of classical mechanics. A free particle is said to have three degrees of freedom because it can move independently in any of the three spatial dimensions. If the particle is constrained to a surface, it has only two degrees of freedom, as its motion is restricted to that surface. This idea is intuitive and straightforward. However, this informal and somewhat imprecise notion does not generalize easily to field theory. Unlike a particle, a field is a physical entity that exists at every point in space and time\footnote{Or at least within a compact spacetime region, which defines the system’s extent, duration, and the limits of our observation.}. Consequently, there is no meaningful way to say that ``the field can freely move in the $x$-direction''.

A more flexible notion, which remains meaningful in field theories, is based on counting the initial data required to \textit{uniquely} determine the time evolution of the physical system under consideration. As a simple example, consider the harmonic oscillator, described by the second-order differential equation
\begin{equation}
\ddot{q}(t) + \omega^2 q(t) = 0\,, \quad \text{with } \omega > 0\,.
\end{equation}
The general solution to this equation is well known:
\begin{equation}
q(t) = A \cos(\omega t) + B \sin(\omega t)\,,
\end{equation}
where $A$ and $B$ are arbitrary real constants. To determine a \textit{unique} time evolution, we must specify two initial conditions that fix the values of $A$ and $B$. Concretely, we need to provide the initial position, $q_0 \ce q(t_0)$, and the initial velocity, $\dot{q}_0 \ce \dot{q}(t_0)$, at some initial time $t = t_0$. Since we are free to choose any values for $q_0$ and $\dot{q}_0$, we say that the harmonic oscillator has two \textit{phase space degrees of freedom}. The \textit{configuration space degrees of freedom} are simply half the number of phase space degrees of freedom.

Moreover, one can see that counting the freely specifiable initial data aligns with our earlier, more intuitive notion that the particle can only move in certain directions, but not in others. For instance, if a particle is constrained to a surface, we are free to specify its starting position on the surface and its initial velocity, provided the velocity is tangential to the surface. However, we \textit{cannot} place the particle ``above'' or ``below'' the surface, nor can we assign it an initial velocity that is orthogonal to the surface. Thus, we can freely specify two pieces of initial data for position and two for velocity, yielding four phase space degrees of freedom, or equivalently, two configuration space degrees of freedom---just as expected from our intuition.

All of the above is straightforward. However, the reason for explaining this idea in detail is its usefulness in field theories, as we shall discuss in the next subsection.

\subsection{Degrees of Freedom in Field Theory}
The second definition of degrees of freedom, introduced in the previous subsection, is based on two key ideas: (\textit{i}) equations of motion require a certain amount of initial data to yield a unique solution, and (\textit{ii}) constraints can restrict our freedom in specifying this initial data. These principles extend naturally to field theories with only minor modifications.

To illustrate this point while keeping the discussion simple, we shall temporarily assume a four-dimensional spacetime manifold $\M$, tensorial fields $\Psi$ of unspecified index structure, second-order field equations $\bbE = 0$, and no gauge freedom. In setting up a well-posed initial value problem, the first step is to select a three-dimensional, spacelike Cauchy hypersurface $\Sigma$. On this surface, we must prescribe two fields:
\begin{equation}
\Psi_{\Sigma} \ce \left. \Psi \right|_{\Sigma} \quad \text{and} \quad
\dot{\Psi}_\Sigma \ce \left. \left(\lie_n \Psi \right) \right|_{\Sigma},
\end{equation}
where $\lie_n$ denotes the Lie derivative along the unit timelike normal vector $n$ to $\Sigma$ (see Figure~\ref{fig:IVP}).

\begin{figure}[htb!]
	\centering
	\includegraphics[width=0.7\linewidth]{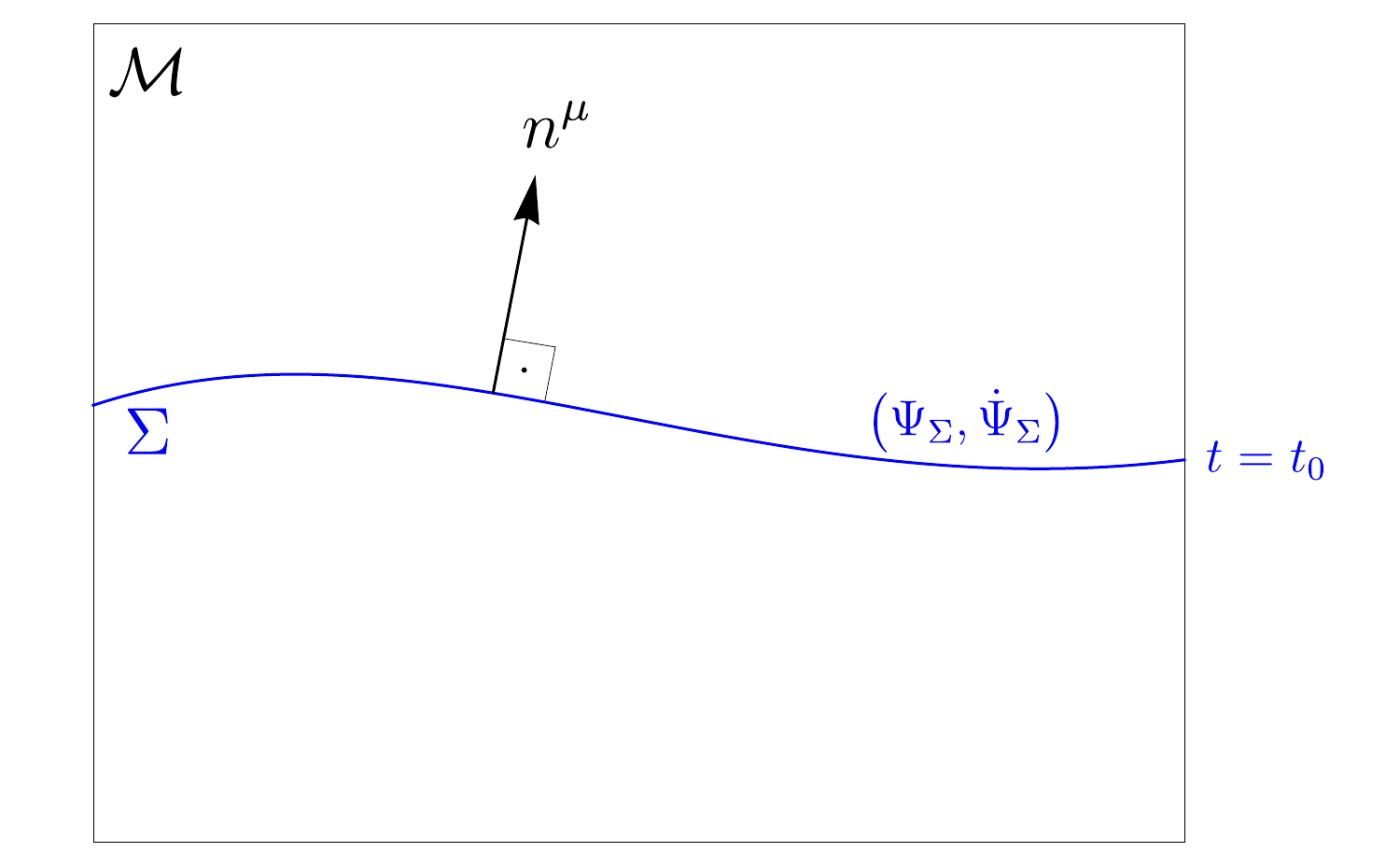}
	\caption{\textit{To formulate a well-posed initial value problem, one must specify the field $\Psi$ and its time derivative, $\dot{\Psi}$, on a Cauchy surface $\Sigma$. The time direction is determined by the normal vector $n$ to $\Sigma$, and the Cauchy surface itself corresponds to the instant of time $t = t_0$.}}
	\label{fig:IVP}
\end{figure}

Intuitively, the hypersurface $\Sigma$ can be understood as an ``instant of time $t=t_0$'', with $\Psi_{\Sigma}$ representing the field configuration at that moment (analogous to the initial position in classical mechanics). Similarly, $\dot{\Psi}_{\Sigma}$ describes the change of $\Psi_{\Sigma}$ over time, evaluated at the instant defined by $\Sigma$. In simpler terms, $\dot{\Psi}_{\Sigma}$ serves as the ``initial velocity'' of $\Psi$.

These are the minimal requirements for setting up an initial value problem. In the absence of gauge symmetries and constraint equations, we can now conclude that a theory described by the field equations $\bbE = 0$ and the field content $\Psi$ propagates $2m$ phase space degrees of freedom, where $m$ is the number of algebraically independent field components contained in~$\Psi$.

The counting changes when not all equations in $\bbE = 0$ are dynamical. If $\bbE = 0$ includes equations that are only first order or even zeroth order in time derivatives, we classify them as constraints rather than dynamical equations. More important than this terminology is the fact that such equations impose restrictions on the fields $\Psi_\Sigma$ and $\dot{\Psi}_\Sigma$ that can be specified on $\Sigma$. Consequently, the number of phase space degrees of freedom is reduced by the number of constraints on $\Psi_\Sigma$ and $\dot{\Psi}_{\Sigma}$.

If, in addition to constraints, we also allow for gauge symmetries\footnote{It is well known that whenever gauge symmetries are present, they are accompanied by constraint equations.}, we must account for the fact that gauge freedom allows us to fix some of the field components of $\Psi$ arbitrarily. This means that analyzing only the dynamical equations, constraint equations, and initial values of the fields does not necessarily reveal the true physical degrees of freedom when gauge symmetries are present. Instead, this approach gives us a mixture of physical and gauge degrees of freedom without distinguishing between them. To properly count the physical degrees of freedom, we must also impose gauge-fixing conditions.

All of this is well known and conceptually clear. However, there is value in explicitly spelling out these ideas because, despite their intuitive nature, they can be challenging to implement in practice. When encountering a new field theory---one that, more often than not, possesses some gauge symmetry---we naturally ask: How many degrees of freedom does it propagate? In light of our discussion so far, we can refine this question as follows: Given a set of field equations, how can we determine how much free data must be specified and how many gauge modes need to be fixed in order to obtain a unique solution?

In some cases, such as electromagnetism, this question can be answered directly through simple considerations. Of course, systematic methods exist for addressing this type of question across large classes of theories. The most prominent approach is the Dirac-Bergmann algorithm. However, this method has its limitations~\cite{SundermeyerBook, Seiler:1995, DAmbrosio:2023}, and recent studies have shown that entire classes of teleparallel gravity theories---whether based on torsion or non-metricity---\textit{cannot} be treated using the Dirac-Bergmann method~\cite{DAmbrosio:2023}. The reason why the Dirac-Bergmann approach can fail is a seemingly innocuous assumption that is tacitly present in the formalism: it presupposes that the time-evolution of primary constraints gives rise to \emph{linear equations} for the Lagrange multipliers. However, in $f(\bbQ)$ gravity and a wide range of other theories~\cite{SundermeyerBook, Wipf:1993, DAmbrosio:2023} this assumption is not valid. Rather, one finds first order partial differential equations for the Lagrange multipliers, causing the method to collapse, except in very fortunate cases~\cite{Wipf:1993}.

This is where an alternative method, motivated by works of Einstein, Cartan, Kuranishi, and Seiler, proves more effective. The core idea dates back to Einstein, who used it as a guiding principle in his search for a unified field theory of gravity and electromagnetism. Subsequent refinements by the mathematicians Cartan, Kuranishi, and Seiler transformed Einstein’s initial insights into a powerful tool for analyzing partial differential equations in physical theories. In particular, this method provides a way to count degrees of freedom without encountering the same limitations as the Dirac-Bergmann algorithm. To begin with, there are no Lagrange multipliers involved. Furthermore, the whole method is tailored to deal with partial differential equations.

In the next subsection, following Einstein’s book~\cite{EinsteinBook}, we introduce the basic ideas of the method by using the relativistic wave equation of a scalar field as an example. The concepts we introduce along the way will re-emerge in later subsections. Einstein's simple method is thus a good way for building up intuition and a deeper understanding for the formalism that will be developed in these subsections.

\subsection{Einstein's Way of Counting}\label{ssec:EinsteinsMethod}
In his quest for a unified field theory, Einstein encountered the fact that field equations do not completely determine the fields. As we have seen earlier, there is always some free data that must be specified. Moreover, when testing different field equations for the same set of fields, one finds that some equations require less free data than others. In this sense, one can say that these equations determine the field ``stronger'' than others, which led Einstein to introduce the concept of \textit{strength}~\cite{EinsteinBook}.

The core idea behind the concept of strength is to expand the field into a formal Taylor series. This series is then substituted into the field equations. By performing an order-by-order analysis, one can determine which Taylor coefficients are fixed by the field equations and which ones need to be specified by hand. The fewer the number of freely specifiable Taylor coefficients, the stronger the field is determined by the field equations (i.e., the fewer the degrees of freedom).

To illustrate this concept in practice and clarify its relation to the number of degrees of freedom, we follow Einstein by considering the simple case of a single scalar field $\Phi$ that obeys the relativistic wave equation
\begin{align}\label{eq:WaveEquation}
    \Phi_{tt} - \Phi_{xx} - \Phi_{yy} - \Phi_{zz} = 0\,.
\end{align}
We use the shorthand notation $\Phi_{\mu\nu}$ to indicate second order derivatives instead of the more traditional but bulkier $\partial_\mu\partial_\nu\Phi$. We now assume that $\Phi$ is analytic in the point $p$ with coordinates $p^\mu=(p^t, p^x, p^y, p^z)$, allowing us to expand it as Taylor series around $p$:
\begin{align}\label{eq:TaylorExpansionPhi}
    \Phi(x) &= \Phi(p) + \Phi_\mu(p)\,(x^\mu-p^\mu) + \frac12 \Phi_{\mu\nu}(p)\,(x^\mu-p^\mu)(x^\nu-p^\nu)\notag\\
    &\phantom{=}+ \frac{1}{3!}\Phi_{\mu\nu\rho}(p)\,(x^\mu-p^\mu)(x^\nu-p^\nu)(x^\rho-p^\rho) + \dots\,,
\end{align}
where $\Phi(p)$, $\Phi_\mu(p)$, $\Phi_{\mu\nu}(p)$, and $\Phi_{\mu\nu\rho}(p)$ are the zeroth, first, second, and third-order Taylor coefficients evaluated at $p$. Taking into account that partial derivatives commute, we find that there are four coefficients $\Phi_\mu(p)$, ten coefficients $\Phi_{\mu\nu}(p)$, and twenty coefficients $\Phi_{\mu\nu\rho}(p)$. In general, at order $n$ in the Taylor expansion, there are
\begin{align}\label{eq:SquareBracket4D}
    \left[\begin{matrix}
    4 \\
    n
    \end{matrix}\right]
\end{align}
coefficients, where we defined the square bracket as
\begin{align}
    \left[\begin{matrix}
    m \\
    n
    \end{matrix}\right]
    \ce \begin{pmatrix}
        m + n - 1 \\
        n
    \end{pmatrix}
    = \frac{(m+n-1)!}{n!(m-1)!}\,.
\end{align}
Observe that~\eqref{eq:SquareBracket4D} also makes sense for $n = 0$. In this case, it simply gives $1$, which corresponds to the single Taylor coefficient at order zero. If we did not have the wave equation at our disposal, we would need to specify all these Taylor coefficients to fully determine $\Phi$. However, the wave equation~\eqref{eq:WaveEquation} implies that there are relations between the different Taylor coefficients. By substituting the expansion~\eqref{eq:TaylorExpansionPhi} into the wave equation~\eqref{eq:WaveEquation} and analyzing it order by order, we find that the zeroth and first-order terms are annihilated. This is expected since the wave equation is second-order. As a result, the zeroth and first-order Taylor coefficients are completely unconstrained by the wave equation. For the second-order term, however, we find the following single relation:
\begin{align}\label{eq:Rel1}
    \Phi_{tt}(p) - \Phi_{xx}(p) - \Phi_{yy}(p) - \Phi_{zz}(p) &= 0\,.
\end{align}
At third order, four distinct relations emerge:
\begin{align}\label{eq:Rel2}
    \Phi_{ttx}(p) - \Phi_{xxx}(p) - \Phi_{xyy}(p) - \Phi_{xzz}(p) &= 0 \notag\\
    \Phi_{tty}(p) - \Phi_{xxy}(p) - \Phi_{yyy}(p) - \Phi_{yzz}(p) &= 0 \notag\\
    \Phi_{ttz}(p) - \Phi_{xxz}(p) - \Phi_{yyz}(p) - \Phi_{zzz}(p) &= 0 \notag\\
    \Phi_{ttt}(p) - \Phi_{txx}(p) - \Phi_{tyy}(p) - \Phi_{tzz}(p) &= 0\,.
\end{align}
It is clear how to systematically construct all possible relations at any order in the Taylor expansion: The relation~\eqref{eq:Rel1} corresponds to evaluating the wave equation at the point $p$, while the relation~\eqref{eq:Rel2} is obtained by differentiating the wave equation with respect to $x$, $y$, $z$, and $t$, and then evaluating the differentiated equations at the point $p$. Similarly, the relations for the fourth-order Taylor coefficients are obtained by differentiating the wave equation twice and then evaluating at $p$. This procedure can be continued \textit{ad infinitum}, to any order of the Taylor expansion. From elementary combinatorial considerations, one can then infer that at order $n$, the wave equation imposes precisely
\begin{align}\label{eq:NumberOfRelations}
   \left[\begin{matrix}
        4\\
        n-2
    \end{matrix}\right]
    = \frac16 (n-1)\,n\,(n+1)
\end{align}
relations between the $n$-th order Taylor coefficients. Observe that~\eqref{eq:NumberOfRelations} also makes sense for $n=0$ and $n=1$. In both cases one simply finds zero, as we had already established. Each relation between Taylor coefficients has to be considered as a constraint imposed on us by the wave equation. However, many Taylor coefficients remain unconstrained. The number of these free Taylor coefficients at order $n$ is measured by
\begin{align}
    z \ce 
    \begin{pmatrix}
        \text{number of} \\
        \text{Taylor coefficients}\\
        \text{at order }n
    \end{pmatrix}
    -
    \begin{pmatrix}
        \text{number of}\\
        \text{relations}\\
        \text{at order }n
    \end{pmatrix}
    =
    \left[\begin{matrix}
        4 \\
        n
    \end{matrix}\right]
    -
    \left[\begin{matrix}
        4\\
        n-2
    \end{matrix}\right]\,.
\end{align}
The first half of the above equation is the general definition of $z$, which is also applicable to other field theories (excluding gauge symmetry for the time being). The second half simply represents what this definition amounts to in the example of the scalar field described by the relativistic wave equation.

It is easy to see that, in our case, the number $z$ is positive for all $n$. This has an important implication: We can fix the coefficients of the Taylor expansion order by order, while being certain that the $n$-th order relations do not impose constraints on the coefficients of order less than $n$ respectively \footnote{To see this, note that $z$ can only be negative if the number of independent relations at a given order exceeds the number of Taylor coefficients at that order. However, this can only occur if there is at least one relation that involves a Taylor coefficient of order less than $n$.}. In general, $z$ can become negative, which prevents us from determining the Taylor coefficients order by order. This limitation of Einstein's method was later overcome by Cartan and Kuranishi. We will discuss their solution in Subsection~\ref{ssec:CKAlgorithm}.

At this point, we recall our definition of degrees of freedom, which relies on counting the free data we must specify on a Cauchy surface to generate a unique solution from the field equations. This is not the same as counting the Taylor coefficients that are not constrained by the field equations. In fact, there are infinitely many free Taylor coefficients. However, $z$ still provides information about the free data. To illustrate this, we first give a qualitative argument, which we will later turn into a more mathematically precise statement: Since the wave equation is second-order, we must specify two functions in order to obtain a unique $\Phi$ as a solution to these equations. Thus, we can think of $\Phi$ as depending on two freely specifiable functions. Given that the field equations do not determine all Taylor coefficients of $\Phi$, it follows that some of these free Taylor coefficients represent the free functions of the initial value problem.

Our reasoning suggests the following strategy: We compare the number of free Taylor coefficients $z$ to the number of Taylor coefficients of a single function of four spacetime coordinates. The latter is simply given by~\eqref{eq:SquareBracket4D}. Clearly, both numbers diverge as $n \to \infty$. However, we expect $z$ to diverge at a faster but constant rate, since it contains, according to our reasoning, Taylor coefficients of at least two unknown functions. Thus, the ratio of $z$ to~\eqref{eq:SquareBracket4D} should yield a finite number at all orders of $n$. Indeed, we find
\begin{align}\label{eq:EinsteinExpansion}
    z &=
    \left[\begin{matrix}
        4\\
        n
    \end{matrix}\right]
    -
    \left[\begin{matrix}
        4\\
        n-2
    \end{matrix}\right]
     = 
     \left[\begin{matrix}
        4\\
        n
    \end{matrix}\right]
     -
    \left[\begin{matrix}
        4\\
        n
    \end{matrix}\right]
    \frac{6(n+1)}{(n+2)(n+3)} \notag\\
    &= \left[\begin{matrix}
        4\\
        n
    \end{matrix}\right]
    \left(z_0 - \frac{z_1}{n} + \mathcal{O}(1/n^2)\right)\,,
\end{align}
where, in the last step, we Taylor-expanded $\frac{(n-1)n}{(n+2)(n+3)}$ around $n = \infty$, and $z_0$, $z_1$ are the zeroth and first-order Taylor coefficients of this expansion, respectively. Following Einstein's nomenclature, we refer to $z_0$ as the \textit{compatibility coefficient} and to $z_1$ as the \textit{strength}. In the case of the wave equation, these coefficients turn out to be
\begin{align}\label{eq:zValues}
    z_0 &= 0\,, & z_1 &= 6\,.
\end{align}
Consistent with our expectation, we find that $z$ divided by~\eqref{eq:SquareBracket4D} has a finite limit as $n \to \infty$. From our discussion, we also expect that $z$ grows at a rate proportional to the number of free functions present in the general solution $\Phi$. Therefore, we need to interpret the numbers $z_0$ and $z_1$, as they determine the rate at which $z$ grows. To do this, we must clarify what it means for $\Phi$ to depend on unknown functions. In the most general case, $\Phi$ can depend on $f_1$ functions of one coordinate, $f_2$ functions of two coordinates, $f_3$ functions of three coordinates, and $f_4$ functions of four coordinates. If we Taylor-expand $\Phi$ and account for these dependencies, we obtain a total of
\begin{align}\label{eq:T}
    T \ce \sum_{k=1}^{4} f_k \begin{pmatrix}
        k  + n -1\\
        n
    \end{pmatrix} 
\end{align}
Taylor coefficients. This number must equal the number of Taylor coefficients left undetermined by the wave equation. This is true because the only remaining freedom in specifying $\Phi$ is the freedom to choose initial data, i.e., the freedom to choose two functions of three coordinates. We are therefore led to the condition
\begin{align}
    z &\overset{!}{=} T &&\Longrightarrow & \left(z_0-f_4\right) + \left(z_1 - 3 f_3\right)\frac{1}{n} + \O\left(\frac{1}{n^2}\right) &= 0\,.
\end{align}
On the right-hand side, we only kept terms up to first order in $1/n$. Solving this equation order by order gives us 
\begin{align}
    f_4 &= z_0\, &&\text{and} & f_3 &= \frac13 z_1\,.
\end{align}
In combination with~\eqref{eq:zValues}, we thus conclude that $\Phi$ contains $f_4 = 0$ functions of four coordinates and $f_3 = 2$ functions of three coordinates. This is in perfect agreement with our expectations, and it clarifies the meaning of the terms \textit{compatibility coefficient} and \textit{strength}: If $z_0$ is not zero, the general solution would contain arbitrary functions of all four coordinates, causing the equations to fail to be deterministic. Equations compatible with classical determinism must satisfy $z_0 = 0$. The strength $z_1$ then tells us how strongly a solution is determined by the field equations and how much freedom remains in choosing initial data. At this point, we conclude that the relativistic wave equation propagates two phase space degrees of freedom (i.e., one configuration space degree of freedom).

This simple example illustrates how a general procedure for extracting the number of physical degrees of freedom, which is applicable to more general types of fields and equations, can be developed. In broad terms, one proceeds as follows: 
\begin{enumerate} 
    \item Given a field $\Psi$, expand it in a Taylor series around some point $p$ and determine the number of Taylor coefficients at order $n$. Note that this number depends on the number of algebraically independent field components contained in $\Psi$, so it will, in general, differ from~\eqref{eq:SquareBracket4D}. 
    \item Plug the Taylor series into the field equations $\bbE = 0$ and evaluate them at $p$. Then determine all relations for the Taylor coefficients, order by order. 
    \item Determine the number $z$ of free Taylor coefficients by subtracting the number of independent relations at order $n$ from the number of Taylor coefficients at the same order. This step may require taking into account additional constraints on the fields $\Psi$ and possibly also gauge redundancy. 
    \item Assume that the general solution $\Psi$ to the field equations $\bbE$ can depend on $f_i$ functions of $i$ coordinates. This leads to an expression similar to~\eqref{eq:T}, which needs to be equated to $z$. 
    \item Finally, solve the equation $z = T$ for $f_4$ and $f_3$. For equations compatible with classical determinism, one should always find $f_4 = 0$, while $f_3$ directly gives the number of phase space degrees of freedom. 
\end{enumerate} 
In his book~\cite{EinsteinBook}, Einstein applies the steps outlined above to electromagnetism and general relativity. He finds that in both cases, $f_4 = 0$ and $f_3 = 4$, which reproduces the expected results. However, it also becomes clear that this method becomes more and more cumbersome as the field equations become more complicated and when gauge symmetries need to be taken into account.

In the next subsection, we highlight a few more limitations of this method and provide a sketch of how they can be overcome. The main part of this paper will then be devoted to fully developing this method and illustrating its use in various field theories.

\subsection{Limitations of Einstein's Method and Cartan's Refinements}\label{ssec:Limitations}

At the end of the previous subsection, we outlined the steps required to count degrees of freedom using Einstein’s method. However, this approach has a potential point of failure: it assumes that all relations among Taylor coefficients imposed by the equations of motion can be systematically determined order by order. If, for some reason, we cannot predict the number of independent relations at each order $n$, then step 3---determining the number $z$ of free Taylor coefficients---becomes problematic. In particular, if higher-order relations impose conditions on lower-order coefficients, previous results must be re-evaluated, potentially leading to inconsistencies or rendering the method unusable.  

Additionally, Einstein’s method does not explicitly account for gauge theories, and its connection to the initial value formulation of partial differential equations remains somewhat vague. The only clear intersection occurs when $T$ (the expected number of free functions) is introduced and equated to $z$.  

These concerns were central to Cartan’s correspondence with Einstein on this topic~\cite{DebeverBook}. In response, Cartan refined and formalized the approach, ultimately developing a mathematically rigorous theory of partial differential equations (PDEs)~\cite{Cartan:1930, CartanBook}. In its modern formulation using jet bundles, this theory will be discussed in detail in the next section. For now, we highlight its key result: the Cartan-Kuranishi theorem (see Subsection~\ref{ssec:CKAlgorithm}).  

This theorem states that, under mild assumptions, any system of PDEs can be transformed into an equivalent system that allows to systematically construct formal power series solutions order by order. Crucially, this result holds independently of spacetime dimension, PDE order, or whether the equations describe a topological field theory, a gauge theory, or any other type of field theory.  

Most importantly, the theorem ensures that the number of relations at every order can be predicted, allowing $z$ to be determined consistently at all orders as a well-defined, positive quantity. This foundational result eliminates the ambiguities in Einstein’s original method and provides a systematic framework for counting degrees of freedom. The theorem was then further used by Seiler~\cite{Seiler:1995, Seiler:1995b, Seiler:2000, SeilerBook}. In this work, we build on these results, introduce refinements, offer a more transparent physical interpretation, and present numerous examples and case studies.

\subsection{Conventions and Roadmap to the Main Result}
In what follows, we are loosely guided by Einstein's method for analyzing field theories. Our first step is to shift our perspective on partial differential equations (PDEs). Rather than viewing them as equations to be solved through integration, we reinterpret them as (nonlinear) constraint equations between fields and their derivatives. This transition begins in Subsection~\ref{ssec:BasicsOfVectorJetBundles}, where we introduce vector bundles and jet bundles. These concepts allow us to treat fields and their derivatives as independent entities living on a manifold—distinct from the spacetime manifold—known as the jet bundle. In Subsection~\ref{ssec:JetBundlePDEs}, we leverage this insight to interpret PDEs as equations that impose relations among these a priori independent fields and their derivatives. In other words, PDEs define submanifolds within the jet bundle\footnote{This situation is similar to what happens in the Dirac-Bergmann approach to constrained Hamiltonian systems. In that framework, the fields and their conjugate momenta are initially treated as independent variables that together define the phase space. However, the existence of constraints imposes relations among the fields and momenta. These relations restrict the physically allowed states to a submanifold of the phase space, often referred to as the constraint surface. It is on this submanifold that the true physical dynamics unfolds.}.

Next, we introduce the concepts of prolongations and projections. Recall that in Subsection~\ref{ssec:EinsteinsMethod}, we took derivatives of the wave equation and evaluated them at $p$, the point around which the formal Taylor series of the scalar field was expanded. This process generated new relations of higher order among the Taylor coefficients of $\Phi$. A prolongation achieves the same goal: it systematically generates new relations among fields and their higher-order derivatives.

Projection, on the other hand, is a concept not explicitly encountered in our discussion of Einstein’s method. However, we will quickly see its usefulness: it enables us to uncover hidden integrability conditions. This is a crucial step, as it helps overcome a key limitation of Einstein’s method. If integrability conditions exist but are not accounted for, the order-by-order construction of a formal power series solution fails\footnote{An example where hidden integrability conditions arise, obstructing an order-by-order construction \`{a} la Einstein, is Proca's theory of massive photons. Later, we will see how to overcome this issue systematically by employing the Cartan-Kuranishi algorithm.}.

In Subsection~\ref{ssec:Symbol}, we introduce another essential tool absent from Einstein’s method: the symbol of a PDE. The symbol provides valuable information about the highest-order derivatives appearing in a PDE. Some physicists may be familiar with the concept of a kinetic matrix, which governs the second-order time derivatives of a system. The symbol generalizes this idea to arbitrary orders of differentiation and to all types of derivatives---not just time derivatives. Moreover, the symbol allows us to check for constraints and integrability conditions, enabling us to explicitly construct these equations.

From our study of symbols, we are naturally led to the concept of involutive equations. Qualitatively speaking, these are the best-behaved equations, permitting the unobstructed, order-by-order construction of formal power series solutions. When Einstein developed his method, he was fortunate to study only involutive equations (even though he was unaware of the concept). However, not every physical equation is involutive, and we provide several examples throughout the text, particularly in Subsection~\ref{ssec:Examples}.

In Subsection~\ref{ssec:CKAlgorithm}, we present the first major result: the Cartan-Kuranishi algorithm (cf. Algorithm~\ref{alg:CK}). This algorithm takes any system of PDEs (under mild technical assumptions) and, if the system is not already involutive, systematically transforms it into an equivalent involutive system—one that preserves the original solution space. This procedure ensures that any PDE system can be completed into an involutive form, allowing for a systematic order-by-order construction of formal power series solutions.

In Subsection~\ref{ssec:OrderByOrder}, we demonstrate how such solutions can be explicitly constructed and emphasize the crucial role of working with involutive equations. Subsection~\ref{ssec:HilbertPolynomial} then introduces tools to quantify the size of the solution space, which will later be essential for counting degrees of freedom.

Before proceeding with this count, we must first discuss how gauge symmetries manifest in the jet bundle formalism. This is the focus of Subsection~\ref{ssec:GaugeTheory}. With these foundations in place, we arrive at a concrete algorithm for determining the number of independent degrees of freedom. To illustrate the method, we provide a variety of examples (see Subsection~\ref{ssec:Examples}).\medskip

\paragraph{\underline{Conventions:}}
Throughout this work, we adopt the following conventions:
\begin{itemize}
    \item The spacetime manifold is denoted by $\M$. It is $n$-dimensional, with coordinates $x^\mu = \{x^1, x^2, \dots, x^n\}$. Note that the index $\mu$ runs from $1$ to $n$, \textit{not} from $0$ to $n-1$ as is more common in the physics literature. When necessary, $x^n$ serves as the time coordinate, while $\{x^1, x^2, \dots, x^{n-1}\}$ are spatial coordinates.
    \item Lowercase Greek indices such as $\mu$, $\nu$, and $\rho$ are always reserved for spacetime indices.
    \item Multi-indices are denoted by lowercase boldface Roman letters, e.g., $\bfm$ (see Definition~\ref{def:multiindex}).
    \item Fields, or collections of fields, are denoted by $v^A$, where the index $A$ labels the components. For example, in a theory with a scalar field $\Phi$ and a vector field $A^\mu$ in $n=4$ dimensions, we write  
    \begin{align*}  
        v^A = (v^1, v^2, v^3, v^4, v^5) = (A^1, A^2, A^3, A^4, \Phi).  
    \end{align*}
    \item The letters $q$, $r$, and $s$ have fixed meanings:
    \begin{itemize}
        \item $q$ refers to the order of a PDE (see Definition~\ref{def:Rq})
        \item $r$ counts the number of times a PDE has been prolonged (see Definition~\ref{def:Prolongation})
        \item $s$ counts the number of times a PDE has been prolonged and then projected back (see Definition~\ref{def:Projection})
    \end{itemize}
    \item Unless stated otherwise, $m$ denotes the number of algebraically independent components of $v^A$. By ``algebraically independent'', we mean that symmetries of the tensor components have been taken into account. For example:
    \begin{itemize}
        \item A metric tensor generally has 16 components, but only 10 are independent due to its symmetry.
        \item The electromagnetic field strength tensor $F_{\mu\nu}$, which satisfies $F_{\mu\nu} = - F_{\nu\mu}$, has only six independent components due to its antisymmetry in $\mu$ and $\nu$.
    \end{itemize}
\end{itemize}

\section{The Jet Bundle Approach to Differential Equations}\label{sec:JetBundleApproach}

\subsection{Basics of Vector Bundles and Jet Bundles}\label{ssec:BasicsOfVectorJetBundles}

We aim to describe field equations for scalars, vectors, metrics, and other tensorial fields using a precise mathematical framework. To achieve this, we introduce the language of fiber bundles and jet bundles, which provide a powerful, coordinate-independent description of such equations and help formalize the method of counting degrees of freedom. For an accessible introduction to fiber bundles, we refer to Baez and Muniain~\cite{BaezBook}, while our discussion of jet bundles and the formal theory of partial differential equations follows, in part, Seiler~\cite{SeilerBook}.

While fiber and jet bundles might initially seem abstract, they are, in fact, highly practical tools. We will introduce only as much mathematical detail as needed, always grounding the discussion in physical examples. Readers familiar with electromagnetism, Proca's equation, general relativity, and related theories will find it easy to understand the concepts that follow.

To build intuition, let us first outline the fundamental ideas behind fiber and jet bundles and explain their relevance. Fiber bundles allow for an algebraic, coordinate-free description of tensor fields. Jet bundles take this a step further, treating \textit{derivatives of tensor fields} as independent algebraic variables. By combining these two frameworks, we gain a new perspective: instead of viewing field equations as differential equations to be \textit{integrated}, we can reinterpret them as non-linear algebraic constraints among independent variables---spacetime coordinates, fields, and their derivatives up to any given order.

To illustrate this shift, consider the relativistic wave equation~\eqref{eq:WaveEquation}, previously discussed in Subsection~\ref{ssec:EinsteinsMethod}. Traditionally, we solve it by integrating to find the field~$\Phi$. In the jet bundle formulation, however, we treat the spacetime coordinates $(t, x, y, z)$, the field~$\Phi$, and its derivatives as independent variables in a larger space---the jet bundle. The wave equation then acts as a constraint: for example, the second time derivative $\Phi_{tt}$ cannot be specified freely but must be related to the spatial derivatives $\Phi_{xx}$, $\Phi_{yy}$, and $\Phi_{zz}$. This viewpoint aligns precisely with our earlier discussion of Einstein's method for counting degrees of freedom.

A jet bundle of order $n$ provides a space where a tensor field and all its derivatives up to order $n$ are treated as independent entities. This naturally lends itself to the study of formal power series, such as Taylor expansions, without concerns about convergence, as we always work at a finite order.

In summary, jet bundles offer the ideal mathematical setting for constructing formal power series solutions to differential equations while treating these equations as constraints on the coefficients of such series. This is precisely the approach we employed in Subsection~\ref{ssec:EinsteinsMethod} when analyzing the relativistic wave equation.

To clarify these basic ideas, we begin by considering a vector field $V$ on an $n$-dimensional manifold $\M$. From the traditional differential-geometric perspective, a vector field assigns a tangent vector to each point $p \in \M$. Specifically, at each point $p \in \M$, the vector field $V$ is represented as an element of the tangent space $T_p \M$. Since we can define a tangent space $T_p \M$ at every point of $\M$, we can think of $\M$ as being covered by its tangent spaces, as illustrated in Figure~\ref{fig:TpS1}. In this figure, the tangent spaces appear to intersect or overlap. However, these spaces are independent of each other, and the intersections are merely a result of our visual representation.
\begin{figure}[ht]
    \centering
    \includegraphics[width=0.5\linewidth]{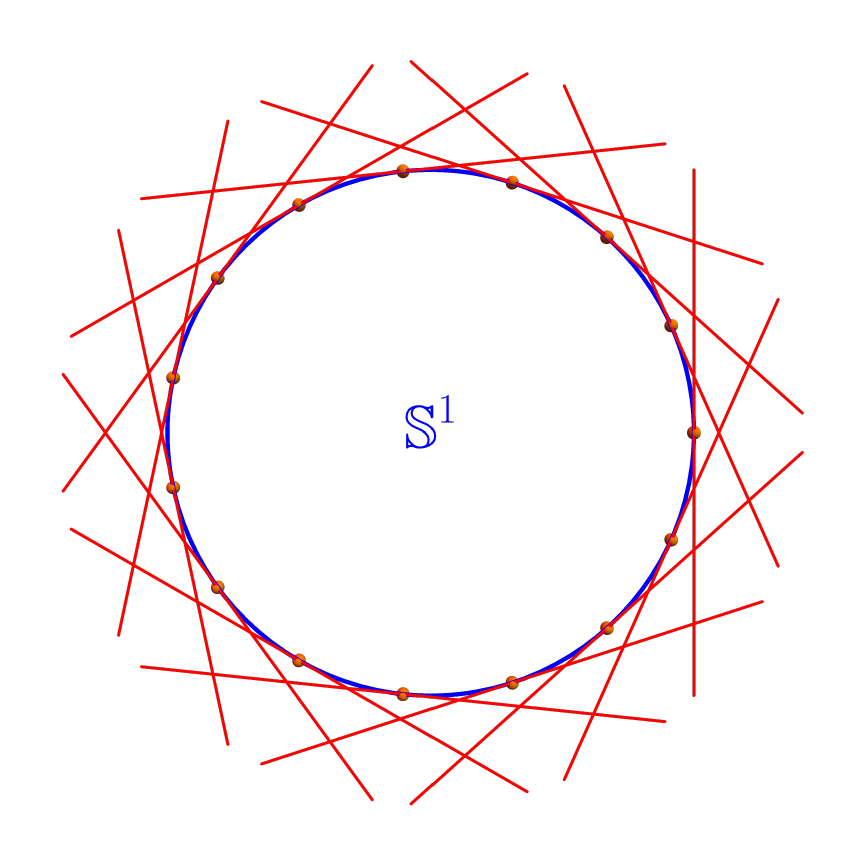}
    \caption{The tangent space to the circle $\mathbb{S}^1$ consists of the collection of all lines, which are tangent to $\mathbb{S}^1$. Here we only show $15$ such lines.}
    \label{fig:TpS1}
\end{figure}

To resolve this issue, we switch to a more accurate representation, as shown in Figure~\ref{fig:TpS1Fibered}. Here, we still see how $\mathbb{S}^1$ is covered by its tangent spaces, but the spaces are now aligned like fibers, which do not intersect. The collection of all these tangent spaces forms the tangent bundle $T\M$. To correctly reflect the idea that the tangent spaces are independent and do not overlap, we take the \textit{disjoint} union of all tangent spaces. This is expressed as
\begin{align}
    T\M &\ce \bigsqcup_{p \in \M} T_p \M = \bigcup_{p \in \M} \{ p \} \times T_p \M\,,
\end{align}
where $\bigsqcup$ denotes the disjoint union. Alternatively, this can be written as the union of spaces of the form $\{ p \} \times T_p \M$, which represents the union of each point of $\M$ with its corresponding tangent space. The tangent bundle $T\M$ is our first and most significant example of a fiber bundle. We refer to the tangent space $T_p \M$ associated with each point $p \in \M$ as the fiber of $T\M$ over $p$. In Figure~\ref{fig:TpS1Fibered}, this corresponds to a line.
\begin{figure}[ht]
    \centering
    \includegraphics[width=0.5\linewidth]{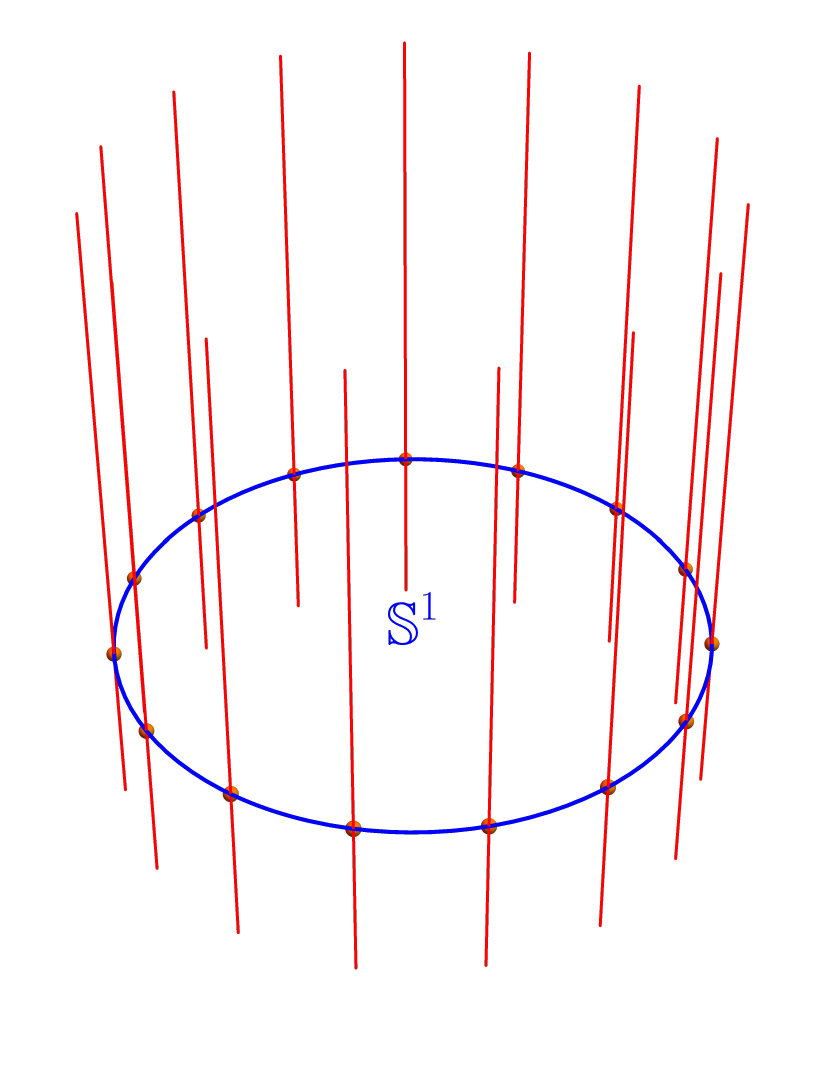}
    \caption{In this visual representation of the tangent bundle $T\mathbb{S}^1$ the tangent spaces are aligned like fibers, which no longer intersect.}
    \label{fig:TpS1Fibered}
\end{figure}

Given the tangent bundle $T\M$, there is a natural projection map $\pi : T\M \to \M$, which acts on a generic element $\{ y \} \times T_y \M$ of $T\M$ as follows:
\begin{align}
    \{ p \} \times T_p \M \quad &\mapsto \quad \pi(\{ p \} \times T_p \M) = p\,.
\end{align}
In other words, $\pi$ projects each fiber onto its base point $p$. In Figure~\ref{fig:TpS1Fibered}, this corresponds to mapping each red line to a single point on the blue circle. Thus, $\pi$ is a surjective map, but not injective. In practice, this means that the inverse map $\pi^{-1}$ takes each point $p$ to the entire tangent space $T_p \M$, not just a single point.

Additionally, we can construct another map that allows us to view the vector field $V$ as a map from $\M$ to $T\M$. This map, denoted by $s : \M \to T\M$, is called a \textbf{section}, and it represents the idea that a vector field assigns a tangent vector to each point of $\M$. To ensure this works properly, we require that $\pi \circ s = \textsf{id}_{\M}$, where $\textsf{id}_{\M}$ is the identity map on $\M$. In simpler terms, if we first map a point $p$ from $\M$ to $T\M$ and then project it back to $\M$, we should end up at the original point $p$. This condition holds only if $s$ maps each point $p$ to a vector in the tangent space $T_p \M$ associated with that point, not to a vector in some other tangent space $T_{p'} \M$ with $p' \neq p$. Thus, a section $s$ is required to map each point $p$ to a tangent vector in $T_p \M$. This is what we mean when we say that ``a vector field assigns a tangent vector to each point of $\M$''. Finally, this leads us to the formal definition of a fiber bundle.\medskip

\begin{definition}[Fiber bundle]
    A \textbf{fiber bundle} of dimension $m$ is a quadruple $(B, \E, \pi, F)$ consisting of a \textbf{base space} $B$, a \textbf{total space} $\E$, a surjective \textbf{projection} $\pi:\E\to B$, and a \textbf{fiber} $F$. The inverse of $\pi$ is required to map every point $p$ of $B$ to an $m$-dimensional real vector space, i.e.,
    \begin{align*}
        F \ce \pi^{-1}(p) \simeq \mathbb{R}^m\,.
    \end{align*}
\end{definition}\medskip

This definition captures the essence of a fiber bundle: a space $\E$ that locally looks like the product of $B$ and $F$, but may have a more intricate global structure. A particularly important class of fiber bundles arises when the fibers themselves are vector spaces, leading to the notion of \textbf{vector bundles}. 

Given a vector bundle $(B, \E, \pi, F)$, we can construct new bundles known as \textbf{jet bundles}, which encode not only the field values but also their derivatives. The \textbf{first-order jet bundle} $(B, J_1\E, \pi, \pi^1_0)$ retains $B$ as its base space, but its total space is now enlarged to $J_1\E$. The coordinates on this space include the vector field components $v^A$ and their first-order derivatives,
\begin{align}
    p^A_\mu \ce \PD{v^A}{x^\mu}\,.
\end{align}
Additionally, there are now two surjective projection maps: the original $\pi: \E \to B$ and a new projection $\pi^1_0: J_1\E \to B$.

More generally, we define the \textbf{$q$-th order jet bundle} over the vector bundle $(B, \E, \pi, F)$. The fibers of this bundle are coordinatized by the vector field $v^A$, its first-order derivatives $\partial_\mu v^A$, its second-order derivatives $\partial_\mu\partial_\nu v^A$, and so on up to the $q$-th order derivatives. The structure of the jet bundle is further organized by a sequence of surjective projection maps:
\begin{align}
    \pi^{q+r}_q: J_{q+r}\E \to J_q \E\,,
\end{align}
where we naturally identify $J_0\E$ with $\E$ and $\pi^0_0$ with $\pi$. For simplicity, we often refer to $J_q\E$ as the $q$-th order jet bundle, rather than writing out the full tuple $(B, J_q\E, \pi, \dots)$.

To better understand the role of jet bundles, let us now examine a concrete example: a second-order jet bundle.
\begin{example}[Second-order jet bundle]\label{ex:J2E}
    As a base space, we take $B = \mathbb{R}^2$ with Cartesian coordinates $(x, y)$. We consider a single scalar field $\Phi: B \to \mathbb{R}$, which also serves as a coordinate for the fibers of the total space $\E$. The fibers of the second-order jet bundle $J_2\E$ are then coordinatized by
    \begin{align*}
        (\Phi, \partial_x\Phi, \partial_y\Phi, \partial_x\partial_x\Phi, \partial_x\partial_y \Phi, \partial_y\partial_y\Phi)\,.
    \end{align*}
    Thus, $J_2\E$ is a six-dimensional vector space, naturally spanned by the elements listed above. A typical vector in this space can be expressed as a linear combination:
    \begin{align*}
        a_1\, \Phi + a_2\, \partial_x \Phi + a_3\, \partial_y\Phi + a_4 \, \partial_x\partial_x \Phi + a_5\, \partial_x \partial_y \Phi + a_6\, \partial_y\partial_y \Phi\,,
    \end{align*}
    where the coefficients $a_i \in \mathbb{R}$ are constants.

    Notice that by appropriately choosing these constants, we can reconstruct the second-order Taylor expansion of $\Phi$ around any chosen point. This highlights the key role of jet bundles in encoding the local behavior of fields and their derivatives. In particular, this formulation directly connects to the example of the relativistic wave equation discussed in Subsection~\ref{ssec:EinsteinsMethod}.
\end{example}\medskip

The above example illustrates how jet bundles provide a framework for discussing \textit{formal} Taylor expansions of functions and, more generally, tensor fields. Additionally, determining the dimension of the fiber $F$ of a $q$-th order jet bundle $J_q\E$ over an $n$-dimensional base space $B$ is a straightforward combinatorial exercise. Specifically, the number of derivatives of order at most $q$ for a single function of $n$ variables is given by
\begin{align}
    \begin{pmatrix}
        n+q\\
        q
    \end{pmatrix}.
\end{align}
Since we typically deal with $m$-component tensor fields, we multiply this count by $m$ to obtain the \textbf{fiber dimension}:
\begin{align}\label{eq:FiberDimension}
    \dim_F J_q\E = m\binom{n+q}{q}\,.
\end{align}
The subscript $F$ serves as a reminder that this is the dimension of the \textit{fiber} of $J_q\E$, \textit{not} the total space $J_q\E$ itself. 

As a quick verification, inserting $n = 2$, $q = 2$, and $m = 1$ into the formula confirms the result for Example~\ref{ex:J2E}, yielding $\dim_F J_2\E = 6$, as expected. Thus, the fiber dimension effectively counts the number of terms in the formal Taylor expansion of our vector field~$v^A$.

Beyond their role in Taylor expansions, fiber bundles provide a natural framework for defining partial differential equations and for understanding how to differentiate such equations in a systematic way.

To fully harness the power of jet bundles, we must introduce additional notation and definitions. We begin with the concept of \textit{multi-indices}, which significantly streamline our notation.\medskip

\begin{definition}[Multi-index $\bfm$]\label{def:multiindex}
    A \textbf{multi-index} is an $n$-tuple $\bfm\ce [m_1, \dots, m_n]$ of non-negative integers $m_i\in\bbN_0$. The \textbf{length} of a multi-index is defined as 
    \begin{align*}
        |\bfm|\ce m_1 + \cdots + m_n\,.    
    \end{align*}
    Multi-indices can be added or subtracted component-wise, provided all resulting entries remain non-negative. We also introduce the shorthand notation
    \begin{align*}
        \bfm \pm a_i \ce [m_1, \dots, m_i \pm a, \dots, m_n]
    \end{align*}
    for any $a \in \bbN_0$ such that $m_i \pm a \geq 0$ (where the index $i$ of $a_i$ refers to the specific position at which $a$ is added to the multi-index).
\end{definition}

We now introduce the field variables and jet variables, which serve as the coordinates of the $q$-th order jet bundles $J_q\E$.\medskip

\begin{definition}[Field variables $v^A$ and jet variables $p^A_{\bfm}$]\label{def:JetVariables}
    The fibers of the $q$-th order jet bundle $J_q\E$ are coordinatized by the \textbf{field variables} $v^A$ and their derivatives with respect to the base space variables $(x^1, \dots, x^n) \in B$ up to order $q$. These derivatives, known as \textbf{jet variables}, are denoted as
    \begin{equation*}
        p^A_{\bfm} \ce \frac{\pd^{|\bfm|}v^A}{\pd(x^1)^{m_1} \cdots \pd(x^n)^{m_n}}, \quad \text{for} \quad 0 < |\bfm| \leq q.
    \end{equation*}
\end{definition}

To illustrate the newly introduced concept of jet variables, we turn to an example from electromagnetism.

\begin{example}[Jet variables $p^A_{\bfm}$ in electromagnetism]
    Throughout this paper, both classical electromagnetism and Proca's modified version will serve to illustrate different concepts.

    In both cases, the theory is described by a vector field $A^\mu$ on Minkowski space. In the language of fiber bundles, $A^\mu$ is a section of the tangent bundle. More importantly for our purposes, the equations governing $A^\mu$ are of second order, meaning we require the second order jet bundle $J_2\E$ to formulate the theory. Here, the field variables are identified as
    \begin{align*}
        v^A \equiv A^\mu\,,
    \end{align*}
    where the index $A$ (ranging over the components of the tensor field) corresponds to the spacetime index $\mu$. In a four-dimensional spacetime, both indices take values from $1$ to $4$. 

    To determine the jet variables $p^A_{\bfm}$, we proceed as follows: The index $A$ is still identified with $\mu$, and the length of the multi-index $\bfm$ satisfies $0<|\bfm| \leq 2$, as we consider the second order jet bundle ($q=2$). This allows us to distinguish between first and second order jet variables.

    The first order jet variables are simply
    \begin{align*}
        p^\mu_1 &\equiv \PD{A^\mu}{x^1}\,, & p^\mu_2 &\equiv \PD{A^\mu}{x^2}\,, & p^\mu_3 &\equiv \PD{A^\mu}{x^3}\,, & p^\mu_4 &\equiv \PD{A^\mu}{x^4}\,.
    \end{align*}
    The second order jet variables yield ten independent expressions:
    \begin{align*}
        p^\mu_{11} &\equiv \frac{\partial^2 A^\mu}{\partial x^1 \partial x^1}\,, & p^\mu_{12} &\equiv \frac{\partial^2 A^\mu}{\partial x^1 \partial x^2}\,, & p^\mu_{13} &\equiv \frac{\partial^2 A^\mu}{\partial x^1 \partial x^3}\,, & p^\mu_{14} &\equiv \frac{\partial^2 A^\mu}{\partial x^1 \partial x^4}\,,\\
        p^\mu_{22} &\equiv \frac{\partial^2 A^\mu}{\partial x^2 \partial x^2}\,, & p^\mu_{23} &\equiv \frac{\partial^2 A^\mu}{\partial x^2 \partial x^3}\,, & p^\mu_{24} &\equiv \frac{\partial^2 A^\mu}{\partial x^2 \partial x^4}\,,\\
        p^\mu_{33} &\equiv \frac{\partial^2 A^\mu}{\partial x^3 \partial x^3}\,, & p^\mu_{34} &\equiv \frac{\partial^2 A^\mu}{\partial x^3 \partial x^4}\,,\\
        p^\mu_{44} &\equiv \frac{\partial^2 A^\mu}{\partial x^4 \partial x^4}\,.
    \end{align*}
    Together, the field variables, first order jet variables, and second order jet variables define the coordinates of the jet bundle $J_2\E$. Since we left the index $\mu$ unspecified, there is a total of $4\times 4 = 16$ first order jet variables and $4\times 10 = 40$ second order jet variables. By adding the four zeroth order jet variables (i.e., the field components $A^\mu$), we get a total of $4 + 16 + 40 = 60$ variables which coordinatize the jet bundle $J_2\E$. This is the same number one obtains from the fiber dimension~\eqref{eq:FiberDimension}: $\dim_F J_2\E = 4 \binom{4+2}{2} = 60$.
\end{example}

As suggested in the example above, jet bundles provide a natural framework for describing field equations. The next subsection explores how they achieve this.

\subsection{Jet Bundle Description of PDEs}\label{ssec:JetBundlePDEs}
The jet bundle approach to partial differential equations (PDEs) introduces a fundamental shift in perspective. Traditionally, PDEs are viewed as equations that require integration to obtain solutions. In general, solving PDEs is a highly nontrivial task, and explicit solutions are often available only in special cases. For instance, the general solution of Einstein's field equations remains unknown, with only a few exact solutions---typically possessing some degree of symmetry---being derivable through direct integration.

The jet bundle perspective, however, offers a radically different viewpoint. Instead of treating a PDE as an equation that must be integrated, we can reinterpret it as a system of non-linear algebraic equations that impose constraints on the fiber coordinates of the jet bundle. This viewpoint naturally leads to the fundamental question: How many fiber coordinates can be freely specified, and how many are constrained by the PDE? 

This question is central to determining the degrees of freedom described by a PDE, aligning closely with Einstein's strategy for analyzing physical theories.

To formalize these ideas, we treat the total space $J_q\E$ of the $q$-th order jet bundle as a manifold, while PDEs define a submanifold $\R_q \subset J_q\E$. Moreover, since the laws of nature should not depend on the choice of coordinate system, we exclude PDEs that impose constraints on the coordinate system itself. This motivates the following definition.

\begin{definition}[Fibered submanifold of $J_q\E$]\label{def:FiberedSubmanifold}
    A \textbf{fibered submanifold} $\R_q$ of the total space $J_q\E$ of a $q$-th order jet bundle is a subspace of $J_q\E$ which can be projected onto the base space $B$.
\end{definition}

We can visualize $\R_q$ as a hypersurface embedded into a larger space, which represents $J_q\E$. The condition that $\R_q$ can be projected back to $B$ is harder to visualize. However, we will shortly see an example where this projection-condition is violated because $\R_q$ imposes restrictions on the coordinates of the base space $B$. Before that, however, we introduce the jet-bundle-theoretic definition of a PDE.\medskip

\begin{definition}[Differential equation $\R_q$ of order $q$]\label{def:Rq}
    A (non-linear) system of partial differential equations (PDEs) of order $q$ is a fibered submanifold $\R_q$ of the jet bundle $J_q\E$. Locally, the differential equation can be represented by the map
    \begin{align*}
        \bbE: J_q\E &\to \E'\\
        (x^\mu, v^A, p^A_{\bfm}) &\mapsto \bbE^\tau(x^\mu, v^A, p^A_{\bfm})\,,
    \end{align*}
    where $\E'$ is another bundle over $B$, such that $\R_q$ is the kernel of this map. In shorthand notation, we write:
    \begin{align*}
       \R_q : \left\{\bbE^\tau(x^\mu, v^A, p^A_{\bfm}) = 0\right..
   \end{align*}
    This expression represents the local form of the $q$-th order PDE.
\end{definition}\medskip

Throughout this work, we will often refer to $\R_q$ as a differential equation or PDE, rather than explicitly using $\bbE^\tau$. However, it is important to keep in mind that $\bbE^\tau$ represents a system of PDEs, organized as a vector. This is indicated by the index $\tau$, which runs from $1$ to $\ell$, where $\ell$ denotes the number of algebraically independent equations.

While this definition may appear more abstract than the conventional approach in physics, we will clarify its meaning and explore its implications through three examples of increasing complexity.

\begin{example}[The wave equation from the jet bundle perspective]\label{ex:WaveEqJetBundle}
    For simplicity, we work in $(1+1)$-dimensional Minkowski space $\M$, which serves as the base space, i.e., $B = \M$. We use the standard coordinates $(t, x)$ and consider a scalar field $\Phi:B\to \bbR$ satisfying the wave equation:
    \begin{align*}
        \Phi_{tt} - \Phi_{xx} = 0\,.
    \end{align*}
    Since this is a second-order PDE ($q=2$), the total space $J_2\E$ of the jet bundle is coordinatized by
    \begin{align*}
        (\Phi, \Phi_t, \Phi_x, \Phi_{tt}, \Phi_{tx}, \Phi_{xx})\,,
    \end{align*}
    as established in Example~\ref{ex:J2E}. We now define the map
    \begin{align*}
        \bbE:J_2\E &\to \E'\\
        (t, x, \Phi, p^A_{\bfm}) &\mapsto \bbE^\tau(t, x, \Phi, p^A_{\bfm}) = \Phi_{tt} - \Phi_{xx}\,,
    \end{align*}
    where the jet variables are $p^A_{\bfm} = (\Phi_t, \Phi_x, \Phi_{tt}, \Phi_{tx}, \Phi_{xx})$, and the index $\tau=1$ since there is only one equation. 

    By definition, $\R_2$ is the submanifold of $J_2\E$ where $\bbE$ vanishes, i.e., the hypersurface satisfying $\Phi_{xx} = \Phi_{tt}$. This allows us to coordinatize $\R_2$ by
    \begin{align*}
        (\Phi, \Phi_t, \Phi_x, \Phi_{tt}, \Phi_{tx})\,.
    \end{align*}
    Thus, $\R_2$ has one dimension less than $J_2\E$, yet it imposes no constraints on the base space coordinates $(t, x)$. This confirms that $\R_2$ is a fibered submanifold of $J_2\E$.
\end{example}\medskip

In our next example, we illustrate how a submanifold $\R_q$ can fail to be fibered.
\begin{example}[Non-fibered submanifold]\label{ex:NonFiberedSubmanifold}
    As the base space, we take the two-dimensional Euclidean plane $B=\bbR^2$, equipped with standard coordinates $(x,y)$. We consider a scalar field $\Phi:B\to\bbR$ satisfying the system of PDEs:
    \begin{align*}
        \left(\Phi_x-1\right)\Phi_{xx} &= 0,\\
        \Phi_{xy} - \Phi_{yy} &= 0\,.
    \end{align*}
    Since this is a second-order system ($q=2$), the total space $J_2\E$ is coordinatized by six variables, as in the previous examples. The system of PDEs imposes algebraic constraints on these variables, which we express via the map
    \begin{align*}
        \bbE:J_2\E &\to \E',\\
        (x,y, \Phi, p^A_{\bfm}) &\mapsto \bbE^\tau(x,y,\Phi, p^A_{\bfm}) = \begin{pmatrix}
            \left(p_x-1\right)p_{xx} \\
            p_{xy} - p_{yy}
        \end{pmatrix},
    \end{align*}
    where we have introduced the shorthand notation $p_{x} = \Phi_x$, $p_{xx} = \Phi_{xx}$, $p_{xy} = \Phi_{xy}$, and $p_{yy} = \Phi_{yy}$, following Definition~\ref{def:JetVariables}. 

    Unlike Example~\ref{ex:WaveEqJetBundle}, where $\tau=1$, here we have $\tau=1,2$ since the system consists of two equations for a single field. While $\R_2$ is still a subspace of $J_2\E$, it is \textit{not} a fibered submanifold because its dimension is not constant. 

    To see this, we analyze the equation $\left(p_x-1\right)p_{xx} = 0$. If $p_x - 1 \neq 0$, then we must have $p_{xx} = 0$, while the second equation implies $p_{xy} = p_{yy}$. In this case, $\R_2$ forms a four-dimensional hypersurface within the six-dimensional space $J_2\E$. However, if $p_x = 1$, the first equation becomes trivial, imposing no constraint on $p_{xx}$, meaning $\R_2$ can have a different dimension in certain regions of the Euclidean plane. 

    The key issue is that the dimension of $\R_2$ depends on the value of $(x,y)$. For regions where $p_x - 1 \neq 0$, $\R_2$ has dimension four, but this condition restricts the base space coordinates $(x,y)$. Consequently, $\R_2$ fails to be a fibered submanifold of $J_2\E$.
\end{example}

As a final example, we consider Einstein's field equations in General Relativity (GR). Traditionally, these equations are expressed as a system of ten partial differential equations (PDEs) governing the ten components of the metric tensor. In this formulation, they are often written in matrix form. However, in the jet bundle framework, Einstein's equations are reformulated as ten nonlinear algebraic equations for the second-order jet variables of the metric tensor. In this perspective, the equations are naturally organized as a vector.

\begin{example}[Einstein's Field Equations in the Jet Bundle Perspective]\label{ex:FieldEqOfGR}
    In the standard formulation of General Relativity, Einstein's field equations are a system of ten coupled second-order PDEs for the ten independent components of the metric tensor $g_{\mu\nu}$. Explicitly, they are given by
    \begin{align*}
        R_{\mu\nu} - \frac{1}{2} R g_{\mu\nu} + \Lambda g_{\mu\nu} = 8\pi G\, T_{\mu\nu}\,,
    \end{align*}
    where $R_{\mu\nu}$ is the Ricci curvature tensor, $R \ce g^{\mu\nu}R_{\mu\nu}$ is the Ricci scalar, $\Lambda$ is the cosmological constant, and $T_{\mu\nu}$ is the energy-momentum tensor. 

    In the jet bundle formulation, we reinterpret these equations as algebraic constraints on the coordinates of the second-order jet bundle $J_2\E$. The total space $J_2\E$ is coordinatized by the independent components of the metric tensor, their first derivatives (which form the Christoffel symbols), and their second derivatives (which determine the Riemann curvature tensor):
    \begin{align*}
        (g_{\mu\nu}, p^A_{\bfm}) \equiv (g_{\mu\nu}, \partial_\lambda g_{\mu\nu}, \partial_\sigma\partial_\lambda g_{\mu\nu})\,.
    \end{align*}
    The field equations are more conveniently written as
    \begin{align*}
        \E_{\mu\nu} \ce R_{\mu\nu} - \frac12 R g_{\mu\nu} + \Lambda g_{\mu\nu} - 8\pi G\, T_{\mu\nu} = 0\,.
    \end{align*}
    In the traditional tensor-language, $\mathcal{E}_{\mu\nu}$ would be a symmetric tensor, which in a given chart can be written as matrix of the form
   \begin{align*}
        \begin{pmatrix}
            \E_{11} & \E_{12} & \E_{13} & \E_{14} \\
            \E_{12} & \E_{22} & \E_{23} & \E_{24} \\
            \E_{13} & \E_{23} & \E_{33} & \E_{34} \\
            \E_{14} & \E_{24} & \E_{34} & \E_{44}
        \end{pmatrix}
    \end{align*}
    In the jet bundle approach, however, the field equations define a fibered submanifold $\R_2 \subset J_2\E$, given locally by the kernel of the map
    \begin{align*}
        \bbE: J_2\E &\to \E' \\
        (x^\mu, g_{\mu\nu}, p^\alpha_{\bfm}) &\mapsto \bbE^\tau(x^\mu, g_{\mu\nu}, p^\alpha_{\bfm}) = \begin{pmatrix}
            \E_{11} \\
            \E_{12} \\
            \E_{13} \\
            \E_{14} \\
            \E_{22} \\
            \E_{23} \\
            \E_{24} \\
            \E_{33} \\
            \E_{34} \\
            \E_{44}
        \end{pmatrix}\,.
    \end{align*}
    Here, the index $\tau$ runs from $1$ to $10$, corresponding to the ten algebraically independent Einstein equations. As anticipated, we see that the Einstein field equations are now organized in a vector. The order in which these equations appear in the vector is irrelevant. The only important thing is that all algebraically independent equations are represented in the vector. 

    Lastly, we point out that the jet bundle $J_2 \E$ is coordinatized by the metric $g_{\mu\nu}$, its first derivatives $\partial_\lambda g_{\mu\nu}$, and its second order derivatives $\partial_\sigma\partial_\lambda g_{\mu\nu}$. Taking into account the symmetry of the metric and the fact that partial derivatives commute, this gives us a total of $10 + 4\times 10 + 10\times 10 = 150$ jet variables. Again, this number could also have been obtained using the fiber dimension~\eqref{eq:FiberDimension}: $\dim_F J_2\E = 10 \binom{4+2}{2} = 150$.
\end{example}\medskip

Now that we have established how to describe differential equations in the jet bundle framework, we turn our attention to special operations that can be performed on these equations. In particular, we focus on \textit{prolongations} and \textit{projections}, which play a crucial role in extracting meaningful information from PDEs.

\subsection{Prolongations and Projections}\label{ssec:ProlongationsProjections}
When discussing Einstein's method for determining the number of degrees of freedom in Subsection~\ref{ssec:EinsteinsMethod}, we analyzed derivatives of the wave equation. This approach allowed us to extract information about higher-order Taylor coefficients. In the jet bundle framework, the idea of differentiating equations can be naturally translated, leading to the following definition.\medskip
\begin{definition}[Formal derivative operator $D_\mu$]\label{def:Di}
    A local representation of a $q$-th order differential equation can be formally differentiated. The \textbf{formal derivative operator} acting on $\bbE^\tau$ is defined as
    \begin{align*}
        D_\mu \bbE^\tau \ce \PD{\bbE^\tau}{x^\mu} + \sum_{A = 1}^\ell \PD{\bbE^\tau}{v^A}p^A_\mu + \sum_{A = 1}^\ell \sum_{0<|\bfm|\leq q} \PD{\bbE^\tau}{p^A_{\bfm}}p^A_{\bfm + 1_\mu}
    \end{align*}
    with $p^A_\mu\ce \PD{v^A}{x^\mu}$ and for all $\mu\in\{1,\dots, n\}$ and $\tau\in\{1, \dots, \ell\}$.
\end{definition}\medskip
We recall that the notation $\bfm + 1_\mu$ indicates that the $\mu$-th entry of the multi-index $\bfm$ is increased by one (see Definition~\ref{def:multiindex}). 

At first glance, the formula for the formal derivative may seem intricate, but in practice, it is straightforward to work with and produces the expected results. Essentially, it corresponds to taking a derivative with respect to the coordinate $x^\mu$, but reformulated in the language of jet bundles. To illustrate this, we compute the formal derivative of the Euler-Lagrange equation as an example.\medskip

\begin{example}[Formal derivative of the Euler-Lagrange equation]
    For simplicity, we consider the one-dimensional Euler-Lagrange equation for a function $x(t)$, where $t$ is the time coordinate:
    \begin{align*}
        \underbrace{\PD{L}{x} - \frac{\dd}{\dd t}\PD{L}{\dot{x}}}_{= \bbE} = 0\,.
    \end{align*}
    We assume the standard case where the Lagrangian $L$ depends only on $x$ and $\dot{x}$, without higher-order time derivatives. This results in a second-order differential equation for $x$, and the jet bundle $J_2\E$ is coordinatized by $(x, \dot{x}, \ddot{x})$. 

    Applying the definition of the formal derivative, we compute:
    \begin{align*}
        D_t\bbE &= \PD{\bbE}{t} + \sum_{A=1}^{\ell}\PD{\bbE}{x}p^{A}_\mu + \sum_{A=1}^\ell \sum_{0<|\bfm|\leq 2}\PD{\bbE}{p^{A}_{\bfm}}p^{A}_{\bfm + 1_\mu} \\
        &= \underbrace{\PD{\bbE}{t}}_{=0} + \PD{\bbE}{x}\dot{x} + \PD{\bbE}{\dot{x}} \ddot{x} + \PD{\bbE}{\ddot{x}}\dddot{x} \\
        &= \PD{}{x}\left(\PD{L}{x} - \frac{\dd}{\dd t}\PD{L}{\dot{x}}\right)\dot{x} + \PD{}{\dot{x}}\left(\PD{L}{x} - \frac{\dd}{\dd t}\PD{L}{\dot{x}}\right)\ddot{x} - \PD{}{\ddot{x}}\left(\frac{\dd}{\dd t}\PD{L}{\dot{x}}\right)\dddot{x}\,.
    \end{align*}
    Here, we used that $L$ has no explicit time dependence, i.e., $\PD{L}{t} = 0$, and that $\PD{L}{\ddot{x}} = 0$. Moreover, since we have a single-component PDE ($\ell=1$), the sum over $A$ collapses. However, the sum over the multi-index $\bfm$ contributes the third and fourth terms in the second line. 

    Notably, $D_t\bbE$ is a third-order equation, exactly matching the result obtained by taking the total time derivative of the Euler-Lagrange equation.
\end{example}

Equipped with the formal derivative operator as a new tool, we now introduce \textbf{prolonged equations}. These consist of the original PDE along with its formal derivatives. Naturally, this means that prolonged equations are of higher order than the original PDE.

\begin{definition}[Prolongation]\label{def:Prolongation}
    The \textbf{prolongation} of a differential equation $\R_q$ is a differential equation $\R_{q+1}\subset J_{q+1}\E$, where the order has been increased by one. In a local representation we can write the prolongation of $\R_q$ as
    \begin{align*}
        \R_{q+1}:
        \begin{cases}
            \bbE^\tau = 0\\
            D_\mu \bbE^\tau = 0
        \end{cases}
    \end{align*}
    for all $\mu\in\{1,\dots, n\}$ and for all $\tau\in\{1, \dots, \ell\}$. A differential equation can be prolonged several times, giving rise to an equation $\R_{q+r}$ which consists of the original equation and its first $r$ formal derivatives,
    \begin{align*}
        \R_{q+r}:
        \begin{cases}
            \bbE^\tau = 0\\
            D_{\mu_1}\bbE^\tau = 0\\
            \phantom{D_{\mu_1}}\vdots \\
            \underbrace{D_{\mu_r} D_{\mu_{r-1}}\cdots D_{\mu_1}}_{r-\text{times}} \bbE^\tau = 0
        \end{cases}\,.
    \end{align*}
    We call the integer number $r>0$ the \textbf{prolongation order}.
\end{definition}\medskip
The concept of prolongation provides a systematic way to generate higher-order equations by iteratively applying the formal derivative operator. To see this in action, we consider an important example from physics: the prolongation of Maxwell’s equations.

\begin{example}[Prolongation of Maxwell's equations]\label{ex:ProlongationMaxwell}
    We consider the base space $B=\bbR^n$, the total space $\E = TB$ (tangent bundle to $B$) and the vector field $A^\mu$, which is defined as section of $\E$. This vector field is subjected to Maxwell's equations
    \begin{align*}
        \R_2 : 
      \underbrace{  \begin{cases}
            \partial_\mu \left(\partial^\mu A^\nu - \partial^\nu A^\mu\right) 
        \end{cases}}_{= \bbE^\nu}=0\,,
    \end{align*}
    where the indices have been raised with the Minkowski metric. The formal derivative of Maxwell's equations is simply
    \begin{align*}
        D_\rho \bbE^\nu = \partial_\rho \partial_\mu \left(\partial^\mu A^\nu - \partial^\nu A^\mu \right) = 0\,.
    \end{align*}
    The prolongation thus reads
    \begin{align*}
        \R_3 : 
        \begin{cases}
                \partial_\mu \left(\partial^\mu A^\nu - \partial^\nu A^\mu\right) = 0 \\
                &\\
                \partial_\rho \partial_\mu \left(\partial^\mu A^\nu - \partial^\nu A^\mu \right) = 0
        \end{cases}\,.
    \end{align*}
\end{example}\medskip

In the previous subsection, we introduced the concept of a fibered submanifold, a fundamental component in our definition of a PDE. One might naturally expect that prolonging a PDE that forms a fibered submanifold would result in another fibered submanifold. However, this is not always the case. The following example demonstrates how the prolongation of a fibered submanifold can lead to a non-fibered submanifold.

\begin{example}[Prolongation to a non-fibered submanifold]\label{ex:NonFiberedProlongation}
We consider a two-dimensional manifold $\M$ with coordinates ${x, y}$ and a second-order PDE system for a scalar field $\Phi$ given by
\begin{align*}
    \R_2:
    \begin{cases}
        \Phi_{xx} - \frac{1}{2} (\Phi_{yy})^2 &= 0 \\
        \Phi_{yy} - \Phi_{xy} &= 0
    \end{cases}\,.
    \end{align*}
    To obtain the prolonged system, we compute the $x$- and $y$-derivatives of both equations and include them in the original system. This yields the third-order system
    \begin{align*}
        \R_3:
    \begin{cases}
        \Phi_{xx} - \frac{1}{2} (\Phi_{yy})^2 &= 0 \\
        \Phi_{yy} - \Phi_{xy} &= 0 \\
        \\
        \Phi_{xxx} - \Phi_{xyy} \Phi_{yy} &= 0 \\
        \Phi_{xyy} - \Phi_{xxy} &= 0 \\
        \Phi_{xxy} - \Phi_{yyy} \Phi_{yy} &= 0 \\
        \Phi_{yyy} - \Phi_{xyy} &= 0
    \end{cases}\,.
\end{align*}
A fundamental issue arises when we add the fourth, fifth, and sixth equations (counting from the top). After basic algebraic manipulations, we obtain
\begin{align*}
    (\Phi_{yy} - 1) \Phi_{yyy} = 0.
\end{align*}
This equation has the same pathological structure as the one encountered in Example~\ref{ex:NonFiberedSubmanifold}. Applying the same reasoning, we conclude that $\R_3$ is a non-fibered submanifold of $J_3\E$.
\end{example}

The problem with non-fibered submanifolds is that they require case distinctions. In the example above, we must separately analyze the cases $\Phi_{yy} - 1 = 0$ and $\Phi_{yyy} = 0$, as the PDE behaves differently in each scenario. Such distinctions complicate our goal of systematically determining all constraint equations between the Taylor coefficients of a formal power series solution. To avoid these complications, we restrict our attention to so-called \textit{regular equations}.

\begin{definition}[Regular equation]\label{def:RegularEquation}
A differential equation $\R_q$ of order $q$, as defined in Definition~\ref{def:Rq}, is called \textbf{regular} if all its prolongations $\R_{q+r}$ form fibered submanifolds for all $r \geq 0$.
\end{definition}

In practice, a non-regular equation can often be made regular by imposing suitable restrictions on $\R_q$. Therefore, focusing exclusively on regular equations is a relatively mild assumption that does not significantly limit the generality of our approach.

Next, we introduce the concept of projecting prolonged equations. This technique proves to be particularly useful in uncovering hidden integrability conditions and constraint equations.

\begin{definition}[Projection]\label{def:Projection}
    A prolonged differential equation of order $q+1$ can be projected back to a $q$-th order equation via the surjective map
    \begin{align*}
        \pi^{q+1}_q: J_{q+1}\E \to J_q\E.
    \end{align*}
    The first projection of $\R_{q+1}$ is given by
    \begin{align*}
        \R^{(1)}_q \ce \pi^{q+1}_q(\R_{q+1}) \subset J_q\E.
    \end{align*}
    In practical terms, $\R^{(1)}_q$ is obtained from $\R_{q+1}$ by removing all equations of order $q+1$ while retaining those of order $q$ and lower. This process extends naturally to cases where an equation $\R_q$ has been prolonged $r+s$ times and subsequently projected $s$ times:
    \begin{align*}
        \R^{(s)}_{q+r} \ce \pi^{q+r+s}_{q+r}(\R_{q+r+s}) \subset J_{q+r}\E.
    \end{align*}
    We refer to the integer $s \geq 0$ as the \textbf{projection order}.
\end{definition}

At first glance, this definition may seem abstract. Let’s break it down: Starting with an equation $\R_q$, we know that it can be prolonged to $\R_{q+1}$. While $\R_q$ defines a submanifold in the jet bundle $J_q\E$, the prolonged equation $\R_{q+1}$ defines a submanifold in the higher-order jet bundle $J_{q+1}\E$. The projection process then maps $\R_{q+1}$ back into $J_q\E$, yielding a new equation $\R^{(1)}_q$.

Here, the superscript $(1)$ denotes that we have performed a single projection, while the subscript $q$ indicates that the resulting equation remains of order $q$. Importantly, we write $\R^{(1)}_q$ instead of simply $\R_q$ to emphasize a key fact: in general,
\begin{align*}
\R^{(1)}_q \neq \R_q.
\end{align*}
That is, the projected equation typically defines a \textit{different} submanifold in $J_q\E$ than the original one! This distinction plays a crucial role in analyzing the structure of differential equations and their hidden constraints.

We illustrate this phenomenon with an example. Before doing so, we note that the concept of projection, as described in the definition above, naturally extends to higher orders of prolongation. In practice, projection is straightforward: when mapping from $J_{q+r+s}\E$ to $J_{q+r}\E$, we simply ``forget'' all \textit{independent} equations of order $q+r+1$ or higher. The following example demonstrates this process.

\begin{example}[Projection of the prolonged Maxwell equations]\label{ex:ProjectionMaxwell}
    In Example~\ref{ex:ProlongationMaxwell}, we encountered the prolonged Maxwell equations. For convenience, we recall them here:
    \begin{align*}
    \R_3 :
        \begin{cases}
            \partial_\mu \left(\partial^\mu A^\nu - \partial^\nu A^\mu\right) = 0 \\
            \\
            \partial_\rho \partial_\mu \left(\partial^\mu A^\nu - \partial^\nu A^\mu \right) = 0
        \end{cases}.
    \end{align*}
    These equations define a submanifold in the jet bundle $J_3\E$. Applying the projection $\pi^{3}_2: J_3\E \to J_2\E$, we obtain $\R^{(1)}_2$. Since projection involves discarding third-order equations while retaining those of second order and lower, we find
    \begin{align*}
        \R^{(1)}_2 :
        \begin{cases}
            \partial_\mu \left(\partial^\mu A^\nu - \partial^\nu A^\mu\right) = 0
        \end{cases}.
    \end{align*}
    Clearly, this is just the original Maxwell equation, implying that in this case, $\R^{(1)}_2 = \R_2$.
\end{example}

Projection might initially appear to be the inverse operation of prolongation, and the previous example may have reinforced this impression. However, as we have already indicated, this is by no means always the case! In fact, situations where $\R^{(1)}_q$ \textit{does not} coincide with the original system $\R_q$ are of particular interest. In such cases, the process of prolongation followed by projection can reveal so-called \textit{integrability conditions}\footnote{In the language of physics, these are often referred to as constraint equations.}. As a concrete example, we now examine Proca’s equations.

\begin{example}[$\R^{(1)}_2$ is \textbf{not} equal to $\R_2$ for Proca's equations]\label{ex:R1qNotEqualR1}
    Proca’s equations are structurally similar to Maxwell’s equations, differing only by the presence of a mass term. To set the stage, we assume that $\M$ is an $n$-dimensional manifold equipped with the Minkowski metric $\eta_{\mu\nu} = \text{diag}(-1, +1, \dots, +1)$, which we use to raise and lower indices. In natural units, Proca’s equations take the form
    \begin{align*}
        \R_2:
    \begin{cases}
        \partial_\mu \left(\partial^\nu A^\mu - \partial^\mu A^\nu\right) + m^2 A^\nu = 0
    \end{cases},
    \end{align*}
    where $m > 0$ denotes the photon mass and $A^\mu$ is the vector potential. This system consists of $n$ equations. The prolonged system is obtained immediately:
    \begin{align*}
        \R_3:
    \begin{cases}
        \partial_\mu \left(\partial^\nu A^\mu - \partial^\mu A^\nu\right) + m^2 A^\nu = 0 \\
        \\
        \partial_\rho\partial_\mu \left(\partial^\nu A^\mu - \partial^\mu A^\nu\right) + m^2 \partial_\rho A^\nu = 0
    \end{cases},
    \end{align*}
    where $\rho$ ranges from $1$ to $n$, giving a total of $n + n \times n$ equations that define a submanifold in $J_3\E$. At first glance, projecting $\R_3$ seems straightforward: the top $n$ equations are second order, while the bottom $n^2$ equations are third order. It would thus appear that projection simply returns the original second-order system. However, this conclusion is incorrect. The mistake lies in projecting the system before simplifying it and identifying the \textit{independent} equations.

    To rectify this, we divide the third-order equations into two groups: one where $\rho = \nu$, and another where $\rho \neq \nu$. Summing over the $n$ equations with $\rho = \nu$, we obtain
    \begin{align*}
        \partial_\rho \partial_\mu \left(\partial^\rho A^\mu - \partial^\mu A^\nu\right) + m^2 \partial_\rho A^\rho = 0.
    \end{align*}
    Since $\partial_\rho \partial_\mu$ is symmetric in $\mu$ and $\rho$, whereas $\partial^\rho A^\mu - \partial^\mu A^\nu$ is  antisymmetric, this equation simplifies to
    \begin{align*}
        \partial_\rho A^\rho = 0.
    \end{align*}
    Thus, we have uncovered a first-order equation hidden within the third-order equations! The simplified prolonged system now reads
    \begin{align*}
        \R_3:
        \begin{cases}
            \partial_\mu \left(\partial^\nu A^\mu - \partial^\mu A^\nu\right) + m^2 A^\nu = 0 \\
            \\
            \partial_\rho A^\rho = 0 \\
            \partial_\rho\partial_\mu \left(\partial^\nu A^\mu - \partial^\mu A^\nu\right) + m^2 \partial_\rho A^\nu = 0
        \end{cases}.
\end{align*}

    The projection $\pi^3_2 : J_3\E \to J_2 \E$ instructs us to ``forget'' all third-order equations and retain only those of second order or lower. Consequently, we find
    \begin{align*}
        \R^{(1)}_2:
    \begin{cases}
        \partial_\mu \left(\partial^\nu A^\mu - \partial^\mu A^\nu\right) + m^2 A^\nu = 0 \\
        \\
        \partial_\rho A^\rho = 0
    \end{cases}
    =
    \begin{cases}
    \partial_\mu \partial^\mu A^\nu - m^2 A^\nu = 0 \\
    \\
    \partial_\rho A^\rho = 0
    \end{cases}.
    \end{align*}

    As anticipated, in the case of Proca’s equations, we find that $\R^{(1)}_2 \neq \R_2$. This confirms that prolongation followed by projection can expose hidden constraint equations.
\end{example}

What we have uncovered here is a well-known fact about Proca’s equations. Typically, $\R_2$ and $\R^{(1)}_2$ are considered equivalent, but our analysis shows that $\R^{(1)}_2 \neq \R_2$, meaning they define distinct submanifolds in $J_2\E$.

This discrepancy is not a cause for concern. In Subsection~\ref{ssec:CKAlgorithm}, where we discuss the Cartan-Kuranishi theorem, we will resolve this apparent tension. The key idea is that $\R^{(1)}_2$ and $\R_2$ are equivalent in the sense that they possess the same solution space. Consequently, if our goal is to integrate the equations and find specific solutions, it does not matter which system we use. However, if we follow Einstein’s procedure for determining the degrees of freedom, the distinction between $\R^{(1)}_2$ and $\R_2$ becomes significant.

In this sense, $\R^{(1)}_2$ and $\R_2$ are \textit{not} equivalent: they define different submanifolds that behave differently when we perform Einstein’s analysis. Using $\R_2$ leads to complications, as the hidden constraint equation interferes with our attempt to determine all constraints on the Taylor coefficients. This issue does not arise when working with $\R^{(1)}_2$, where the constraint $\partial_\rho A^\rho = 0$ is made explicit from the outset\footnote{In fact, as we will see in Subsection~\ref{ssec:Examples}, it is $\R^{(2)}_2$, rather than $\R^{(1)}_2$, that fully resolves this issue.}.

To demonstrate that $\R^{(1)}_2$ and $\R_2$ define distinct subspaces, we introduce the concept of the dimension of a submanifold $\R_q$. This notion will also be useful in later discussions. The idea is straightforward: the jet bundle $J_q\E$ is coordinatized by $m\binom{n+q}{q}$ jet variables, where $m$ is the number of field components, $n$ is the spacetime dimension, and $q$ is the order of the jet bundle. Each algebraically independent equation in $\R_q$ eliminates one coordinate in the jet bundle, meaning that one coordinate can be expressed in terms of others. Thus, the dimension of the hypersurface $\R_q$ in $J_q\E$ is given by the difference between the number of jet variables and the number of algebraically independent equations.

\begin{definition}[Dimension of $\R_q$]\label{def:DimRq}
    The dimension of the submanifold determined by $\R_q$ is defined as
    \begin{align*}
        \dim \R_q \coloneqq m \binom{n+q}{q} - \text{\# of algebraically independent equations in $\R_q$} \equiv \dim_F J_q\E - \ell\,,
    \end{align*}
    where $\ell$ is the number of all algebraically independent equations of order $q$ or less.
\end{definition}\medskip

For the original Proca equations, we have $m = n$ vector field components (where $m$ should not be confused with the photon mass!), $q = 2$, and $n$ algebraically independent equations. This yields
\begin{align*}
\dim \R_2 = \frac{1}{2} n^2 (n+3).
\end{align*}
In the case of $\R^{(1)}_2$, we have one additional algebraically independent equation, so
\begin{align*}
\dim \R^{(1)}_2 = \frac{1}{2} n^2 (n+3) - 1,
\end{align*}
which is different from the previous result. This confirms that $\R^{(1)}_2$ and $\R_2$ define genuinely distinct spaces.

Before concluding this subsection, we revisit an important observation from the previous example: prolonging once produced an integrability condition, which we then isolated by projecting back. This is a completely general phenomenon and plays a crucial role in our analysis. As we have noted, attempting to construct a formal power series solution for an equation where integrability conditions have \textit{not} been identified is doomed to fail.

The reason is straightforward: our construction begins with an equation $\R_q$, which allows us to determine some of the $q$-th order coefficients of the power series while leaving others undetermined. We then prolong to $\R_{q+1}$ to obtain constraint equations for the $(q+1)$-th order coefficients. However, integrability conditions arise at lower orders. If they emerge after prolongation, they force us to re-examine the coefficients of order $q$ or lower. As a result, the systematic order-by-order construction of a formal power series solution breaks down. This brings us to the concept of \textit{formal integrability}.
 
\begin{definition}[Formal Integrability of $\R_q$]\label{def:FormalIntegrability}
    A differential equation $\R_q$ is called \textbf{formally integrable} if
    \begin{align*}
        \R^{(1)}_{q+r} = \R_{q+r} \quad \text{for all } r \geq 0,
    \end{align*}
    i.e., if it already contains all its integrability conditions.
\end{definition}\medskip

This definition requires that at every prolongation order $r$, no new integrability conditions arise. At first glance, verifying this condition appears challenging, as it involves checking infinitely many equations of the form $\R^{(1)}_{q+r} = \R_{q+r}$. However, we will soon see that for certain classes of equations, this task is significantly simplified (see Theorem~\ref{thm:ConsequencesOfInvolutiveSymbols} and Corollary~\ref{cor:CriterionOfInvolutivity}).

In Subsection~\ref{ssec:CKAlgorithm}, we will explore how the Cartan-Kuranishi algorithm provides a systematic procedure for transforming any given equation into one that is formally integrable. However, before we can fully appreciate this result, we must first introduce the concept of the \textit{symbol of a PDE}.

\subsection{The Symbol of a PDE}\label{ssec:Symbol}

The key insight from the last subsection was that projecting the prolongation of a PDE $\R_q$ does not necessarily reproduce the original system. In general, we have
\begin{align*}
    \R^{(1)}_q \neq \R_q\,.
\end{align*}
This discrepancy arises due to the presence of equations that either reduce to identities or introduce new integrability conditions. This raises an important question: How can we systematically determine whether identities and/or integrability conditions appear when prolonging a given equation $\R_q$?

Answering this question (and correctly identifying these identities and integrability conditions) is crucial for our approach. Recall that our goal is to construct a formal power series solution to $\R_q$ and determine how many Taylor coefficients remain undetermined by the equations. These undetermined coefficients correspond to our freedom in specifying initial data and are directly related to the number of degrees of freedom. However, the emergence of integrability conditions disrupts the systematic order-by-order construction of a Taylor series.

At order $q$, we can substitute the $q$-th order Taylor expansion into $\R_q$ and determine a subset of its Taylor coefficients from the resulting algebraic equations. If $\R_q$ is formally integrable, we can prolong the system indefinitely without encountering new integrability conditions. Consequently, we can determine the Taylor coefficients at order $q+r$ for any $r \geq 0$ without affecting those obtained at lower orders.

If, however, $\R_q$ is \textit{not} formally integrable, a prolongation will eventually produce an integrability condition at some order $q+r$ for $r > 0$. Since integrability conditions necessarily arise at \textit{lower} orders than $q+r$, we must \textit{recompute} certain lower-order Taylor coefficients. In this sense, integrability conditions act as corrections, forcing us to revise our earlier computations.

Now that we recognize the significance of this question, we introduce a powerful tool that allows us to systematically address it using simple linear algebra methods: the \textit{symbol} $\S_q$ of the equation $\R_q$. To motivate this concept and uncover its intuitive meaning, we begin by examining a special class of PDEs: Quasi-linear first-order partial differential equations. 

\paragraph{Interlude: Quasi-linear First-order PDEs and the Kinetic Matrix}\phantom{.}\newline
\textit{We consider the initial value problem for a class of first-order PDEs that are quasi-linear in a sense we will clarify shortly. The problem is formulated as\footnote{This discussion is deliberately kept brief, omitting many details. For an introduction to the methods discussed here, see, for instance,~\cite{MiersemannBook}. For a pedagogical treatment with applications to physics, see Appendix A.2 and A.3 of~\cite{DAmbrosio:2022}.} }
\begin{align}\label{eq:IVP}
    \begin{cases}
        \displaystyle\sum_{\mu =1}^n M^{(\mu)}(v, x) \partial_{\mu} v^{A} + L(x) v^{A} + V(x) = 0\\
        \\
        \left. v^{A}\right|_{\Sigma} = f^{A}(y^1, \dots, y^{n-1})\,.
    \end{cases}
\end{align}
\textit{According to our conventions, $v^{A}$ represents a collection of fields with a total of $m$ algebraically independent components, } 
\begin{align}  
    v^{A} = (v^1, v^2, \dots, v^m).
\end{align}  
\textit{The first-order derivatives $\partial_\mu v^{A}$ appear multiplied by $n$ matrices $M^{(\mu)}$ of dimension $m\times m$. Each value of  $\mu$ corresponds to a distinct matrix $M^{(\mu)}$. Furthermore, there can be an $m\times m$ matrix multiplying $v^{A}$, and an $m$-dimensional vector $V(x)$. The equation is \emph{quasi-linear} because the matrices $M^{(\mu)}$ are allowed to depend on $v^{A}$.}

\textit{The second line of \eqref{eq:IVP} imposes the initial condition. Here, $\Sigma$ denotes a Cauchy surface of dimension $n-1$. The initial data consists of an arbitrary function $f^{A}(y^1, \dots, y^{n-1})$, specifying the values of $v^{A}$ on $\Sigma$. The coordinates $\{y^1, \dots, y^{n-1}\}$ parametrize $\Sigma$.} 

\textit{Finding a general solution to this system of PDEs is often highly nontrivial, if not impossible. Nevertheless, determining whether a unique solution exists for a given initial dataset $f^{A}$ is remarkably straightforward. First, we note that the Cauchy surface $\Sigma$ can always be defined by an equation $\chi(x) = 0$, where $\chi$ is a scalar function with a non-vanishing gradient, } 
\begin{align}
    \nabla\chi(x) \neq 0 \quad \text{for all} \quad x = (x^1, \dots, x^n)\,.
\end{align}  
\textit{Furthermore, we can always introduce coordinates adapted to this formulation of the initial value problem~\eqref{eq:IVP}. The change of coordinates is defined by a map $\phi: \M \to \M$, whose $n$-th component is given by $\chi(x)$:}
\begin{align}
    \phi^n \ce \chi(x^1, \dots, x^n)\,.
\end{align}
\textit{We further demand that the map is smooth and invertible, which implies that its Jacobian matrix $J$ is well-defined with a nonzero determinant,} 
\begin{align}
    \det J = \det\left(\PD{\phi^\mu}{x^\nu}\right) \neq 0\,.
\end{align}
\textit{This ensures that the transformation is locally invertible. So far, we have achieved the following: The Cauchy surface is represented by the constraint equation $\chi(x) = 0$. The condition $\nabla\chi(x) \neq 0$ guarantees that $\Sigma$ has a well-defined normal vector everywhere, given by}
\begin{align}
    \vec{n} \ce \nabla\chi(x)\,.
\end{align}  
\textit{In the new coordinates $(\phi^1, \dots, \phi^n)$, the Cauchy surface is simply given by $\phi^n = 0$. Thus, the coordinates $(\phi^1, \dots, \phi^{n-1})$ parametrize the surface $\Sigma$, while changes in $\phi^n$ correspond to displacements in the direction normal to $\Sigma$, i.e., along $\vec{n}$. This geometric picture will be crucial in understanding the argument that follows.}

\textit{We now rewrite the system~\eqref{eq:IVP} in these new coordinates. To that end, we use the transformation } 
\begin{align}
    \PD{v^{A}}{x^\mu} = \PD{u^{A}}{\phi^{\lambda}} \PD{\phi^{\lambda}}{x^\mu} \equiv \left(J\cdot\nabla_\phi u^{A}\right)_\mu\,,
\end{align}
\textit{where we define}
\begin{align}
    u^{A}(\phi) \ce v^{A}(x(\phi))\,,
\end{align}
\textit{and where $\nabla_\phi$ denotes the gradient with respect to the coordinates $\phi$. From this relation, we immediately deduce that knowing the initial data $f^{A}$ on $\Sigma$ allows us to determine the $m\times (n-1)$ partial derivatives}
\begin{align}
    \PD{u^{A}}{\phi^{j}}, \quad j \in \{1, \dots, n-1\}\,,
\end{align}
\textit{when evaluated on $\Sigma$. To see this explicitly, we use the definition of the partial derivative:}
\begin{align}
    \left.\PD{u^{A}}{\phi^{j}}\right|_{\Sigma} 
    &\ce \lim_{\epsilon\to 0} \frac{u^{A}(\phi^1, \dots, \phi^{j} + \epsilon, \dots, \phi^{n-1}, 0) - u^{A}(\phi^1, \dots, \phi^{j}, \dots, \phi^{n-1}, 0)}{\epsilon} \notag\\
    &= \lim_{\epsilon\to 0} \frac{f^{A}(\phi^1, \dots, \phi^{j} + \epsilon, \dots, \phi^{n-1}, 0) - f^{A}(\phi^1, \dots, \phi^{j}, \dots, \phi^{n-1}, 0)}{\epsilon} \notag\\
    &= \PD{f^{A}}{\phi^{j}}\,.
\end{align}
\textit{Thus, these derivatives are fully determined by the initial data $f^{A}$, as claimed. However, the $m$ partial derivatives}
\begin{align}
    \PD{u^{A}}{\phi^n} = \PD{u^{A}}{\chi}
\end{align}
\textit{remains undetermined because} 
\begin{align}\label{eq:NotInitialData}
    \left. \PD{u^{A}}{\chi}\right|_{\Sigma} &= \lim_{\epsilon\to 0} \frac{u^{A}(\phi^1, \dots, \phi^{n-1}, \epsilon) - u^{A}(\phi^1, \dots, \phi^{n-1}, 0)}{\epsilon} \notag\\
    &= \lim_{\epsilon\to 0} \frac{u^{A}(\phi^1, \dots, \phi^{n-1}, \epsilon) - f^{A}(\phi^1, \dots, \phi^{n-1}, 0)}{\epsilon}\,.
\end{align}
\textit{This expression is not determined by $f^A$ because $u^{A}(\phi^1, \dots, \phi^{n-1}, \epsilon)$ is evaluated at a point \textit{off} the surface $\Sigma$. Recall that $\Sigma$ is located at $\phi^n = 0$, and any displacement along $\phi^n$ corresponds to moving away from $\Sigma$ in the direction of the normal vector $\vec{n}$. Equation~\eqref{eq:NotInitialData} also reveals an important insight: If we can determine $\left.\PD{u^{A}}{\chi}\right|_{\Sigma}$, then we can extend $u^{A}$ beyond~$\Sigma$. Specifically, if $\left.\PD{u^{A}}{\chi}\right|_{\Sigma}$ is known, then}
\begin{align}\label{eq:FormalIntegration}
    u^{A}(\phi^1, \dots, \phi^{n-1}, \epsilon) = f^{A} + \epsilon\,\left. \PD{u^{A}}{\chi}\right|_\Sigma + \mathcal{O}(\epsilon^2)\,.
\end{align}
\textit{This means that we can formally integrate the equation and determine $u^{A}$ in a neighborhood of~$\Sigma$. At this point, the PDE~\eqref{eq:IVP} becomes crucial. In the adapted coordinate system, it takes the form}
\begin{align}
    \sum_{\mu=1}^n \tilde{M}^{(\mu)} n_\mu \PD{u^{A}}{\chi} 
    + \sum_{\mu=1}^n \tilde{M}^{(\mu)} \PD{\phi^{j}}{x^\mu} \PD{u^{A}}{\phi^{j}} 
    + \tilde{L} u^{A} + \tilde{V} = 0\,,
\end{align}
\textit{where $n_\mu$ are the components of the normal vector, given by $n_\mu = \PD{\chi}{x^\mu}$, and $\tilde{M}^{(\mu)}$, $\tilde{L}$, and $\tilde{V}$ are the same matrices and vectors as before, but expressed in the new coordinate system.}

\textit{Evaluating this equation on $\Sigma$ yields the schematic form}
\begin{align}
    \sum_{\mu=1}^n \left.\tilde{M}^{(\mu)} n_\mu \PD{u^{A}}{\chi}\right|_\Sigma = \text{terms known on } \Sigma\,.
\end{align}
\textit{On the right-hand side, everything is fully determined by the initial data $f^{A}$. However, on the left-hand side, the partial derivative $\PD{u^{A}}{\chi}$ appears—this is not directly determined by the initial data.}

\textit{It follows that if we can solve this equation for $\PD{u^{A}}{\chi}$, we can integrate the PDE~\eqref{eq:IVP} and determine $u^{A}$ in a neighborhood of $\Sigma$ using~\eqref{eq:FormalIntegration}. Thus, establishing whether the initial value problem~\eqref{eq:IVP} has a unique solution reduces to a problem of linear algebra: We must determine whether the matrix}
\begin{align}
    \sum_{\mu=1}^n \tilde{M}^{(\mu)} n_\mu
\end{align}
\textit{is invertible. What is crucial for our discussion is that the above matrix carries additional information: it can reveal the presence of hidden constraint equations or integrability conditions. To make this explicit, let us denote}
\begin{align}
    \mathcal{K} \ce \sum_{\mu=1}^n \tilde{M}^{(\mu)} n_\mu\,,
\end{align}
\textit{and assume that it is degenerate, meaning its rank satisfies $\rank\mathcal{K} = r < m$. In this case, the equation for $\PD{u^{A}}{\chi}$ does not have a unique solution. More importantly, the degeneracy of $\mathcal{K}$ implies the existence of either constraints or identities within the PDE system~\eqref{eq:IVP}.}

\textit{To see this explicitly, we apply the Gauss elimination algorithm to bring $\mathcal{K}$ into row-echelon form:}
\begin{align}
    \tilde{\mathcal{K}} =
    \begin{pmatrix}
        \tilde{\mathcal{K}}_{11} & \tilde{\mathcal{K}}_{12} & \tilde{\mathcal{K}}_{13} & \cdots & \tilde{\mathcal{K}}_{1(m-1)} & \tilde{\mathcal{K}}_{1 m} \\
        0 & \tilde{\mathcal{K}}_{22} & \tilde{\mathcal{K}}_{23} & \cdots & \tilde{\mathcal{K}}_{2(m-1)} & \tilde{\mathcal{K}}_{2m} \\
        \vdots & & \ddots & & \ddots & \vdots \\
        0 & 0 & 0 & \cdots & 0 & \tilde{\mathcal{K}}_{rm} \\ \hline
        0 & 0 & 0 & \cdots & 0 & 0 \\
        \vdots & \vdots & \vdots & \cdots & \vdots & \vdots \\
        0 & 0 & 0 & \cdots & 0 & 0
    \end{pmatrix}.
\end{align}
\textit{Here, the tilde denotes the row-echelon form of $\mathcal{K}$. Since we assumed $\rank\mathcal{K} = r$, only the first $r$ rows contain nonzero entries, while the remaining $m-r$ rows (below the horizontal line) are entirely zero.}

\textit{Since the transformation to row-echelon form involves only linear operations, we can extend these operations to the entire PDE system~\eqref{eq:IVP}, which then takes the schematic form}
\begin{align}
    \tilde{\mathcal{K}}\, \PD{u^{A}}{\chi} = \text{terms known on } \Sigma\,.
\end{align}
\textit{Explicitly, this becomes}
\begin{align}
    \begin{pmatrix}
        \tilde{\mathcal{K}}_{11} & \tilde{\mathcal{K}}_{12} & \tilde{\mathcal{K}}_{13} & \cdots & \tilde{\mathcal{K}}_{1(m-1)} & \tilde{\mathcal{K}}_{1 m} \\
        0 & \tilde{\mathcal{K}}_{22} & \tilde{\mathcal{K}}_{23} & \cdots & \tilde{\mathcal{K}}_{2(m-1)} & \tilde{\mathcal{K}}_{2m} \\
        \vdots & & \ddots & & \ddots & \vdots \\
        0 & 0 & 0 & \cdots & 0 & \tilde{\mathcal{K}}_{rm} \\ \hline
        0 & 0 & 0 & \cdots & 0 & 0 \\
        \vdots & \vdots & \vdots & \cdots & \vdots & \vdots \\
        0 & 0 & 0 & \cdots & 0 & 0
    \end{pmatrix}
    \begin{pmatrix}
        \partial_\chi u^1\\
        \partial_\chi u^2\\
        \vdots\\
        \partial_\chi u^r \\
        \partial_\chi u^{r+1} \\
        \vdots\\
        \partial_\chi u^{m}
    \end{pmatrix}
    = \text{terms known on } \Sigma\,.
\end{align}
\textit{From this form, we see that only the first $r$ equations determine the $\chi$-derivatives of $u^{A}$. The remaining $m-r$ equations are either lower-order constraints on the initial data or trivial identities (e.g. $0=0$). Thus, the degeneracy of $\mathcal{K}$ directly implies the presence of hidden constraints or redundancies in the PDE system.}

\textit{These considerations extend naturally to higher-order systems. In the case of second-order systems, the matrix $\mathcal{K}$, which multiplies the second-order time derivatives, is often referred to as the \emph{kinetic matrix}. The techniques discussed here have been applied to kinetic matrices in $f(\bbQ)$ gravity in~\cite{DAmbrosio:2023}, and we refer the reader to that work for further details.}  

\textit{With this, we are now equipped to answer the question posed at the beginning of this subsection: How can we systematically determine whether identities and/or integrability conditions arise when prolonging a given equation $\R_q$? The key lies in analyzing the matrix that multiplies the highest-order derivatives.}\medskip

By extending our analysis from quasi-linear first-order PDEs to \emph{arbitrary} PDEs and considering all highest-order derivatives---not just those in the ``time'' direction---we arrive at the following definition:

\begin{definition}[The symbol $\S_q$ of $\R_q$]\label{def:Symbol}
    Let $\R_q$ be described by the equation $\bbE^\tau(x^\mu, v^{A}, p^{A}_{\bfm})=0$ for $\tau \in \{1, \dots, \ell\}$. The \textbf{symbol} $\S_q$ of $\R_q$ is then defined as the solution space to the following system of linear equations in the unknowns $\xi^{A}_{\bfm}$:
    \begin{align*}
        \S_q: \left\{\sum_{A=1}^{m} \sum_{|\bfm|=q} \frac{\partial \bbE^\tau}{\partial p^{A}_{\bfm}} \xi^{A}_{\bfm} = 0\right. \quad \text{for all} \quad \tau \in \{1, \dots, \ell\}.
    \end{align*}
    By abuse of language, we also refer to the matrix $\frac{\partial \bbE^\tau}{\partial p^{A}_{\bfm}}$ as the \textbf{symbol} and denote it by $\S_q$.
\end{definition}

To better understand the matrix $\frac{\partial \bbE^\tau}{\partial p^{A}_{\bfm}}$ and the system of linear equations it defines, let us examine its size. In the matrix-vector product described above, the indices $A$ and $\bfm$ are summed over, while $\tau$ is free. Since the matrix-vector multiplication produces a column vector, we infer that $\tau$ labels the rows of that vector, and there are $\ell$ rows in total. Therefore, the matrix $\frac{\partial \bbE^\tau}{\partial p^{A}_{\bfm}}$ has $\ell$ rows (i.e., there are as many rows as algebraically independent equations).

The index $A$ corresponds to the components of the field. For example, if $v^{A}$ is a scalar field, then $A=1$; if $v^{A}$ is a vector field, then $A \in \{1, \dots, n\}$; for a symmetric $(0,2)$ tensor field, $A \in \{1, \dots, \frac{n(n+1)}{2}\}$, and so on. The multi-index $\bfm$ represents all partial derivatives of order $q$. In $n$ dimensions, the number of independent $q$-th order derivatives is given by the binomial coefficient:
\begin{align*}
    \binom{n-1+q}{n-1}.
\end{align*}
Thus, the vector $\xi^{A}_{\bfm}$, where $|\bfm|=q$, has a total of $m \binom{n-1+q}{n-1}$ components, which means the matrix must have this many columns.

Taking these considerations into account, we conclude that the symbol is a matrix of size $\ell \times m \binom{n-1+q}{n-1}$. By experimenting with different values, we observe that $m \binom{n-1+q}{n-1}$ is generally larger than $\ell$, except in the special case of a single scalar field ($m=1$) in one dimension ($n=1$), which satisfies a single equation ($\ell=1$) of order one ($q=1$). This indicates that, in most cases, the symbol matrix is longer (more columns) than it is tall (number of rows).

Another way to state this fact is that the system of linear equations describing the symbol contains more variables than equations. This observation is crucial, as we will later be interested in the size of the solution space of $\S_q$. From linear algebra, we recall that the size, or dimension, of the solution space to a system of linear equations is determined by the number of variables that can be freely chosen to specify a particular solution. Since there are always fewer equations than variables in our case\footnote{Because the number of equations $\ell$ is in general smaller than the number of variables $m\binom{n-1+q}{n-1}$.}, there will always be some undetermined variables that we can freely choose to fix a particular solution.

In the best case, the symbol has maximal rank, i.e., $\rank \S_q = \ell$, which implies that all equations of order $q$ are independent. In this case, there are exactly $m \binom{n-1 +q}{n-1} - \ell$ free variables. On the other end of the spectrum, if $\S_q$ has minimal rank\footnote{Note that $\S_q$ cannot have a rank of zero, because of how $\S_q$ is constructed. In fact, $\rank \S_q = 0$ would imply that the $q$-th order PDE $\R_q$ \textbf{does not} contain any equations of order $q$, which is a contradiction.}, its rank would be~$1$. In this case, there are $m \binom{n-1 +q}{n-1} - 1$ free variables. The general case is described by the following formula:
\begin{align}\label{eq:DimOfSq}
    \dim \S_q = m \binom{n-1 +q}{n-1} - \rank \S_q \,.
\end{align}
For future reference we note that the rank of $\S_q$ measures how many independent $q$-th order equations are present in $\R_q$. This number can be equal or lower to the number of algebraically independent equations in $\R_q$. That is because $\R_q$ may contain independent equations of order less than $q$. At this point, it is convenient to introduce the notion of \emph{principal} and \emph{parametric derivatives}.

\begin{definition}[Principal and parametric derivatives]\label{def:PrincipalParametricDerivatives}
    Given an equation $\R_q$, we call every $q$-th order jet variable we can solve for a \textbf{principal derivative}. All other jet variables, namely the ones that cannot be solved for, are called \textbf{parametric derivatives}.
\end{definition}

The concept of principal and parametric derivatives is best described using an analogy and an example. For the analogy, consider the linear system of equations
\begin{align}
    \begin{cases}
        2x - z &= 3 \\
        x + y &= 1
    \end{cases}\,,
\end{align}
which can be solved for $x$ and $y$:
\begin{align}
    x &= \frac32 + \frac12 z\notag\\
    y &= -\frac12 - \frac12 z\,.
\end{align}
Here, $x$ and $y$ are analogous to principal derivatives, because we were able to solve for them. On the other hand, $z$ corresponds to a parametric derivative, because it appears on the right hand side of the solved-for variables. In this sense, $z$ parametrizes the solution space: for every choice of $z$, we obtain a distinct solution to the linear system of equations. Also, note that in matrix-vector notation we would say that the above system of linear equations is described by a matrix of rank $2$ and that the solution space is $1$-dimensional. This is where a connection to the symbol $\S_q$ appears: The rank of the symbol is analogous to the rank of the matrix of the linear system, telling us how many independent equations we have to work with. On the other hand, the dimension of the symbol is analogous to the dimension of the solution space of the linear system. It tell us how many variables are left to parametrize the solution space. 

These considerations are reinforced by the following example based on PDEs.

\begin{example}[Principal and parametric derivatives]
    We work in three dimensions and we consider a scalar field $\Phi$ governed by the following system of non-linear partial differential equations:
    \begin{align*}
        \R_2 : 
        \begin{cases}
            \Phi_{xx} + \Phi_{yy} - \Phi_x \Phi_y &= 0\\
            \Phi_{xy} - \Phi_{yz} + \Phi &= 0
        \end{cases}\,.
    \end{align*}
    We can solve this system, for instance, for $\Phi_{xx}$ and $\Phi_{xy}$, thus obtaining
    \begin{align*}
        \Phi_{xx} &= - \Phi_{yy} + \Phi_x \Phi_y \\
        \Phi_{xy} &= \Phi_{yz} - \Phi\,.
    \end{align*}
    The jet variables $\Phi_{xx}$ and $\Phi_{xy}$ are principal derivatives, while $\Phi_{xz}$, $\Phi_{yy}$, $\Phi_{yz}$, and $\Phi_{zz}$ are parametric derivatives. We could of course also have selected other variables, such as $\Phi_{yy}$ and $\Phi_{yz}$, to play the role of principal derivatives. Observe, however, that $\Phi_{xz}$ and $\Phi_{zz}$ do not appear in the equations. Therefore, they are always parametric derivatives.

    More generally, we can determine the symbol of this system and conclude that its rank is~$2$. Thus we can always solve for two of the highest order derivatives, which is the same as saying that we can always select two such derivatives as principal derivatives. 

    The solution space of the symbol turns out to have dimension
    \begin{align*}
        \dim \S_2 = \binom{3-1+2}{3-1} - \rank \S_2 = 6-2 = 4\,.
    \end{align*}
    Thus, as expected, there are four parametric derivatives which, as the name suggests, parametrize the solution space of the symbol. 
\end{example}
We reiterate the crucial point that for an equation $\R_q$ with symbol $\S_q$ there are precisely $\rank\S_q$ principal derivatives and $\dim\S_q$ parametric derivatives. Looking back at our interlude on quasi-linear first order partial differential equations, we realize that the matrix $\mathcal{K}$ corresponds to a part of the symbol $\S_1$. Moreover, from the row-echelon form $\tilde{\mathcal{K}}$, which has non-maximal rank, we can see the intuitive meaning of $\rank \S_q < \ell$: Whenever the rank of $\S_q$ is \emph{not maximal}, there are highest order derivatives we \emph{cannot} solve for and which are also \emph{not determined by the initial data}. Moreover, this also implies the presence of constraints or identities. 

Since the rank of $\S_q$ is related to the dimension of $\S_q$ via equation~\eqref{eq:DimOfSq}, we can also use $\dim\S_q$ as an indicator for occurrences of integrability conditions or identities. This will be the content of Theorem~\ref{thm:DimR1q}. Before we discuss this result, however, it is instructive to examine some concrete examples of symbols and determine the symbol of a prolonged equation. We begin with a simple example:

\begin{example}[The symbol of the non-linear equation in~\ref{ex:NonFiberedProlongation}]\label{ex:SymbolOfNonLinearEquation}
For clarity, we repeat the system of PDEs from Example~\ref{ex:NonFiberedProlongation}:
    \begin{align*}
        \R_2:
        \begin{cases}
            \Phi_{xx} - \frac12 \left(\Phi_{yy}\right)^2 &= 0\\
            \Phi_{yy} - \Phi_{xy}  &= 0
        \end{cases}\,.
    \end{align*}
    For this system, $\bbE^\tau$ is defined as
    \begin{align*}
        \bbE^{1} &= \Phi_{xx} - \frac12 \left(\Phi_{yy}\right)^2\\
        \bbE^{2} &= \Phi_{yy} - \Phi_{xy}\,.
    \end{align*}
    According to Definition~\ref{def:Symbol}, we need to take derivatives of $\bbE^\tau$ with respect to the \textit{highest-order} jet variables. In this case, these variables are
    \begin{align*}
        (\Phi_{xx}, \Phi_{xy}, \Phi_{yy})\,.
    \end{align*}
    We differentiate in the following order: first with respect to $\Phi_{xx}$, then with respect to $\Phi_{xy}$, and finally with respect to $\Phi_{yy}$. This results in the following symbol:
    \begin{align*}
        \S_2 = 
        \begin{pmatrix}
            \displaystyle\PD{\bbE^{1}}{\Phi_{xx}} & \displaystyle\PD{\bbE^{1}}{\Phi_{xy}} & \displaystyle\PD{\bbE^{1}}{\Phi_{yy}} \\[15pt]
            \displaystyle\PD{\bbE^{2}}{\Phi_{xx}} & \displaystyle\PD{\bbE^{2}}{\Phi_{xy}} & \displaystyle\PD{\bbE^{2}}{\Phi_{yy}}
        \end{pmatrix}
        =
        \begin{pmatrix}
            1 & 0 & -\Phi_{yy} \\[5pt]
            0 & -1 & 1
        \end{pmatrix}\,.
    \end{align*}
    Notice that changing the order of the highest-order jet variables corresponds to swapping the columns of the symbol. We do have the freedom to alter the ordering of jet variables (and thus the columns), and we will revisit this flexibility later. Regardless of the ordering, however, the rank of the symbol remains unaffected. For any ordering, we have
    \begin{align*}
        \rank \S_2 = 2\,.
    \end{align*}
    Clearly, this is consistent with the fact that $\rank \S_2$ measures the number of independent equations of order $2$, which is the same as the number of principal derivatives. Evidently, the $\R_2$ in this example contains exactly two such equations. Furthermore, for the dimension we obtain
    \begin{align}
        \dim\S_2 = \binom{2-1+2}{2-1} - \rank\S_2 = 3 - 2 = 1\,.
    \end{align}
    Indeed, there is only one parametric derivative spanning the solution space of $\S_2$.
\end{example}\medskip

Next, we consider Maxwell's and Proca's equations. These equations are not only more physically relevant, but we will also encounter them repeatedly, as their familiarity will help illustrate several novel concepts.

\begin{example}[The symbol of Maxwell's and Proca's equations]\label{ex:SymbolMaxwellProca}
    In Example~\ref{ex:R1qNotEqualR1}, we already encountered Proca's equations, which can be written as
    \begin{align*}
        \R_2 : 
        \begin{cases}
            \partial_\mu \left(\partial^\nu A^\mu - \partial^\mu A^\nu\right) + m^2\, A^\nu = 0\,.
        \end{cases}
    \end{align*}
    If we set the photon mass $m$ to zero, we recover Maxwell's equations. However, when it comes to computing the symbol, we do not need to distinguish between Proca and Maxwell, as both equations share the same symbol. This is because the symbol only depends on the highest-order derivatives, which in both cases are contained in the term $\partial_\mu \left(\partial^\nu A^\mu - \partial^\mu A^\nu\right)$.

    To compute the symbol, we need to differentiate this term with respect to all possible second-order derivatives $\partial_\alpha \partial_\beta A^\gamma$. Sticking to index notation, we find the symbol:
    \begin{align*}
        \S_2 = \delta\ud{\nu}{\gamma} \eta^{\alpha\beta} - \frac12 \left(\delta\ud{\beta}{\gamma}\eta^{\alpha\nu} + \delta\ud{\alpha}{\gamma}\eta^{\beta\nu}\right)\,.
    \end{align*}
    This is a rather abstract expression, which hides important properties of the matrix, such as its dimensions and rank, which we discussed earlier. In $n$ spacetime dimensions, we expect the symbol to have $n$ rows and $n\binom{n-1+2}{n-1} = \frac12 n^2 (n+1)$ columns. In the physically relevant case of $n=4$ spacetime dimensions, this results in a matrix of size $4 \times 40$.

    Given these large numbers, determining the rank of $\S_2$ may seem daunting. However, it turns out to be relatively straightforward. First, we observe that $\nu$ labels the equations in $\R_2$, which corresponds to the rows of the matrix $\S_2$. The indices $\alpha$ and $\beta$ belong to the multi-index $\bfm$ and indicate which second-order partial derivative we are considering. Finally, the index $\gamma$ corresponds to the index $A$, which specifies which component of the vector field $A^\gamma$ we are referring to. Together, the indices $\alpha$, $\beta$, and $\gamma$ specify a particular element $\partial_\alpha \partial_\beta A^\gamma$ in the symbol, or equivalently, a column in $\S_2$.

    As in Example~\ref{ex:SymbolOfNonLinearEquation}, we must choose an ordering for the second-order jet variables, i.e., for the variables of the form $\partial_\alpha \partial_\beta A^\gamma$. We select a different ordering than in the previous example and fix $n=4$ to simplify the notation:
    \begin{align*}
        (\partial_1 \partial_1 A^1, \partial_2 \partial_2 A^2, \partial_3 \partial_3 A^3, \partial_4 \partial_4 A^4, \dots)\,,
    \end{align*}
    where the dots represent the remaining 36 variables. Their ordering is not relevant to us, as we have chosen a convenient ordering for the first four jet variables, which allows us to directly read off the rank of the symbol. Specifically, for the first four variables, we have $\alpha=\beta=\gamma$, which reduces the first four columns of the symbol to:
    \begin{align*}
        \S_2 = \eta^{\alpha \alpha}\,.
    \end{align*}
    More explicitly and correctly, this is:
    \begin{equation}
        \S_2 = \bordermatrix{ 
                  & \partial_1 \partial_1 A^1 & \partial_2 \partial_2 A^2 & \partial_3 \partial_3 A^3 & \partial_4 \partial_4 A^4 & \dots\cr
          \nu = 1 & -1 & 0 & 0 & 0 & \dots\cr
          \nu = 2 & 0 & 1 & 0 & 0 & \dots \cr
          \nu = 3 & 0 & 0 & 1 & 0 & \dots \cr
          \nu = 4 & 0 & 0 & 0 & 1 & \dots
          }\,.
    \end{equation}
    This form clearly indicates that there are additional columns, which we have suppressed. Thanks to our clever choice of column ordering, we can now directly read off the rank of $\S_2$:
    \begin{align*}
        \rank \S_2 = 4\,.
    \end{align*}
    Since this index computation is valid in any dimension $n$, and always results in the first $n$ columns of $\S_2$ being diagonal elements $\eta^{\alpha \beta}$, we conclude that the rank of $\S_2$ is:
    \begin{align*}
        \rank \S_2 = n
    \end{align*}
    in any dimension $n$, as the Minkowski metric has rank $n$. This number is again consistent with the fact that $\rank\S_2$ measures the number of independent second-order equations.
\end{example}

As anticipated, the symbol provides valuable information about integrability conditions or constraint equations. In fact, we can explicitly construct these integrability conditions using only linear algebra. To that end, the following theorem~\cite{SeilerBook} is crucial.

\begin{theorem}[Dimension of $\R^{(1)}_q$]\label{thm:DimR1q}
    If the solution space of $\S_{q+1}$ is a vector bundle, then
    \begin{align}
        \dim \R^{(1)}_q = \dim \R_{q+1} - \dim \S_{q+1}\,
    \end{align}
    at every point in the jet bundle. In other words, the dimension of $\R^{(1)}_q$, which is obtained by projecting $\R_{q+1}$ back into $J_q\E$, is equal to the dimension of $\R_{q+1}$ minus the number of freely specifiable jet variables of order $q+1$.
\end{theorem}\medskip
The proof of this theorem provides valuable insight into the symbol, the dimension of its solution space, and related concepts.\medskip

\begin{proof}
    This theorem concerns the dimensions of three different spaces: the submanifolds $\R_{q+1}$ and $\R^{(1)}_q$, and the solution space $\S_{q+1}$. Suppose the original PDE is given by
    \begin{align}
        \R_q : \left\{ \bbE^\tau(x^\mu, v^A, p^A_{\bfm}) = 0 \right.\,.
    \end{align}
    The prolonged PDE that defines the submanifold $\R_{q+1}$ is
    \begin{align}
        \R_{q+1}:
        \begin{cases}
            \bbE^\tau(x^\mu, v^A, p^A_{\bfm}) &= 0 \\
            D_\nu \bbE^\tau(x^\mu, v^A, p^A_{\bfm}) &= 0
        \end{cases}\,.
    \end{align}
    The submanifold $\R^{(1)}_q$ is obtained by projecting $\R_{q+1}$. It may or may not coincide with the original PDE $\R_q$. Whether $\R^{(1)}_q = \R_q$ depends on whether integrability conditions and/or identities arise during the prolongation. The solution space $\S_{q+1}$ is defined by the linear equation
    \begin{align}\label{eq:LinEq}
        \S_{q+1} : \left\{ \sum_{A=1}^{m} \sum_{|\bfm|=q+1} \frac{\partial D_\nu \bbE^\tau}{\partial p^A_{\bfm}} \xi^A_{\bfm} = 0 \right.\,,
    \end{align}
    where $\frac{\partial D_\nu \bbE^\tau}{\partial p^A_{\bfm}}$ represents the symbol of $\R_{q+1}$. According to equation~\eqref{eq:DimOfSq}, the dimension of $\S_{q+1}$ is given by
    \begin{align}
        \dim \S_{q+1} = m 
        \begin{pmatrix}
            n + q \\ 
            n - 1
        \end{pmatrix}
        - \rank \S_{q+1}\,.
    \end{align}

    To clarify, the dimension measures how many of the highest-order jet variables can be freely specified. In other words, the total number of highest-order jet variables, minus the number of independent equations of order $q$, tells us how many of these variables can be solved for in the linear equation~\eqref{eq:LinEq}.

    For the dimension of $\R_{q+1}$, we recall that Definition~\ref{def:DimRq} applies. Denote the number of independent equations in $\R_q$ by $\ell$. The prolonged equation $\R_{q+1}$ then contains $\ell + n \ell$ independent equations, because the prolongation produces $n\ell$ new equations of order $q+1$. Then, according to Definition~\ref{def:DimRq}, we have
    \begin{align}\label{eq:DimRqp1}
    \dim \R_{q+1} = m \binom{n+q+1}{q+1} - \ell - n \ell\,.
    \end{align}

    Since $\R_{q+1}$ defines a submanifold of $J_{q+1}\E$, we can also determine its dimension in another way. Recall from elementary calculus that the dimension of a submanifold is the same as the number of independent functions required to describe it locally. This number is given by the rank of the Jacobian matrix. The Jacobian of $\R_{q+1}$ is obtained by taking derivatives of $\bbE^\tau$ and $D_\nu \bbE^\tau$ with respect to $p^{A}_{\bfm}$ for $|\bfm|\leq q+1$, and with respect to $v^{A}$. We find that the Jacobian naturally compartmentalizes into six distinct blocks:
    \begin{align}
    \left( \begin{array}{|c|c|c|} \hline
    & & \\[-5pt]
        \boldsymbol{0} & \displaystyle \PD{\bbE^\tau}{p^{A}_{\bfm}} \quad \text{with }|\bfm|\leq q & \displaystyle \PD{\bbE^\tau}{v^{A}} \\
        & & \\[-5pt] \hline
        & & \\[-5pt]
        \displaystyle \PD{D_\mu \bbE^\tau}{p^{A}_{\bfm}} \quad \text{with } |\bfm| = q+1 & \displaystyle \PD{D_\mu \bbE^\tau}{p^{A}_{\bfm}} \quad \text{with } |\bfm|\leq q & \phantom{.} \quad \displaystyle \PD{D_\mu \bbE^\tau}{v^{A}} \qquad \phantom{.} \\[-5pt]
        & & \\ \hline
    \end{array} \right)\,.
    \end{align}
    The three upper blocks are obtained from the original PDE $\bbE^\tau$ by taking derivatives with respect to $v^{A}$ and the jet variables $p^{A}_{\bfm}$ up to and including order $|\bfm| = q$. Since the original equation is of order $q$, the left upper block is necessarily filled with zeros. However, the prolonged equation is of order $q+1$, which is why the left lower block is non-zero. Notice that this block contains precisely the symbol of the equation $\R_{q+1}$ (see equation~\eqref{eq:LinEq}).

    Now, we show that the dimension of $\R_{q+1}$ can be written in terms of the dimensions of two other spaces: $\R_q$ and $\S_{q+1}$. To see this, observe that the middle and right blocks in the top row of the Jacobian are simply the Jacobian of $\R_q$. As mentioned earlier, the bottom left block represents the symbol $\S_{q+1}$ of the prolonged equation. Therefore, we have
    \begin{align}\label{eq:RelatingDimensions}
    \dim \R_q + \dim \S_{q+1} &=  m \binom{n+q}{q} - \ell + m \binom{n+q}{n-1} - \rank \S_{q+1} \notag\\
    &= m \left( \frac{(n+q)!}{q!\,n!} + \frac{(n+q)!}{(n-1)!(q+1)!} \right) - \ell - \rank \S_{q+1} \notag\\
    &= m \frac{(n+q)!}{(q+1)!(n-1)!} \left( \frac{q+1}{n} + 1 \right) - \ell - \rank \S_{q+1} \notag\\
    &= m \frac{(n+q)!}{(q+1)!(n-1)!} \frac{n+q+1}{n} - \ell - \rank \S_{q+1} \notag\\
    &= m \binom{n+q+1}{q+1} - \ell - n\ell + \Delta = \dim \R_{q+1} + \Delta\,.
    \end{align}
    To get to the last line, we used that $\rank \S_{q+1} = n\ell - \Delta$, where $\Delta$ measures the deviation from the maximal possible rank. Finally, we also used equation~\eqref{eq:DimRqp1}. After rearranging this equation, we obtain
    \begin{align}\label{eq:IntermediateStep}
        \dim \R_q - \Delta = \dim \R_{q+1} - \dim \S_{q+1}\,,
    \end{align}
    which looks almost like what we are trying to prove. The difference is that Theorem~\ref{thm:DimR1q} involves $\dim\R^{(1)}_{q}$, rather than $\dim\R_q$ and $\Delta$. To relate the former to the latter two quantities, we need to introduce two case distinctions. In the first case, we assume that $\S_{q+1}$ has maximal rank. If this is the case, then $\Delta = 0$ and all rows of the lower left block of the Jacobian are linearly independent, which means it is impossible to construct algebraically independent equations of order less than $q+1$. As we know, this means that projecting $\R_{q+1}$ back into $J_q\E$ gives us $\R^{(1)}_q = \R_q$. Thus, equation~\eqref{eq:IntermediateStep} turns into
    \begin{align}
        \dim \R^{(1)}_q = \dim \R_{q+1} -\dim \S_{q+1}\,.
    \end{align}
    This is precisely what we set out to prove. In the second scenario, the symbol $\S_{q+1}$ \textit{does not} have maximal rank. Thus, not all rows appearing in the left lower block of the Jacobian are linearly independent! What this means, is that we can perform linear operations which turn at least one of the rows (at most all except one row) into a bunch of zeros. In turn this means that not all equations of order $q+1$ were independent of each other (recall Example~\ref{ex:R1qNotEqualR1} where we showed that this happens for the Proca equation). In particular, we now have $\Delta > 0$. Since $\Delta$ is the difference between $n \ell$ and the number of independent equation of order $q+1$, it measures how many of the $n\ell$ equations generated by the prolongation fail to actually be of order $q+1$. In other words, $\Delta$ turns into a measure for the number of hidden integrability conditions. Thus, \begin{align}
        \dim\R_{q} - \Delta = \dim \R^{(1)}_q
    \end{align}
    by definition. We finally conclude that
    \begin{align}
        \dim \R^{(1)}_q = \dim\R_{q+1} - \dim \S_{q+1}\,.
    \end{align}
    Since the solution space of $\S_{q+1}$ is a vector bundle, this result actually holds at every point. 
\end{proof}\medskip

From Theorem~\ref{thm:DimR1q} and its proof we can deduce further useful consequences. First of all, the proof shows us how integrability conditions and/or identities can be constructed: In order to determine the rank of $\S_{q+1}$, it is necessary to perform linear operations on its matrix. By recording these linear operations and then applying them to $\R_{q+1}$, one can directly construct integrability conditions or identities, if they occur. 

Moreover, we saw how the dimensions of $\R_{q}$, $\R_{q+1}$, and $\S_{q+1}$ are related to each other and what role the rank of $\S_{q+1}$ in determining the occurrence of integrability conditions plays. From these relations we deduce the following Corollary. \medskip

\begin{corollary}[Criterion for presence of integrability conditions]\label{cor:CriterionForIntCond}
    Let $\R_q$ be a PDE of order $q$, and let $\S_{q+1}$ be the symbol of its prolongation $\R_{q+1}$. Then:
    \begin{enumerate}
        \item If $\rank \S_{q+1} = n \ell$, then no integrability conditions arise, and $\R^{(1)}_q = \R_q$.
        \item If $\rank \S_{q+1} < n \ell$, then integrability conditions or identities appear, leading to $\R^{(1)}_q \neq \R_q$.
        \item The number of integrability conditions (or identities) is given by $n\ell - \rank \S_{q+1}$.
    \end{enumerate}
\end{corollary}

\begin{proof}
    The proof follows directly from Theorem~\ref{thm:DimR1q} and its proof. If the rank of $\S_{q+1}$ is maximal, meaning $\rank \S_{q+1} = n \ell$, then all equations in $\R_{q+1}$ contribute new independent equations, and no additional conditions arise. In this case, $\R^{(1)}_q = \R_q$.

    However, if $\rank \S_{q+1} < n \ell$, then some of the equations in $\R_{q+1}$ must be dependent, indicating the presence of integrability conditions or identities. The number of such conditions is precisely the deficiency in rank, given by $n \ell - \rank \S_{q+1}$.

    Since these conditions result from the linear dependencies in $\S_{q+1}$, they can be systematically derived by performing row operations on the matrix representing $\S_{q+1}$, as outlined in the proof of Theorem~\ref{thm:DimR1q}. Applying these operations to the equations in $\R_{q+1}$ yields the explicit form of the integrability conditions or identities.
\end{proof}

To illustrate these mathematical results, we revisit Maxwell's and Proca's equations as examples.

\begin{example}[The dimension of $\R^{(1)}_2$ for Proca and Maxwell]\label{ex:DimR12ProcaAndMaxwell}
    From Example~\ref{ex:ProjectionMaxwell}, we know that Maxwell's equations do not contain any hidden integrability conditions. However, the same does not hold for Proca's equations, as seen in Example~\ref{ex:R1qNotEqualR1}. Let us verify these findings using Corollary~\ref{cor:CriterionForIntCond}. We work in $n$ spacetime dimensions and compute the dimensions of three different spaces.

    \paragraph{Maxwell's Equations.} We start with Maxwell's equations:
    \begin{align*}
        \dim \R_2 &= n \binom{n+2}{2}-n = \frac12 n^2 (n+3)\,, \\
        \dim \S_3 &= n \binom{n+2}{n-1} - \rank \S_3 = n \binom{n+2}{n-1} - n^2 = \frac16 n^2 (n+1)(n+2) - n^2\,, \\
        \dim \R_3 &= m \binom{n+q+1}{q+1} - n \ell - \ell = \frac16 (n+1)(n+2)(n+3) - n - n^2\,.
    \end{align*}
    Here, we used $m = n$, $\ell = n$, and $q = 2$. We find:
    \begin{align*}
        \dim \R^{(1)}_2 = \dim\R_3 - \dim\S_3 = \frac12 n^2 (n+3) = \dim\R_2\,.
    \end{align*}
    According to Corollary~\ref{cor:CriterionForIntCond}, this confirms that there are no integrability conditions.

    \paragraph{Proca's Equations.} Now, we turn to Proca's equations. In this case, we find:
    \begin{align*}
        \dim \R_2 &= n \binom{n+2}{2} - n = \frac12 n^2 (n+3)\,, \\
        \dim \S_3 &= n \binom{n+2}{n-1} - \rank \S_3 = n \binom{n+2}{n-1} - n^2 + 1 = \frac16 n^2 (n+1)(n+2) - n^2 + 1\,, \\
        \dim \R_3 &= m \binom{n+q+1}{q+1} - n \ell - \ell = \frac16 (n+1)(n+2)(n+3)- n - n^2\,.
    \end{align*}
    Observe that the dimensions of $\R_2$ and $\R_3$ are identical to those in the Maxwell case. The key difference lies in the dimension of $\S_3$, which stems from its rank. This discrepancy is easily explained: as noted in Example~\ref{ex:R1qNotEqualR1}, a specific linear combination of third-order terms reduces to a first-order expression. As a result, the rank of $\S_3$ decreases by one compared to the Maxwell case, since the rank of a matrix is sensitive to linear dependencies. From the dimensions computed above we obtain
    \begin{align*}
        \dim \R^{(1)}_2 = \dim\R_3 - \dim\S_3 = \frac12 n^2 (n+3) - 1 = \dim\R_2 - 1 < \dim\R_2\,.
    \end{align*}
    By Corollary~\ref{cor:CriterionForIntCond}, we conclude that one integrability condition must be present, which is consistent with our prior findings.
\end{example}

In what follows, the symbol will play a crucial role, revealing more information than we have already uncovered. To make this information more accessible, we can simplify the symbol's structure. This process consists of two steps: First, we exploit the freedom to rearrange the columns of the symbol, grouping them into distinct classes. Then, we apply the Gauss algorithm to bring the symbol into row-echelon form.

The classes mentioned above correspond to classes of multi-indices, defined as follows.

\begin{definition}[The class of a multi-index]\label{def:class}
    A multi-index $\bfm = [m_1, \dots, m_n]$ is said to be of \textbf{class} $k$ if its first non-zero entry is $m_k$.
\end{definition}\medskip

For example, in four dimensions, the multi-indices $\bfm_1 = [0,0,0,1]$, $\bfm_2 = [1,1,0,1]$, and $\bfm_3 = [0,3,4,1]$ belong to class $4$, class $1$, and class $2$, respectively. 

This notion naturally extends to jet variables $p^{A}_{\bfm}$. Consider a scalar field $\Phi$ and the multi-indices $\bfm_1$, $\bfm_2$, and $\bfm_3$ from above. Then, we classify the corresponding jet variables as follows:
\begin{align*}
    \partial_{\bfm_{1}}\Phi &= \partial_{x^4}\Phi&& \rightarrow &&\text{class $4$}\\
    \partial_{\bfm_{2}}\Phi &= \partial_{x^{1}}\partial_{x^{2}}\partial_{x^{4}}\Phi &&\rightarrow &&\text{class $1$}\\
    \partial_{\bfm_{3}}\Phi &= \partial^3_{x^{2}}\partial^4_{x^{3}}\partial_{x^{4}}\Phi &&\rightarrow &&\text{class $2$}\,.
\end{align*}

Since jet variables can be grouped into classes, it follows that the columns of the symbol matrix $\S_q$ can also be arranged according to the classes of $p^{A}_{\bfm}$ with $|\bfm| = q$. Let us now establish some basic properties of these classes.

First, we note that the number of classes always matches the dimension of the underlying manifold. That is, there are precisely $n$ classes. However, the size of each class (i.e., the number of jet variables $p^{A}_{\bfm}$ with $|\bfm| = q$ in each class) depends on the order $q$ and on the number of algebraically independent field components $m$. To make the discussion more concise, we focus on the case of a scalar field, where $m=1$. The general case follows by multiplying the numbers obtained by $m\geq 1$.

The simplest case is $q=1$, where each class contains exactly one element. For $q=2$, the distribution is as follows:
\begin{itemize}
    \item Class $n$: $1$ element
    \item Class $n-1$: $2$ elements
    \item Class $n-2$: $3$ elements
    \item $\dots$
    \item Class 1: $n$ elements
\end{itemize}
This pattern arises because the elements under consideration take the form $\partial_{x^{i}}\partial_{x^{j}}\Phi$. For an element to belong to class $n$, we must have $i=j=n$, as any other choice would place it in a lower class. Similarly, an element belongs to class $n-1$ if either $i=n$ and $j=n-1$, or $i=j=n-1$. Proceeding in this manner, we find that, in general, class $k$ contains $n+1-k$ elements.

As a consistency check, summing over all class sizes must yield the total number of jet variables of order $q=2$ for a scalar field. Indeed, we verify:
\begin{align}
    \underbrace{\sum_{k=1}^n (n+1-k)}_{\text{sum over class sizes}} 
    &= (n+1)\sum_{k=1}^n 1 - \sum_{k=1}^n k\notag \\
    &= n(n+1) - \frac{n(n+1)}{2}\notag \\
    &= \underbrace{\frac{n(n+1)}{2}}_{\text{total number of jet variables}}.
\end{align}

To generalize this result for arbitrary order $q$, we can interpret the multi-index $\bfm$ as a sequence of $q$ bins, each containing a number from $1$ to $n$. If a bin contains the number $k$, it indicates the presence of one occurrence of $x^k$ in the multi-index. Fixing the number of bins to be $q$ ensures that the multi-index has length $q$, corresponding to a $q$-th order derivative operator.

Next, we proceed with filling the bins. The first bin, marked by the red $\red{k}$ below, determines the class we are considering---specifically, class $k$. For the multi-index to belong to class $k$, all other bins must be filled with numbers greater than or equal to $k$. Since partial derivatives commute (i.e., the order in which we take derivatives does not matter), we can always arrange the bins in ascending order:
\begin{align}
    \underbrace{\Boxed{\prescript{}{\phantom{q}}{\red{k}}_{\phantom{q}}}\Boxed{\prescript{}{\phantom{q}}{\blue{k_1}}}\Boxed{\prescript{}{\phantom{q}}{\blue{k_2}}}\,\cdots\,\Boxed{\blue{k_{q-1}}}}_{q \text{ bins}}
    \quad \text{with } \red{k} \leq \blue{k_1} \leq \blue{k_2} \leq \dots \leq \blue{k_{q-1}} \leq n.
\end{align}
We refer to this as the \textbf{lexicographic ordering}\footnote{For instance, $\partial_{x^2}\partial_{x^1}\partial_{x^3}\Phi$ is \textit{not} in lexicographic order. However, swapping the first two derivatives yields $\partial_{x^1}\partial_{x^2}\partial_{x^3}\Phi$, which is indeed lexicographically ordered.}. Determining class sizes now becomes straightforward. For $k=n$, we see that all blue indices must be equal, i.e., $k_1 = k_2 = \dots = k_{q-1} = n$. Consequently, class $n$ always has size $1$, as there is only one possible assignment for the blue indices. 

For a general $k$, more possibilities arise. The key question then becomes: How many distinct values can the blue indices take?

Clearly, in $n$ dimensions, there are $n-k+1$ possible values for the blue~$k$'s, as these correspond to all values greater than or equal to~$\red{k}$. Moreover, the blue $k$'s define a differential operator of order $q-1$. Thus, our task reduces to counting how many distinct operators of the form
\begin{align}
    \partial_{\blue{x^{k_{1}}}}\partial_{\blue{x^{k_{2}}}}\cdots\partial_{\blue{x^{k_{q-1}}}}
\end{align}
can be constructed, given that the indices $k_i$ can only take $(n-k+1)$ distinct values (i.e., they live in an $(n-k+1)$-dimensional space). Since the number of columns in the symbol $\S_q$ equals the number of $q$-th order jet variables in $n$ dimensions, which is given by $\binom{n-1+q}{n-1}$, we can determine the class sizes by replacing $n$ with $n-k+1$ and $q$ with $q-1$. This yields the expression
\begin{align}\label{eq:ClassSize}
    \csize(\class k) = \binom{n+q-k-1}{n-k}\,.
\end{align}

As in our earlier example for $q=2$, we can verify that the sum of all class sizes recovers the total number of jet variables of order $q$:
\begin{align}\label{eq:SumOverAllClassSizes}
    \underbrace{\sum_{k=1}^n \binom{n+q-k-1}{n-k}}_{\text{sum of all class sizes}} 
    = \underbrace{\binom{n+q-1}{n-1}}_{\text{number of jet variables in } J_q\E}\,.
\end{align}

Introducing the change of variable $k\mapsto n-p$, we obtain an alternative expression for the class size:
\begin{align}
    \csize(\class(n-p)) = \frac{1}{p!} \prod_{i=1}^{p} (q+i-1)\,.
\end{align}
Notably, the right-hand side is a polynomial in $q$, whose degree increases with $p$. This leads to the sequence:
\begin{align}\label{eq:SequenceOfClassSizes}
    1 = \csize(\class n) \leq \csize(\class(n-1)) \leq \dots \leq \csize(\class 1) = \binom{n+q-2}{n-1}\,,
\end{align}
where equality holds only for $q=1$. 

Examining the derivations of equations~\eqref{eq:ClassSize} and~\eqref{eq:SequenceOfClassSizes}, we observe that they remain valid regardless of the number of field components. For a general field $v^{A}$ with $m$ components, the right-hand side of~\eqref{eq:ClassSize} is simply scaled by $m$:
\begin{align}\label{eq:GeneralClassSize}
    \csize(\class k) = m\binom{n+q-k-1}{n-k}\,.
\end{align}
Similarly, the sequence~\eqref{eq:SequenceOfClassSizes} extends to:
\begin{align}\label{eq:GeneralSequenceOfClassSizes}
    m = \csize(\class n) \leq \csize(\class(n-1)) \leq \dots \leq \csize(\class 1) = m\binom{n+q-2}{n-1}\,.
\end{align}


What these fundamental results reveal is that we can systematically group the columns of the symbol according to the multi-index classes, with the number of columns in each class given by~\eqref{eq:GeneralClassSize}. To maintain consistency, we adopt the convention of ordering these columns from left to right following the sequence in~\eqref{eq:GeneralSequenceOfClassSizes}.

The resulting structure of the symbol can then be schematically represented as follows:\footnote{The depicted structure corresponds to the case $m=1$ and $q=2$, but the underlying principle extends to arbitrary values of $m$ and $q$.}
\begin{align}\label{eq:GroupingSymbolColumnsByClasses}
\renewcommand\arraystretch{1.3}
\begin{blockarray}{ccccccccc}
\text{\small class $n$} & \multicolumn{3}{c}{\small \text{class}\, $n-1$} & \multicolumn{4}{c}{\small \text{class } $n-2$} & \phantom{\cdots} \cdots \phantom{\cdots} \\
\begin{block}{(c|ccc|cccc|c)}
* & * & * & * & * & * & * & * & \cdots\ * \ \cdots\\
* & * & * & * & * & * & * & * & \cdots\ * \ \cdots\\
* & * & * & * & * & * & * & * & \cdots\ * \ \cdots\\
\vdots & \vdots & \vdots & \vdots & \vdots & \vdots & \vdots & \vdots & \ddots \\
\end{block}
\end{blockarray}\,.
\end{align}
Here, the symbols $*$ serve as placeholders for the entries of the matrix. Organizing the symbol in this manner accomplishes the first step outlined earlier—rearranging its columns so that they are grouped into classes. The next step is to apply the Gauss algorithm to transform the symbol into row-echelon form. Since this procedure is always feasible, the symbol ultimately takes on a simplified structure, which can be schematically represented as follows:
\begin{align}\label{eq:SymbolInRowEchelonForm}
\renewcommand\arraystretch{1.3}
\begin{blockarray}{lcccccccccccc}
&&& \text{\small class $n$} & \multicolumn{3}{c}{\small \text{class}\, $n-1$} & \multicolumn{4}{c}{\small \text{class } $n-2$} & \phantom{\cdots} \cdots \phantom{\cdots} & \phantom{c} \\
\begin{block}{lc(c@{}c|ccc|cccc|c@{}c)}
\beta^{(n)}_q &  && \red{\bullet} & \blue{*} & \blue{*} & \blue{*} & \blue{*} & \blue{*} & \blue{*} & \blue{*} & \cdots\ \blue{*} \ \cdots & \phantom{c}\\
\BAhhline{~~~---------~}
\multirow{3}*{$\beta^{(n-1)}_q$}  &  && 0 & \red{\bullet} & \blue{*} & \blue{*} & \blue{*} & \blue{*} & \blue{*} & \blue{*} & \cdots\ \blue{*} \ \cdots & \phantom{c} \\
  &  && 0 & 0 & \red{\bullet} & \blue{*} & \blue{*} & \blue{*} & \blue{*} & \blue{*} & \cdots\ \blue{*} \ \cdots & \phantom{c} \\
  &  && 0 & 0 & 0 & \red{\bullet} & \blue{*} & \blue{*} & \blue{*} & \blue{*} & \cdots\ \blue{*} \ \cdots & \phantom{c}  \\
\BAhhline{~~~---------~}
\multirow{4}*{$\beta^{(n-2)}_q$} &  && 0 & 0 & 0 & 0 & \red{\bullet} & \blue{*} & \blue{*} & \blue{*} & \cdots\ \blue{*} \ \cdots & \phantom{c} \\
 &  && 0 & 0 & 0 & 0 & 0 & \red{\bullet} & \blue{*} & \blue{*} & \cdots\ \blue{*} \ \cdots & \phantom{c}  \\
 &  && 0 & 0 & 0 & 0 & 0 & 0 & \red{\bullet} & \blue{*} & \cdots\ \blue{*} \ \cdots & \phantom{c} \\
 &  && 0 & 0 & 0 & 0 & 0 & 0 & 0 & \red{\bullet} & \cdots\ \blue{*} \ \cdots & \phantom{c} \\
\BAhhline{~~~---------~}
 \phantom{\beta} \vdots & && \vdots & \vdots & \vdots & \vdots & \vdots & \vdots & \vdots & \vdots & \ddots & \phantom{c} \\
\end{block}
\end{blockarray}\,.
\end{align}
The blue $\blue{*}$ entries differ from the black $*$ in the symbol described in~\eqref{eq:GroupingSymbolColumnsByClasses}, as the Gauss algorithm requires multiplying rows by real numbers and adding or subtracting rows to bring the symbol into row-echelon form, as shown in~\eqref{eq:SymbolInRowEchelonForm}. The blue $\blue{*}$ can thus be interpreted as the entries of the symbol after the Gauss algorithm has been applied. The red $\red{\bullet}$ represent the pivots---a concept familiar from linear algebra. While the blue $\blue{*}$ entries may or may not be zero, the pivots must be nonzero. This leads us to the meaning of the horizontal lines and the $\beta$'s.  

If we closely examine~\eqref{eq:SymbolInRowEchelonForm}, we see that the horizontal lines are positioned so that all pivots between two vertical lines are grouped together in a block. For example, moving the second vertical line up by one row would enclose only two pivots from the $\class n-1$. Conversely, moving the same line down by one row would exclude the first pivot of the $\class n-2$ from its block.

This observation suggests that, just as columns are grouped into classes, we can also group the rows. While the number of columns in each class is well-defined (see equation~\eqref{eq:GeneralClassSize}), the number of rows in each class is less straightforward to determine. There are two reasons for this: First, while the number of columns is determined purely by combinatorial factors involving $k$, $m$, $n$, and $q$, the number of rows depends on the number of equations in $\R_q$. Second, the Gauss algorithm may reveal that some rows are linear combinations of others. Once the row-echelon form is reached, only linearly independent rows remain, and it is possible that some rows at the bottom of the matrix contain only zeros, providing no new information.

This motivates the introduction of the integers $\beta^{(k)}_q$, which represent the number of rows at order $q$ that belong to the $\class k$. These values can only be determined once the symbol has been transformed into row-echelon form.

Since the symbol in the form~\eqref{eq:SymbolInRowEchelonForm} and the $\beta$'s play a crucial role in what follows, we now introduce two formal definitions to capture these concepts.

\begin{definition}[The symbol in solved form]\label{def:SymbolSolvedForm}
    A symbol $\S_q$ is said to be in \textbf{solved form} if its columns are ordered according to the multi-index classes of the corresponding $p^{A}_{\bfm}$, with $|\bfm|=q$, in descending order from left to right, and if it has been brought into row-echelon form.
\end{definition}\medskip

It is important to emphasize that the order within a given class is not significant; it can be fixed on a case-by-case basis. However, we will often choose a lexicographic ordering within each class.

\begin{definition}[Characters $\beta^{(k)}_q$]\label{def:beta}
    Given a symbol $\S_q$ in solved form, we define the \textbf{characters} of the symbol, denoted by $\beta^{(k)}_q$, as the integer values corresponding to the number of rows of class $k$ in $\S_q$.
\end{definition}\medskip

Before discussing the importance of these two notions, let us consider a concrete example of a second-order symbol in solved form.

\begin{example}[A symbol in solved form and its $\beta$'s]
    We consider a single scalar field $\Phi$ that obeys the second-order equations
    \begin{align*}
        \R_2 : 
        \begin{cases}
            \Phi_{yy} - \frac{1}{2} \Phi_{xx}^2 &= 0 \\
            \Phi_{xy} &= 0 \\
            \Phi_{zz} + y\, \Phi_{xz} &= 0
        \end{cases}\,.
    \end{align*}
    The second-order jet variables, written in lexicographic order, are
    \begin{align*}
        (\Phi_{xx}, \Phi_{xy}, \Phi_{xz}, \Phi_{yy}, \Phi_{yz}, \Phi_{zz})\,.
    \end{align*}
    Therefore, the symbol of this second-order system, when written in lexicographic order, takes the form
    \begin{align*}
        \begin{blockarray}{cccccc}
            \Phi_{xx} & \Phi_{xy} & \Phi_{xz} & \Phi_{yy} & \Phi_{yz} & \Phi_{zz} \\
            \begin{block}{(cccccc)}
                -\Phi_{xx} & 0 & 0 & 1 & 0 & 0 \\
                0 & 1 & 0 & 0 & 0 & 0 \\
                0 & 0 & y & 0 & 0 & 1 \\
            \end{block}
        \end{blockarray}\,,
    \end{align*}
    where we wrote the second-order jet variables over the corresponding columns to emphasize the ordering prescription we chose. The first step to bringing this symbol into solved form is to regroup the columns according to the multi-index classes in descending order. This leads to the \emph{equivalent} symbol
    \begin{align*}
        \begin{blockarray}{cccccc}
            \Phi_{zz} & \Phi_{yy} & \Phi_{yz} & \Phi_{xx} & \Phi_{xy} & \Phi_{xz} \\
            \begin{block}{(c|cc|ccc)}
                0 & 1 & 0 & -\Phi_{xx} & 0 & 0 \\
                0 & 0 & 0 & 0 & 1 & 0 \\
                1 & 0 & 0 & 0 & 0 & y \\
            \end{block}
        \end{blockarray}\,,
    \end{align*}
    where the vertical lines indicate the boundaries between multi-index classes. We emphasize that both symbols contain exactly the same information, only arranged differently. Once the symbol is ordered correctly, we apply the Gauss algorithm to bring it into row-echelon form. After elementary algebraic manipulations, we obtain
    \begin{align*}
        \begin{blockarray}{cccccc}
            \Phi_{zz} & \Phi_{yy} & \Phi_{yz} & \Phi_{xx} & \Phi_{xy} & \Phi_{xz} \\
            \begin{block}{(c|cc|ccc)}
                \red{1} & 0 & 0 & 0 & 0 & y \\ \BAhhline{------}
                0 & \red{1} & 0 & -\Phi_{xx} & 0 & 0 \\ \BAhhline{------}
                0 & 0 & 0 & 0 & \red{1} & 0 \\
            \end{block}
        \end{blockarray}\,.
    \end{align*}
    We indicate the pivots in red, just as in~\eqref{eq:SymbolInRowEchelonForm}.
    The horizontal lines indicate where the class of a row ends and the next one begins. In this form, it is easy to read off the $\beta$'s. We find
    \begin{align*}
        \beta^{(3)}_2 = 1, \quad \beta^{(2)}_2 = 1, \quad \beta^{(1)}_2 = 1\,.
    \end{align*}
\end{example}\medskip

In our effort to define physical degrees of freedom, prolongations of differential equations play an important role. How does the solved form of a symbol change under prolongation? To answer this question, we first note that the same linear operations used to bring the symbol into solved form can also be applied to the differential equations in $\R_q$. Therefore, we can assume, without loss of generality, that any $\R_q$ we consider automatically gives us a symbol $\S_q$ in solved form.

Prolonging $\R_q$ simply means taking derivatives with respect to $x^{1}, x^{2}, \dots, x^{n}$ of each equation in $\R_q$. To see what effect this has on $\S_q$, we consider a first order equation $\R_1$ in solved form in three dimensions. Its associated symbol $\S_1$ then has the schematic form\footnote{We use the reversed lexicographic ordering.}
    \begin{align}
       \S_1 \sim \begin{blockarray}{c|c|c}
           \Phi_{z} & \Phi_{y} & \Phi_{x} &\\
            \begin{block}{c|c|c}
                \red{\bullet_z} & *_{z_1} & *_{z_2} \\ \BAhhline{---}
                0 & \red{\bullet_y} & *_{y_1} \\ \BAhhline{---}
                0 & 0 & \red{\bullet_x}\\
            \end{block}
        \end{blockarray}\,.
    \end{align}
The red bullets indicate the pivot elements and we index them by $x$, $y$, and $z$ in order to see to which class they belong to. Similarly, we index the entries to the right of the pivots by $z_1$, $z_2$, and $y_1$.

In a first step, we prolong each row of $\R_1$ by coordinates which are of lower or equal class as the row itself. That is, the first row is prolonged by $z$, $y$, and $x$, since it is of class $z$. The second row is prolonged by $y$ and $x$, since it is of class $y$. Finally, the third row is prolonged only by $x$. These operations give rise to the following prolonged symbol $\S_2$\footnote{This is of course not the full symbol, since we did not prolong each row by all coordinates.}
    \begin{align}\label{eq:ProlongSymbol1}
        \S_2 \sim \begin{blockarray}{c|cc|ccc}
           \Phi_{zz} & \Phi_{zy} & \Phi_{yy} & \Phi_{zy} & \Phi_{yx} & \Phi_{xx} &\\
            \begin{block}{c|cc|ccc}
                \red{\bullet'_z} & *'_{z_1} & 0 & *_{z_2} & 0 & 0 \\
                0 & \red{\bullet'_z} & *'_{z_1} & 0 & *'_{z_2} & 0 \\ 
                0 & 0 & 0 & \red{\bullet'_z} & *'_{z_1} & *'_{z_2}\\ \BAhhline{------}
                0 & 0 & \red{\bullet'_y} & 0 & *'_{y_1} & 0   \\
                0 & 0 & 0 & 0 & \red{\bullet'_y} & *'_{y_1} \\ \BAhhline{------}
                0 & 0 & 0 & 0 & 0 & \red{\bullet'_x} \\
            \end{block}
        \end{blockarray}\,.
    \end{align}
Here, the primes in $\red{\bullet'}$ and $*'$ remind us that these entries were obtained from the unprimed ones by taking a derivative. By looking at the above schematic prolongation of $\R_1$, we see a clear pattern: prolonging a row of class $k$ by $x^1$, $x^2$, $\dots$, $x^k$ creates $k$ new rows which preserve the row-echelon form. In fact, if we swap the third and the fourth row in the above prolonged symbol matrix, which is of course a valid operation, we obtain again a row-echelon form. 

The same is not true if a row of class $k$ is prolonged by the variables $x^{k+1}$, $x^{k+2}$, $\dots$, $x^n$. This creates $n-k$ new rows, but they are in general not in row-echelon form. It is precisely these rows which force us to re-compute the row-echelon form of a prolonged symbol. This insight is more important than it might seem at first glance, and it leads us to introduce the notion of multiplicative variables.

\begin{definition}[Multiplicative variables]\label{def:ClassPreservingVariables}
    Given a symbol $\S_q$ in solved form, we associate with a row of class $k$ the \textbf{multiplicative variables} $x^{1}$, $x^2$, $\dots$, $x^k$.
\end{definition}

As we will see in the next section, there are special types of equations, known as \emph{involutive equations}, where all independent equations at any prolongation order can be obtained by prolonging each equation only with respect to its multiplicative variables. This offers considerable advantages in constructing formal power series solutions to these equations, since the row-echelon form does not need to be recomputed at each order.

\subsection{Involutive Equations}\label{ssec:InvolutivePDEs}
In the previous subsection, we mentioned that we can always assume a symbol to be in solved form without loss of generality. This is because any linear operation required to bring a generic symbol into its solved form can also be applied to $\R_q$, from which the symbol is derived. Consequently, throughout this subsection, we will always assume that any symbol $\S_q$ we encounter is in solved form.

We also briefly introduced the concept of multiplicative variables. We will now explore this concept further and derive an important consequence from it. The central question that guides our discussion is: \textit{Can we devise a simple method for determining all independent equations resulting from the prolongations of $\R_q$?} Developing such a method would be highly beneficial for constructing formal power series solutions to $\R_q$.

To answer this question, let us analyze how the symbol changes when $\R_q$ is prolonged to $\R_{q+1}$. From our discussion of multiplicative variables in the previous subsection, we know that prolongations have two distinct effects: By examining the symbol, as we did in~\eqref{eq:ProlongSymbol1}, we observe that prolongation by multiplicative variables maintains the row-echelon form, thereby producing a new set of linearly independent higher-order equations. However, the same does \textit{not} hold for non-multiplicative variables: in this case, the row-echelon form is disrupted and must be recalculated for the new symbol $\S_{q+1}$. As a result, we cannot easily determine the independent equations of $\R_{q+r}$ for all orders $r>0$.

How many independent equations arise when we prolong $\R_q$ by its multiplicative variables? In the previous subsection, we introduced the characters $\beta^{(k)}_q$ to count the number of equations of class $k$ contained in $\R_q$. It follows that if we prolong $\R_q$ only by its multiplicative variables, the number of independent equations in $\R_{q+1}$ is given by
\begin{align}\label{eq:LowerBound}
    \sum_{k=1}^{n} k \beta^{(k)}_q\,,
\end{align}
since each class $k$ contains precisely $k$ multiplicative variables. This quantity provides a lower bound on the total number of independent equations in~$\R_{q+1}$.

However, to fully prolong $\R_q$, we must also consider the non-multiplicative variables. This requires recomputing the row-echelon form of $\R_{q+1}$, a process that may reveal additional independent equations. There are, nonetheless, special cases in which no new equations emerge. Specifically, when the rank of $\S_{q+1}$ coincides with the lower bound in~\eqref{eq:LowerBound}, prolongation by non-multiplicative variables yields no further independent equations.

This equivalence arises because $\rank \S_{q+1}$ counts the total number of independent rows in $\S_{q+1}$, which is precisely the number of independent equations of order $q+1$ in~$\R_{q+1}$. Therefore, if
\begin{align*}
    \rank \S_{q+1} = \sum_{k=1}^{n} k \beta^{(k)}_q\,,
\end{align*}
it is sufficient to prolong $\R_q$ solely by its multiplicative variables. In all other cases, prolongation by non-multiplicative variables becomes necessary, making it essential to recompute the row-echelon form of $\R_{q+1}$. This procedure may uncover additional independent equations of order $q+1$.

This observation is fundamental and merits its own formal definition:
\begin{definition}[Involutive symbol]\label{def:InvolutiveSymbol}
    A symbol $\S_q$ is said to be \textbf{involutive} if its rank satisfies
    \begin{align*}
        \rank \S_{q+1} = \sum_{k=1}^{n} k \beta^{(k)}_q\,.
    \end{align*}
    In this case, the quantity $\sum_{k=1}^{n} k \beta^{(k)}_q$ equals the total number of independent equations of order $q+1$ contained in~$\R_{q+1}$.
\end{definition}

This definition captures the essence of our discussion in the preceding paragraphs. It is important to note, however, that it concerns only the independent equations of order $q+1$. It does not address equations of lower order, which may arise as integrability conditions (see Corollary~\ref{cor:CriterionForIntCond}).

Moreover, the definition applies specifically to the symbol. In practice, however, we are typically more interested in the differential equation $\R_q$ itself, from which the symbol $\S_q$ is derived. We therefore extend the notion of involutivity as follows:
\begin{definition}[Involutive equation]\label{def:InvolutiveEquation}
    An equation $\R_q$ is said to be \textbf{involutive} if it is formally integrable and if its symbol is involutive.
\end{definition}

This extended definition includes a crucial addition: the requirement that $\R_q$ be formally integrable. This condition matters for two key reasons. First, a formally integrable equation includes all of its integrability conditions. If its symbol is also involutive, then it is possible to determine all independent equations of order $q+r$ for any prolongation order $r \geq 1$. In other words, involutive equations behave predictably under prolongation, which makes it particularly straightforward to construct formal power series solutions.

Second, this condition highlights that formal integrability and involutivity are logically independent. An equation may be formally integrable without being involutive. In this case, although all integrability conditions are contained in the equation, there is no systematic way to determine all principal derivatives at higher prolongation orders. Conversely, an equation $\R_q$ may possess an involutive symbol $\S_q$ without being formally integrable. In such a case, we can predict all principal derivatives to any order of prolongation, but we lack knowledge of the integrability conditions and, consequently, of all independent equations.

Both scenarios prevent us from systematically constructing formal power series solutions. This is why involutivity is such a fundamental concept.

To make this discussion more concrete, we will now illustrate one of these possibilities with an explicit example. We consider an equation that turns out to be formally integrable, but with a non-involutive symbol.

\begin{example}[Formal integrability does not imply involutivity]
Consider the case of a two dimensional manifold $\M$, coordinatized by $x$ and $y$. Let $\Phi$ be a scalar field satisfying the second-order PDE
\begin{align*}
    \R_2 :
    \begin{cases}
        \Phi_{xx} = 0 \\
        \Phi_{yy} = 0
    \end{cases}\,.
\end{align*}
This system is formally integrable, as it contains all its integrability conditions. Indeed, it is straightforward to determine its general solution:
\begin{align*}
    \Phi(x,y) = A + x B + y C + x y D\,,
\end{align*}
where $A$, $B$, $C$, and $D$ are real constants. The system is thus clearly integrable. However, it is not involutive.

To see this, note that the system contains one equation of class $1$ ($\Phi_{xx} = 0$) and one equation of class $2$ ($\Phi_{yy} = 0$), so that $\beta^{(1)}_2 = 1$ and $\beta^{(2)}_2 = 1$. The prolongation of the system reads
\begin{align*}
    \R_3 :
    \begin{cases}
        \Phi_{xxx} = 0 \\
        \Phi_{xxy} = 0 \\
        \Phi_{xyy} = 0 \\
        \Phi_{yyy} = 0 \\
        \\
        \Phi_{xx} = 0 \\
        \Phi_{yy} = 0
    \end{cases}
\end{align*}
and gives rise to the symbol
\begin{align*}
    \S_3 =
    \begin{pmatrix}
        1 & 0 & 0 & 0 \\
        0 & 1 & 0 & 0 \\
        0 & 0 & 1 & 0 \\
        0 & 0 & 0 & 1
    \end{pmatrix}\,.
\end{align*}
This symbol has rank $4$. However, the sum of the characters yields
\begin{align*}
    \sum_{k=1}^2 k \, \beta^{(k)}_2 = 1 \times 1 + 2 \times 1 = 3\,.
\end{align*}
Therefore,
\begin{align*}
    \rank \S_3 = 4 \neq 3 = \sum_{k=1}^2 k \, \beta^{(k)}_2\,.
\end{align*}
The system $\R_2$ is thus not involutive, despite being formally integrable. This means that, although we know all integrability conditions to any order of prolongation, we cannot predict the principal derivatives solely by prolongation with respect to multiplicative variables.

Indeed, if we prolong $\R_2$ only by its multiplicative variables, we obtain the independent equations
\begin{align*}
    \Phi_{xxx} &= 0  & &\text{(from prolonging $\Phi_{xx}$ with respect to $x$)} \\
    \Phi_{xyy} &= 0  & &\text{(from prolonging $\Phi_{yy}$ with respect to $x$)} \\
    \Phi_{yyy} &= 0  & &\text{(from prolonging $\Phi_{yy}$ with respect to $y$)}\,.
\end{align*}
To generate the equation $\Phi_{xxy} = 0$, we need to prolong the class $1$ equation $\Phi_{xx} = 0$ by $y$, which is not a multiplicative variable.
\end{example}

This example clearly illustrates that formal integrability does not imply involutivity. However, due to the simplicity of the system, it does not exhibit the need to re-compute the row-echelon form of $\S_{q+1}$ upon prolongation. This technical aspect becomes more relevant in more complex systems and should not detract from the significance of involutivity.

Equations whose symbols are involutive possess remarkable structural properties, even when they are not formally integrable. These properties are summarized in the following theorem (for a proof see~\cite{SeilerBook}).

\begin{theorem}[Involutive symbols and their consequences for prolongations]\label{thm:ConsequencesOfInvolutiveSymbols}
    Let $\S_q$ be the involutive symbol of an equation $\R_q$. Then the following statements hold:
    \begin{itemize}
        \item[(i)] The symbol $\S_{q+1}$ of the prolonged equation $\R_{q+1}$ is also involutive.
        \item[(ii)] All integrability conditions of $\R^{(1)}_{q+1}$ are obtained by prolonging the integrability conditions of $\R^{(1)}_q$, that is, $\R^{(1)}_{q+1} = (\R^{(1)}_q)_{+1}$.
    \end{itemize}
\end{theorem}

The significance of this theorem is twofold. First, it requires only that $\S_q$ be involutive; no assumption about the formal integrability of $\R_q$ is necessary. Second, if an equation $\R_q$ gives rise to an involutive symbol, then the symbols of all its prolongations are involutive as well. This guarantees that we can systematically identify all principal derivatives of $\R_q$ at any order of prolongation—a fact anticipated in our earlier discussion.

Part~(ii) of the theorem further asserts that, if the symbol is involutive, then any integrability conditions that arise at order $q+1$ are merely prolongations of those already present at order $q$. In particular, we obtain the following corollary:

\begin{corollary}[Criterion of involution]\label{cor:CriterionOfInvolutivity}
    The equation $\R_q$ is involutive if and only if its symbol $\S_q$ is involutive and $\R^{(1)}_q = \R_q$.
\end{corollary}

Recall from Definition~\ref{def:FormalIntegrability} that formal integrability requires $\R^{(1)}_{q+r} = \R_{q+r}$ for all $r \geq 0$. This is an infinite sequence of conditions, which is typically impractical to verify. However, if the symbol $\S_q$ is involutive, this sequence collapses to a single, simple condition: checking whether $\R^{(1)}_q = \R_q$ suffices to decide formal integrability. This highlights the power of involutive systems.

At first glance, one might be discouraged by the impression that involutive equations are rare. Theorem~\ref{thm:ConsequencesOfInvolutiveSymbols} and Corollary~\ref{cor:CriterionOfInvolutivity} might seem irrelevant to physical systems if these systems are not involutive. For instance, Proca's equations are not involutive, and therefore the methods developed here cannot be directly applied to construct formal power series solutions or to count degrees of freedom.

However, this concern is unfounded. In fact, any equation $\R_q$ can be \emph{completed} to an equivalent equation that is involutive. Involutive equations are therefore not an exception, but rather the rule. This fundamental result, due to Cartan and Kuranishi, is the subject of the next subsection.

\subsection{The Cartan-Kuranishi Algorithm}\label{ssec:CKAlgorithm}
Let us pause to recapitulate the steps we have taken so far and to summarize the assumptions underlying our construction. Our objective has been to start from a set of PDEs of order $q$, to construct a formal power series solution to these equations, and to use this construction to count the physical degrees of freedom described by the system.

This led us to a detailed study of PDEs in the jet bundle formalism. By employing jet bundles, we shifted our perspective: PDEs are now understood as (possibly non-linear) algebraic equations imposing constraints on jet variables. From this viewpoint, the space of all $q$-th order jet variables can be treated as a manifold, and a system of PDEs is interpreted as defining a submanifold of this jet space.

However, this interpretation implicitly relies on an important assumption, which we identified in Subsection~\ref{ssec:JetBundlePDEs}: Not every equation $\R_q$ defines a submanifold of $J_q\E$ in a well-behaved way. More precisely, the set of all possible $q$-th order PDEs can be subdivided into two categories: those which define fibered submanifolds of $J_q\E$, and those which require a case distinction and therefore fail to define such submanifolds.

Crucially, we observed that restricting ourselves to equations defining fibered submanifolds does not entail a significant loss of generality. In fact, for any equation $\R_q$ which does \textit{not} define a fibered submanifold, it is always possible to impose additional conditions---i.e., to append further equations---such that the resulting system does describe a fibered submanifold.

For this reason, we can, without loss of generality, restrict our attention to equations $\R_q$ that define fibered submanifolds, as formalized in Definition~\ref{def:Rq}. Nevertheless, even when $\R_q$ describes a fibered submanifold, its prolongations may fail to do so. We encountered this phenomenon explicitly in Example~\ref{ex:NonFiberedProlongation}. Once again, this issue can be resolved without sacrificing generality: by imposing suitable additional conditions, we can ensure that all prolongations of $\R_q$ describe fibered submanifolds. This leads us to the notion of \emph{regular equations} (cf.~Definition~\ref{def:RegularEquation}).

In summary, we may consistently and without loss of generality focus our analysis on regular equations, whose prolongations always describe fibered submanifolds.
As our discussion of involutive symbols and involutive equations has shown, the set of regular equations admits a further subdivision, which is schematically represented in Figure~\ref{fig:Summary2}. Specifically, we can distinguish four distinct subsets within the set of regular equations:

\begin{itemize}
    \item Equations which are formally integrable (blue region).
    \item Equations whose symbol $\S_q$ is involutive (red region).
    \item Equations which are both formally integrable and possess an involutive symbol (green region), i.e., involutive equations in the strict sense.
    \item Equations which are neither formally integrable nor possess an involutive symbol (yellow region surrounding the other regions).
\end{itemize}

This visualization reinforces the impression that involutive equations, which lie in the intersection of the blue and red regions, may at first appear to be the exception rather than the rule.
\begin{figure}[htb!]
	\centering
	\includegraphics[width=0.75\linewidth]{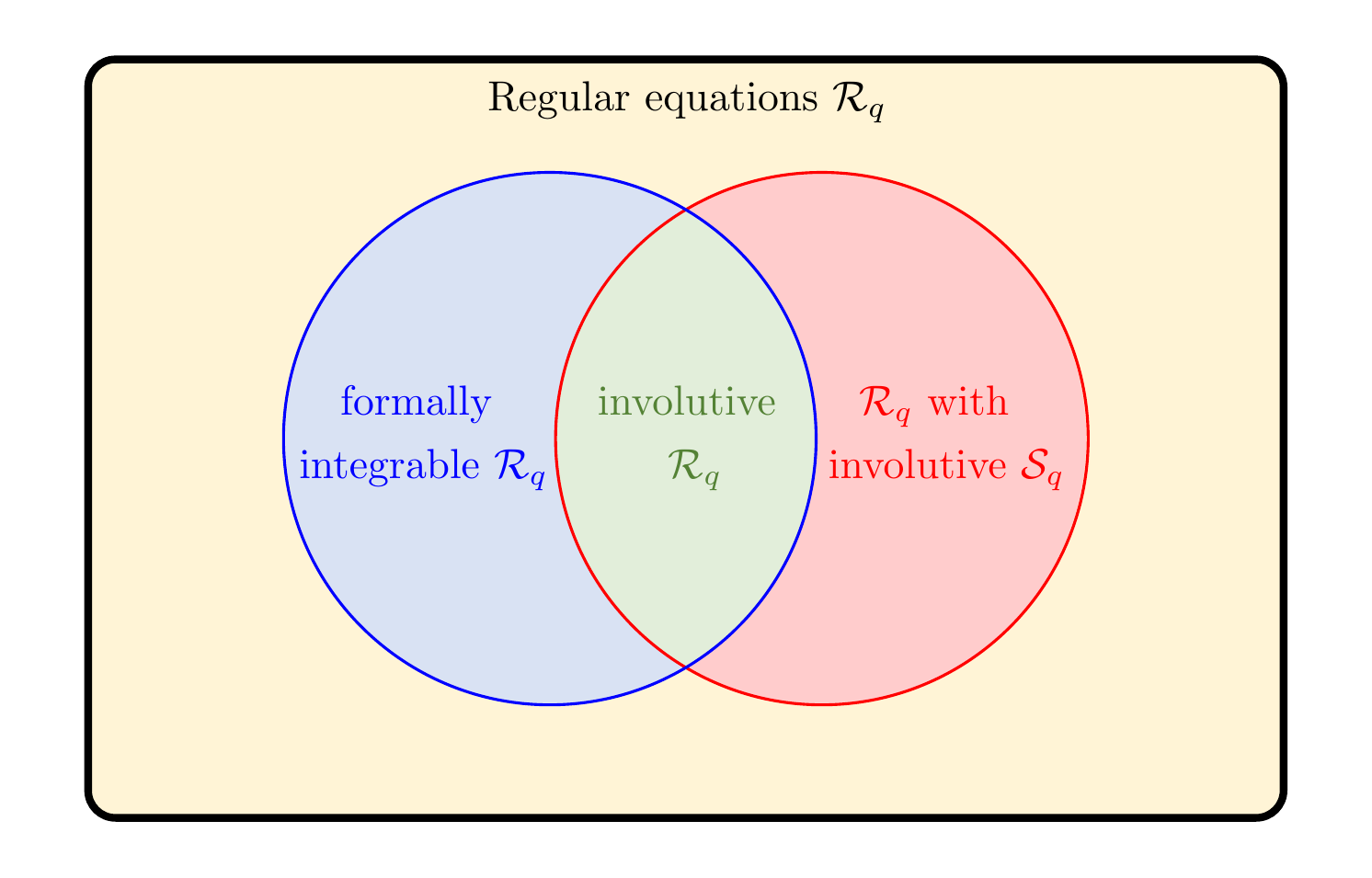}
	\caption{\textit{Within the set of regular equations we can distinguish between equations which are formally integrable (blue set), equations which possess an involutive symbol (red set), and equations which have neither of these properties (yellow set). If an equation is formally integrable and it also possess an involutive symbol, it is called an involutive equation and it lies in the intersection set (green set).}}
	\label{fig:Summary2}
\end{figure}
The observation that involutive equations may seem exceptional is, fortunately, misleading. The seminal work of Cartan and Kuranishi reassures us that this is not the case. Their theorem establishes that any regular equation can be systematically completed to an involutive equation without changing its solution space. More precisely, by performing a finite number of prolongations and projections, one can always obtain an equivalent involutive equation.

This result is formalized in the following theorem:

\begin{theorem}[Cartan-Kuranishi]
    For every regular equation $\R_q$ there exist two integers $r$ and $s$ such that the equation $\R^{(s)}_{q+r}$ is an involutive equation which possesses the same solution space as the original equation $\R_q$.
\end{theorem}

This theorem has profound consequences. It ensures that the method of involutive completion is always applicable to regular equations and thus fully justifies our focus on involutive equations in the analysis of formal integrability and the construction of formal power series solutions. Figure~\ref{fig:VisualizationCKTheorem} illustrates how the Cartan-Kuranishi theorem maps any regular equation to an equivalent involutive equation.

\begin{figure}[htb!]
	\centering
	\includegraphics[width=0.75\linewidth]{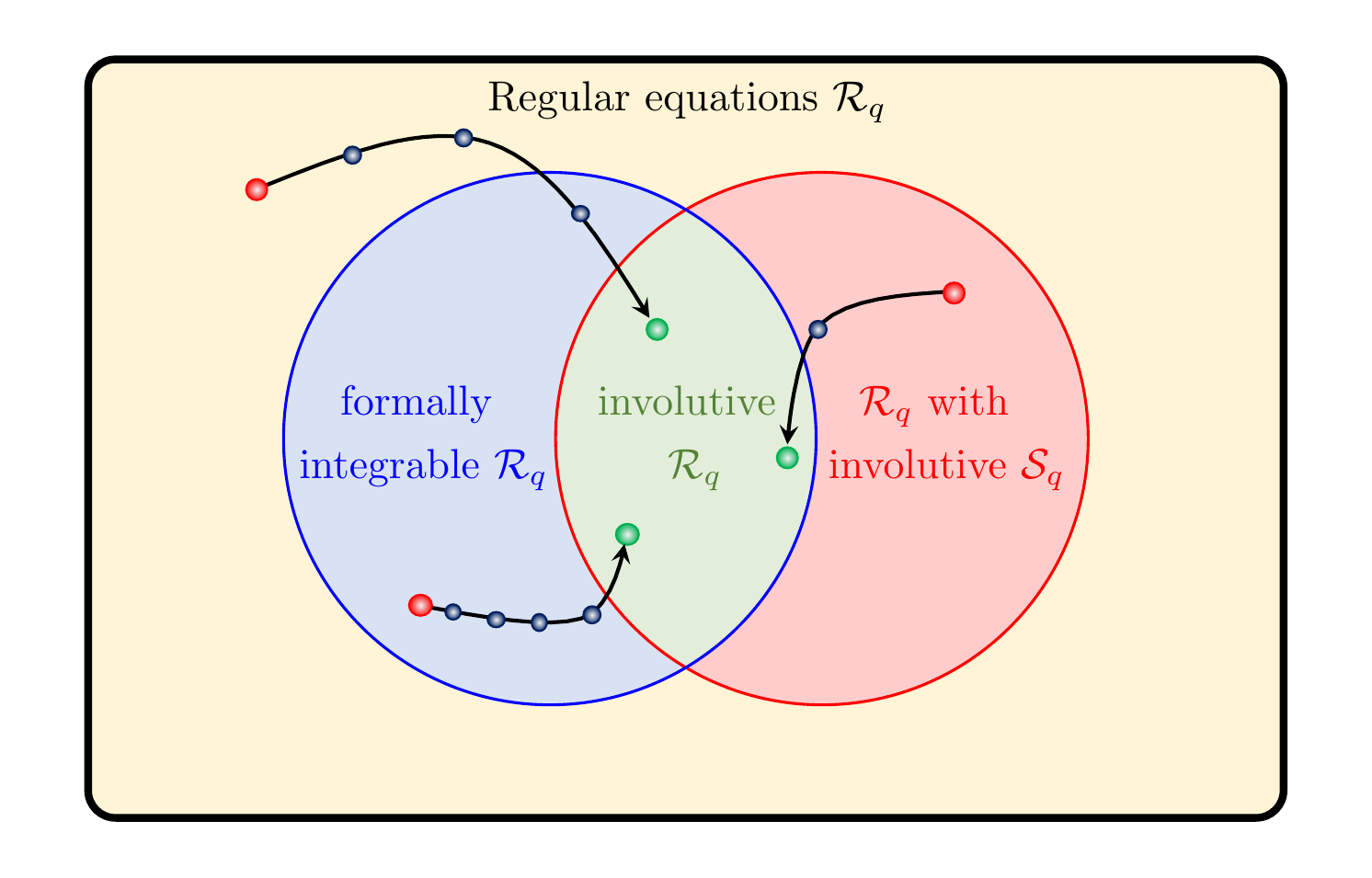}
	\caption{\textit{The Cartan–Kuranishi theorem guarantees that any regular equation (represented by red dots) can be transformed, after a finite sequence of prolongations and projections (indicated by blue dots), into an equivalent equation that is involutive (depicted by green dots within the green region). The procedure described by the Cartan–Kuranishi algorithm can be visualized as moving an equation $\R_q$ through the space of regular equations until it reaches the subset of involutive equations.}}
	\label{fig:VisualizationCKTheorem}
\end{figure}
Indeed, we can algorithmically complete any regular equation $\R_q$ to an equivalent involutive equation $\R^{(s)}_{q+r}$. That is, we can always construct an equation which contains all its integrability conditions, whose prolongations behave predictably, and for which we can systematically build a formal power series solution order by order. Crucially, since the Cartan-Kuranishi theorem guarantees that $\R_q$ and $\R^{(s)}_{q+r}$ possess the same solution space, this procedure effectively yields a formal power series solution to the original system.

In summary, we can always determine such a formal power series solution for any regular equation $\R_q$, because we can always complete the equation to an involutive one. Once this is achieved, we are in a position to systematically count the physical degrees of freedom encoded in the PDE system. This will be the subject of Section~\ref{sec:FormalPSSExtnsionToGT}.

To present the Cartan-Kuranishi algorithm, we employ pseudo-code and introduce the notation $\#\text{MV}(\S_q)$. This quantity, referred to as the ``number of multiplicative variables,'' is defined by
\begin{align}
    \#\text{MV}(\S_q) \ce \sum_{k=1}^n k\, \beta^{(k)}_q\,,
\end{align}
where the characters $\beta^{(k)}_q$ are determined from the symbol $\S_q$. Additionally, we use the notation $\S^{(s)}_{q+r}$ to denote the symbol associated with an equation $\R^{(s)}_{q+r}$ that has undergone $r$ prolongations and $s$ projections.

\begin{algorithm}[hb!]
\caption{(Cartan-Kuranishi)}\label{alg:CK}
Input: Equation $\R_q$
\begin{algorithmic}[1]
\State Set $r=0$ and $s=0$ (initialize the counters for prolongations and projections)
\State Compute the prolongation $\R_{q+1}$
\State Compute the symbols $\S_q$ and $\S_{q+1}$
\Repeat\{ \hfill \textcolor{myGreen}{$\longleftarrow$ (Start outer loop)} 
    \Repeat\{ \hfill \textcolor{red}{$\longleftarrow$ (Start inner loop: Make symbol involutive)}
    \State r = r+1 (counter for prolongations);
    \State Compute the prolongation $\R^{(s)}_{q+r+1}$;
    \State Extract the symbol $\S^{(s)}_{q+r+1}$;
    \Until{$\#\text{MV}(\S^{(s)}_{q+r})$ = $\rank\S^{(s)}_{q+r+1}$}\} \hfill \textcolor{red}{$\longleftarrow$ (End inner loop: Symbol is involutive)}
    \If {$\dim \R^{(s)}_{q+r+1} - \dim\S^{(s)}_{q+r+1} < \dim\R^{(s)}_{q+r}$} \hfill \textcolor{blue}{$\longleftarrow$ Start check for int. conditions}
        \State s = s+1 (counter for projections);
        \State Compute the projections $\R^{(s)}_{q+r}$ and $\R^{(s)}_{q+r+1}$;
        \State Extract the symbols $\S^{(s)}_{q+r}$ and $\S^{(s)}_{q+r+1}$;
    \EndIf \State \textbf{End if} \hfill \textcolor{blue}{$\longleftarrow$ Integrability conditions identified and added}
\Until{$R^{(s)}_{q+r}$ is involutive}\} \hfill \textcolor{myGreen}{$\longleftarrow$ (End outer loop)}
\State Return $\R^{(s)}_{q+r}$
\end{algorithmic}
Output: Involutive equation $\R^{(s)}_{q+r}$ with the same solution space as $\R_q$
\end{algorithm}

Let us dissect the Cartan-Kuranishi algorithm step by step and subsequently illustrate it with a concrete example of physical relevance. The algorithm takes as input an equation $\R_q$ and, after a finite number of prolongations $r$ and projections $s$, produces an involutive equation $\R^{(s)}_{q+r}$ that possesses the same solution space as the original system. The algorithm is structured around two key criteria:
\begin{itemize}
    \item Criterion for an involutive symbol (line 9):
    \begin{align}
        \#\text{MV}(\S_q) = \sum_{k=1}^n k\,\beta^{(k)}_q = \rank \S_{q+1}
    \end{align}
    (cf. Definition~\ref{def:InvolutiveSymbol} and discussion around it)
    \item Criterion for the presence of integrability conditions (line 10):
    \begin{align}
        \dim \R_{q+1} - \dim \S_{q+1} < \dim \R_{q}
    \end{align}
    (cf. Corollary~\ref{cor:CriterionForIntCond} and discussion around it)
\end{itemize}
These criteria reflect the fact that an involutive equation contains all its integrability conditions and possesses an involutive symbol. Consequently, the algorithm is designed to verify both conditions. To start, we require knowledge of $\R_{q+1}$ as well as the symbols $\S_q$ (which determines the characters $\beta^{(k)}_q$) and $\S_{q+1}$. This explains the purpose of the first three steps of the algorithm.

Once the prolongation $\R_{q+1}$ and the symbols $\S_q$ and $\S_{q+1}$ are determined, we can bring $\S_q$ into row-echelon form, deduce the values of $\beta^{(k)}_q$, compute the rank of $\S_{q+1}$, and evaluate the dimensions of $\R_{q}$, $\R_{q+1}$, and $\S_{q+1}$ (see equation~\eqref{eq:DimOfSq} and Definition~\eqref{def:DimRq}). We can then assess whether the two criteria are satisfied. If the equation fails to meet either condition, the algorithm prompts specific operations (prolongations and projections) to remedy this.

For clarity, it is useful to divide the algorithm into three logical blocks and analyze their roles:
\begin{itemize}
    \item The first block is an outer loop, running from line~$4$ to line~$15$. Its purpose is to iterate the other two blocks until we arrive at an involutive equation $\R^{(s)}_{q+1}$. Importantly, the two inner blocks are not independent of each other.
    \item The second block is an inner loop, spanning lines~$5$ to $9$. Here, we check whether the symbol is involutive. If it is not, we must prolong $\R_q$ until an involutive symbol is obtained. However, this block is not independent of the third block.
    \item The third block covers lines~$10$ to $14$. It checks for hidden integrability conditions. If none are present (i.e., the condition on line~$10$ evaluates to \texttt{False}), there are no consequences for the second block. If integrability conditions are found, we perform a projection from the last prolongations, effectively adding the integrability conditions to lower-order equations. Since this can alter the symbols, we must return to the second block and reassess whether the symbols remain involutive.
\end{itemize}

To understand how these blocks interact, let us examine three illustrative cases:

In the first case, we start with an equation $\R_q$ and execute the initial steps of the algorithm. We then enter the second block and check whether $\S_q$ is involutive. If it is, no prolongation is required (so $r = 0$), and we proceed to the third block. If there are no integrability conditions, the third block is skipped (so $s = 0$), and the algorithm terminates. We conclude that $\R_q$ is involutive.

In the second case, we assume that $\S_q$ is not involutive. That is, $\#\text{MV}(\S_q) \neq \rank \S_{q+1}$. This requires us to execute the second block: we increment $r$ from $0$ to $1$, compute the prolongation $\R_{q+2}$, extract the symbol $\S_{q+2}$, and check whether $\#\text{MV}(\S_{q+1}) = \rank \S_{q+2}$. If this condition is satisfied, we proceed to the third block. There, we evaluate whether
\begin{align}
    \dim \R_{q+2} - \dim \S_{q+2} < \dim \R_{q+1}.
\end{align}
Suppose this condition evaluates to \texttt{False}. In that case, the third block is skipped, and the algorithm terminates. The output is the involutive equation $\R_{q+1}$.

In the third and final case, everything proceeds as in the second case, except that the condition in the third block now evaluates to \texttt{True}. We then increment $s$ from $0$ to $1$ and compute the projections $\R^{(1)}_{q+1}$ and $\R^{(1)}_{q+2}$, along with the corresponding symbols $\S^{(1)}_{q+1}$ and $\S^{(1)}_{q+2}$. Next, on line~$15$, we must determine whether $\R^{(1)}_{q+1}$ is involutive. This requires checking whether the symbol $\S^{(1)}_{q+1}$ is involutive. If it is, the third block can be skipped (since the integrability conditions have already been incorporated), and the algorithm terminates. We then conclude that $\R^{(1)}_{q+1}$ is an involutive equation.

These simple case studies illustrate the interaction between the algorithm's blocks. Of course, more intricate scenarios are possible. Nevertheless, even these basic examples may seem somewhat abstract. To make the procedure more tangible and to develop familiarity with the Cartan-Kuranishi algorithm, we will next apply it to the Proca equation.
\begin{example}[Completion of the Proca equation to an involutive system]\label{ex:CKAAppliedToProca}
    We have already used the Proca equation on several occasions to illustrate key concepts. This is advantageous here, as we can refer back to those earlier examples for intermediate results needed to demonstrate the Cartan-Kuranishi algorithm. For the sake of brevity, we omit detailed computations and focus on the essential steps and results.

    Our starting point is the Proca equation in $n = 4$ dimensions with $m = 4$ vector field components. The second-order ($q = 2$) equation reads
    \begin{align*}
        \R_2 :  
        \begin{cases}
            \partial^\nu \partial_\mu A^\mu - \partial_\mu \partial^\mu A^\nu + m^2 A^\nu = 0\,.
        \end{cases}
    \end{align*}
    We now proceed step by step through the Cartan-Kuranishi algorithm (the numbering below is unrelated to that used in Algorithm~\ref{alg:CK}).
    \begin{enumerate}
        \item Initialize $r = 0$ and $s = 0$.
        \item Compute the prolongation $\R_3$ (as done in Example~\ref{ex:R1qNotEqualR1}).
        \item Determine the symbols $\S_2$ and $\S_3$ from $\R_2$ and $\R_3$, respectively (see Example~\ref{ex:SymbolMaxwellProca}).
        \item Bring $\S_2$ into row-echelon form and extract the characters $\beta^{(k)}_2$. We find
        \begin{align*}
            \beta^{(1)}_2 = 0\,, \quad \beta^{(2)}_2 = 0\,, \quad \beta^{(3)}_2 = 1\,, \quad \beta^{(4)}_2 = 3\,.
        \end{align*}
        Consequently,
        \begin{align*}
            \#\text{MV}(\S_2) = \sum_{k=1}^4 k\, \beta^{(k)}_2 = 3 \times 1 + 4 \times 3 = 15\,.
        \end{align*}
        \item Compute the rank of $\S_3$, obtaining $\rank \S_3 = 15$.
        \item Since $\#\text{MV}(\S_2) = \rank \S_3 = 15$, the symbol is involutive, and no further prolongations are needed. We proceed to check for integrability conditions.
        \item From Example~\ref{ex:DimR12ProcaAndMaxwell}, we already know that
        \begin{align*}
            \dim \R_2 = 56 \quad \text{and} \quad \dim \R_3 - \dim \S_3 = 55\,,
        \end{align*}
        indicating the presence of one integrability condition. We therefore execute the third block of the algorithm.
        \item Set $s = 1$ to account for the projection. The resulting projected system $\R^{(1)}_2$ was computed in Example~\ref{ex:R1qNotEqualR1}:
        \begin{align*}
            \R^{(1)}_2: 
            \begin{cases}
                \partial_\mu \partial^\mu A^\nu - m^2 A^\nu = 0\\[0.5em]
                \partial_\mu A^\mu = 0
            \end{cases}\,.
        \end{align*}
        Additionally, compute the prolongation $\R^{(1)}_3$ and the corresponding symbols $\S^{(1)}_2$ and~$\S^{(1)}_3$.
        \item Since the third block was executed, we return to the second block and verify whether the symbol $\S^{(1)}_2$ remains involutive. We find that the $\beta$'s are unchanged and that $\rank \S^{(1)}_3 = \rank \S_3$. Therefore, $\S^{(1)}_2$ is still involutive.
        \item Next, we re-examine the presence of integrability conditions. We compute
        \begin{align*}
            \dim \R^{(1)}_2 = 55 \quad \text{and} \quad \dim \R^{(1)}_3 - \dim \S^{(1)}_3 = 51\,.
        \end{align*}
        Since these numbers differ, there are $55 - 51 = 4$ additional integrability conditions to be incorporated.
        \item To account for them, increase $s$ from $1$ to $2$ and compute the projection $\R^{(2)}_2$. According to Definition~\ref{def:Projection}, this projection is defined by $\R^{(2)}_2 \ce \pi^{2+2}_2(\R_2)$, meaning we prolong $\R_2$ twice and retain only equations of order $2$ or lower. Using the known prolongation $\R_3$, we obtain
        \begin{align*}
            \R^{(2)}_2 :
            \begin{cases}
                \partial_\mu \partial^\mu A^\nu - m^2 A^\nu = 0\\[0.5em]
                \partial_\nu \partial_\mu A^\mu = 0 \\[0.5em]
                \partial_\mu A^\mu = 0
            \end{cases}\,.
        \end{align*}
        Observe that the four missing integrability conditions correspond to the prolongations of $\partial_\mu A^\mu = 0$, i.e., they are given by $\partial_\nu \partial_\mu A^\mu = 0$.
        \item Verify whether the new symbol $\S^{(2)}_2$ is involutive. This time, the presence of additional second-order equations alters the characters:
        \begin{align*}
            \beta^{(1)}_2 = 1\,, \quad \beta^{(2)}_2 = 1\,, \quad \beta^{(3)}_2 = 2\,, \quad \beta^{(4)}_2 = 4\,.
        \end{align*}
        From these, we compute
        \begin{align*}
            \sum_{k=1}^{4} k\, \beta^{(k)}_2 = 1 \times 1 + 2 \times 1 + 3 \times 2 + 4 \times 4 = 25\,,
        \end{align*}
        and verify that
        \begin{align*}
            \rank \S^{(2)}_3 = 25\,.
        \end{align*}
        Thus, the symbol $\S^{(2)}_2$ is involutive.
        \item Finally, we check for remaining integrability conditions:
        \begin{align*}
            \dim \R^{(2)}_2 = 51 \quad \text{and} \quad \dim \R^{(2)}_3 - \dim \S^{(2)}_3 = 51\,.
        \end{align*}
        The equality confirms that no further integrability conditions are present.
        \item The Cartan-Kuranishi algorithm terminates. The system $\R^{(2)}_2$, as given in step 11, is an involutive equation equivalent to the original Proca equation.
    \end{enumerate}
\end{example}

It may appear surprising that the final system $\R^{(2)}_2$ contains an apparently redundant equation. Indeed, as previously noted, the equation $\partial_\nu \partial_\mu A^\mu = 0$ is merely the prolongation of the integrability condition $\partial_\mu A^\mu = 0$ and does not seem to provide new information. While it is true that this second-order equation can be obtained by differentiating $\partial_\mu A^\mu = 0$, its explicit inclusion is essential. Only when this equation is part of the system does $\R^{(2)}_2$ attain the formal property of involutivity. As we have emphasized throughout the preceding sections, it is precisely this property that enables the systematic construction of formal power series solutions and the precise counting of degrees of freedom. From this perspective, including the seemingly redundant equation $\partial_\nu \partial_\mu A^\mu = 0$ is indispensable.

In the next section, we will finally explain how to construct formal power series solutions for involutive systems and how to count their degrees of freedom.

\section{Formal Power Series Solutions and Extension to Gauge Theories}\label{sec:FormalPSSExtnsionToGT}

Consider a generic field $v^{A}$, where $A$ labels the individual components, and fix a point $x_0$ in the $n$-dimensional base space $B$. We can formally expand $v^{A}$ around $x_0$ in a Taylor series of the form
\begin{align}\label{eq:FormalPowerSeries}
    v^{A}(x) &= \sum_{|\bfm|=0}^{\infty} \frac{p^{A}_{\bfm}(x_0)}{\bfm!}(x - x_0)^{\bfm}\,,
\end{align}
where the jet variables $p^{A}_{\bfm}(x_0)$ correspond to the Taylor coefficients. We adopt the standard multi-index notation
\begin{align}
    \bfm! &\ce m_1! m_2! \cdots m_n! \notag\\
    (x - x_0)^{\bfm} &\ce (x^1 - x^1_0)^{m_1} (x^2 - x^2_0)^{m_2} \cdots (x^n - x^n_0)^{m_n}\,.
\end{align}
For notational convenience, we set $p^{A}_{\bfm}(x_0) = v^{A}(x_0)$ when $|\bfm| = 0$.

This is a \emph{formal} power series in the sense that it is not required to converge. Its primary purpose is to provide insight into the equations $\R_q$ by constructing a formal solution. The procedure is analogous to the example of the scalar field discussed in Subsection~\ref{ssec:EinsteinsMethod}. Specifically, we substitute the formal power series~\eqref{eq:FormalPowerSeries} into $\R_q$ and evaluate the result at $x = x_0$. In local coordinates, this leads to a system of algebraic equations of the form
\begin{align}
    \R_q &:
    \begin{cases}
        \bbE^\tau(x_0, p^{A}_{\bfm}) = 0 & \text{with } 0 \leq |\bfm| \leq q\,.
    \end{cases}
\end{align}

When evaluating the series at $x = x_0$, all terms except the zeroth-order term vanish. However, since $\R_q$ involves derivatives of $v^{A}$ up to order $q$, finitely many terms of the series survive in the expression for $\bbE^\tau$ evaluated at $x = x_0$. This results in algebraic equations for finitely many jet variables $p^{A}_{\bfm}$. Typically, there are not enough equations to fully determine all the jet variables. Nevertheless, these equations impose relations among the $p^{A}_{\bfm}$, thereby constraining them. In geometric terms, the equation $\R_q$ defines a subspace of the jet bundle~$J_q \E$.

This shift in perspective is precisely what we introduced in Subsection~\ref{ssec:JetBundlePDEs}. Moreover, we are not limited to considering $\R_q$ alone. By employing the notion of prolongations, we can systematically generate higher-order equations $\R_{q+r}$, which in turn yield new algebraic constraints on additional jet variables. Schematically, this process takes the form
\begin{align}\label{eq:OrderByOrderAlgEq}
    \R_{q+1} &:
    \begin{cases}
        D_{\mu_1}\bbE^\tau(x_0, p^{A}_{\bfm}) = 0 & \text{with } 0 \leq |\bfm| \leq q + 1
    \end{cases} \notag\\
    \notag\\
    \R_{q+2} &:
    \begin{cases}
        D_{\mu_2} D_{\mu_1} \bbE^\tau(x_0, p^{A}_{\bfm}) = 0 & \text{with } 0 \leq |\bfm| \leq q + 2
    \end{cases} \notag\\
    \notag\\
    \vdots \notag\\
    \notag\\
    \R_{q+r} &:
    \begin{cases}
        D_{\mu_r} D_{\mu_{r-1}} \cdots D_{\mu_1} \bbE^\tau(x_0, p^{A}_{\bfm}) = 0 & \text{with } 0 \leq |\bfm| \leq q + r
    \end{cases}\,.
\end{align}
There is no reason to terminate this process at any finite order $r > 0$. We can continue indefinitely, thereby obtaining an infinite sequence of algebraic relations for an infinite set of Taylor coefficients.

The objective, however, is not to explicitly solve these algebraic equations and reconstruct the full power series for $v^{A}$. Rather, this formal procedure is used to analyze how much freedom remains in specifying a solution to $\R_q$. Qualitatively, the more Taylor coefficients that can be determined from $\R_q$ and its prolongations, the fewer degrees of freedom remain in the solution. Conversely, the more coefficients that remain unconstrained, the greater the freedom in specifying a solution.

As an illustrative example, consider an equation $\R_q$ that can be solved exactly and whose solution is a polynomial of degree $q$. In this case, the formal power series solution coincides with the polynomial itself. Consequently, only finitely many Taylor coefficients are nonzero, and these coefficients are simply constants. The freedom in specifying a solution thus amounts to choosing these constant coefficients, where each distinct choice yields a different solution.

This stands in sharp contrast to a more interesting case, such as Maxwell's equations. Here, we know that the freedom in specifying solutions goes beyond selecting constants. Gauge freedom implies that we can arbitrarily prescribe one component of the vector potential, meaning we can freely choose a function of four coordinates. Additionally, we can specify the initial data $(A^\mu|_{\Sigma}, \dot{A}^\mu|_{\Sigma})$ on a spacelike initial value surface $\Sigma$. This amounts to the freedom of choosing six functions of three coordinates. All of this is, of course, well-known. The connection to our current discussion is that, if we attempted to construct a formal power series solution of Maxwell's equations, this freedom would manifest as a much larger set of undetermined Taylor coefficients. Moreover, these coefficients would not merely be constants but functions of four, three, two, or even a single coordinate.

Our objective is now to develop a systematic way of quantifying this freedom in specifying a formal solution and relating it to the familiar notion of \emph{degrees of freedom}. To this end, we will restrict our attention to involutive equations $\R_q$. This is not a severe limitation, since any regular equation can be completed to an involutive one via the Cartan--Kuranishi Algorithm~\ref{alg:CK}. Furthermore, involutive equations possess many desirable properties that simplify the construction of formal power series solutions: they include all integrability conditions, and the parametric derivatives at any order of prolongation can be predicted systematically. This latter property will play a key role in the next subsection.

\subsection{Order-by-Order Construction, Cartan Characters, and Hilbert Polynomials}\label{ssec:OrderByOrder}

As stated above, our goal is to construct a formal power series solution to an involutive equation $\R_q$ order by order. To achieve this, we need to solve the algebraic equations appearing in~\eqref{eq:OrderByOrderAlgEq}. Each such equation determines some of the jet variables $p^{A}_{\bfm}$ in terms of others. A generic solution to one of these algebraic equations takes the form
\begin{align}\label{eq:solAlgEq}
    p^{A}_{\bfm} = f^{A}_{\bfm}(x^\mu, v^{A}, p^{A}_{\mathbf{n}})\,,
\end{align}
where the functions $f^{A}_{\bfm}$ arise naturally from solving the equations, and the $p^{A}_{\mathbf{n}}$ denote jet variables distinct from those on the left-hand side. This equation naturally distinguishes between principal and parametric derivatives (see Definition~\ref{def:PrincipalParametricDerivatives}): the jet variables on the left-hand side are principal derivatives, while those appearing inside $f^{A}_{\bfm}$ on the right-hand side are parametric derivatives.

The condition that a formal power series solution can be constructed order by order translates into the requirement
\begin{align}\label{eq:OrderByOrderCond}
    |\mathbf{n}| \leq |\bfm|\,.
\end{align}
In other words, at order $q+r = |\bfm|$, the principal derivatives can only depend on parametric derivatives of order less than or equal to $|\bfm|$. This has a crucial consequence: when we move to the next order---that is, when we prolong $\R_{q+r}$ to $\R_{q+r+1}$ and solve the resulting algebraic equations---the principal derivatives at order $q+r+1$ will depend only on principal derivatives already determined at lower orders and on lower-order parametric derivatives. Since the lower-order principal derivatives can themselves be expressed in terms of parametric derivatives via~\eqref{eq:solAlgEq}, it follows that the principal derivatives at order $q+r+1$ ultimately depend only on a subset of parametric derivatives, most of which already appear at order $q$. 

If the condition~\eqref{eq:OrderByOrderCond} is violated, however, it may happen that solving a higher-order equation forces us to revisit and re-compute principal derivatives at lower orders. In that case, an order-by-order construction of the solution becomes impossible.

Fortunately, involutive equations are structured precisely so that condition~\eqref{eq:OrderByOrderCond} is satisfied at all orders of prolongation. To see this, recall from Subsection~\ref{ssec:Symbol} that the symbol $\S_q$ of an equation $\R_q$ can always be brought into row-echelon form, where the pivot elements correspond to the principal derivatives and the non-pivot elements to the parametric derivatives. Since the same linear operations used to bring $\S_q$ into row-echelon form can be applied to $\R_q$ itself, we can systematically rewrite the nonlinear equations to isolate the principal derivatives. Moreover, when written in this \emph{solved form} (see Definition~\ref{def:SymbolSolvedForm}), condition~\eqref{eq:OrderByOrderCond} is automatically satisfied at order $q$.

If $\R_q$ is involutive, then it suffices to prolong it with respect to its multiplicative variables only. As discussed previously, the row-echelon form of $\S_q$ is preserved under this operation, allowing us to identify the pivot elements systematically at each order of prolongation. This ensures that condition~\eqref{eq:OrderByOrderCond} holds at every order $q+r$. In summary, the involutivity of $\R_q$ guarantees that an order-by-order construction of formal solutions is always possible.

To make our discussion more precise, we now compute how many principal derivatives appear at any order $q + r$. Recall that the number of principal derivatives of class $k$ appearing in $\R_q$ is measured by the characters $\beta^{(k)}_q$. Our task is therefore to predict the characters $\beta^{(k)}_{q + r}$ at higher orders in terms of the characters at order $q$. The following theorem provides the answer.

\begin{theorem}[Number of principal derivatives at order $q + r$]\label{thm:NoOfPrincipalD}
    If $\S_q$ is an involutive symbol with characters $\beta^{(k)}_q$, then the characters of its prolongations are given by
    \begin{align*}
        \beta^{(k)}_{q + r} = \sum_{i = k}^{n} \binom{r + i - k - 1}{r - 1} \beta^{(i)}_q
    \end{align*}
    for any prolongation order $r\geq 1$.
\end{theorem}

\medskip
\textit{Proof.} We proceed by induction on $r$. For the base case $r = 1$, the formula simplifies to
\begin{align*}
    \beta^{(k)}_{q + 1} = \sum_{i = k}^{n} \underbrace{\binom{i - k}{0}}_{=1} \beta^{(i)}_q = \beta^{(k)}_q + \beta^{(k + 1)}_q + \dots + \beta^{(n)}_q\,.
\end{align*}
This is indeed correct. Prolonging an equation of class $k$ or higher in $\R_q$ with respect to $x^{k}$ produces an equation of class $k$. Since $\R_q$ contains $\beta^{(k)}_q$ equations of class $k$, $\beta^{(k + 1)}_q$ equations of class $k + 1$, and so on, the total number of class $k$ equations in $\R_{q + 1}$ is precisely $\beta^{(k)}_q + \beta^{(k + 1)}_q + \dots + \beta^{(n)}_q$. Moreover, since $\R_q$ is assumed involutive, there are no further independent equations of class $k$ in $\R_{q + 1}$. Thus, the formula holds for $r = 1$.

For the inductive step, assume the formula holds for $r - 1$. When prolonging $\R_{q + r - 1}$ to $\R_{q + r}$, the same argument as above shows that
\begin{align*}
    \beta^{(k)}_{q + r} = \sum_{i = k}^{n} \beta^{(i)}_{q + r - 1}\,.
\end{align*}
By the induction hypothesis, we can express $\beta^{(i)}_{q + r - 1}$ as
\begin{align*}
    \beta^{(i)}_{q + r - 1} = \sum_{j = i}^{n} \binom{r + j - i - 2}{r - 2} \beta^{(j)}_q\,.
\end{align*}
Plugging this into the previous equation yields
\begin{align*}
    \beta^{(k)}_{q + r} 
    &= \sum_{i = k}^{n} \sum_{j = i}^{n} \binom{r + j - i - 2}{r - 2} \beta^{(j)}_q \\
    &= \sum_{j = k}^{n} \sum_{i = k}^{j} \binom{r + j - i - 2}{r - 2} \beta^{(j)}_q \\
    &= \sum_{j = k}^{n} \binom{r + j - k - 1}{r - 1} \beta^{(j)}_q\,,
\end{align*}
where in the last step we performed the sum over $i$. This completes the proof. \hfill $\Box$

\medskip

Theorem~\ref{thm:NoOfPrincipalD} equips us with a concrete tool to predict the number of principal derivatives in involutive equations at any order of prolongation. However, our ultimate goal is to quantify the amount of freedom we possess in specifying a formal power series solution to $\R_q$. Since principal derivatives, by definition, are precisely those determined by $\R_q$ and its prolongations, they do not directly inform us about this freedom. What we truly need is a way to count the number of \emph{parametric derivatives} at any order of prolongation. These derivatives remain undetermined by the differential equations and represent the freedom we have in constructing formal solutions. In physical applications, this freedom typically corresponds to the ability to choose initial data or boundary conditions.

Principal and parametric derivatives are closely related. Recall that the total number of $q$-th order derivatives of class $k$ for a field with $m$ components in $n$ dimensions is given by equation~\eqref{eq:ClassSize}, namely
\begin{align}
    m\binom{n + q - k - 1}{n - k}\,.
\end{align}
Among these, exactly $\beta^{(k)}_q$ derivatives are principal derivatives. By definition, all remaining derivatives of class $k$ are parametric derivatives. Therefore, the difference between the total number of class $k$ derivatives and the number of principal derivatives yields the number of parametric derivatives of class $k$. This number is so important that it deserves a name of its own.

\begin{definition}[Cartan characters $\alpha^{(k)}_q$]\label{def:CartanCharacters}
    Let $\R_q$ be an involutive $q$-th order equation in $n$ dimensions for $m$ field components. The \textbf{Cartan characters} $\alpha^{(k)}_q$ are defined as
    \begin{align*}
        \alpha^{(k)}_q \coloneqq m\binom{n + q - k - 1}{n - k} - \beta^{(k)}_q \quad \text{for } k = 1, \dots, n\,.
    \end{align*}
    The Cartan character $\alpha^{(k)}_q$ counts the number of class $k$ parametric derivatives present in~$\R_q$.
\end{definition}

For involutive equations, it is particularly straightforward to predict the number of parametric derivatives at any order of prolongation. Indeed, the relationship between the characters $\beta^{(k)}_q$ and the Cartan characters $\alpha^{(k)}_q$ allows us to deduce the following lemma as a direct consequence of Theorem~\ref{thm:NoOfPrincipalD}.

\begin{lemma}[Number of parametric derivatives at order $q + r$]\label{lem:NoOfParametricDer}
    Let $\S_q$ be an involutive symbol with Cartan characters $\alpha^{(k)}_q$. Then, the Cartan characters of its prolongations are given by
    \begin{align*}
        \alpha^{(k)}_{q + r} = \sum_{i = k}^{n} \binom{r + i - k - 1}{r - 1} \alpha^{(i)}_q\
    \end{align*}
    for any prolongation order $r\geq 1$.
\end{lemma}

This lemma finally equips us with a systematic tool to quantify the amount of freedom available when constructing a formal power series solution to an involutive equation $\R_q$. Specifically, it allows us to count how many parametric derivatives remain undetermined at each order, and hence, to characterize the space of solutions.

Let us also briefly connect this discussion to the structure of the symbol. From the definition of the solved form of $\S_q$ (see Definition~\ref{def:SymbolSolvedForm}), it is clear that the rank of $\S_q$ is given by
\begin{align}
    \operatorname{rank}\S_q = \sum_{k = 1}^{n} \beta^{(k)}_q\,,
\end{align}
since $\beta^{(k)}_q$ counts the number of class $k$ pivots in the row-echelon form of $\S_q$. Thus, the total number of pivots is simply the sum of the principal characters. Furthermore, from equation~\eqref{eq:DimOfSq}, the dimension of the symbol $\S_q$ is
\begin{align}
    \dim \S_q 
    &= m\binom{n - 1 + q}{n - 1} - \operatorname{rank}\S_q \notag\\
    &= m\sum_{k = 1}^{n} \binom{n + q - k - 1}{n - k} - \sum_{k = 1}^{n} \beta^{(k)}_q \notag\\
    &= \sum_{k = 1}^{n} \alpha^{(k)}_q\,,
\end{align}
where in the second line we used equation~\eqref{eq:SumOverAllClassSizes}, and in the third line, we applied Definition~\ref{def:CartanCharacters}.

In summary, the dimension of $\S_q$ is precisely the sum of all Cartan characters; that is, it counts the total number of parametric derivatives. This result is intuitively clear: When $\S_q$ is viewed as a system of linear equations (see Definition~\ref{def:Symbol}), only the pivot variables (the principal derivatives) are determined. The remaining variables, which correspond to the parametric derivatives, span the solution space of the symbol. Thus, the dimension of $\S_q$ coincides with the number of parametric derivatives.

Using Theorem~\ref{thm:NoOfPrincipalD} and Lemma~\ref{lem:NoOfParametricDer}, we can now generalize our results to any order of prolongation. Specifically, we obtain
\begin{align}
    \operatorname{rank} \S_{q + r} &= \sum_{k = 1}^{n} \binom{r + k - 1}{r} \beta^{(k)}_q \notag\\
    \dim \S_{q + r} &= \sum_{k = 1}^{n} \binom{r + k - 1}{r} \alpha^{(k)}_q\,.
\end{align}
The meaning of these numbers is clear: The rank of $\S_{q + r}$ measures the total number of principal derivatives at order $q + r$, while $\dim \S_{q + r}$ counts the total number of parametric derivatives at that order. The latter quantity is particularly important, as it quantifies the amount of freedom available when specifying a solution. This observation motivates the following definition.

\begin{definition}[Hilbert polynomial]\label{def:HilbertPolynomial}
    Let $\R_q$ be an involutive equation with Cartan characters $\alpha^{(k)}_q$. The total number of parametric derivatives present in the prolongation $\R_{q + r}$ for $r \geq 0$ is called the \textbf{Hilbert polynomial} and is defined as
    \begin{equation*}
        H_q(r) \coloneqq \dim \S_{q + r} = \sum_{k = 1}^{n} \binom{r + k - 1}{r} \alpha^{(k)}_q\,.
    \end{equation*}
\end{definition}

It is important to emphasize that the number of parametric derivatives is precisely the number of Taylor coefficients which remain undetermined in the formal power series solution to $\R_q$. Thus, the Hilbert polynomial can be interpreted as a measure of the functional freedom inherent in the equation.

At this point, it may not be immediately obvious why $H_q(r)$ is indeed a polynomial in $r$. However, this becomes clear upon closer inspection of the binomial coefficients. For $k = 1$, the binomial simplifies to
\begin{align}
    \binom{r + k - 1}{r} = \binom{r}{r} = 1\,,
\end{align}
which is simply a constant, i.e., a polynomial of degree zero. For the terms with $2\leq k \leq n$, the binomial evaluates to
\begin{align}
    \binom{r+k-1}{r} &= \frac{(r+k-1)!}{r!(k-1)!} = \frac{\overbrace{1\cdot 2\cdots (r-1)\cdot r}^{=r!} \cdot (r+1)(r+2)\cdots(r+k-1)}{r!(k-1!)} \notag\\
    &= \frac{(r+1)(r+2)\cdots(r+k-1)}{(k-1)!} = \underbrace{\frac{1}{(k-1)!}}_{\text{numerical factor}}\underbrace{\prod_{i=1}^{k-1}(r+i)}_{\text{polynomial in $r$}}
\end{align}
For the second equality, we factored out $r!$ from the factorial $(r + k - 1)!$. We can now clearly see that the numerator is a product of $k - 1$ linear factors of the form $(r + i)$. Therefore, each term in the sum is a polynomial in $r$ of degree $k - 1$. 

It follows that the Hilbert polynomial $H_q(r)$ is a polynomial in $r$ of degree at most $n - 1$. The precise degree depends on which Cartan characters $\alpha^{(k)}_q$ are non-zero. In particular, if $\alpha^{(n)}_q = 0$, the degree of $H_q(r)$ is strictly less than $n - 1$.

Before proceeding with our general discussion, let us illustrate the concepts of Cartan characters and the Hilbert polynomial with a concrete example.

\begin{example}[Cartan characters and Hilbert polynomial of the Proca equation]\label{ex:CCandHPOfProca}
    In Example~\ref{ex:CKAAppliedToProca}, we applied the Cartan--Kuranishi algorithm to the Proca equation. In the process, we determined the characters $\beta^{(k)}_2$ associated with the involutive prolongation $\R^{(2)}_2$. From these, the Cartan characters can be readily computed:
    \begin{align*}
        \alpha^{(1)}_2 &= 15\,, & \alpha^{(2)}_2 &= 11\,, & \alpha^{(3)}_2 &= 6\,, & \alpha^{(4)}_2 &= 0\,.
    \end{align*}
    The corresponding Hilbert polynomial is then given by
    \begin{align*}
        H_2(r) = \sum_{k = 1}^{4} \binom{r + k - 1}{r} \alpha^{(k)}_2 = 32 + 20r + 3r^2\,.
    \end{align*}
\end{example}

As we will see shortly, the fact that $\alpha^{(4)}_2$ vanishes is of particular significance. Specifically, if the Cartan character of highest class does \emph{not} vanish, this indicates that either gauge freedom has not been properly accounted for, or that the system of equations under consideration is pathological. The precise implications of a non-vanishing highest Cartan character will be discussed in the next subsection. Gauge freedom, on the other hand, will be treated in detail in Subsection~\ref{ssec:GaugeTheory}.

\subsection{Hilbert Polynomials and the Presence of Arbitrary Functions}\label{ssec:HilbertPolynomial}

As alluded to in the previous subsection, we generally expect that a solution $v^{A}$ to an equation $\R_q$ depends on a certain number of freely specifiable functions. For instance, in the case of Maxwell's equations, we observed that one component of the vector potential $A^{\mu}$ can be freely chosen due to gauge symmetry---corresponding to one arbitrary function of four coordinates. Additionally, the initial data $(A^{\mu}|_{\Sigma}, \dot{A}^{\mu}|_{\Sigma})$ allows us to specify six functions of three coordinates freely.

More generally, we now formalize the idea that in order to construct a solution to $\R_q$, we are required to choose certain functions. These functions may depend on all coordinates or only on a subset of them, just as in the Maxwell example. More precisely, we assume that the solution $v^{A}$ can be expressed as\footnote{This is known as an algebraic representation. Not every solution can be written in this form. Many differential equations admit solutions in terms of integrals. A well-known example is the d'Alembert solution to the one-dimensional wave equation~\cite{SeilerBook}. We will not consider integral representations here.}
\begin{align}\label{eq:AlgebraicRepresentation}
    v^{A} = v^{A}(x^\mu, F_1, \dots, F_k, \dots, F_n)\,,
\end{align}
where $F_k$ collectively denotes all freely specifiable functions of $k$ coordinates which appear in the solution $v^{A}$. It is important to emphasize that some of the $F_k$ may be absent from a particular solution and that the specific functions present may depend on different sets of $k$ coordinates. For instance, we may encounter a solution $v^{A}$ which depends on two functions of two coordinates, such as $f(x^{1}, x^{2})$ and $g(x^{3}, x^{4})$, where $f$ and $g$ are arbitrary. Hence, in this case we would have $F_2 = \{f(x^1, x^2), g(x^3, x^4)\}$.

In typical gauge theories, one generally encounters the freedom to choose one or more functions of all $n$ coordinates, corresponding to the gauge freedom. Additionally, in virtually all field theories, the freedom to specify initial or boundary data introduces an arbitrariness in the form of functions depending on $n-1$ coordinates.

Our next task is to understand how the arbitrary functions in the algebraic representation~\eqref{eq:AlgebraicRepresentation} are related to the Taylor coefficients of the formal power series expansion. This connection is established by Taylor-expanding the right-hand side of~\eqref{eq:AlgebraicRepresentation} and comparing it with the formal power series expansion~\eqref{eq:FormalPowerSeries}. 

Suppose that the solution $v^{A}$ depends on $f_k$ arbitrary functions of $k$ coordinates (that is, $f_k$ counts how many functions of $k$ coordinates are contained in the collection $F_k$). Then, the Taylor expansion of~\eqref{eq:AlgebraicRepresentation} will contain exactly
\begin{align}\label{eq:Tqr}
    T_q(r) \ce \sum_{k=1}^{n}f_k\binom{k + q + r -1}{q+r}
\end{align}
arbitrary Taylor coefficients at order $q+r$. If the Taylor expansion is to agree with the formal power series solution, this number of arbitrary coefficients must match the total number of parametric derivatives appearing at order $q+r$. As we have seen, this number is measured by the Hilbert polynomial. We are therefore led to impose the condition
\begin{align}\label{eq:HequalT}
    H_q(r) = \sum_{k=1}^{n}\binom{r+k-1}{r}\alpha^{(k)}_q \overset{!}{=} \sum_{k=1}^{n}f_k\binom{k + q + r -1}{q+r} = T_q(r)\,.
\end{align}

We recall that the Hilbert polynomial is of degree at most $n-1$, provided that $\alpha^{(n)}_q$ is not zero. The right-hand side of~\eqref{eq:HequalT} is also a polynomial in $r$, as can be seen from the following manipulation:
\begin{align}
    \binom{k+q+r-1}{q+r} &= \frac{(k+q+r-1)!}{(q+r)!(k-1)!}\notag\\
    &= \frac{\overbrace{1\cdot 2\cdots (q+r-1)(q+r)}^{=(q+r)!}(q+r+1)\cdots(q+r+k-1)}{(q+r)!(k-1)!}\notag\\
    &= \frac{1}{(k-1)!}\prod_{i=1}^{k-1}(q+r+i)\,.
\end{align}
This expression is clearly a polynomial in $r$ of degree $k-1$. Plugging this back into the expression for $T_q(r)$, we conclude that the right-hand side of~\eqref{eq:HequalT} is a polynomial of degree at most $n-1$.

We can now relate the Cartan characters $\alpha^{(k)}_q$ to the numbers $f_k$ by comparing the coefficients of $r$, $r^2$, \dots, $r^{n-1}$ on both sides of~\eqref{eq:HequalT}. The simplest relation follows by comparing the coefficients of the highest-degree term, $r^{n-1}$. This yields
\begin{align}
    f_n = \alpha^{(n)}_q\,.
\end{align}
This is an important result: it tells us that whenever the highest Cartan character $\alpha^{(n)}_q$ is non-zero, there are precisely $\alpha^{(n)}_q$ arbitrary functions of $n$ variables in the general solution $v^{A}$ to $\R_q$.

This is typically the case in gauge theories \emph{before} gauge symmetry has been fixed. If such a non-zero $\alpha^{(n)}_q$ appears in a theory without gauge symmetry, this indicates a pathology, since in this case a unique solution cannot be obtained from initial or boundary conditions alone. Instead, one would have to specify a certain number of arbitrary functions of $n$ coordinates, which from a physical standpoint signals a breakdown of classical determinism.

As an explicit example, recall from Example~\ref{ex:CCandHPOfProca} that for the involutive version of the Proca equation we found $\alpha^{(4)}_2 = 0$. Since the Proca theory has no gauge symmetry and is deterministic, this result is fully consistent with our discussion. Had we found a non-vanishing $\alpha^{(4)}_2$, we would have encountered a serious problem.

The relation between the remaining Cartan characters and the numbers $f_k$ with $1\leq k<n$ is more subtle. To find these relations, one has to solve a recursion relation. While this can be done in complete generality for arbitrary $n$ and $q$ (see~\cite{SeilerBook}), the resulting formulas are cumbersome and offer limited insight. Instead, we present here the solutions for $n=4$ and $q=1$ or $q=2$, which cover many practical applications.

For first-order equations in four dimensions, the relations are
\begin{align}\label{eq:fFirstOrder}
    f_1 &= \alpha^{(1)}_1 - \alpha^{(2)}_1, & f_2 &= \alpha^{(2)}_1 - \alpha^{(3)}_1, & f_3 &= \alpha^{(3)}_1 - \alpha^{(4)}_1, & f_4 = \alpha^{(4)}_1\,.
\end{align}
For second-order equations in four dimensions, the relations read
\begin{align}\label{eq:fSecondOrder}
    f_1 &= \alpha^{(1)}_2 - 2\alpha^{(2)}_2 + \alpha^{(3)}_2, & f_2 &= \alpha^{(2)}_2 - 2\alpha^{(3)}_2 + \alpha^{(4)}_2, & f_3 &= \alpha^{(3)}_2 - 2\alpha^{(4)}_2, & f_4 &= \alpha^{(4)}_2\,.
\end{align}
Using the general solution to the recursion relations derived in~\cite{SeilerBook}, one can show that for any dimension $n\geq 2$ and any order $q\geq 1$,
\begin{align}\label{eq:GeneralRelation}
    f_n &= \alpha^{(n)}_q\,, & f_{n-1} &= \alpha^{(n-1)}_q - q\, \alpha^{(n)}_q\,.
\end{align}
These relations provide an important clue towards the interpretation of the Cartan characters and the structure of the general solution to the equation $\R_q$.

\subsection{Gauge-Correction for Cartan Characters and Hilbert Polynomials}\label{ssec:GaugeTheory}

By equating the number of parametric derivatives to the number of free Taylor coefficients, we were able to relate the Cartan characters to the occurrence of freely specifiable functions in the solution to $\R_q$. In particular, we found that $\alpha^{(n)}_q$ always counts the number of arbitrary functions of all $n$ coordinates that need to be specified in order to solve $\R_q$. This observation, together with the explicit formulas~\eqref{eq:fFirstOrder},~\eqref{eq:fSecondOrder}, and in particular~\eqref{eq:GeneralRelation}, leads us to propose the following tentative interpretation:
\begin{itemize}
    \item For a well-posed physical theory, we require $\alpha^{(n)}_q = 0$. If this is not the case, initial or boundary conditions alone do not suffice to determine a unique solution to the equation $\R_q$. This signals a breakdown of classical determinism.
    \item Provided that $\alpha^{(n)}_q = 0$, the Cartan character $\alpha^{(n-1)}_q$ measures the amount of initial or boundary data that must be specified to uniquely solve $\R_q$. In this sense, it directly quantifies the number of degrees of freedom propagated by the equation.
\end{itemize}

This interpretation is corroborated by the case of Proca's equations. In Example~\ref{ex:CKAAppliedToProca} and Example~\ref{ex:CCandHPOfProca}, we found that $\alpha^{(4)}_2 = 0$ and $\alpha^{(3)}_2 = 6$, which implies $f_4 = 0$ and $f_3 = 6$. That is, no arbitrary functions of four coordinates appear in the solution, and six functions of three coordinates must be specified as initial or boundary data. This is precisely in agreement with the well-known physical properties of Proca's theory.

However, so far we have neglected to take gauge symmetry into account. Our current interpretation may fail in this context because, in gauge theories, we typically need to specify functions of all $n$ coordinates, but this does \textit{not} imply a breakdown of determinism. Rather, it simply reflects the fact that an equation $\R_q$ for a gauge field determines a solution $v^{A}$ only up to gauge transformations. More precisely, if $v^{A}$ solves $\R_q$ for a given set of initial or boundary data, then so does any other field $\tilde{v}^{A}$ which is related to $v^{A}$ by a gauge transformation. Such a transformation generically involves arbitrary functions of all $n$ coordinates. Consequently, it is impossible for $\R_q$ to determine $v^{A}$ completely, and we should naturally expect $\alpha^{(n)}_q \neq 0$ in the presence of gauge freedom.

A simple example will convince us of the validity of these arguments and will explicitly demonstrate that our tentative interpretation requires refinement when applied to gauge theories.

\begin{example}[Cartan characters of Maxwell's equations]\label{ex:CartanCharactersMaxwell}
    Unlike Proca's equations, Maxwell's equations feature gauge symmetry. Starting from the equation
    \begin{align*}
        \partial^\nu \partial_\mu A^\mu - \partial_\mu \partial^\mu A^\nu = 0\,,
    \end{align*}
    we run through the machinery of the Cartan-Kuranishi algorithm to ensure that we are working with an involutive equation (our interpretation can only be valid for involutive equations). Remarkably, the Cartan-Kuranishi algorithm reveals that Maxwell's equations in the form given above are already involutive. Thus, the algorithm terminates without requiring any prolongations or projections and yields the following Cartan characters:
    \begin{align*}
        \alpha^{(1)}_2 = 16\,, \qquad \alpha^{(2)}_2 = 12\,, \qquad \alpha^{(3)}_2 = 7\,, \qquad \alpha^{(4)}_2 = 1\,.
    \end{align*}
    As expected, $\alpha^{(4)}_2$ is non-zero, immediately signaling the breakdown of our tentative interpretation. However, the fact that $\alpha^{(4)}_2 = 1$ is fully consistent with our knowledge that in Maxwell's theory there exists one gauge degree of freedom---corresponding to one arbitrary function of four coordinates---that needs to be specified in order to obtain a particular solution.

    On the other hand, the interpretation of the Cartan character $\alpha^{(3)}_2$ is now problematic. Our earlier reading of $\alpha^{(n-1)}_q$ as counting the number of initial or boundary conditions fails in this case. This is further reflected in the corresponding numbers $f_k$ of free functions:
    \begin{align*}
        f_1 = -1\,, \qquad f_2 = -1\,, \qquad f_3 = 5\,, \qquad f_4 = 1\,.
    \end{align*}
    Instead of the expected four functions serving as initial or boundary data, we are now led to the nonsensical conclusion that there should be five, accompanied by negative numbers of functions. Clearly, our tentative interpretation does not hold in the presence of gauge symmetry and requires revision.
\end{example}
Of course, we understand the source of the failure and how to remedy it: we need to incorporate gauge transformations into the formalism. Let us assume that the base space~$B$ is $n$-dimensional and that $v^{A}$ is a gauge field. We say that $\tilde{v}^{A}$ is the gauge-transformed field of $v^{A}$ if it can be written as 
\begin{align}\label{eq:GaugeTransformationLaw} \tilde{v}^{A} = \Gamma\big(x^\mu, v^{A}, \lambda^{(0)}_a(x), \partial\lambda^{(1)}_a(x), \dots, \partial^{\bfm}\lambda^{(p)}_a(x)\big) \quad \text{with } p \ce |\bfm|,. 
\end{align} 
Here, $\Gamma$ is a function of the $n$ base coordinates $x^\mu$, of the field $v^{A}$ itself, and of auxiliary fields $\lambda^{(i)}_a(x)$, each depending on all $n$ coordinates. These fields may appear algebraically (as in $\lambda^{(0)}_a(x)$) or accompanied by derivatives. For later convenience, we introduce the numbers $\gamma_\ell$, defined as follows: \begin{itemize} \item $\gamma_0$: the number of fields that enter algebraically in the transformation law~\eqref{eq:GaugeTransformationLaw}. \item $\gamma_i$ with $0 < i \leq p$: the number of fields that enter with $i$-th order derivatives in the transformation law~\eqref{eq:GaugeTransformationLaw}. \end{itemize}

We illustrate the structure of~\eqref{eq:GaugeTransformationLaw} and the counting of $\gamma_\ell$ with two familiar examples.

\begin{example}[Electromagnetism and General Relativity]\label{ex:EDandGR}
    As is well-known, the gauge transformation law of electromagnetism has the form
    \begin{align*}
        A_\mu \quad\mapsto\quad \tilde{A}_\mu = A_\mu + \partial_\mu \Phi\,,
    \end{align*}
    where $\Phi$ is some scalar field. This has precisely the form of~\eqref{eq:GaugeTransformationLaw}. The field $\tilde{A}_\mu$, which is obtained from $A_\mu$ by a gauge transformation, is written in terms of $A_\mu$ itself and a function of all $n$ coordinates. We have the identification
    \begin{align*}
        \partial_\mu \Phi \equiv \partial_\mu \lambda^{(1)}(x)\,.
    \end{align*}
    It thus follows that in electromagnetism we have
    \begin{align*}
        \gamma_\ell =
        \begin{cases}
            1 & \text{when } \ell = 1 \\
            0 & \text{otherwise}
        \end{cases}
    \end{align*}
    In the case of GR, gauge transformations are induced by diffeomorphism. Under a diffeomorphism which takes $x^\mu$ to $\tilde{x}^\mu(x)$, where $\tilde{x}^\mu$ should be understood as four functions of $x^\mu$, the metric transforms as
    \begin{align*}
        g_{\mu\nu} \quad\mapsto\quad \tilde{g}_{\mu\nu} = \PD{x^\alpha}{\tilde{x}^\mu}\PD{x^\beta}{\tilde{x}^\nu} g_{\alpha\beta}(x) \,.
    \end{align*}
    This is again a transformation law of the form~\eqref{eq:GaugeTransformationLaw}. The field $\tilde{g}_{\mu\nu}$ is expressed in terms of $g_{\mu\nu}$ itself and four function $\tilde{x}^\mu(x)$. With the identification
    \begin{align*}
        \PD{\tilde{x}^\mu}{x^\nu} \equiv \partial_\nu (\lambda^{(1)})^\mu(x)\,,
    \end{align*}
    we conclude that the numbers $\gamma_\ell$ satisfy the relation
    \begin{align*}
        \gamma_\ell =
        \begin{cases}
            4 & \text{when } \ell = 1\\
            0 & \text{otherwise}
        \end{cases}\,.
    \end{align*}
    Observe that $\tilde{g}_{\mu\nu} = \left(\partial_\alpha (\lambda^{(1)})^\mu\right)^{-1}\left(\partial_\alpha (\lambda^{(1)})^\nu\right)^{-1} g_{\alpha\beta}(x)$, which is consistent with~\eqref{eq:GaugeTransformationLaw}.
\end{example}

At this point, we draw attention to a potential ambiguity in the definition of $\gamma_\ell$ introduced above. In the examples discussed so far, the gauge parameters appeared exclusively with first-order derivatives in the transformation law. However, in more general cases, the transformation may involve a mixture of $\lambda$-fields: some appearing algebraically, others with first-order, second-order, or even higher-order derivatives. This raises an important question: how should one properly count the fields when multiple types of $\lambda$'s are present?

To illustrate the problem, let us consider another example. This time, we promote Proca's theory to a gauge theory using St\"{u}ckelberg's trick. In other words, we introduce a scalar field $\pi$, which serves to render the theory gauge-invariant under a certain set of transformations.

\begin{example}[Proca \`{a} la St\"{u}ckelberg]\label{ex:ProcaStueckelberg} 
The Proca Lagrangian is typically written as 
\begin{align*} 
    \mathcal{L} = -\frac14 F_{\mu\nu} F^{\mu\nu} - \frac12 m^2 A_\mu A^\mu\,, 
\end{align*} where $F_{\mu\nu} \ce \partial_\mu A_\nu - \partial_\nu A_\mu$. The tensor $F_{\mu\nu}$ is clearly invariant under the transformation $A_\mu \mapsto \tilde{A}_\mu = A_\mu + \partial_\mu \Phi$, but the mass term $m^2 A_\mu A^\mu$ prevents the Lagrangian from being invariant as a whole.
This issue can be resolved by introducing a scalar field $\pi$, leading to a modified Lagrangian of the form
\begin{align*}
    \mathcal{L} = -\frac14 F_{\mu\nu} F^{\mu\nu} - \frac12 m^2 (A_\mu + \partial_\mu \pi)(A^\mu + \partial^\mu \pi)\,.
\end{align*}
This new Lagrangian is now invariant under the simultaneous transformation
\begin{align*}
    A_\mu \quad &\mapsto \quad \tilde{A}_\mu = A_\mu + \partial_\mu \Phi \\
    \pi \quad &\mapsto \quad \tilde{\pi} = \pi - \Phi\,.
\end{align*}
For more details on the St\"{u}ckelberg trick and this re-formulation of Proca's theory, see for instance~\cite{Heisenberg:2018}. What is crucial for us is that this transformation can be written in a form compatible with~\eqref{eq:GaugeTransformationLaw}. To that end, we recall that in the jet bundle formalism, we combine all fields into a single vector. Thus, we define
\begin{align*}
    v^{A} = 
    \begin{pmatrix}
        A^\mu \\
        \pi
    \end{pmatrix}
\end{align*}
and similarly for the transformations defined above. This is now consistent with~\eqref{eq:GaugeTransformationLaw}. 

Now, how should we identify the functions $\lambda^{(i)}$? We observe that $\Phi$ appears both algebraically and with its first derivative. Should we set $\lambda^{(0)}(x) = \Phi$? Or perhaps $\lambda^{(1)}_\mu(x) = \partial_\mu \Phi$? Another option is to set both $\lambda^{(0)}(x) = \Phi$ and $\lambda^{(1)}_\mu(x) = \partial_\mu \Phi$ simultaneously. These three choices correspond to three distinct possibilities for the numbers $\gamma_\ell$:
\begin{align*}
    (\gamma_0, \gamma_1) \overset{?}{=} 
    \begin{cases}
        (1, 0) \\
        (0, 1) \\
        (1, 1)
    \end{cases}
\end{align*}
Which of these choices is correct?
\end{example}

To determine the correct choice, we need to understand the purpose of the $\gamma_\ell$'s. Recall that when we construct a formal power series solution to a given set of PDEs, we always end up with Taylor coefficients which are determined by the equations themselves, provided we specify a certain number of other Taylor coefficients by hand. These freely specifiable Taylor coefficients encode our freedom to choose initial data or boundary data, as well as our freedom to fix a gauge. The total number of Taylor coefficients at order $q+r$ which are \emph{not} determined by the PDE itself, i.e., the ones over which we have some control, is measured by the Cartan characters $\alpha^{(k)}_{q+r}$. In particular, $\alpha^{(n)}_{q+r}$ measures how many Taylor coefficients can be attributed to the presence of arbitrary functions of $n$ coordinates at order $q+r$.

We now begin to understand why~\eqref{eq:GaugeTransformationLaw} yields ambiguous results in the case of the Proca-St\"{u}ckelberg theory. The Cartan characters are derived directly from the field equations and they do not know anything about gauge transformations. On the other hand, the  $\gamma$'s only know about the gauge transformations, but nothing about the field equations. Hence, to resolve the problem, we have to look at how the $\gamma$'s enter the field equations after a gauge transformation has been applied. Only in this way can we understand how many Taylor coefficients can be fixed by gauge transformations, which is tantamount to knowing the $\gamma$'s.

To make the argument more precise, we first need to state the Proca-St\"{u}ckelberg equations:
\begin{align}
    \R_2 :
\begin{cases}
    \partial^\nu \partial_\mu A^\mu - \partial_\mu \partial^\mu A^\nu + m^2 A^\nu + m\,\partial^\nu \pi &= 0\\
    \\
    \partial^\mu \partial_\mu \pi + m\, \partial_\mu A^\mu &= 0
\end{cases}\,.
\end{align}
If $v^{A} = (A^\mu, \pi)$ is a solution to these equations, then so is $\tilde{v}^{A} = (\tilde{A}^\mu, \tilde{\pi})$, if it is related to~$v^{A}$ via the gauge transformation described in Example~\ref{ex:ProcaStueckelberg}. Of course, the equations are perfectly invariant under this gauge transformation, so nothing seems to be gained. The trick is to observe that $v^{A}$ and $\tilde{v}^{A}$ are genuinely \textit{different} fields and that their difference stems from using the gauge transformation to fix some of their components in different ways. An other observation is that $\pi$, despite transforming algebraically under gauge transformations, always enters with derivatives in the field equations. Thus, also $\Phi$ enters only with derivatives. When using the gauge transformation to fix Taylor coefficients in the formal power series solution, we thus only ever have access to $\partial\Phi$, not $\Phi$ itself. It follows that
\begin{align}
    (\gamma_0, \gamma_1) = (0,1)
\end{align}
is the only correct way of defining the $\gamma$'s for the Proca-St\"{u}ckelberg theory. More generally, one can show that under a given gauge transformation of the form~\eqref{eq:GaugeTransformationLaw}, it is possible to fix
\begin{align}\label{eq:GaugePolynomial}
    G_q(r) \ce \sum_{\ell = 0}^p \gamma_\ell \binom{q+r+\ell + n - 1}{q+r+\ell}
\end{align}
Taylor coefficients of functions of $n$ coordinates at order $q+r$ by means of a gauge-fixing. Compare $G_q(r)$ to $T_q(r)$, as defined in equation~\eqref{eq:Tqr}. The binomial of both expressions has a similar structure. The term $q+r$ keeps track of the order in the formal power series solution, while $k$ and $n$ keep track of the number of coordinates that appear in the functions. In $T_q(r)$ we sum over $k$, meaning that we sum over different functions, where $f_k$ determines how many functions of $k$ coordinates are present. In $G_q(r)$, we only consider functions of $n$ coordinates, so we can see this as specializing $T_q(r)$ to the case where $f_k=0$ unless $k=n$. Moreover, the parameter $\ell$ keeps track of the order of differentiation of the function under consideration. That is, $G_q(r)$ knows that the functions we consider involve derivatives. 

To take into account that gauge-fixing allows us to eliminate certain Taylor coefficients in the formal power series solution by hand, we introduce the gauge-corrected Hilbert polynomial.
\begin{definition}[Gauge-corrected Hilbert polynomial]\label{def:GaugeCorrectedHilbertPolynomial}
    Let $H_q(r)$ be the Hilbert polynomial associated with $\R_q$ and $G_q(r)$ the polynomial defined in~\eqref{eq:GaugePolynomial}, which takes into account the gauge freedom~\eqref{eq:GaugeTransformationLaw}. The gauge-corrected Hilbert polynomial is then defined as
    \begin{align*}
        \bar{H}_q(r) \ce H_q(r) - G_q(r)\,.
    \end{align*}
    In words, we subtract from the number of freely specifiable Taylor coefficients at order $q+r$ the number of coefficients which can be fixed by a gauge choice. 
\end{definition}
We originally defined the Hilbert polynomial in terms of Cartan characters. This relationship can be inverted. That is, knowing the Hilbert polynomial allows us to recursively infer the values of the Cartan characters. This is important because once we have determined the gauge-corrected Hilbert polynomial, we can define a new set of Cartan characters, which we call \textbf{gauge-corrected Cartan characters}, using very similar recursion relations.

In order to determine the Cartan characters from $H_q(r)$, we need to introduce the modified Stirling numbers:
\begin{align}\label{eq:ModifiedStirlingNumbers}
    s^{(N)}_k(X) \ce
    \begin{cases}
        0 & \text{if } k<0 \\
        1 & \text{if } k=0 \\
        \sigma^{(N)}_k(X+1, X+2, \dots, X+n) & \text{for } 0<k\leq N
    \end{cases}\,.
\end{align}
Here, $\sigma^{(N)}_k$ are the symmetric elementary polynomial of degrees $k$ in $N$ unknowns (caution: $N$ has nothing to do with spacetime dimension!). These polynomials are defined as
\begin{align}
    \sigma^{(N)}_k(X_1, X_2, \dots, X_n) \ce \sum_{1\leq i_1 < i_2 <\dots < i_k \leq N} X_{i_1} X_{i_2} \cdots X_{i_k}\,.
\end{align}
In the special cases with $k=1$, $k=2$, and $k=N$, the above definition reduces to the simpler expressions
\begin{align}\label{eq:ESPSpecialCases}
    \sigma^{(N)}_1(X_1, X_2, \dots, X_N) &= \sum_{i=1}^N X_i \notag\\
    \sigma^{(N)}_2(X_1, X_2, \dots, X_N) &= \sum_{1\leq i < j \leq N} X_i X_j = \sum_{i=1}^{N-1}\sum_{j=i+1}^N X_i X_j \notag\\
    \sigma^{(N)}_N(X_1, X_2, \dots, X_N) &= X_1 X_2 \cdots X_N\,.
\end{align}
By using modified Stirling numbers, we can re-write binomial coefficients of the form
\begin{align}\label{eq:BinomialToStirling}
    \binom{q+r+n}{q+r} = \frac{1}{n!}\sum_{i=0}^{n} s^{(n)}_{n-1}(q) r^{i}\,.
\end{align}
This is the general form that appears in the Hilbert polynomial. Next, recall that $H_q(r)$ is a polynomial of degree at most $n-1$ in $r$. We would like to write this polynomial as 
\begin{align}
    H_q(r) = \sum_{i=0}^{n-1} h_i r^{i}\,,
\end{align}
where $h_i$ are the coefficients of the monomials $r^{i}$. Using the expression of the Hilbert polynomial $H_q(r)$ found in Definition~\ref{def:HilbertPolynomial}, and equation~\eqref{eq:BinomialToStirling}, we find
\begin{align}
    H_q(r) &= \sum_{k=1}^{n} \binom{r+k-1}{r}\alpha^{(k)}_q \notag\\
    &= \sum_{k=1}^n \left(\frac{1}{(k-1)!} \sum_{i=0}^{k-1} s^{(k-1)}_{k-i-1}(0) \alpha^{(k)}_q r^{i}\right) \notag\\
    &= \sum_{k=0}^{n-1} \left(\sum_{i=0}^{k-1} \frac{\alpha^{(k+1)}_q}{k!} s^{(k)}_{k-i}(0) r^{i}\right) \notag\\
    &= \sum_{i=0}^{n-1}\left(\sum_{k=i}^{n-1}\frac{\alpha^{(k+1)}_q}{k!}s^{(k)}_{k-i}(0)\right) r^{i}
\end{align}
from which we can now read off the coefficients $h_i$ as 
\begin{align}\label{eq:hi}
    h_i = \sum_{k=i}^{n-1}\frac{\alpha^{(k+1)}_q}{k!}s^{(k)}_{k-i}(0)\,.
\end{align}
To get from the first to the second line, we used~\eqref{eq:BinomialToStirling}. On the third line we re-defined the summation index $k$ as $k+1$. This has the effect of shifting the summation range from $k=1,\dots,n$ to $k=0,\dots,n-1$. Finally, on the last line we swapped the two sums. This is necessary since we want the sum over $i$ outside of the bracket, so that we obtain an expression of the form $\sum_{i=0}^{n-1}h_i r^{i}$. In swapping the sums, we also adjusted the summation ranges of both sums. In the case of the sum over $k$, we can start the summation at $k=i$, instead of $k=0$. This is possible since the modified Stirling numbers $s^{(k)}_{k-i}$ are zero for $k-i<0$ (see equation~\eqref{eq:ModifiedStirlingNumbers}). Similarly, we can extend the sum over $i$ to go up to $n-1$, instead of just $k-1$. The reason is again that $s^{(k)}_{k-i}=0$ for $i\geq k$.

From~\eqref{eq:hi}, we obtain simple expressions for the cases where $i=n-1$ or $i=0$:
\begin{align}\label{eq:hn-1h0}
    h_{n-1} &= \sum_{k=0}^{n-1} \frac{\alpha^{(k+1)}_q}{k!} \underbrace{s^{(k)}_{k-n+1}(0)}_{=0 \text{ unless } k=n-1} = \frac{\alpha^{(n)}_q}{(n-1)!} \notag\\
    h_0 &= \sum_{k=0}^{n-1} \frac{\alpha^{(k+1)}_q}{k!} \underbrace{s^{(k)}_k(0)}_{=k!} = \sum_{k=1}^{n} \alpha^{(k)}_q\,.
\end{align}
Observe that the first equation allows us to solve for $\alpha^{(n)}_q$. If we know the value of $h_{n-1}$, we find that $\alpha^{(n)}_q$ is given by
\begin{align}
    \alpha^{(n)}_q = (n-1)!\,h_{n-1}\,.
\end{align}
This now allows us to recursively determine all the other Cartan characters. To see this, we set $i=k-1$, for $1\leq k \leq n-1$, and we compute
\begin{align}\label{eq:hcoeff}
    h_{k-1} &= \sum_{j=k-1}^{n-1} \frac{\alpha^{(j+1)}_q}{j!} s^{(j)}_{j-k+1}(0) = \frac{\alpha^{(k)}_q}{(k-1)!} \underbrace{s^{(k-1)}_0(0)}_{=1} + \sum_{j=k}^{n-1} \frac{\alpha^{(j+1)}_q}{j!}s^{(j)}_{j-k+1}(0) \notag\\
    &= \frac{\alpha^{(k)}_q}{(k-1)!} + \sum_{j=k+1}^{n} \frac{\alpha^{(j)}_q}{(j-1)!}s^{(j-1)}_{j-k}(0)\,.
\end{align}
Observe that the sum on the second line only involves Cartan characters of order $j>k$. Thus, we can solve for $\alpha^{(k)}_q$ and express it in terms of higher order Cartan characters as follows:
\begin{align}
    \alpha^{(n)}_q &= (n-1)!\, h_{n-1} \notag\\
    \alpha^{(k)}_q &= (k-1)!\, h_{k-1} - \sum_{j=k+1}^{n} \frac{(k-1)!}{(j-1)!}\alpha^{(j)}_q s^{(j-1)}_{j-k}(0) & \text{for} && 1\leq k \leq n-1 \,.
\end{align}
As claimed, we can recursively compute the Cartan characters from the coefficients of the Hilbert polynomial. Using virtually the same computational steps, we can determine the gauge-corrected Cartan characters from the gauge-corrected Hilbert polynomial. To that end, we express
\begin{align}
    \bar{H}_q(r) = \sum_{i=0}^{n-1} \bar{h}_i r^{i}
\end{align}
and, using Definition~\ref{def:GaugeCorrectedHilbertPolynomial} together with the definition of the modified Stirling numbers, we obtain for the coefficients of the gauge-corrected Hilbert polynomial
\begin{align}\label{eq:hbar}
    \bar{h}_k &= h_k - \frac{1}{(n-1)!}\sum_{\ell=0}^p \gamma_\ell \, s^{(n-1)}_{n-k-1}(q+\ell) & \text{for} && 1\leq k \leq n-1\,.
\end{align}
The gauge-corrected Cartan characters are then found recursively to be given by
\begin{align}\label{eq:RecRelGaugeCorrCartanChar}
    \bar{\alpha}^{(n)}_q &= (n-1)!\, \bar{h}_{n-1} \notag\\
    \bar{\alpha}^{(k)}_q &= (k-1)!\, \bar{h}_{k-1} - \sum_{j=k+1}^{n} \frac{(k-1)!}{(j-1)!}\bar{\alpha}^{(j)}_q s^{(j-1)}_{j-k}(0) & \text{for} && 1\leq k \leq n-1 \,.
\end{align}
In a gauge theory, we have to equate the gauge-corrected Hilbert polynomial $\bar{H}_q(r)$, rather than $H_q(r)$, to $T_q(r)$. If we do so, we obtain analogously to the previous subsection that for any $n\geq 2$ and any order $q\geq 1$ the general power series solution to $\R_q$ contains 
\begin{align}\label{eq:GaugeCorrectedfs}
    f_n &= \bar{\alpha}^{(n)}_q, && & f_{n-1} &= \bar{\alpha}^{(n-1)}_q - q\, \bar{\alpha}^{(n)}_q
\end{align}
functions of $n$ and $n-1$ coordinates, respectively. To see whether the gauge-correction cures the problem we had with our tentative interpretation of Cartan characters we discussed at the beginning of this subsection, we consider again an example. 
\begin{example}[Gauge-corrected Cartan characters for Maxwell's equations]\label{ex:CartanCharactersMaxwellCorrected}
    In Example~\ref{ex:CartanCharactersMaxwell} we saw that without taking gauge freedom into account, we obtain $\alpha^{(4)}_2 = 1$ and $\alpha^{(3)}_2 = 5$ for Maxwell's equations. This implies the presence of one function of four coordinates and five functions of three coordinates in the power series solution to Maxwell's equations. This is clearly incorrect.

    However, if we take gauge freedom into account, we should work with the gauge-corrected Cartan characters. To determine those, we first need to compute the gauge-corrected Hilbert polynomial. Since $\gamma_\ell = 1$ for $\ell=1$, we obtain (cf. equation~\eqref{eq:GaugePolynomial} and Definition~\ref{def:GaugeCorrectedHilbertPolynomial})
    \begin{align*}
        \bar{H}_2(r) = \sum_{k=1}^4 \binom{r+k-1}{r}\alpha^{(k)}_2 - \binom{6+r}{3+r} = 16 + 12 r + 2 r^2\,.
    \end{align*}
    Because $\bar{H}_q(r)$ has to be of the form $\sum_{i=0}^{n-1}\bar{h}_i r^{i}$, we can simply read off the coefficients $\bar{h}_i$:
    \begin{align*}
        \bar{h}_0 &= 16\,, & \bar{h}_1 &= 12\,, & \bar{h}_2 &= 2\,, & \bar{h}_3 &= 0\,.
    \end{align*}
    Now we can use the recursion relations~\eqref{eq:RecRelGaugeCorrCartanChar} to determine the gauge-corrected Cartan characters. We find
    \begin{align*}
        \bar{\alpha}^{(4)}_2 &= 3!\, \bar{h}_3 = 0 \\
        \bar{\alpha}^{(3)}_2 &= 2 \bar{h}_2 - 2 \bar{\alpha}^{(4)}_2 = 4 \\
        \bar{\alpha}^{(2)}_2 &= \bar{h}_1 - \frac32 \bar{\alpha}^{(3)}_2 - \frac{11}{6} \bar{\alpha}^{(4)}_2 = 6 \\
        \bar{\alpha}^{(1)}_2 &= \bar{h}_0 - \bar{\alpha}^{(2)}_2 - \bar{\alpha}^{(3)}_2 - \bar{\alpha}^{(4)}_2 = 6\,.
    \end{align*}
    Of particular interest are $\bar{\alpha}^{(4)}_2$ and $\bar{\alpha}^{(3)}_2$. These values for the Cartan characters now imply that there are (see equation~\eqref{eq:GaugeCorrectedfs})
    \begin{align*}
        f_4 &= 0 & \text{and} && f_3 = 4
    \end{align*}
    functions of four and tree coordinates, respectively, in the general power series solution to Maxwell's equations. This is now in agreement with our tentative interpretation: The gauge-corrected Cartan character $\bar{\alpha}^{(4)}_2$ vanishes, signaling that the theory respects classical determinism. Moreover, $\bar{\alpha}^{(3)}_2$ corresponds precisely to the amount of initial data one needs to provide in order to solve Maxwell's equations.
\end{example}

In the next subsection, we will bring everything together and show how the degrees of freedom of any field theory can be computed from $\alpha^{(k)}_q$, $\beta^{(k)}_q$, and $\gamma_\ell$.


\subsection{Counting Degrees of Freedom}
After extending the formalism to accommodate gauge theories, we are now ready to count the degrees of freedom for any theory described by an equation $\R_q$. Without loss of generality, we assume that $\R_q$ is an involutive equation. If it is not, we can apply the Cartan-Kuranishi Algorithm~\ref{alg:CK} to complete the non-involutive equation $\R_q$ into an equivalent involutive system $\R^{(s)}_{q+r}$ for some integers $r, s \geq 0$.

Furthermore, let us assume that we are working in $n$ dimensions and that $\R_q$ describes the evolution of a gauge field $v^A$, which transforms according to the gauge transformation law~\eqref{eq:GaugeTransformationLaw}. Under these assumptions, we know that the gauge-corrected Hilbert polynomial (cf. Definition~\ref{def:GaugeCorrectedHilbertPolynomial}) counts the number of free Taylor coefficients at order $q+r$.

The key idea behind counting degrees of freedom is straightforward: Any solution $v^A$ to $\R_q$ depends on a set of functions that must be specified. Some of these functions merely account for gauge redundancy, while others arise from initial or boundary data. Our goal is to count the latter, as they represent the physical degrees of freedom.

The formalism presented here does not grant direct access to these functions. Instead, it allows us to count the undetermined Taylor coefficients at any order $q+r$. Furthermore, it enables us to distinguish between coefficients arising from gauge freedom and those that remain undetermined even after gauge fixing. These undetermined Taylor coefficients must, in one form or another, correspond to the physical degrees of freedom.

Now, we bring everything together. Since degrees of freedom correspond to \textit{functions} that can be freely specified on an initial surface or boundary, we should compare $\bar{H}_q(r)$---the number of free Taylor coefficients at order $q+r$ after accounting for gauge redundancies---to the number of Taylor coefficients of a single function of $n$ coordinates at the same order.

Both quantities grow as $r$ increases, since any formal power series contains infinitely many Taylor coefficients. However, we expect $\bar{H}_q(r)$ to grow at most as fast as the number of free Taylor coefficients of a single function, since the gauge-corrected Hilbert polynomial encodes information about at least one, and potentially several, functions. This observation motivates the following definition.

\begin{definition}[The strength] 
    Let $\bar{H}_q(r)$ be a gauge-corrected Hilbert polynomial. The \textbf{strength} is defined as the ratio   \begin{align*} 
        Z_q(r) \ce \frac{\bar{H}_q(r)}{\binom{n+q+r-1}{q+r}}\,. 
    \end{align*} 
    This quantity measures the number of free Taylor coefficients at order $q+r$ (after removing gauge degrees of freedom) relative to the number of Taylor coefficients of a single function of $n$ coordinates at the same order. 
\end{definition}

Note that $\bar{H}_q(r)$ depends on the spacetime dimension $n$, the number of fields $m$, the order $q$ of the equation $\R_q$, and other properties of these equations via the Cartan characters $\alpha^{(k)}_q$. In contrast, the binomial factor in the denominator depends only on $n$ and $q$ and carries no information about $\R_q$. Both terms also depend on $r$, which is the only parameter unrelated to the physical system under study---it simply serves to track the order of the power series expansion.

The dependence on $r$ can be eliminated by considering the limit $r \to \infty$ of the strength. This limit always exists. To see why, recall from Subsection~\ref{ssec:HilbertPolynomial} that the Hilbert polynomial is a polynomial of degree at most $n-1$, whereas the binomial coefficient $\binom{n+q+r-1}{q+r}$ is a polynomial of degree exactly $n-1$. Consequently, the limit $r \to \infty$ exists, and using the gauge-corrected Hilbert polynomial (cf. Definition~\ref{def:GaugeCorrectedHilbertPolynomial}), we obtain
\begin{align}
    \lim_{r\to\infty} Z_q(r) = \bar{\alpha}^{(n)}_q = \alpha^{(n)}_q - \sum_{\ell = 0}^{p}\gamma_\ell\,.
\end{align}

From our previous discussions, we know that $\alpha^{(n)}_q$ counts the number of independent functions of $n$ coordinates appearing in the power series solution $v^{A}$, while the sum $\sum_{\ell = 0}^{p} \gamma_\ell$ counts the number of gauge functions present in the gauge transformation law~\eqref{eq:GaugeTransformationLaw}. For a physically well-behaved theory, these two quantities must be equal, meaning we impose the condition
\begin{align}\label{eq:CompatibilityCondition}
    \alpha^{(n)}_q \overset{!}{=} \sum_{\ell=0}^p \gamma_\ell\,.
\end{align}
As a consequence, the gauge-corrected Hilbert polynomial has a degree of at most $n-2$, rather than $n-1$. In other words, the number of free Taylor coefficients it counts grows more slowly than $\binom{n+q+r-1}{q+r}$, which corresponds to the number of free Taylor coefficients of a function of $n$ coordinates. This key observation motivates the introduction of the concepts of \textit{compatibility} and the \textit{compatibility coefficient}.

\begin{definition}[Compatible equations and the compatibility coefficient] 
    We define the \textbf{compatibility coefficient} $Z^{(0)}_q$ as
    \begin{align*} 
        Z^{(0)}_q \ce \lim_{r\to\infty} Z_q(r) = \bar{\alpha}^{(n)}_q = \alpha^{(n)}_q - \sum_{\ell=0}^p \gamma_\ell\,. 
    \end{align*} 
    An equation $\R_q$ is said to be \textbf{compatible} if $Z^{(0)}_q = 0$. A compatible equation is one in which no additional free functions of $n$ coordinates appear in the formal power series solution.
\end{definition}
The rational function $Z_q(r)$ contains more information than just $Z^{(0)}_q$. So far, we have only extracted the zeroth-order term in its expansion around $r = \infty$. Since $Z_q(r)$ is a rational function with a denominator of degree $n-1$ and a numerator of degree at most $n-1$, its Taylor expansion around $r = \infty$ takes the form
\begin{align}\label{eq:ExpansionZ}
    Z_q(r) = Z^{(0)}_q + \frac{Z^{(1)}_q}{r} + \mathcal{O}\left(\frac{1}{r^2}\right)\,. 
\end{align}
This should be reminiscent of the expansion~\eqref{eq:EinsteinExpansion} discussed in Subsection~\ref{ssec:EinsteinsMethod}, where we examined Einstein's method for counting degrees of freedom. In that context, we also introduced the concepts of compatibility and strength. Here, we have arrived at these notions again, but through a different approach---one that resolves the difficulties and limitations of Einstein’s method, as outlined in Subsection~\ref{ssec:Limitations}. To complete the picture, we now demonstrate how to isolate the next term in the expansion~\eqref{eq:ExpansionZ} and use it to compute the degrees of freedom. The term $Z^{(1)}_q$ is extracted from the expansion~\eqref{eq:ExpansionZ} as follows:
\begin{align}
    Z^{(1)}_q &\ce \lim_{r\to\infty}r\left(Z_q(r) - Z^{(0)}_q\right) =\notag\\
    &= (n-1)\left(\frac12 n\,\alpha^{(n)}_q + \alpha^{(n-1)}_q - \sum_{\ell = 0}^{p}\left\{\frac12 n + q + \ell\right\}\gamma_\ell\right)\,.
\end{align}
For a compatible equation $\R_q$, the coefficient $Z^{(1)}_q$ simplifies to
\begin{align}\label{eq:Z1q}
    Z^{(1)}_q = (n-1)\left(\alpha^{(n-1)}_q - q\sum_{\ell = 0}^{p}\gamma_\ell - \sum_{\ell = 1}^{p} \ell \gamma_\ell\right)\,. 
\end{align}
We claim that this expression directly measures the physical degrees of freedom propagated by the involutive equation $\R_q$. To clarify the interpretation of $Z^{(1)}_q$, it is helpful to rewrite the Cartan character $\alpha^{(n-1)}_q$ in terms of $\beta^{(n-1)}_q$. From Definition~\ref{def:CartanCharacters}, we have
\begin{align} 
    \alpha^{(n-1)}_q = q\,m - \beta^{(n-1)}_q\,. 
\end{align}
Substituting this into the expression for $Z^{(1)}_q$, we obtain
\begin{align}\label{eq:FinalZ1}
    Z^{(1)}_q = (n-1)\left(q\,m - \beta^{(n-1)}_q - q\sum_{\ell = 0}^{p}\gamma_\ell - \sum_{\ell = 1}^{p} \ell \gamma_\ell\right)\,. 
\end{align}
To interpret $Z^{(1)}_q$, let us first consider the case without gauge symmetry, i.e., $\gamma_\ell = 0$. In this scenario, the expression in parentheses simplifies to
\begin{align} 
    q\,m - \beta^{(n-1)}_q\,. 
\end{align}
The first term, $q\,m$, represents the number of field components multiplied by the order of the differential equation. In the absence of gauge symmetry and constraints, this would correspond to $q$ times the degrees of freedom. It is also a measure of how many initial value functions or boundary conditions need to be specified in order to obtain a unique solution to $\R_q$.

However, constraints can exist even in the absence of gauge symmetry. The second term, $\beta^{(n-1)}_q$, accounts for these constraints. To see why, recall that $\beta^{(k)}_q$ counts the number of equations of class $k$. For a $q$-th order system, an equation of class $n$ contains $q$ derivatives with respect to $x^n$. In physical systems, these typically correspond to equations involving second-order time derivatives. On the other hand, equations of class $n-1$ contain at least one derivative with respect to $x^{n-1}$, meaning they lack the highest-order time derivative and are therefore classified as constraints. 

When gauge symmetry is present, two additional correction terms appear:
\begin{align} 
    q\,m - \beta^{(n-1)}_q - q\sum_{\ell = 0}^p \gamma_\ell - \sum_{\ell=1}^p \ell \gamma_\ell\,. 
\end{align}
Since we are dealing with compatible equations, we know that
\begin{align}\label{eq:C}
    \sum_{\ell = 0}^p \gamma_\ell = \alpha^{(n)}_q\,, 
\end{align}
which represents the number of gauge modes. Thus, the term $q \sum_{\ell = 0}^p \gamma_\ell$ removes from $q\,m$ all degrees of freedom that are purely gauge. That is, it removes everything that is not solely fixed by initial or boundary conditions. 

Finally, the term $\sum_{\ell = 1}^p \ell \gamma_\ell$ combines with $\beta^{(n-1)}_q$ to account for all constraints, including those that arise due to gauge symmetry. This leads us to the following conjecture: for an involutive equation $\R_q$ in $n$ spacetime dimensions, the number of configuration space degrees of freedom is given by
\begin{align}\label{eq:FinalFormula}
    \boxed{\text{DOFs} = \frac{Z^{(1)}_q}{(n-1) q} = m - \sum_{\ell = 1}^p \gamma_\ell - \frac{1}{q}\left(\beta^{(n-1)}_q + \sum_{\ell = 1}^p \ell \gamma_\ell\right)}\,.
\end{align}
What remains to be shown is that this formula is indeed well-defined, i.e., that it always yields a non-negative integer. We can at least show positivity for theories which lack gauge symmetry. In the absence of gauge, all $\gamma_\ell$'s are zero. Then~\eqref{eq:FinalFormula} reduces to 
\begin{align}
    \text{DOFs} = m - \frac{1}{q} \beta^{(n-1)}_q\,.
\end{align}
Recall that we can organize jet variables in classes and that these classes have a predictable size (cf. equation~\eqref{eq:ClassSize}). Recall also that we defined the $\beta$'s to be the number of pivot elements within a given class (see, in particular, the schematic row-echelon matrix~\eqref{eq:SymbolInRowEchelonForm}). Evidently, the value of $\beta^{(k)}_q$ cannot exceed the number of elements within the class $k$. Thus, we always have
\begin{align}
    \beta^{(k)}_q \leq \csize(\class k) = m\binom{n+q-k-1}{n-k}\,.
\end{align}
For the special case $k=n-1$ this translates into the inequality
\begin{align}
    \beta^{(n-1)}_q \leq q\,m \,.
\end{align}
It therefore follows that
\begin{align}\label{eq:PositiveDOFsNonGauge}
    \text{DOFs} = m - \frac{1}{q} \beta^{(n-1)}_q \geq 0\,.
\end{align}
In the presence of gauge symmetry, and for theories satisfying the compatibility condition~\eqref{eq:CompatibilityCondition}, we can provide a partial estimate:
\begin{align}
    \text{DOFs} &= m - \sum_{\ell=0}^p \gamma_\ell - \frac{1}{q} \left(\beta^{(n-1)}_q + \sum_{\ell=1}^p \ell \gamma_\ell\right) \notag\\
    &= m - \alpha^{(n)}_q - \frac{1}{q} \left(\beta^{(n-1)}_q + \sum_{\ell=1}^p \ell \gamma_\ell\right) \notag\\
    &= m - \left(m-\beta^{(n)}_q\right) - \frac{1}{q} \left(\beta^{(n-1)}_q + \sum_{\ell=1}^p \ell \gamma_\ell\right) \notag\\
    &= \underbrace{\beta^{(n)}_q - \frac{1}{q}\beta^{(n-1)}_q}_{\geq 0} - \frac{1}{q}\sum_{\ell=1}^p \ell \gamma_\ell \,.
\end{align}
We used the fact that $\alpha^{(n)}_q = m - \beta^{(n)}_q$ and that $\beta^{(n)}_q \geq \beta^{(n-1)}_q$ (see~\cite{SeilerBook} for a proof of the latter). For equations with $q>1$, the first two terms are strictly larger than zero. In order for the whole expression to be larger or equal to zero, we need
\begin{align}\label{eq:UpperBoundSum}
    \frac{1}{q}\beta^{(n-1)}_q + \frac{1}{q}\sum_{\ell=0}^p \ell \gamma_\ell \leq \beta^{(n-1)}_q \quad &\Longleftrightarrow\quad \sum_{\ell=0}^p \ell \gamma_\ell \leq (q-1) \beta^{(n-1)}_q \notag\\
    &\Longrightarrow \sum_{\ell=0}^p \ell \gamma_\ell \leq (q-1) \beta^{(n)}_q \leq q\, m - m
\end{align}
On the second line we used $\beta^{(n-1)}_q \leq \beta^{(n)}_q$ and the fact that $\beta^{(n)}_q \leq \csize(\class n) = m$. Some restriction on the $\gamma_\ell$'s is of course reasonable, since a theory with too much gauge freedom is either trivial (all field components are gauge) or logically impossible (more gauge modes than field components). The above inequality then expresses the condition that gauge symmetry cannot remove more than $q\,m-m$ functions. In fact, $q\,m$ is the total number of available functions in the initial value formulation of a $q$-th order PDE for $m$ field components. Also, $\beta^{(n-1)}_q$, the number of constraints, has an upper limit of $m$. This follows again from $\beta^{(n-1)}_q \leq \beta^{(n)}_q \leq m$, but it is also intuitively clear: If there were $m$ constraints, then none of the $m$ field components would be propagating. We therefore postulate
\begin{align}\label{eq:PostulatedUpperBound}
    \sum_{\ell =0}^p \ell \gamma_\ell \leq (q-1) \beta^{(n-1)}_q
\end{align}
as a reasonable upper bound. Curiously, in all examples we study in Subsection~\ref{ssec:Examples} we find that this inequality is saturated and therefore
\begin{align}\label{eq:PositiveDOFs}
    \text{DOFs} = \beta^{(n)}_q -  \beta^{(n-1)}_q \geq 0\,.
\end{align}
Whether this is true beyond the examples we considered cannot be said at this stage. Also, we have not answered the question whether $\text{DOFs}$ is an integer number. This remains a conjecture until one can demonstrate that
\begin{align}
    \frac{1}{q}\left(\beta^{(n-1)}_q + \sum_{\ell = 1}^p \ell \gamma_\ell\right)
\end{align}
is always an integer. We only know with certainty that $Z^{(1)}_q$ is an integer, which implies that the \emph{phase space degrees of freedom} are integer numbers. 

In the next subsection, we put this conjecture to the test by applying the Cartan-Kuranishi algorithm to various well-known physical systems. This allows us to extract the $\beta$'s and Cartan characters, which in turn enable us to verify that all considered equations are compatible as defined in this subsection (cf.~\eqref{eq:CompatibilityCondition}). Furthermore, we use these quantities to compute the number of physical degrees of freedom. In every case, our formula~\eqref{eq:FinalFormula} yields the correct result.

\subsection{Examples: Electromagnetism, General Relativity, and many more}\label{ssec:Examples}
In the examples that follow, we assume the spacetime dimension is $n = 4$ and that all equations are second-order, so we always have $q = 2$.

For each example, we apply the Cartan-Kuranishi algorithm to the field equations to ensure that they are involutive before using them to compute the degrees of freedom.

Each example is structured as follows: We begin with a table summarizing the input, which includes the field content, the number of field components, the field equations, and the set of integers $\gamma_\ell$.

In the output table, we record the number $r$ of prolongations and the number $s$ of projections required to obtain an involutive equation. Additionally, we report the $\beta$'s and Cartan characters, the Hilbert polynomial, and in the case of gauge theories, also the gauge-corrected Hilbert polynomial and Cartan characters. These quantities are then used to determine the degrees of freedom according to~\eqref{eq:FinalFormula}.

\subsubsection{The relativistic wave equation}\label{sssec:WaveEq}
We begin with a simple example of a field theory without gauge symmetry: The relativistic wave equation for the scalar field $\Phi$. 
A scalar field is the simplest type of quantum field, characterized by assigning a single value (a scalar) to every point in spacetime and transforming trivially under Lorentz transformations. The dynamics of a free scalar field are governed by the Klein–Gordon equation $\Box\Phi=0$, which enforces the relativistic energy–momentum relation and ensures Lorentz invariance. It originates from a Lagrangian containing only the kinetic term of the scalar field $\mathcal{L}=-\frac12\partial_\alpha\Phi\partial^\alpha\Phi$. In this case, it is evident that there is a single physical propagating degree of freedom associated with the scalar field. The canonical momentum is $\pi=\frac{\partial\mathcal{L}}{\partial(\partial_0\Phi)}=\dot{\Phi}$. There are no constraints (primary or gauge): the pair $(\Phi,\pi)$ is unconstrained at each spatial point. In Hamiltonian language, the number of physical propagating degrees of freedom equals the number of independent canonical pairs. Equivalently, the Cauchy problem for $\Box\Phi=0$ needs exactly $\Phi(\vec{x},t_0)$ and $\dot\Phi(\vec{x},t_0)$ to determine the solution. 
Although the counting of degrees of freedom is straightforward in this simplest case, it is nevertheless instructive to examine how the single propagating mode manifests itself within the framework developed in this work.\footnote{It would also be interesting to apply the method to the more nontrivial case of Galileon theories with derivative interactions, although this lies beyond the scope of the present study.}.

\begin{table}[H] 
    \centering
    \begin{tabular}{ll}
        \multicolumn{2}{c}{\textbf{Input}}  \\ \toprule
        \textbf{Field content} $\boldsymbol{v^{A}}$ & $\Phi$\\[5pt]
        \textbf{Number of field components} $\boldsymbol{m}$ & $1$ \\[5pt]
        \textbf{Field equations} $\boldsymbol{\R_q}$ &  $\partial_\alpha \partial^\alpha \Phi = 0$\\[5pt]
         $\boldsymbol{\gamma_\ell}$ & $\gamma_\ell = 0$ $\forall \ell$ \\ \bottomrule
         \\
         \multicolumn{2}{c}{\textbf{Output}}  \\ \toprule
         \textbf{Involutive after $\boldsymbol{s}$ projections} & $s=0$ \\[5pt]
         \textbf{Involutive after $\boldsymbol{r}$ prolongations} & $r=0$ \\[5pt]
         $\boldsymbol{\beta^{(k)}_q}$ & $\beta^{(1)}_2 = 0$, $\beta^{(2)}_2 = 0$, $\beta^{(3)}_2 = 0$, $\beta^{(4)}_2 = 1$ \\[5pt]
         $\boldsymbol{\alpha^{(k)}_q}$ & $\alpha^{(1)}_2 = 4$, $\alpha^{(2)}_2 = 3$, $\alpha^{(3)}_2 = 2$, $\alpha^{(4)}_2 = 0$ \\[5pt]
         $\boldsymbol{H_q(r)}$ & $9 + 6r + r^2$\\[5pt]
         \textbf{Degrees of freedom} & $1$ \\ \bottomrule
    \end{tabular}
\end{table}
First, due to the absence of gauge freedom, all 
$\gamma_\ell$ identically vanish.
 As can be seen from the output table, the wave equation is involutive without the need to perform prolongations nor projections ($r=0$ and $s=0$). All $\beta$'s, except $\beta^{(4)}_2$, vanish. This is consistent with our discussion on constraint equations in the previous subsection: The vanishing of $\beta^{(3)}_2$ signals the absence of constraint equations while $\beta^{(4)}_2 = 1$ tells us that there is one principal derivative of class $4$. In physics terms, this is of course the second order time derivative $\partial_t \partial_t \Phi$.

As expected, $\alpha^{(4)}_2$ is zero, which tells us that there are no arbitrary functions of $4$ coordinates in the solution of the wave equation. On the other hand, $\alpha^{(3)}_2 = 2$ signals the presence of two functions of three coordinates in the general solutions. This is precisely what we would expect from the initial value problem: We need to prescribe the initial ``position'' $\Phi|_\Sigma$ and the initial ``velocity'' $\dot{\Phi}|_\Sigma$ on a Cauchy surface $\Sigma$ in order to obtain a unique solution. 

Finally, using~\eqref{eq:FinalFormula} we find that the relativistic wave equation propagates $1$ degree of freedom, as expected. \newpage

\subsubsection{Maxwell's equations}\label{sssec:Maxwell}
This is our first example of a gauge theory. 
Physically, the Maxwell equations describe the propagation of massless spin-1 excitations, the photons, which have two transverse polarization states due to gauge invariance and the masslessness of the field. The traditional way of counting works as follows: The vanishing of the temporal momentum $\Pi^0=0$ defines a primary constraint, $\mathcal{C}_1=\Pi^0=0$. Requiring its preservation under time evolution yields the secondary constraint (obtained via the Poisson bracket with the Hamiltonian density $\mathcal{H}$) $\mathcal{C}_2=\dot{\Pi}^0=\left\{\mathcal{C}_1,\mathcal{H} \right\}=\partial_i\Pi^i$=0. This primary constraint is first-class, as their mutual Poisson brackets vanish $\left\{\mathcal{C}_1,\mathcal{C}_2 \right\}=0$. In the Dirac–Bergmann framework, each first-class constraint removes one configuration-space degree of freedom and one conjugate momentum, so together they eliminate two phase-space dimensions. This accounts for the removal of the unphysical longitudinal mode, leaving precisely two propagating degrees of freedom corresponding to the transverse polarizations of the massless vector field. The first-class character of the constraints reflects the presence of a gauge symmetry in the theory $A_\mu \to A_\mu + \partial_\mu\theta$. At the Hamiltonian level, this is manifested by the temporal component of the vector field $A_0$
  appearing as a Lagrange multiplier, $\mathcal{H}=\int d^3x \left( \Pi_i\Pi^i-A_0 \partial_i \Pi^i -\frac14 F_{ij}F^{ij}\right)$. The constraint associated with $A_0$
  is first-class, which implies that one eliminates not only the explicit 
$A_0$ dependence in the Hamiltonian but also the dependence on the longitudinal component of the canonical momentum $\partial_i \Pi^i$. This gives us $4-1-1 = 2$ degrees of freedom. 

\begin{table}[H] 
    \centering
    \begin{tabular}{ll}
        \multicolumn{2}{c}{\textbf{Input}}  \\ \toprule
        \textbf{Field content} $\boldsymbol{v^{A}}$ & $A^\mu$\\[5pt]
        \textbf{Number of field components} $\boldsymbol{m}$ & $4$ \\[5pt]
        \textbf{Field equations} $\boldsymbol{\R_q}$ &  $\partial_\alpha \partial^\alpha A^\mu - \partial_\alpha \partial^\mu A^\alpha = 0$\\[5pt]
         $\boldsymbol{\gamma_\ell}$ & $\gamma_1 = 1$ \\ \bottomrule
         \\
         \multicolumn{2}{c}{\textbf{Output}}  \\ \toprule
         \textbf{Involutive after $\boldsymbol{s}$ projections} & $s=0$\\[5pt]
         \textbf{Involutive after $\boldsymbol{r}$ prolongations} & $r=0$ \\[5pt]
         $\boldsymbol{\beta^{(k)}_q}$ & $\beta^{(1)}_2 = 0$, $\beta^{(2)}_2 = 0$, $\beta^{(3)}_2 = 1$, $\beta^{(4)}_2 = 3$ \\[5pt]
         $\boldsymbol{\alpha^{(k)}_q}$ & $\alpha^{(1)}_2 = 16$, $\alpha^{(2)}_2 = 12$, $\alpha^{(3)}_2 = 7$, $\alpha^{(4)}_2 = 1$ \\[5pt]
         $\boldsymbol{H_q(r)}$ & $36 + \frac{73}{3}r + \frac{9}{2}r^2 + \frac16 r^3$\\[5pt]
         $\boldsymbol{\bar{H}_q(r)}$ & $16 + 12 r + 2 r^2$\\[5pt]
         $\boldsymbol{\bar{\alpha}^{(k)}_q}$ & $\bar{\alpha}^{(1)}_2 = 6$, $\bar{\alpha}^{(2)}_2 = 6$, $\bar{\alpha}^{(3)}_2 = 4$, $\bar{\alpha}^{(4)}_2 = 0$ \\[5pt]
         \textbf{Degrees of freedom} & $2$ \\ \bottomrule
    \end{tabular}
\end{table}

It is now instructive to compare this standard counting with the corresponding degree-of-freedom counting in the new formalism.
Our formula~\eqref{eq:FinalFormula} achieves the same result by systematically analyzing the field equations, the field content, and the gauge symmetries. 
As discussed in Example~\ref{ex:EDandGR}, all but one of the coefficients $\gamma_\ell$ vanish. The remaining coefficient, $\gamma_1 = 1$, reflects the fact that the gauge transformation in electromagnetism depends on the first derivative of an arbitrary scalar~field.

From the output table, we see that Maxwell's equations are involutive, meaning no prolongations or projections were necessary ($r=s=0$). The table also confirms that $\alpha^{(4)}_2 = 1$, which aligns with our expectations: one component of $A^\mu$ is a gauge mode. Consequently, Maxwell's equations satisfy the compatibility condition defined in the previous subsection: $\bar{\alpha}^{(4)}_2 = \alpha^{(4)}_2 - \gamma_1 = 0$.

Moreover, the value $\beta^{(3)}_2 = 1$ indicates the presence of a single constraint equation---none other than the well-known Gauss constraint. The gauge-corrected Cartan character $\bar{\alpha}^{(3)}_2 = 4$ signals the presence of four initial-value functions (two ``positions'' and two ``velocities'', so to speak). Applying formula~\eqref{eq:FinalFormula}, we conclude that Maxwell's equations propagate two physical degrees of freedom.

\newpage

\subsubsection{Proca's equations}\label{sssec:Proca}
The Proca theory includes an explicit mass term, which significantly alters the physical degrees of freedom. The Lagrangian now contains a quadratic term in $A_\mu$
 without derivatives, resulting in the loss of gauge symmetry $A_\mu \to A_\mu + \partial_\mu\theta$. The temporal component 
$A_0$ remains non-propagating, giving rise to a primary constraint $\mathcal{C}_1=\Pi^0=0$. However, the presence of the mass term modifies the secondary constraint, which now takes the form $\mathcal{C}_2=\partial_i\Pi^i-m^2A_0=0$. Unlike the massless case, the Poisson bracket between the primary and secondary constraints does not vanish $\left\{\mathcal{C}_1,\mathcal{C}_2 \right\}=m^2$, so the constraints are second-class. In Dirac’s terminology, each second-class constraint eliminates one phase-space degree of freedom, or equivalently, half a configuration-space degree of freedom per constraint. Here we have two second-class constraints ($\Pi^0$ and the $\partial_i\Pi^i-m^2A_0$), which together remove one configuration-space degree of freedom. Starting from 4 components $A_\mu$
  and subtracting this single eliminated mode, we are left with 3 physical degrees of freedom.

\begin{table}[H] 
    \centering
    \begin{tabular}{ll}
        \multicolumn{2}{c}{\textbf{Input}}  \\ \toprule
        \textbf{Field content $\boldsymbol{v^{A}}$} & $A^\mu$\\[5pt]
        \textbf{Number of field components $\boldsymbol{m}$} & $4$ \\[5pt]
        \textbf{Field equations} $\boldsymbol{\R_q}$ &  $\partial_\alpha \partial^\alpha A^\mu - \partial_\alpha \partial^\mu A^\alpha - m^2 A^\mu = 0$\\[5pt]
         $\boldsymbol{\gamma_\ell}$ & $\gamma_\ell = 0$ $\forall \ell$ \\ \bottomrule
         \\
         \multicolumn{2}{c}{\textbf{Output}}  \\ \toprule
         \textbf{Involutive after $\boldsymbol{s}$ projections} & $s=2$ \\[5pt]
         \textbf{Involutive after $\boldsymbol{r}$ prolongations} & $r=0$ \\[5pt]
         $\boldsymbol{\beta^{(k)}_q}$ & $\beta^{(1)}_2 = 1$, $\beta^{(2)}_2 = 1$, $\beta^{(3)}_2 = 2$, $\beta^{(4)}_2 = 4$ \\[5pt]
         $\boldsymbol{\alpha^{(k)}_q}$ & $\alpha^{(1)}_2 = 15$, $\alpha^{(2)}_2 = 11$, $\alpha^{(3)}_2 = 6$, $\alpha^{(4)}_2 = 0$ \\[5pt]
         $\boldsymbol{H_q(r)}$ & $32 + 20 r + 3 r^2$\\[5pt]
         \textbf{Degrees of freedom} & $3$ \\ \bottomrule
    \end{tabular}
\end{table}
Since the mass term of Proca's theory explicitly breaks gauge symmetry, all coefficients $\gamma_\ell$ vanish.
The mass term also renders Proca's equations non-involutive. As shown in Example~\ref{ex:CKAAppliedToProca}, achieving an involutive system requires two prolongations followed by two projections (hence $s=2$). In that example, we also determined the $\beta$'s and found that $\beta^{(3)}_2 = 2$, indicating the presence of two constraint equations. A closer look at line 11 of Example~\ref{ex:CKAAppliedToProca} confirms this: there, $\R^{(2)}_2$ contains the constraint $\partial_\mu A^\mu = 0$ along with its prolongation $\partial_\nu \partial_\mu A^\mu = 0$.

At first glance, subtracting this constraint twice may seem counterintuitive. However, from the perspective of the jet bundle formalism, this is both natural and necessary. As discussed after Equation~\eqref{eq:FinalZ1}, when completing a system of PDEs to an involutive system, we must include prolongations of constraint equations. This shifts $\beta^{(3)}_2$ from zero to $2$, making it proportional to $q=2$. The same phenomenon occurs here. Moreover, if our conjecture from the previous subsection is correct, we should always expect $\beta^{(n-1)}_q$ to be proportional to $q$. Only then is $q\,m - \beta^{(n-1)}_q$ guaranteed to be divisible by $q$.

We note that for Proca’s equations, $\alpha^{(4)}_2$ vanishes. This aligns with the absence of gauge symmetry and satisfies the definition of compatibility. Finally, using our formula~\eqref{eq:FinalFormula}, we find the expected three degrees of freedom.

\newpage

\subsubsection{Proca-St\"{u}ckelberg equations}\label{sssec:ProcaStuckelberg}
Proca's theory can be reformulated as a gauge theory using St\"uckelberg's trick, as mentioned in Example~\ref{ex:ProcaStueckelberg}. To restore gauge symmetry, we introduce an additional field, $\pi$. Consequently, compared to standard Proca theory, the field content is augmented. This also introduces an additional field equation, obtained by varying the Proca-St\"uckelberg action with respect to $\pi$ (see, for instance,~\cite{Heisenberg:2018} for more details).
\begin{table}[H] 
    \centering
    \begin{tabular}{ll}
        \multicolumn{2}{c}{\textbf{Input}}  \\ \toprule
        \textbf{Field content} $\boldsymbol{v^{A}}$ & $(A^\mu, \pi)$\\[5pt]
        \textbf{Number of field components} $\boldsymbol{m}$ & $4 + 1 = 5$ \\[5pt]
        \textbf{Field equations} $\boldsymbol{\R_q}$ & $\partial_\alpha \partial^\alpha A^\mu - \partial_\alpha \partial^\mu A^\alpha - m^2 A^\mu - m\, \partial^\mu \pi = 0$ \\[5pt]
        & $m\, \partial_\alpha A^\alpha + \partial_\alpha \partial^\alpha \pi = 0$ \\[5pt]
         $\boldsymbol{\gamma_\ell}$ & $\gamma_1 = 1$ \\ \bottomrule
         \\
         \multicolumn{2}{c}{\textbf{Output}}  \\ \toprule
         \textbf{Involutive after $\boldsymbol{s}$ projections} & $s=0$ \\[5pt]
         \textbf{Involutive after $\boldsymbol{r}$ prolongations} & $r=0$ \\[5pt]
         $\boldsymbol{\beta^{(k)}_q}$ & $\beta^{(1)}_2 = 0$, $\beta^{(2)}_2 = 0$, $\beta^{(3)}_2 = 1$, $\beta^{(4)}_2 = 4$ \\[5pt]
         $\boldsymbol{\alpha^{(k)}_q}$ & $\alpha^{(1)}_2 = 20$, $\alpha^{(2)}_2 = 15$, $\alpha^{(3)}_2 = 9$, $\alpha^{(4)}_2 = 1$ \\[5pt]
         $\boldsymbol{H_q(r)}$ & $45 + \frac{91}{3}r + \frac{11}{2}r^2 + \frac16 r^3$\\[5pt]
         $\boldsymbol{\bar{H}_q(r)}$ & $25 + 18 r + 3 r^2$\\[5pt]
         $\boldsymbol{\bar{\alpha}^{(k)}_q}$ &  $\bar{\alpha}^{(1)}_2 = 10$, $\bar{\alpha}^{(2)}_2 = 9$, $\bar{\alpha}^{(3)}_2 = 6$, $\bar{\alpha}^{(4)}_2 = 0$ \\[5pt]
         \textbf{Degrees of freedom} & $3$ \\ \bottomrule
    \end{tabular}
\end{table}

As discussed after Example~\ref{ex:ProcaStueckelberg}, the gauge transformation acting on $A^\mu$ and $\pi$ is characterized by $\gamma_1 = 1$, while all other $\gamma_\ell$ vanish.
As shown in the output table, the Proca-St\"uckelberg equations are involutive, meaning no prolongations or projections are required. From these equations, we extract the $\beta$'s and the Cartan characters. Of particular interest is $\alpha^{(4)}_2 = 1$, which indicates the presence of a gauge mode. However, since $\gamma_1 = 1$, the compatibility condition is satisfied. Moreover, formula~\eqref{eq:FinalFormula} correctly yields the expected three degrees of freedom.
\newpage

\subsubsection{Einstein's field equations in Riemannian normal coordinates}\label{sssec:GR}
So far, we have studied linear equations. With Einstein's field equations, we now encounter a nonlinear system that also exhibits gauge symmetry. 
In GR the fundamental field is $g_{\mu\nu}$. Thus, we would naively expect to find $4\times 4 = 16$ components. However, $g_{\mu\nu}$ is symmetric. With this in mind, we find that $g_{\mu\nu}$ really only possesses ten components. Einstein's field equations are also symmetric in $\mu$ and $\nu$. The same counting argument therefore tells us that there are ten partial differential equations for ten unknown functions.

However, not all of Einstein's field equations are dynamical. In fact, there are four constraint equations, which can be seen as follows: Starting from the Bianchi identity $\nabla_\mu G^{\mu\nu} = 0$,  we expand the contraction over $\mu$. This yields
\begin{align}\label{eq:ExpandedBianchi}
	\nabla_\mu G^{\mu \nu} = \nabla_{\red{0}} G^{\red{0}\nu} + \nabla_{\blue{i}} G^{\blue{i}\nu} = 0 \quad\Longrightarrow\quad \partial_{\red{0}} G^{\red{0}\nu} = - \partial_{\blue{i}} G^{\blue{i}\nu} - \left\{\nu \atop \sigma \mu \right\} \, G^{\sigma\mu} - \left\{\mu \atop \sigma\mu \right\}\, G^{\nu\sigma}\,,
\end{align}
where the curly brackets denote the Christoffel symbols of the metric. At this point, it is important to recall that $G_{\mu\nu}$ contains second order derivatives of the metric. In particular, it contains \textit{at most} second order time derivatives. The first term on the right hand side of~\eqref{eq:ExpandedBianchi} contains third order spatial derivatives because of $\partial_{\blue{i}}$, but no more than second order time derivatives. Similarly, the other two terms contain at most second order time derivatives (the Christoffel symbols contain only first order derivatives). The term on the left hand side, $\partial_{\red{0}} G^{\red{0}\nu}$, looks like it could contain third order time derivatives. That is because $\partial_{\red{0}}$ increases the order of time derivatives by one. However, this can not be the case, since the right hand side contains at most second order time derivatives, as we just have convinced ourselves. Thus, it follows that the components $G^{\red{0}\nu}$ contain only first order time derivatives. This shows that the $G^{\red{0}\nu}$ components of Einstein's field equations are constraint equations, not dynamical equations.

In GR, gauge transformations depend on four arbitrary functions $\xi^{\nu}$ and these functions can be used to choose four of the ten components of $g_{\mu\nu}$ at will. These components cannot be physical, they are mere gauge. Thus, we are left with six components. We can further use the four equations of the Bianchi identity  $\nabla_{\mu} G^{\mu\red{\nu}} = 0$ to get rid of the dependence of 4 components in the equations. It follows that GR propagates $10-4-4 = 2$ degrees of freedom.  

Let us now examine how the degree-of-freedom counting is manifested within our methodology in the jet-bundle framework.
In Example~\ref{ex:EDandGR}, we examined the gauge symmetry of General Relativity and found that $\gamma_1 = 4$, while all other $\gamma_\ell$ vanish.
In order to drastically simplify the computations, we resort to Riemannian normal coordinates. This choice of coordinates gets rid of nonlinearities, as can be seen from the expression for $\R_2$. 
\begin{table}[H] 
    \centering
    \begin{tabular}{ll}
        \multicolumn{2}{c}{\textbf{Input}}  \\ \toprule
        \textbf{Field content} $\boldsymbol{v^{A}}$ & $g_{\mu\nu}$\\[5pt]
        \textbf{Number of field components} $\boldsymbol{m}$ & $10$ \\[5pt]
        \textbf{Field equations} $\boldsymbol{\R_q}$ &  $\eta^{\alpha\beta}\left(\partial_\alpha\partial_\beta g_{\mu\nu} - \partial_\beta \partial_\mu g_{\nu\alpha} + \partial_\nu \partial_\mu g_{\alpha\beta}\right) = 0$\\[5pt]
         $\boldsymbol{\gamma_\ell}$ & $\gamma_1 = 4$ \\ \bottomrule
         \\
         \multicolumn{2}{c}{\textbf{Output}}  \\ \toprule
         \textbf{Involutive after $\boldsymbol{s}$ projections} & $s=0$ \\[5pt]
         \textbf{Involutive after $\boldsymbol{r}$ prolongations} & $r=0$ \\[5pt]
         $\boldsymbol{\beta^{(k)}_q}$ & $\beta^{(1)}_2 = 0$, $\beta^{(2)}_2 = 0$, $\beta^{(3)}_2 = 4$, $\beta^{(4)}_2 = 6$ \\[5pt]
         $\boldsymbol{\alpha^{(k)}_q}$ & $\alpha^{(1)}_2 = 40$, $\alpha^{(2)}_2 = 30$, $\alpha^{(3)}_2 = 16$, $\alpha^{(4)}_2 = 4$ \\[5pt]
         $\boldsymbol{H_q(r)}$ & $90+\frac{184}{3}r + 12 r^2 + \frac23 r^3$\\[5pt]
         $\boldsymbol{\bar{H}_q(r)}$ & $10 + 12 r + 2 r^2$\\[5pt]
         $\boldsymbol{\bar{\alpha}^{(k)}_q}$ & $\bar{\alpha}^{(1)}_2 = 0$, $\bar{\alpha}^{(2)}_2 = 6$, $\bar{\alpha}^{(3)}_2 = 4$, $\bar{\alpha}^{(4)}_2 = 0$ \\[5pt]
         \textbf{Degrees of freedom} & $2$ \\ \bottomrule
    \end{tabular}
\end{table}

Remarkably, Einstein's field equations are already involutive, meaning no prolongations or projections are required, so $r = s = 0$. As shown in the output table, we find $\alpha^{(4)}_2 = 4$ and $\beta^{(3)}_2 = 4$. These values align perfectly with expectations: the highest-order Cartan character confirms the presence of four gauge modes among the ten metric components, while $\beta^{(3)}_2$ indicates that four of Einstein's field equations are constraints.

Since $\gamma_1 = 4$, the compatibility condition $\alpha^{(4)}_2 - \gamma_1 = 0$ is satisfied. Furthermore, using the values from the output table in formula~\eqref{eq:FinalFormula}, we conclude that General Relativity propagates two physical degrees of freedom.

Finally, we point out that the theory could also haven been analyzed in a different coordinate system, without spoiling the final result. More on the role of coordinate systems and under which conditions they have an influence on the values of characters can be found in~\cite{SeilerBook}.
\newpage

\subsubsection{Coincident General Relativity in Riemannian normal coordinates}\label{sssec:CGR}
Coincident General Relativity (CGR) \cite{BeltranJimenez:2017} is a reformulation of General Relativityin where non-metricity, rather than curvature, plays the central role. Despite this fundamental difference, GR and CGR are equivalent in the sense that they share the same solution space (see, for example,~\cite{BeltranJimenez:2017, BeltranJimenez:2018, BeltranJimenez:2019, Heisenberg:2023a}, or~\cite{DAmbrosio:2020a} for a Hamiltonian perspective).
\begin{table}[H] 
    \centering
    \begin{tabular}{ll}
        \multicolumn{2}{c}{\textbf{Input}}  \\ \toprule
        \textbf{Field content} $\boldsymbol{v^{A}}$ & $g_{\mu\nu}$\\[5pt]
        \textbf{Number of field components} $\boldsymbol{m}$ & $10$ \\[5pt]
        \textbf{Field equations} $\boldsymbol{\R_q}$ & $\eta^{\alpha\beta}\left(\partial_\alpha\partial_\beta g_{\mu\nu} - \partial_\beta\partial_\mu g_{\alpha\nu} - \partial_\beta\partial_\nu g_{\alpha\mu}\right.$ \\[5pt] 
        & $\left.+ \eta_{\mu\nu} \partial^\gamma \partial_\beta g_{\alpha\gamma} - \eta_{\mu\nu} \partial^\gamma \partial_\gamma g_{\alpha\beta} + \partial_\mu \partial_\nu g_{\alpha\beta}\right) = 0$ \\[5pt]
         $\boldsymbol{\gamma_\ell}$ & $\gamma_1 = 4$ \\ \bottomrule
         \\
         \multicolumn{2}{c}{\textbf{Output}}  \\ \toprule
         \textbf{Involutive after $\boldsymbol{s}$ projections} & $s=0$ \\[5pt]
         \textbf{Involutive after $\boldsymbol{r}$ prolongations} & $r=0$ \\[5pt]
         $\boldsymbol{\beta^{(k)}_q}$ & $\beta^{(1)}_2 = 0$, $\beta^{(2)}_2 = 0$, $\beta^{(3)}_2 = 4$, $\beta^{(4)}_2 = 6$ \\[5pt]
         $\boldsymbol{\alpha^{(k)}_q}$ & $\alpha^{(1)}_2 = 40$, $\alpha^{(2)}_2 = 30$, $\alpha^{(3)}_2 = 16$, $\alpha^{(4)}_2 = 4$ \\[5pt]
         $\boldsymbol{H_q(r)}$ & $90+\frac{184}{3}r + 12 r^2 + \frac23 r^3$\\[5pt]
         $\boldsymbol{\bar{H}_q(r)}$ & $10 + 12 r + 2 r^2$\\[5pt]
         $\boldsymbol{\bar{\alpha}^{(k)}_q}$ & $\bar{\alpha}^{(1)}_2 = 0$, $\bar{\alpha}^{(2)}_2 = 6$, $\bar{\alpha}^{(3)}_2 = 4$, $\bar{\alpha}^{(4)}_2 = 0$ \\[5pt]
         \textbf{Degrees of freedom} & $2$ \\ \bottomrule
    \end{tabular}
\end{table}

The equivalence between GR and CGR can also be demonstrated in a straightforward manner within the jet-bundle framework. Upon fixing the coincident gauge, the connection becomes trivial, leaving the metric as the sole carrier of the dynamical degrees of freedom A direct comparison of the output table for CGR with that of GR reveals identical results, further reinforcing their correspondence.\newpage

\subsubsection{Massless Fierz-Pauli equations}\label{sssec:FP}
\begin{table}[H] 
    \centering
    \begin{tabular}{ll}
        \multicolumn{2}{c}{\textbf{Input}}  \\ \toprule
        \textbf{Field content} $\boldsymbol{v^{A}}$ & $h_{\mu\nu}$\\[5pt]
        \textbf{Number of field components} $\boldsymbol{m}$ & $10$ \\[5pt]
        \textbf{Field equations} $\boldsymbol{\R_q}$ &  $\eta^{\alpha\beta}\left(\partial_\alpha\partial_\beta h_{\mu\nu} - \partial_\beta\partial_\mu h_{\alpha\nu} - \partial_\beta \partial_\nu h_{\alpha\mu}\right.$\\[5pt]
        & $\left.+ \eta_{\mu\nu} \partial^\gamma \partial_\beta h_{\alpha\gamma} - \eta_{\mu\nu}\partial^\gamma \partial_\gamma h_{\alpha\beta} + \partial_\nu \partial_\mu h_{\alpha\beta}\right) = 0$\\[5pt]
         $\boldsymbol{\gamma_\ell}$ & $\gamma_1 = 4$ \\ \bottomrule
         \\
         \multicolumn{2}{c}{\textbf{Output}}  \\ \toprule
         \textbf{Involutive after $\boldsymbol{s}$ projections} & $s=0$ \\[5pt]
         \textbf{Involutive after $\boldsymbol{r}$ prolongations} & $r=0$ \\[5pt]
         $\boldsymbol{\beta^{(k)}_q}$ & $\beta^{(1)}_2 = 0$, $\beta^{(2)}_2 = 0$, $\beta^{(3)}_2 = 4$, $\beta^{(4)}_2 = 6$ \\[5pt]
         $\boldsymbol{\alpha^{(k)}_q}$ & $\alpha^{(1)}_2 = 40$, $\alpha^{(2)}_2 = 30$, $\alpha^{(3)}_2 = 16$, $\alpha^{(4)}_2 = 4$ \\[5pt]
         $\boldsymbol{H_q(r)}$ & $90+\frac{184}{3}r + 12 r^2 + \frac23 r^3$\\[5pt]
         $\boldsymbol{\bar{H}_q(r)}$ & $10 + 12 r + 2 r^2$\\[5pt]
         $\boldsymbol{\bar{\alpha}^{(k)}_q}$ & $\bar{\alpha}^{(1)}_2 = 0$, $\bar{\alpha}^{(2)}_2 = 6$, $\bar{\alpha}^{(3)}_2 = 4$, $\bar{\alpha}^{(4)}_2 = 0$ \\[5pt]
         \textbf{Degrees of freedom} & $2$ \\ \bottomrule
    \end{tabular}
\end{table}
Given a massless spin-$2$ field $h_{\mu\nu}$, the most general Lagrangian that is local, Lorentz-invariant, and first order in partial derivatives takes the form  
\begin{align}
    \mathcal{L} = \frac{c_1}{2} \partial_\rho h_{\mu\nu} \partial^\rho h^{\mu\nu} 
    - c_2 \partial_\mu h\ud{\mu}{\rho} \partial_\nu h^{\nu\rho} 
    + c_3 \partial_\mu h \partial_\nu h^{\mu\nu} 
    - \frac{c_4}{2} \partial_\mu h \partial^\mu h\,.
\end{align}
Here, $h \ce \eta^{\mu\nu} h_{\mu\nu}$, $\eta_{\mu\nu}$ is the Minkowski metric, and $c_i$ are four real but otherwise arbitrary coefficients.  

By imposing that the $(4,4)$ and $(4,i)$ components of $h_{\mu\nu}$ do not propagate—or equivalently, by requiring that $\mathcal{L}$ is invariant under linearized diffeomorphisms—one arrives at the Fierz-Pauli Lagrangian. This corresponds to the choice $c_2 = c_1$, $c_4 = c_3$, and the normalization $c_1 = 1 = c_3$ (see~\cite{Heisenberg:2018} for more details):  
\begin{align}\label{eq:FPL}
    \mathcal{L} = \frac{1}{2} \partial_\rho h_{\mu\nu} \partial^\rho h^{\mu\nu} 
    - \partial_\mu h\ud{\mu}{\rho} \partial_\nu h^{\nu\rho} 
    + \partial_\mu h \partial_\nu h^{\mu\nu} 
    - \frac{1}{2} \partial_\mu h \partial^\mu h\,.
\end{align}
The resulting field equations can be interpreted as a linearization of Einstein's equations around a Minkowski background. We study the Fierz-Pauli equations here for three key reasons:  
\begin{enumerate}
    \item Distinguishing equivalence from similarity: Two sets of field equations with identical $\alpha$ and $\beta$ characters are not necessarily equivalent. Comparing the output tables of the Fierz-Pauli and GR equations, we find that they share the same $\alpha$'s and $\beta$'s. However, their solution spaces are clearly different. This is because the Cartan characters and $\beta$'s encode information only about the highest-order derivatives in the PDE system. If two systems have the same highest-order structure, they will exhibit identical Cartan characters and $\beta$'s, even if they describe different physics.  
    \item Exploring ghost instabilities: In the next example, we examine what happens when the coefficients $c_i$ are detuned. This introduces ghosts into the theory, and we analyze how this affects the Cartan-Kuranishi algorithm and our method for counting degrees of freedom.
    \item Studying the massive spin-2 field: While a full discussion of massive gravity is beyond the scope of this paper, the massive Fierz-Pauli equations provide a simple yet insightful case for comparison with the massless limit. 
\end{enumerate} 
\newpage

\subsubsection{Detuned massless Fierz-Pauli equations}\label{sssec:DetunedFP}
\begin{table}[H] 
    \centering
    \begin{tabular}{ll}
        \multicolumn{2}{c}{\textbf{Input}}  \\ \toprule
        \textbf{Field content} $\boldsymbol{v^{A}}$ & $h_{\mu\nu}$\\[5pt]
        \textbf{Number of field components} $\boldsymbol{m}$ & $10$ \\[5pt]
        \textbf{Field equations} $\boldsymbol{\R_q}$ & $\eta^{\alpha\beta}\left(\partial_\alpha\partial_\beta h_{\mu\nu} + \partial_\beta\partial_\mu h_{\alpha\nu} + \partial_\beta \partial_\nu h_{\alpha\mu}\right.$\\[5pt]
        & $\left.+ \eta_{\mu\nu} \partial^\gamma \partial_\beta h_{\alpha\gamma} + \eta_{\mu\nu}\partial^\gamma \partial_\gamma h_{\alpha\beta} + \partial_\nu \partial_\mu h_{\alpha\beta}\right) = 0$\\[5pt]
         $\boldsymbol{\gamma_\ell}$ & $\gamma_\ell = 0$ $\forall \ell$ \\ \bottomrule
         \\
         \multicolumn{2}{c}{\textbf{Output}}  \\ \toprule
         \textbf{Involutive after $\boldsymbol{s}$ projections} & $s=0$ \\[5pt]
         \textbf{Involutive after $\boldsymbol{r}$ prolongations} & $r=0$ \\[5pt]
         $\boldsymbol{\beta^{(k)}_q}$ & $\beta^{(1)}_2 = 0$, $\beta^{(2)}_2 = 0$, $\beta^{(3)}_2 = 0$, $\beta^{(4)}_2 = 10 $ \\[5pt]
         $\boldsymbol{\alpha^{(k)}_q}$ & $\alpha^{(1)}_2 = 40$, $\alpha^{(2)}_2 = 30$, $\alpha^{(3)}_2 = 20$, $\alpha^{(4)}_2 = 0$ \\[5pt]
         $\boldsymbol{H_q(r)}$ & $90 + 60 r + 10 r^2$\\[5pt]
         \textbf{Degrees of freedom} & $10$ \\ \bottomrule
    \end{tabular}
\end{table}
This example is again based on the most general Lagrangian that is local, Lorentz-invariant, and first order in partial derivatives. However, this time we do not impose any conditions on the coefficients $c_i$. For concreteness, we choose the coefficients to be $c_1 = 1$, $c_2 = -1$, $c_3 = 1$, and $c_4 = -1$. This choice clearly violates both conditions we spelled out in the previous example and, as is well-known, introduces ghosts into the theory.

However, the presence of ghosts does not affect the methodology discussed in this work. Neither the Cartan-Kuranishi algorithm nor the counting of degrees of freedom is affected by their presence. In fact, the Cartan-Kuranishi algorithm tells us that the resulting equations are involutive, that there are no constraint equations ($\beta^{(3)}_2 = 0$), and that the equations are compatible ($\alpha^{(4)}_2 = 0$). But there are now 10 propagating degrees of freedom, which would not be compatible with a consistent massless spin-2 theory and give rise to ghost instabilities.
\newpage

\subsubsection{Massive Fierz-Pauli equations}\label{sssec:MassiveFP}
Constructing the field equations of the massive spin-$2$ field is straightforward if one recalls the field equations of the massless theory, since both theories share the same kinetic term. What we have to determine is the contribution of the mass term, and we easily find the equations
\begin{align}\label{eq:MassiveSpin2FEQ}
	2 \hat{\mathcal{E}}\du{\mu\nu}{\alpha\beta}h_{\alpha\beta} + m^2_h\left(c_1\, h_{\mu\nu} + c_2\, \eta_{\mu\nu} h\right) = 0\,.
\end{align}
One can verify that under the transformation $h_{\alpha\beta} \mapsto h_{\alpha\beta} + \xi^{\alpha}\partial_\alpha h_{\mu\nu} + \partial_\mu \xi^\alpha h_{\alpha\nu} + \partial_\nu\xi^\alpha h_{\alpha\mu}$, the kinetic term $\hat{\mathcal{E}}\du{\mu\nu}{\alpha\beta}h_{\alpha\beta}$ remains \textit{invariant}, but the terms $h_{\mu\nu}$ and $h$ obviously transform in a non-trivial way. Thus, the mass term breaks the gauge symmetry of the massless spin-$2$ theory. 

To determine whether there are constraints, we could again use the massless theory as a guide line and easily find four constraint equations. However, we can also make use of the fact that  $\hat{\mathcal{E}}\du{\mu\nu}{\alpha\beta}h_{\alpha\beta}$ satisfies a Bianchi identity, namely
\begin{align}
	\partial^\mu \hat{\mathcal{E}}\du{\mu\nu}{\alpha\beta}h_{\alpha\beta} = 0\,.
\end{align}
Therefore, by taking the divergence of the field equation~\eqref{eq:MassiveSpin2FEQ} we find
\begin{align}\label{eq:Spin2Constraints}
	c_1\, \partial^\mu h_{\mu\nu} + c_2 \partial_\nu h = 0\,.
\end{align}
Observe that this gives us four constraints, since this equation only contains first order derivatives and since there is one free index. Next, we need to check if we can obtain more constraint equations. One can convince oneself that it is not possible to generate further constraints by taking more derivative of the field equations or of the constraint equations we already found. What one can do, however, is to check whether all field equations are independent. By taking linear combinations of the field equations, one may find that terms with second order derivatives drop out, thus giving us further constraint equations. One particular linear combination of field equations is obtained by taking the trace:
\begin{align}
	2\eta^{\mu\nu}\hat{\mathcal{E}}\du{\mu\nu}{\alpha\beta}h_{\alpha\beta} + m_h^2\left(c_1 \eta^{\mu\nu} h_{\mu\nu} + c_2 \eta^{\mu\nu}\eta_{\mu\nu} h\right) &= 0\notag\\
	\Longrightarrow -2\partial_\mu\partial_\nu h^{\mu\nu} + 2 \partial_\nu \partial^\nu h + m_h^2\left(c_1+4c_2\right)h = 0\,.
\end{align}
To get to the second line, we used the definition of the Lichnerowicz operator\footnote{$\hat{\mathcal{E}}\du{\mu\nu}{\alpha\beta}$ is the Lichnerowicz operator which acts on $h_{\alpha\beta}$ as
\begin{align}
	\hat{\mathcal{E}}\du{\mu\nu}{\alpha\beta}h_{\alpha\beta} \ce -\frac12\left[\Box h_{\mu\nu}- 2 \partial^\alpha \partial_{(\mu}h_{\nu)\alpha} + \partial_\mu \partial_\nu h -\eta_{\mu\nu}\left(\Box h - \partial^\alpha \partial^\beta  h_{\alpha\beta}\right)\right]\,.
\end{align}}. Next, we can use the constraint equations~\eqref{eq:Spin2Constraints}, which tell us that
\begin{align}
	\partial^{\nu} h = -\frac{c_1}{c_2} \partial_\mu h^{\mu\nu}\,,
\end{align}
to get rid of the $\partial_\nu \partial^\nu h$ term. We then find
\begin{align}
	m^2_h \left(c_1+4c_2\right)h = 2\left(1+\frac{c_1}{c_2}\right)\partial_\mu\partial_\nu h^{\mu\nu}\,.
\end{align}
From this equation we learn that if we choose $c_2 = -c_1$, we obtain a fifth constraint equation, namely
\begin{equation}
	h = 0\,.
\end{equation}
Given that the number of equations is ten and that we have zero gauge symmetry, we can finally conclude
\begin{align}
	\text{d.o.f.} = \begin{cases}
		10 - 0 - 4 = 6 & \text{if } c_2 \neq -c_1\\
		10 - 0 - 5 = 5 & \text{if } c_2 = -c_1
	\end{cases}\,.
\end{align}
We emphasize again that in the case where $c_2 \neq -c_1$, one of the degrees of freedom is a ghost. The healthy, ghost-free theory has to satisfy $c_2 = -c_1$ and it propagates five physical degrees of freedom. As a final note, we observe that in the case where $c_2 = -c_1$, the field equations~\eqref{eq:MassiveSpin2FEQ} can be equivalently rewritten as
\begin{align}
	\begin{cases}
		\hfill \left(\Box - m^2_h\right)h_{\mu\nu} &= 0\\
		\hfill\partial^\mu h_{\mu\nu} - \partial_\nu h &= 0\\
		\hfill h &= 0
	\end{cases}\,.
\end{align}
Let us now examine the counting of degrees of freedom in linearized massive gravity within the jet-bundle framework.
\begin{table}[H] 
    \centering
    \begin{tabular}{ll}
        \multicolumn{2}{c}{\textbf{Input}}  \\ \toprule
        \textbf{Field content} $\boldsymbol{v^{A}}$ & $h_{\mu\nu}$\\[5pt]
        \textbf{Number of field components} $\boldsymbol{m}$ & $10$ \\[5pt]
        \textbf{Field equations} $\boldsymbol{\R_q}$ &  $\mathbb{E}^\text{FP}_{\mu\nu} -\frac14 m^2 \left(h_{\mu\nu} - \eta^{\alpha\beta}h_{\alpha\beta} \eta_{\mu\nu}\right)  = 0$\\[5pt]
         $\boldsymbol{\gamma_\ell}$ & $\gamma_\ell = 0$ $\forall \ell$ \\ \bottomrule
         \\
         \multicolumn{2}{c}{\textbf{Output}}  \\ \toprule
         \textbf{Involutive after $\boldsymbol{s}$ projections} & $s=4$ \\[5pt]
         \textbf{Involutive after $\boldsymbol{r}$ prolongations} & $r=0$ \\[5pt]
         $\boldsymbol{\beta^{(k)}_q}$ & $\beta^{(1)}_2 = 8$, $\beta^{(2)}_2 = 7$, $\beta^{(3)}_2 = 10$, $\beta^{(4)}_2 = 10$ \\[5pt]
         $\boldsymbol{\alpha^{(k)}_q}$ & $\alpha^{(1)}_2 = 32$, $\alpha^{(2)}_2 = 23$, $\alpha^{(3)}_2 = 10$, $\alpha^{(4)}_2 = 0$ \\[5pt]
         $\boldsymbol{H_q(r)}$ & $65 + 38 r + 5 r^2$\\[5pt]
         \textbf{Degrees of freedom} & $5$ \\ \bottomrule
    \end{tabular}
\end{table}
 Since the mass term explicitly breaks linearized diffeomorphisms, we set $\gamma_\ell = 0$. 
After running through the Cartan-Kuranishi algorithm, we find that the massive Fierz-Pauli equations are \emph{not} involutive from the start. It is necessary to prolong and project them four times, i.e., the Cartan-Kuranishi algorithm produces an involutive system for $s=4$. The resulting system of equations, $\R^{(4)}_2$, contains a large number of second order, first order, and even zeroth order equations. This complicates the straightforward interpretation of the $\beta$'s. However, the algorithm produces nothing unexpected: It finds the constraints
\begin{align}
    \eta^{\mu\nu}h_{\mu\nu} &= 0 &&\text{and} & \partial^\mu h_{\mu\nu} - \eta^{\alpha\beta}\partial_\nu h_{\alpha\beta} &= 0\,,
\end{align}
which are the constraints one finds also through other considerations~\cite{Heisenberg:2018}, as well as their first and second order prolongations. 

The Cartan character $\alpha^{(4)}_2 = 0$ is still an indication that $\R^{(4)}_2$ is a compatible set of equations, i.e., one that does not introduce arbitrary functions of four coordinates in the general solution.

Finally, we note that with the $\alpha$ and $\beta$ characters we computed, our formula~\eqref{eq:FinalFormula} produces the expected number---namely five---of physical degrees of freedom.
\newpage

\subsubsection{Detuned Massive Fierz-Pauli equations (Boulware-Deser Ghost)}\label{sssec:DetunedMassiveFP}
\begin{table}[H] 
    \centering
    \begin{tabular}{ll}
        \multicolumn{2}{c}{\textbf{Input}}  \\ \toprule
        \textbf{Field content} $\boldsymbol{v^{A}}$ & $h_{\mu\nu}$\\[5pt]
        \textbf{Number of field components} $\boldsymbol{m}$ & $10$ \\[5pt]
        \textbf{Field equations} $\boldsymbol{\R_q}$ &  $\mathbb{E}^\text{FP}_{\mu\nu} -\frac14 m^2 \left(h_{\mu\nu} + \eta^{\alpha\beta}h_{\alpha\beta} \eta_{\mu\nu}\right)  = 0$\\[5pt]
         $\boldsymbol{\gamma_\ell}$ & $\gamma_\ell = 0$ $\forall \ell$ \\ \bottomrule
         \\
         \multicolumn{2}{c}{\textbf{Output}}  \\ \toprule
         \textbf{Involutive after $\boldsymbol{s}$ projections} & $s=2$ \\[5pt]
         \textbf{Involutive after $\boldsymbol{r}$ prolongations} & $r=0$ \\[5pt]
         $\boldsymbol{\beta^{(k)}_q}$ & $\beta^{(1)}_2 = 4$, $\beta^{(2)}_2 = 4$, $\beta^{(3)}_2 = 8$, $\beta^{(4)}_2 = 10$ \\[5pt]
         $\boldsymbol{\alpha^{(k)}_q}$ & $\alpha^{(1)}_2 = 36$, $\alpha^{(2)}_2 = 26$, $\alpha^{(3)}_2 = 12$, $\alpha^{(4)}_2 = 0$ \\[5pt]
         $\boldsymbol{H_q(r)}$ & $74 + 44 r + 6 r^2$\\[5pt]
         \textbf{Degrees of freedom} & $6$ \\ \bottomrule
    \end{tabular}
\end{table}
As is well-known (see, for instance, \cite{Heisenberg:2018}), the mass term of the massive Fierz-Pauli Lagrangian has to be tuned very carefully in order not to introduce ghost instabilities. In this example, we flip the sign between the two terms of the mass term:
\begin{align}
    \mathcal{L}_\text{mass} \mapsto \mathcal{L}_\text{mass} = \frac18 m^2 \left(\eta^{\alpha\beta} \eta^{\mu\nu} \red{+} \eta^{\mu\alpha} \eta^{\nu\beta}\right)h_{\alpha\beta} h_{\mu\nu}
\end{align}
This manipulation introduces the so-called Boulware-Deser ghost and rather than five degrees of freedom we should find six. First of all, we observe that the detuned massive Fierz-Pauli equations are involutive after $s=2$ projections, rather than $s=4$ like their well-tuned counterparts. Moreover, we find $\alpha^{(4)}_2 = 0$, in agreement with the compatibility condition and classical determinism. Moreover, our formula~\eqref{eq:FinalFormula} tell us that there are indeed six degrees of freedom. 

This is another example which shows us that ghost instabilities represent no problem for the formalism. The Cartan-Kuranishi algorithm terminates in a finite number of steps and it produces the correct number of degrees of freedom including the ghost\footnote{It would be particularly interesting to extend the method presented here to the non-linear dRGT theory\cite{deRham:2010kj}. This theory has been shown to be free of the Boulware–Deser ghost through various approaches, and it would be a compelling exercise to demonstrate, within the jet-bundle framework, that it propagates precisely five physical degrees of freedom.} 
\newpage

\subsubsection{Massive Fierz-Pauli \`a la St\"uckelberg}\label{sssec:MassiveFPStuckelberg}
\begin{table}[H] 
    \centering
    \begin{tabular}{ll}
        \multicolumn{2}{c}{\textbf{Input}}  \\ \toprule
        \textbf{Field content} $\boldsymbol{v^{A}}$ & $(h_{\mu\nu}, A^\mu, \Phi)$\\[5pt]
        \textbf{Number of field components} $\boldsymbol{m}$ & $10 + 4 + 1 = 15$ \\[5pt]
        \textbf{Field equations} $\boldsymbol{\R_q}$ &  $\mathbb{E}_\text{FP} + m^2 \left(h_{\mu\nu} + 2 \partial_{(\mu} A_{\nu)} + 2 \partial_\mu \partial_\nu \Phi\right)$\\[5pt]
        & $- m^2 \eta_{\mu\nu}\left(h + \partial^\alpha A_\alpha + \partial^\alpha \partial_\alpha \Phi\right) = 0$ \\[5pt]
        & $\partial^\alpha \partial_{[\mu}A_{\alpha]} + \eta^{\alpha\beta} \partial_{[\mu}h_{\beta]\alpha} = 0$ \\[5pt]
        & $\partial^\alpha \partial^\beta h_{\alpha\beta} - \eta^{\beta\gamma}\partial^\alpha \partial_\alpha h_{\beta\gamma} = 0$ \\[5pt]
         $\boldsymbol{\gamma_\ell}$ & $\gamma_1 = 5$ \\ \bottomrule
         \\
         \multicolumn{2}{c}{\textbf{Output}}  \\ \toprule
         \textbf{Involutive after $\boldsymbol{s}$ projections} & $s=0$ \\[5pt]
         \textbf{Involutive after $\boldsymbol{r}$ prolongations} & $r=0$ \\[5pt]
         $\boldsymbol{\beta^{(k)}_q}$ & $\beta^{(1)}_2 = 0$, $\beta^{(2)}_2 = 0$, $\beta^{(3)}_2 = 5$, $\beta^{(4)}_2 = 10$ \\[5pt]
         $\boldsymbol{\alpha^{(k)}_q}$ & $\alpha^{(1)}_2 = 60$, $\alpha^{(2)}_2 = 45$, $\alpha^{(3)}_2 = 25$, $\alpha^{(4)}_2 = 5$ \\[5pt]
         $\boldsymbol{H_q(r)}$ & $135 + \frac{275}{3} r + \frac{35}{2} r^2 + \frac{5}{6} r^3$\\[5pt]
         $\boldsymbol{\bar{H}_q(r)}$ & $35 + 30 r + 5 r^2$\\[5pt]
         $\boldsymbol{\bar{\alpha}^{(k)}_q}$ & $\bar{\alpha}^{(1)}_2 = 10$, $\bar{\alpha}^{(2)}_2 = 15$, $\bar{\alpha}^{(3)}_2 = 10$, $\bar{\alpha}^{(4)}_2 = 0$ \\[5pt]
         \textbf{Degrees of freedom} & $5$ \\ \bottomrule
    \end{tabular}
\end{table}
The transition from the massless Fierz-Pauli equations to their massive version is analogous to the transition from Maxwell's equations to Proca's equations: The introduction of a mass term spoils the gauge freedom of the original theory. 

Just as in the case of Proca's theory, we can apply the St\"uckelberg trick to restore gauge symmetry in the massive theory. In the case of the massive Fierz-Pauli equations, this necessitates the introduction of two auxiliary fields. Namely the vector field $A^\mu$ and the scalar $\Phi$. Therefore, there are three sets of field equations to consider: The ten equations stemming from varying the Fierz-Pauli-St\"uckelberg action with respect to $h^{\mu\nu}$, the four equations associated with $A_\mu$, and the one equation related to $\Phi$. In the table we use $\mathbb{E}_\text{FP}$ as a short hand to denote the left hand side of the massless Fierz-Pauli equations. 

The Lagrangian of the massive Fierz-Pauli theory can then be shown to be invariant under the following simultaneous gauge transformations (see~\cite{Heisenberg:2018} for more details):
\begin{align}
    h_{\mu\nu} & \mapsto h_{\mu\nu} + \partial_\mu \xi_\nu + \partial_\nu \xi_\mu && \text{and} & A_\mu & \mapsto A_\mu - \xi_\mu \notag\\
    A_\mu & \mapsto A_\mu + \partial_\mu \theta && \text{and} & \Phi &\mapsto \Phi - \theta\,,
\end{align}
for some arbitrary $1$-form $\xi_\mu$ and some arbitrary scalar field $\theta$. This introduces a total of five gauge fields and an inspection of the field equations reveals that they all enter with derivatives (despite some of the variables transforming in an algebraic way). Thus, we find that $\gamma_1 = 5$, while all other $\gamma$ coefficients vanish. 

In contrast to the previous example, we now find that the Cartan-Kuranishi algorithm terminates without having to execute prolongations or projections $(r=s=0)$. The resulting set of equations is much more straightforward to analyze: There are five constraint equations and ten dynamical equations, in agreement with $\beta^{(3)}_2 = 5$ and $\beta^{(4)}_2 = 10$. The highest order Cartan character signals the presence of five gauge modes, while the gauge-corrected Cartan character $\bar{\alpha}^{(4)}_2 = 0$ assures us that the equations are compatible. From the resulting $\alpha$'s and $\beta$'s one determines that the massive Fierz-Pauli equations expressed \`a la St\"uckelberg still propagate five degrees of freedom, as they should. \newpage

\subsubsection{Massless \texorpdfstring{$2$}{2}-Form}\label{sssec:2Form}
\begin{table}[H] 
    \centering
    \begin{tabular}{ll}
        \multicolumn{2}{c}{\textbf{Input}}  \\ \toprule
        \textbf{Field content} $\boldsymbol{v^{A}}$ & $B_{\mu\nu}$\\[5pt]
        \textbf{Number of field components} $\boldsymbol{m}$ & $6$ \\[5pt]
        \textbf{Field equations} $\boldsymbol{\R_q}$ &  $\eta^{\alpha\beta}\left(\partial_\alpha\partial_\beta B_{\mu\nu} + \partial_\beta \partial_\nu B_{\alpha\mu} - \partial_\beta\partial_\mu B_{\alpha\nu} \right) = 0$ \\[5pt]
         $\boldsymbol{\gamma_\ell}$ & $(\gamma_0, \gamma_1) = (1,2)$ \\ \bottomrule
         \\
         \multicolumn{2}{c}{\textbf{Output}}  \\ \toprule
         \textbf{Involutive after $\boldsymbol{s}$ projections} & $s=0$ \\[5pt]
         \textbf{Involutive after $\boldsymbol{r}$ prolongations} & $r=0$ \\[5pt]
         $\boldsymbol{\beta^{(k)}_q}$ & $\beta^{(1)}_2 = 0$, $\beta^{(2)}_2 = 1$, $\beta^{(3)}_2 = 2$, $\beta^{(4)}_2 = 3$ \\[5pt]
         $\boldsymbol{\alpha^{(k)}_q}$ & $\alpha^{(1)}_2 = 24$, $\alpha^{(2)}_2 = 17$, $\alpha^{(3)}_2 = 10$, $\alpha^{(4)}_2 = 3$ \\[5pt]
         $\boldsymbol{H_q(r)}$ & $54 + \frac{75}{2} r + 8 r^2 + \frac12 r^3$\\[5pt]
         $\boldsymbol{\bar{H}_q(r)}$ & $4 + 5 r + r^2$\\[5pt]
         $\boldsymbol{\bar{\alpha}^{(k)}_q}$ & $\bar{\alpha}^{(1)}_2 = 0$, $\bar{\alpha}^{(2)}_2 = 2$, $\bar{\alpha}^{(3)}_2 = 2$, $\bar{\alpha}^{(4)}_2 = 0$ \\[5pt]
         \textbf{Degrees of freedom} & $1$ \\ \bottomrule
    \end{tabular}
\end{table}
From a purely mathematical standpoint, Maxwell's theory of electromagnetism is the theory of a $1$-form $A_\mu$ whose field equations are governed by its field strength $2$-form $F_{\mu\nu} = \partial_\mu A_\nu - \partial_\nu A_\mu$. This can be generalized to higher order differential forms. In particular, we can consider a $2$-form $B_{\mu\nu}$ with an associated field strength $3$-form $H_{\mu\nu\rho} \ce \partial_\mu B_{\nu\rho} + \partial_\nu B_{\rho\mu} + \partial_\rho B_{\mu\nu}$. This results in a mathematical consistent theory governed by the field equations
\begin{align}
    \partial^\alpha H_{\alpha\mu\nu} = 0\,.
\end{align}
Moreover, this theory of a massless $2$-form enjoys gauge freedom:
\begin{align}\label{eq:B-Gauge}
    B_{\mu\nu} \mapsto B_{\mu\nu} + \partial_\mu \theta_\nu - \partial_\nu \theta_\mu\,.
\end{align}
At first sight, it may seem that the $1$-form $\theta$ represents four gauge modes. However, a closer inspection reveals that only three of the four components of $\theta$ enter the field equations (see, for instance, \cite{Aoki:2022} for a discussion of this point). Following~\cite{Aoki:2022}, we introduce the auxiliary $3$-dimensional vector fields
\begin{align}
    B_i &\ce B_{0i} &&\text{and} & C^{i} &\ce \frac12 \epsilon^{ijk} B_{jk}\,,
\end{align}
in terms of which the gauge transformation law can be restate as
\begin{align}
    \vec{B} &\mapsto \vec{B} + \dot{\vec{\theta}} - \nabla\theta_0 \notag\\
    \vec{C} &\mapsto \vec{C} + \nabla\times \vec{\theta}\,.
\end{align}
Next, we decompose $\vec{B}$, $\vec{C}$, and $\vec{\theta}$ into their respective longitudinal and transversal parts:
\begin{align}\label{eq:TLDecomp}
    \vec{B} &= \vec{B}^\text{T} + \nabla B \notag\\
    \vec{C} &= \vec{C}^\text{T} + \nabla C \notag\\
    \vec{\theta} &= \vec{\theta}^\text{T} + \nabla \theta\,,
\end{align}
with $\nabla\cdot \vec{B}^\text{T} = \nabla\cdot \vec{C}^\text{T} = \nabla\cdot \vec{\theta}^\text{T} = 0$. The gauge transformations can finally be written as
\begin{align}\label{eq:TLGauge}
    \vec{B}^\text{T} &\mapsto \vec{B}^\text{T} + \dot{\theta}\notag\\
    B &\mapsto B + \dot{\theta} - \theta_0 \notag\\
    \vec{C}^\text{T} &\mapsto \vec{C}^\text{T} + \nabla\times \vec{\theta}^\text{T} \notag\\
    C &\mapsto C\,.
\end{align}
An inspection of the field equations reveals that the longitudinal part of $\vec{B}$ does not appear in any equation. Thus, only the scalar $\dot{\theta}$ and the $2$-component, transversal vector field $\vec{\theta}^\text{T}$ play a role in the field equations. Since the scalar field $\dot{\theta}$ enters algebraically in the gauge transformation, while $\vec{\theta}^\text{T}$ appears with first order derivatives, we conclude that
\begin{align}
    (\gamma_0, \gamma_1) = (1,2)\,.
\end{align}
All other $\gamma_\ell$ coefficients are zero. Finding the correct values for the $\gamma_\ell$ coefficients is the only challenging part in this example. Executing the Cartan-Kuranishi algorithm presents no obstacles and we find that the field equations of the massless $2$-form are involutive from the start. The highest order Cartan character turns out to be $\alpha^{(4)}_2 = 3$, which is consistent with our discussion so far, from which we concluded that only three gauge modes are present in the field equations. The compatibility equation $\bar{\alpha}^{(4)}_2 = \alpha^{(4)}_2 - \sum_\ell \gamma_\ell = 0$ is satisfied and from~\eqref{eq:FinalFormula} we find one degree of freedom.

This is the same number reported in~\cite{Aoki:2022}, which was found using the following argument: The $2$-form $B_{\mu\nu}$ has six components and we have the gauge freedom~\eqref{eq:B-Gauge}. Given that $\theta_\mu$ has four components, one would expect to be able to fix four components of $B_{\mu\nu}$. However, due to residual gauge freedom of the form $\theta_\mu \mapsto \theta_\mu + \partial_\mu \Phi$, we can always transform away one of the components of $\theta_\mu$. Hence, only three remain to fix components of $B_{\mu\nu}$. Next, the attention turns to the field equations. Only the components $B_{xy}$, $B_{xz}$, and $B_{yz}$ enter with time derivatives. Thus, only these components can be dynamical. However, three of the six equations are constraints and, as it turns out, if two of these constraints are satisfied, the third one is automatically satisfied. A closer inspection then reveals these equations constrain two of the components $B_{xy}$, $B_{xz}$, $B_{yz}$. This leaves us with only $6-3-2 = 1$ degree of freedom. 
\newpage

\subsubsection{Massive \texorpdfstring{$2$}{2}-Form}\label{sssec:Massive2Form}
\begin{table}[H] 
    \centering
    \begin{tabular}{ll}
        \multicolumn{2}{c}{\textbf{Input}}  \\ \toprule
        \textbf{Field content} $\boldsymbol{v^{A}}$ & $B_{\mu\nu}$\\[5pt]
        \textbf{Number of field components} $\boldsymbol{m}$ & $6$ \\[5pt]
        \textbf{Field equations} $\boldsymbol{\R_q}$ &  $\mathbb{E}^B_{\mu\nu} - 9 m^2 B_{\mu\nu}= 0$ \\[5pt]
         $\boldsymbol{\gamma_\ell}$ & $\gamma_\ell = 0$ $\forall \ell$ \\ \bottomrule
         \\
         \multicolumn{2}{c}{\textbf{Output}}  \\ \toprule
         \textbf{Involutive after $\boldsymbol{s}$ projections} & $s=2$ \\[5pt]
         \textbf{Involutive after $\boldsymbol{r}$ prolongations} & $r=0$ \\[5pt]
         $\boldsymbol{\beta^{(k)}_q}$ & $\beta^{(1)}_2 = 4$, $\beta^{(2)}_2 = 5$, $\beta^{(3)}_2 = 6$, $\beta^{(4)}_2 = 6$ \\[5pt]
         $\boldsymbol{\alpha^{(k)}_q}$ & $\alpha^{(1)}_2 = 20$, $\alpha^{(2)}_2 = 13$, $\alpha^{(3)}_2 = 6$, $\alpha^{(4)}_2 = 0$ \\[5pt]
         $\boldsymbol{H_q(r)}$ & $39 + 22 r + 3 r^2$\\[5pt]
         \textbf{Degrees of freedom} & $3$ \\ \bottomrule
    \end{tabular}
\end{table}
It comes as no surprise that the theory of the massless $2$-form described in the previous example can be modified to include a mass term. The field equations are given by
\begin{align}
    \mathbb{E}^B_{\mu\nu} - 9 m^2 B_{\mu\nu} = 0\,,
\end{align}
where $\mathbb{E}^B_{\mu\nu}$ is a short hand notation for the left hand side of the field equations for the massless $B$-field. 

Just as in Proca's theory or the massive Fierz-Pauli theory, we find that introducing a mass term spoils the gauge freedom present in the massless theory. Thus, $\gamma_\ell = 0$.

In contrast to the massless case, we now find that the Cartan-Kuranishi algorithm requires us to perform two prolongations followed by two projections (i.e., $s=2$, $r=0$), in order to obtain an involutive system. From the computed values of the Cartan characters we infer the validity of the compatibility condition $\alpha^{(4)}_2 = 0$ and we find that the massive theory propagates three physical degrees of freedom.

This number was also derived in, for instance,~\cite{Aoki:2022} using different methods.

\newpage

\subsubsection{Massive \texorpdfstring{$2$}{2}-Form \`a la St\"uckelberg}\label{sssec:Massive2FormStuckelberg}
\begin{table}[H] 
    \centering
    \begin{tabular}{ll}
        \multicolumn{2}{c}{\textbf{Input}}  \\ \toprule
        \textbf{Field content} $\boldsymbol{v^{A}}$ & $(B_{\mu\nu}, A_\mu)$\\[5pt]
        \textbf{Number of field components} $\boldsymbol{m}$ & $6 + 4 = 10$ \\[5pt]
        \textbf{Field equations} $\boldsymbol{\R_q}$ &  $\mathbb{E}^B_{\mu\nu} - 9 m^2 B_{\mu\nu} - 9 m \left(\partial_\mu A_\nu - \partial_\nu A_\mu\right) = 0$ \\[5pt]
        & $\eta^{\alpha\beta}\left(\partial_\alpha \partial_\beta A_\mu - \partial_\beta \partial_\mu A_\alpha - m \partial_\beta B_{\mu\alpha} \right) = 0$\\[5pt]
         $\boldsymbol{\gamma_\ell}$ & $(\gamma_0, \gamma_1) = (1,3)$  \\ \bottomrule
         \\
         \multicolumn{2}{c}{\textbf{Output}}  \\ \toprule
         \textbf{Involutive after $\boldsymbol{s}$ projections} & $s=0$ \\[5pt]
         \textbf{Involutive after $\boldsymbol{r}$ prolongations} & $r=0$ \\[5pt]
         $\boldsymbol{\beta^{(k)}_q}$ & $\beta^{(1)}_2 = 0$, $\beta^{(2)}_2 = 1$, $\beta^{(3)}_2 = 3$, $\beta^{(4)}_2 = 6$ \\[5pt]
         $\boldsymbol{\alpha^{(k)}_q}$ & $\alpha^{(1)}_2 = 40$, $\alpha^{(2)}_2 = 29$, $\alpha^{(3)}_2 = 17$, $\alpha^{(4)}_2 = 4$ \\[5pt]
         $\boldsymbol{H_q(r)}$ & $90 + \frac{371}{6} r + \frac{25}{2} r^2 + \frac23 r^3$\\[5pt]
         $\boldsymbol{\bar{H}_q(r)}$ & $20 + 17 r + 3 r^2$\\[5pt]
         $\boldsymbol{\bar{\alpha}^{(k)}_q}$ & $\bar{\alpha}^{(1)}_2 = 6$, $\bar{\alpha}^{(2)}_2 = 8$, $\bar{\alpha}^{(3)}_2 = 6$, $\bar{\alpha}^{(4)}_2 = 0$ \\[5pt]
         \textbf{Degrees of freedom} & $3$ \\ \bottomrule
    \end{tabular}
\end{table}
In our last example, we apply the St\"uckelberg trick to the massive $2$-form. In order to restore gauge symmetry, we introduce the $1$-form $A_\mu$. The gauge transformations are then given by
\begin{align}
    B_{\mu\nu} &\mapsto B_{\mu\nu} + \partial_{\mu}\theta_\nu - \partial_\nu \theta_\mu  \notag\\
    A_{\mu} &\mapsto A_{\mu} - m \theta_\mu\,,
\end{align}
for some arbitrary $1$-form $\theta_\mu$. To determine the $\gamma_\ell$'s, it is again convenient to decompose $B_{\mu\nu}$ into $\vec{B}$ and $\vec{C}$, which themselves are then decomposed into longitudinal and transversal parts. This was already done in~\eqref{eq:TLDecomp} and the resulting gauge-transformed longitudinal and transversal components are given by~\eqref{eq:TLGauge}.

Because of the mass terms in the field equations, the longitudinal part of $\vec{B}$ no longer drops out. Thus, we no longer have $(\gamma_0, \gamma_1) = (1,2)$. Rather, we now find
\begin{align}
    (\gamma_0, \gamma_1) = (1,3)\,.
\end{align}
The increase of $\gamma_1$ from $2$ to $3$ is justified by noting that the longitudinal part of $\vec{B}$ is given by $\nabla B$, rather than $B$ alone. Thus, $\theta_0$ enters the field equations with a derivative acting on it, and not just algebraically. 

Running through the steps of the Cartan-Kuranishi algorithm presents no problems and we find that the field equations are already involutive with highest-order Cartan character $\alpha^{(4)}_2 = 4$. Hence, the compatibility condition is satisfied, since gauge correction yields $\bar{\alpha}^{(4)}_2 = 0$. 

From the values listed in the output table we finally infer that the massive $2$-form, reformulated \`a la St\"uckelberg, still propagates the expected three degrees of freedom.
\newpage

\subsubsection{Observations and Insights from the Examples}\label{sssec:ObservationsAndInsights}
In this subsection we provided $14$ examples of field theories and how the methods discussed in this work can be applied to them. Since we collected all important outcomes in tables, we can now perform a comparative analysis and extract certain patterns. In particular, we observe the following:

\begin{enumerate}
    \item \emph{Highest order (gauge-corrected) Cartan characters:} In all theories without gauge symmetry we found $\alpha^{(4)}_2 = 0$, which is precisely what we expected based on our discussion of compatible equations and classical determinism.

    In the case of gauge theories, we always found that $\alpha^{(4)}_2$ equals the number of gauge modes that enter in the field equations. This was also expected. Moreover, the gauge-corrected Cartan character $\bar{\alpha}^{(4)}_2$ was zero in every single case, as it should be in order to have  equations compatible with classical determinism.
    \item \emph{Ordering of $\beta^{(k)}_q$, $\alpha^{(k)}_q$, and $\bar{\alpha}^{(k)}_q$:} In every example we studied we observe that
    \begin{align}
        \beta^{(1)}_2 \leq \beta^{(2)}_2 \leq \beta^{(3)}_2 \leq \beta^{(4)}_2
    \end{align}
    and
    \begin{align}\label{eq:alphaOrder}
        \alpha^{(1)}_2 \geq \alpha^{(2)}_2 \geq \alpha^{(3)}_2 \geq \alpha^{(4)}_2\,.
    \end{align}
    This is a pattern that is true in any dimension $n$ and for any order $q$. One can show (see for instance~\cite{SeilerBook}) that for an involutive equation $\R_q$ the characters satisfy
    \begin{align}
        \beta^{(1)}_q \leq \beta^{(2)}_q \leq \dots \leq \beta^{(n-1)}_q \leq \beta^{(n)}_q \notag\\
        \alpha^{(1)}_q \geq \alpha^{(2)}_q \geq \dots \geq \alpha^{(n-1)}_q \geq \alpha^{(n)}_q\,.
    \end{align}
    This property is useful for cross-checking whether the computed $\alpha$'s and $\beta$'s are reasonable, or whether an error occurred.
    
    Unfortunately, no such order exists for the gauge-corrected Cartan characters $\bar{\alpha}^{(k)}_q$. In most gauge theories we studied, $\bar{\alpha}^{(1)}_2$ is smaller or equal to $\bar{\alpha}^{(2)}_2$, which itself is larger or equal to $\bar{\alpha}^{(3)}_2$. The only exception, and the only instance in which the order~\eqref{eq:alphaOrder} holds for gauge-corrected Cartan characters, is Proca \`a la St\"uckelberg~\ref{sssec:ProcaStuckelberg}.

    \item \emph{Value of $\alpha^{(k)}_q + \beta^{(k)}_q$:} In every example we can observe that the sum $\alpha^{(k)}_2 + \beta^{(k)}_2$ is a multiple of the number of field components. Even in theories where several fields are present (see for instance the example on the massive $2$-form \`a la St\"uckelberg~\ref{sssec:Massive2FormStuckelberg}). More precisely, we find
    \begin{align}
        \alpha^{(4)}_2 + \beta^{(4)}_2 &= m \notag\\
        \alpha^{(3)}_2 + \beta^{(3)}_2 &= 2m \notag\\
        \alpha^{(2)}_2 + \beta^{(2)}_2 &= 3m \notag\\
        \alpha^{(1)}_2 + \beta^{(1)}_2 &= 4m\,.
    \end{align}
    This pattern is a peculiarity of second order field equations and it is easy to explain. Recall that the Cartan characters are defined as (cf. Definition~\ref{def:CartanCharacters})
    \begin{align}
        \alpha^{(k)}_q = m\binom{n + q - k - 1}{n - k} - \beta^{(k)}_q\,.
    \end{align}
    Thus, the sum of $\alpha^{(k)}_q + \beta^{(k)}_q$ is always equal to the binomial on the right hand side. For equations of order 
    $q=1,2,3$ in any dimension $n$ we find the following values:
    \begin{align}
        \alpha^{(k)}_q + \beta^{(k)}_q = m\binom{n + q - k - 1}{n - k} =
        \begin{cases}
            m & \text{for } q=1 \\
            m(1+n-k) & \text{for } q=2 \\
            \frac12 m (1+n-k)(2+n-k) & \text{for } q = 3
        \end{cases}\,.
    \end{align}
    This reproduces precisely our observations. It is nevertheless an important observation, since this is a simple cross-check one can perform to verify whether the computed $\alpha$ and $\beta$ characters are correct. 

    \item \emph{A pattern which relates the $\gamma_\ell$'s to the $\beta$'s:} This pattern is more interesting than the previous one, but we cannot offer an explanation for it. In all gauge theories, we can observe that $\beta^{(3)}_2$ is equal to the highest non-zero $\gamma_\ell$. If this is the only $\gamma_\ell$ that appears, then $\beta^{(2)}_2 = \beta^{(1)}_2 = 0$. In some cases there is a second $\gamma_\ell$. Then $\beta^{(2)}_2$ is equal to the second $\gamma_\ell$ and $\beta^{(1)}_2 = 0$.
    \begin{align}
        &\text{Maxwell's equations:} & \phantom{\gamma_0} & \phantom{= 0} & \gamma_1 &= 1\,, & \beta^{(1)}_2 &= 0\,, & \beta^{(2)}_2 &= 0\,, & \beta^{(3)}_2 &= \gamma_1 \notag\\
        &\text{Proca-St\"uckelberg:} & \phantom{\gamma_0} & \phantom{= 0} & \gamma_1 &= 1\,, & \beta^{(1)}_2 &= 0\,, & \beta^{(2)}_2 &= 0\,, & \beta^{(3)}_2 &= \gamma_1 \notag\\
        &\text{Einstein's equations:} & \phantom{\gamma_0} & \phantom{= 0} & \gamma_1 &= 4\,, & \beta^{(1)}_2 &= 0\,, & \beta^{(2)}_2 &= 0\,, & \beta^{(3)}_2 &= \gamma_1 \notag\\
        &\text{Coincident GR:} & \phantom{\gamma_0} & \phantom{= 0} & \gamma_1 &= 4\,, & \beta^{(1)}_2 &= 0\,, & \beta^{(2)}_2 &= 0\,, & \beta^{(3)}_2 &= \gamma_1 \notag\\
        &\text{Massless Fierz-Pauli:} & \phantom{\gamma_0} & \phantom{= 0} & \gamma_1 &= 4\,, & \beta^{(1)}_2 &= 0\,, & \beta^{(2)}_2 &= 0\,, & \beta^{(3)}_2 &= \gamma_1 \notag\\
        &\text{Fierz-Pauli-St\"uckelberg:} & \phantom{\gamma_0} & \phantom{= 0} & \gamma_1 &= 5\,, & \beta^{(1)}_2 &= 0\,, & \beta^{(2)}_2 &= 0\,, & \beta^{(3)}_2 &= \gamma_1 \notag\\
        &\text{Massless $2$-form:} & \gamma_0 &= 1 & \gamma_1 &= 2\,, & \beta^{(1)}_2 &= 0\,, & \beta^{(2)}_2 &= \gamma_0\,, & \beta^{(3)}_2 &= \gamma_1 \notag\\
        &\text{$2$-form \`a la St\"uckelberg:} & \gamma_0 &=1 & \gamma_1 &= 3\,, & \beta^{(1)}_2 &= 0\,, & \beta^{(2)}_2 &= \gamma_0\,, & \beta^{(3)}_2 &= \gamma_1
    \end{align}

    This is a remarkable pattern and relation because the $\beta$'s are computed from the symbol of $\R_q$, which knows nothing of the $\gamma_\ell$'s. Moreover, in the absence of gauge symmetry all the $\beta$'s are generally different from zero. 
    
    \item \emph{Number of prolongations $r$ and number of projections $s$:} Recall that the method for counting degrees of freedom discussed here is based on involutive equations $\R_q$. The Cartan-Kuranishi algorithm allows us to turn every equation into an involutive one. To achieve that, a certain number of prolongations and projections might be necessary.
    However, in all but four examples we found $r=s=0$. This means that in the majority of the examples the equations were involutive to begin with. The four examples which make up the exceptions are all massive theories: Proca's theory~\ref{sssec:Proca}, the massive Fierz-Pauli equation~\ref{sssec:MassiveFP}, the detuned massive Fierz-Pauli equation~\ref{sssec:DetunedMassiveFP}, and the massive $2$-form~\ref{sssec:Massive2Form}. In these theories we found $s\geq 2$ and $r=0$. At first one might hypothesize that $s\geq 2$ is a consequence of breaking gauge symmetry by the introduction of a mass term. This hypothesis seems to be supported by the fact that in the St\"uckelberg formulation of these theories, where gauge symmetry is restored by the introduction of appropriate St\"uckelberg fields, one finds again $s=0$. However, the breaking of gauge symmetry cannot fully account for why $s>0$. For instance, in the case of the (massless) detuned Fierz-Pauli equations, gauge symmetry is spoiled by ill-adjusted coefficients in $\mathcal{L}$, but the Cartan-Kuranishi algorithm still produces $r=s=0$.

    Rather than in the breaking of gauge symmetry, the reason for $s>0$ lays in the particular way in which this symmetry is broken. To understand this, let us denote the field equations of the massless theories by
    \begin{align}
        \R_q: \left\{\ \E^\bullet_0 = 0\right.\,,
    \end{align}
    where $\bullet$ stands for an unspecified number of indices, in order to account for the different tensorial structures of the equations. Because of gauge symmetry, these equations satisfy contracted Bianchi identities:
    \begin{align}
        \partial_\alpha\E^{\alpha \bar{\bullet}}_0 = 0\,,
    \end{align}
    where $\bar{\bullet}$ stands for the uncontracted indices. When we promote a gauge theory to a massive theory, we add a mass term to the Lagrangian. This term enters without derivatives and thus the massive field equations take the form
    \begin{align}
        \R_q: \left\{\ \E^\bullet_0 - m^2 \Psi^\bullet = 0\right.\,,
    \end{align}
    where $m$ is the mass and $\Psi^\bullet$ the field in question (i.e., it could be a vector, a $2$-form, a $(0,2)$ tensor or any other kind of tensor).

    Next, the Cartan-Kuranishi algorithm demands that we compute the prolongation of these massive field equations. If we prolong and simultaneously trace over the index of the derivative operator $\partial_\alpha$ and one of the free indices of $\E^\bullet$, which we can always do since it merely corresponds to taking linear combinations of equations, we obtain
    \begin{align}
        \partial_\alpha \E^{\alpha\bar{\bullet}}_0 - m^2 \partial_\alpha \Psi^{\alpha\bar{\bullet}} = 0\,.
    \end{align}
    Since $\E^\bullet_0$ by itself is still gauge-invariant, the Bianchi identities for this tensor are still true. Thus, the above equation reduces to
    \begin{align}
        \partial_\alpha \Psi^{\alpha\bar{\bullet}} = 0\,.
    \end{align}
    In other words, we find that the prolongation of massive field equations always produces a lower order equation. Thus, $\R^{(1)}_q \neq \R_q$ and this implies $s>0$.

    It is also interesting to consider what happens in the St\"uckelberg formulation of these theories. Because of the introduction of St\"uckelberg fields, we now have more equations to deal with. Moreover, the equation $\E^\bullet_0 - m^2 \Psi^\bullet = 0$ gets modified and the St\"uckelberg fields appear in it. However, we can always choose a gauge in which the St\"uckelberg fields vanish, then the massive equation $\E^\bullet_0 - m^2 \Psi^\bullet = 0$ is restored (i.e., the additional terms drop out without affecting the structure of $\E^\bullet_0 - m^2 \Psi^\bullet = 0$) and the additional equations, which were obtained by taking variations of the St\"uckelberg Lagrangian with respect to the St\"uckelberg fields, reduce to constraint equations of the form $\partial_\alpha \Psi^{\alpha\bar{\bullet}} =0$. Hence, the prolongation no longer produces a new equation and one thus finds $\R^{(1)}_q = \R_q$. This implies $s=0$, in agreement with what we found in the three relevant examples: Proca \`a la St\"uckelberg~\ref{sssec:ProcaStuckelberg}, massive Fierz-Pauli \`a la St\"uckelberg~\ref{sssec:MassiveFPStuckelberg}, and the massive $2$-form \`a la St\"uckelberg~\ref{sssec:Massive2FormStuckelberg}.

    \item \emph{Reading off the degrees of freedom from the (gauge-corrected) Hilbert polynomial:}
    The two Hilbert polynomials can be written as $H_q(r) = \sum_{i=0}^{n-1} h_i r^{i}$ and $\bar{H}_q(r) = \sum_{i=0}^{n-1} \bar{h}_i r^{i}$, respectively. As we know, for compatible equations we always have $h_{n-1} = \bar{h}_{n-1} = 0$. Curiously, in all examples we studied we also found that 
    \begin{align}
        \text{DOFs} = 
        \begin{cases}
            h_2 & \text{for non-gauge theories} \\
            \bar{h}_2 & \text{for gauge theories}
        \end{cases}\,.
    \end{align}
    This can be understood analytically. As we will now show, it is true in general that $h_{n-2}$ and $\bar{h}_{n-2}$ are proportional to the number of degrees of freedom. The exact equality we found in all examples is a peculiarity of working with second order equations in four dimensions.

    In what follows we focus on $\bar{h}_{n-2}$, since analogous statements for $h_{n-2}$ follow from it as special cases. To begin with, we use~\eqref{eq:hbar} to write $\bar{h}_{n-2}$ as \begin{align}
        \bar{h}_{n-2} &= h_{n-2} - \frac{1}{(n-1)!} \sum_{\ell = 0}^p \gamma_\ell s^{(n-1)}_1(q+\ell) \notag\\
        &= h_{n-2} - \frac{1}{(n-1)!} \sum_{\ell = 0}^p \gamma_\ell \left[\frac12 n (n-1) + (n-1)(q+\ell)\right] \notag\\
        &= h_{n-2} - \frac{1}{(n-2)!} \sum_{\ell = 0}^p \gamma_\ell \left[\frac12 n + q + \ell\right]\,,
    \end{align}
    where we used that $s^{(n-1)}_1(q+\ell) = \sum_{i=1}^{n-1}(i+q+\ell) = \frac12 n (n-1) + (n-1)(q+\ell)$. Next, we use~\eqref{eq:hi} to express $h_{n-2}$ as
    \begin{align}
        h_{n-2} &= \sum_{k=n-2}^{n-1} \frac{\alpha^{(k+1)}_q}{k!}s^{(k)}_{k-n+2}(0) \notag\\
        &= \frac{\alpha^{(n-1)}_q}{(n-2)!}\underbrace{s^{(n-2)}_0(0)}_{=1} + \frac{\alpha^{(n)}_q}{(n-1)}\underbrace{s^{(n-1)}_1(0)}_{=\frac12 n (n-1)} \notag\\
        &= \frac{\alpha^{(n-1)}_q}{(n-2)!} + \frac12 n \frac{\alpha^{(n)}_q}{(n-2)!}\,.
    \end{align}
    After plugging this into the expression for $\bar{h}_{n-2}$, we obtain
    \begin{align}
        \bar{h}_{n-2} &= \frac{\alpha^{(n-1)}_q}{(n-2)!} + \frac12 \frac{n}{(n-2)!}\left[\alpha^{(n)}_q - \sum_{\ell = 0}^p \gamma_\ell\right] - \frac{1}{(n-2)!} \sum_{\ell = 0}^p \gamma_\ell (q+\ell)\,.
    \end{align}
    For compatible theories, the square bracket vanishes. We are then left with
    \begin{align}
        \bar{h}_{n-2} &= \frac{1}{(n-2)!}\left(\alpha^{(n-1)}_q - q \sum_{\ell = 0}^p \gamma_\ell - \sum_{\ell = 0}^p \ell \gamma_\ell\right)\,.
    \end{align}
    From~\eqref{eq:Z1q} we recognize the term in the round bracket to be the strength divided by $n-1$, i.e. $\frac{Z^{(1)}_q}{(n-1)}$. According to~\eqref{eq:FinalFormula}, the degrees of freedom are given by
    \begin{align}
        \text{DOFs} &= \frac{Z^{(1)}_q}{(n-1)q}\,.
    \end{align}
    We finally conclude that $\bar{h}_{n-2}$ can be written as
    \begin{align}
        \bar{h}_{n-2} = \frac{1}{(n-2)!} \frac{Z^{(1)}_q}{n-1} = \frac{q}{(n-2)!} \text{DOFs}\,.
    \end{align}
    The coefficient $\bar{h}_{n-2}$ is thus always proportional to the degrees of freedom. However, only in cases where $q = (n-2)!$ is it exactly equal to the $\text{DOFs}$. In particular, this happens for $n=4$ and $q=2$, which are the parameters we used in every example studied in this subsection. 

    \item \emph{Degrees of freedom from $\alpha^{(3)}_2$, $\bar{\alpha}^{(3)}_2$, and $\beta^{(4)}_2 - \beta^{(3)}_2$:} In all examples we can simply read of the number of phase space degrees of freedom either from $\alpha^{(3)}_2$ or $\bar{\alpha}^{(3)}_2$. The former encodes the phase space degrees of freedom for non-gauge theories, while the latter encodes them for gauge theories. This is easy to understand. From the definition of the Cartan characters~\ref{def:CartanCharacters} we obtain
    \begin{align}
        \alpha^{(n-1)}_q = m \binom{n+q-(n-1)-1}{n-(n-1)} - \beta^{(n-1)}_q = m \, q - \beta^{(n-1)}_q\,.
    \end{align}
    A comparison with~\eqref{eq:FinalFormula} shows that this is exactly $q$ times the degrees of freedom in the case where all the $\gamma_\ell$'s vanish.

    For the gauge theory case, we need to use the recursion relations~\eqref{eq:RecRelGaugeCorrCartanChar} and~\eqref{eq:hbar}. For $k=n-1$ we obtain
    \begin{align}
        \bar{\alpha}^{(n-1)}_q &= (n-2)! \bar{h}_{n-2} - \frac{(n-2)!}{n!} \underbrace{\bar{\alpha}^{(n)}_q}_{=0}s^{(n-1)}_{1}(0) \notag\\
        &= (n-2)!\left(h_{n-2} - \frac{1}{(n-1)!} \sum_{\ell=0}^p \gamma_\ell \underbrace{s^{(n-1)}_{1}(q+\ell)}_{=(n-1)(q+\ell+1)}\right) \notag\\
        &= (n-2)!\left(\frac{\alpha^{(n-1)}_q}{(n-2)!} + \frac{\alpha^{(n)}_q}{(n-1)!}\underbrace{s^{(n-1)}_1(0)}_{=n-1} - \frac{1}{(n-2)!} \sum_{\ell=0}\gamma_\ell\left\{q+\ell+1\right\}\right) \notag\\
        &= \alpha^{(n-1)}_q + \underbrace{\alpha^{(n)}_q - \sum_{\ell=0}^p \gamma_\ell }_{=0} - q\sum_{\ell = 0}^p \gamma_\ell - \sum_{\ell=1}^p \ell \gamma_\ell\,.
    \end{align}
    Using $\alpha^{(n-1)}_q = q\,m - \beta^{(n-1)}_q$, we deduce that
    \begin{align}
        \bar{\alpha}^{(n-1)}_q = m\, q - q\sum_{\ell = 0}^p\gamma_\ell - \beta^{(n-1)}_q - \sum_{\ell=0}^p \ell\gamma_\ell\,,
    \end{align}
    which is exactly equal to $q$ times~\eqref{eq:FinalFormula}.

    In almost all cases, the configuration space degrees of freedom can also be directly computed from $\beta^{(4)}_2 - \beta^{(3)}_2$. This is true in all examples we studied, except in the four massive theories. In particular this holds irrespective of whether the corresponding theory is a gauge theory or not and it also holds when there are ghost instabilities.

    Notice that we encountered the difference $\beta^{(4)}_2 - \beta^{(3)}_2$ in equation~\eqref{eq:PositiveDOFs}, which we obtained by postulating the upper bound~\eqref{eq:PostulatedUpperBound} for gauge theories. Recall that \eqref{eq:PositiveDOFs} is obtained by saturating the upper bound. Interestingly, all gauge and non-gauge theories seem to obey this equation. The pattern is only broken by massive theories where gauge symmetry was not restored using the St\"uckelberg trick. At this stage we cannot offer any explanation for this pattern. 
    
    \item \emph{Structure of the (gauge-corrected) Hilbert polynomial:} 
    In all examples we studied, we found that the Hilbert polynomial of non-gauge theories possess non-negative, integer coefficients. For gauge theories, it is the gauge-corrected Hilbert polynomial which has this property.

    This pattern can easily be understood for non-gauge theories. We only need to assume compatible equations. From the recursion relation for the $h_k$ (see equations~\eqref{eq:hn-1h0}and~\eqref{eq:hcoeff}) we obtain
    \begin{align}
        h_{n-1} &= \frac{\alpha^{(n)}_q}{(n-1)!} = 0\notag\\
        h_0 &= \sum_{k=1}^n \alpha^{(k)}_q\,.
    \end{align}
    On the first line we used the compatibility condition. The coefficient $h_0$ is clearly a non-negative integer since every $\alpha^{(k)}_q$ in the sum can only attain non-negative integer values. These two relations are true for any dimension $n$ and any order $q$.

    Next we consider $h_{n-2}$. From bullet point 6 we already know that
    \begin{align}
        h_{n-2} = \frac{q}{(n-1)!} \text{DOFs}\,.
    \end{align}
    As we showed in~\eqref{eq:PositiveDOFsNonGauge}, the number of degrees of freedom in the absence of gauge symmetry is always positive. Thus, $h_{n-2}$ is a positive number for any dimension $n\geq 1$ and any order $q\geq 1$. However, it is not an integer for any value of $n$. From now on we shall therefore focus on $n=4$ and $q=2$, which covers all examples we studied.

    In the $n=4$ case, the only remaining coefficient to consider is $h_1$. We begin with arbitrary $n$ and then specialize to $n=4$. The relevant coefficient to compute is $h_{n-3}$. From~\eqref{eq:hcoeff} we obtain
    \begin{align}
        h_{n-3} &= \sum_{k=n-3}^{n-1} \frac{\alpha^{(k+1)}_q}{k!} s^{(k)}_{k-n-3}(0) \notag\\
        &= \frac{\alpha^{(n-2)}_q}{(n-3)!} s^{(n-3)}_0(0) + \frac{\alpha^{(n-1)}_q}{(n-2)!} s^{(n-2)}_1(0) + \frac{\alpha^{(n)}_q}{(n-1)!} s^{(n-1)}_2(0) \notag\\
        &= \frac{1}{(n-3)!} \left[\alpha^{(n-2)}_q + \frac12 (n-1)\alpha^{(n-1)}_q + \frac{n (3n-1)}{4!} \alpha^{(n)}_q\right]\,,
    \end{align}
    where we used the definitions~\eqref{eq:ModifiedStirlingNumbers} and~\eqref{eq:ESPSpecialCases} to evaluate the modified Stirling numbers:
    \begin{align}
        s^{(n-3)}_0(0) &= 1 \notag\\
        s^{(n-2)}_1(0) &= \sum_{i=1}^{n-2} i = \frac12 (n-2)(n-1) \notag\\
        s^{(n-1)}_2(0) &= \sum_{i=1}^{n-2}\sum_{j=i+1}^{n-1} i j = \frac{1}{4!} (n-2)(n-1)n(3n-1)\,.
    \end{align}
    For compatible equations, and using $\alpha^{(n-1)}_q = m \, q - \beta^{(n-1)}_q = q \text{DOFs}$, the expression for $h_{n-3}$ simplifies to
    \begin{align}
        h_{n-3} = \frac{1}{(n-3)!} \left[\alpha^{(n-2)}_q + \frac12 (n-1) q \text{DOFs}\right]\,.
    \end{align}
    This is again a non-negative number for any $n$ and any $q$. Also, $q \text{DOFs}$ is always an integer. However, $h_{n-3}$ is only an integer if $q$ or $\text{DOFs}$ is divisible by $2$ and the bracket is divisible by $(n-3)!$. This is the case for $n=4$ and $q=2$. Thus, $h_1$ is a positive integer.

    We also see that in the case of non-gauge theories the coefficients of the Hilbert polynomial can be ordered as
    \begin{align}
        h_0 \geq h_1 \geq h_2 \geq h_3 \geq 0\,.
    \end{align}
    This is indeed confirmed in the relevant examples we studied. For gauge theories, however, no such ordering exists. It is true in most examples we considered, but we also found counter-examples to this ordering.

    Moreover, it is not easy to verify analytically that all coefficients are non-negative integers. For $\bar{h}_{n-1}$ it is still true that
    \begin{align}
        \bar{h}_{n-1} = \frac{\bar{\alpha}^{(n)}_q}{(n-1)!} = 0\,,
    \end{align}
    when compatibility is employed. For $\bar{h}_{n-2}$ we already saw in bullet point 6 that
    \begin{align}
        \bar{h}_{n-2} = \frac{q}{(n-2)!} \text{DOFs}\,.
    \end{align}
    By using the upper bound~\eqref{eq:UpperBoundSum} we can again conclude that this is a positive integer for $n$ and $q$ chosen accordingly. However, the remaining coefficients are not easy to analyze. For instance, $\bar{h}_0$ is now given by
    \begin{align}
        \bar{h}_0 = \sum_{k=1}^n \alpha^{(k)}_q - \sum_{\ell=0}^p \gamma_\ell \binom{q+\ell + n-1}{n-1}\,.
    \end{align}
    It is not immediately clear that the second sum is smaller than the first one. It is not even clear that the second sum is always an integer!

    Similarly, we encounter difficulties when analyzing $\bar{h}_{n-3}$. From~\eqref{eq:hbar} we obtain
    \begin{align}\label{eq:Intermedhbar}
        \bar{h}_{n-3} &= h_{n-3} - \frac{1}{(n-1)!} \sum_{\ell = 0}^p \gamma_\ell s^{(n-1)}_2(q+\ell) \notag\\
        &= \frac{1}{(n-3)!} \left[\alpha^{(n-2)}_q + \frac12 (n-1)\alpha^{(n-1)}_q + \frac{n (3n-1)}{4!} \alpha^{(n)}_q\right]\notag\\
        &\phantom{=} - \frac{1}{(n-1)!} \sum_{\ell = 0}^p \gamma_\ell s^{(n-1)}_2(q+\ell)\,.
    \end{align}
    To proceed with the computations we need to determine the modified Stirling number $s^{(n-1)}_2(q+\ell)$. Using~\eqref{eq:ESPSpecialCases} we find
    \begin{align}
        s^{(n-1)}_2(q+\ell) &= \sum_{i=1}^{n-2} \sum_{j=i+1}^{n-1} (i+q+\ell)(j+q+\ell) \notag\\
        &= \sum_{i=1}^{n-2} \sum_{j=i+1}^{n-1}\left\{i j + i (q+\ell) + j (q+\ell) + (q+\ell)^2\right\}\,.
    \end{align}
    It is convenient to evaluate this expression term by term:
    \begin{align}
        \sum_{i=1}^{n-2} \sum_{j=i+1}^{n-1} i j &= \frac{1}{4!}(n-2)(n-1)n(3n-1)  \notag\\
        \sum_{i=1}^{n-2} \sum_{j=i+1}^{n-1} i (q+\ell) &= \frac16 (n-2)(n-1)n (q+\ell) \notag\\
        \sum_{i=1}^{n-2} \sum_{j=i+1}^{n-1} j (q+\ell) &= \frac13 (n-2)(n-1) n (q+\ell) \notag\\
        \sum_{i=1}^{n-2} \sum_{j=i+1}^{n-1} (q+\ell)^2 &= \frac12 (n-2)(n-1) (q+\ell)^2\,.
    \end{align}
    Notice that every term is proportional to $(n-2)(n-1)$. Thus, when plugging the modified Stirling number back into~\eqref{eq:Intermedhbar}, we find
    \begin{align}
        \bar{h}_{n-3} &= \frac{1}{(n-3)!} \left[\alpha^{(n-2)}_q + \frac12 (n-1)\alpha^{(n-1)}_q + \frac{n (3n-1)}{4!} \alpha^{(n)}_q\right] \notag \\
        &\phantom{=} - \frac{1}{(n-3)!} \sum_{\ell=0}^p \gamma_\ell \left\{\frac{n(3n-1)}{4!} + \frac{n}{2}(q+\ell) + \frac12 (q+\ell)^2\right\}\,.
    \end{align}
    Observe that for compatible equations, the two terms proportional to $\frac{1}{4!}n (3n-1)$ cancel each other. Moreover, we can use the formula for the degrees of freedom to write
    \begin{align}
        \bar{h}_{n-3} &= \frac{1}{(n-3)!}\left[\alpha^{(n-2)}_q + \frac12 (n-1) q \text{DOFs} - \sum_{\ell=0}^p \gamma_\ell\left\{\frac{q(q+1)}{2} + \frac{\ell (\ell+1)}{2} + q\ell\right\} \right]\,.
    \end{align}
    Unsurprisingly, the first two terms in the square bracket are the same as the ones for $h_{n-3}$. It is the third term which introduces problems. In this case it is clear that the sum is always an integer, but it is difficult to see under which condition it is smaller than the two other terms in the square bracket. Thus we can only conclude that $\bar{h}_{n-3}$ is an integer, but we cannot decide whether it is always non-negative. 
\end{enumerate}

\subsection{Application to \texorpdfstring{$f(\bbQ)$}{f(Q)} Gravity}\label{ssec:fQGravity}
In recent years, $f(\bbQ)$ gravity has been a very active field of research. In particular, it has sparked a lot of activities on black hole physics~\cite{Wang:2021, DAmbrosio:2021b, Bahamonde:2022, Javed:2023, Junior:2023}, cosmology~\cite{BeltranJimenez:2019c, Capozziello:2022, DAmbrosio:2020c, DAmbrosio:2021, Dimakis:2022, Frusciante:2021, Heisenberg:2022}, wormholes~\cite{Banerjee:2021, Parsaei:2022, Mustafa:2023} and exotic stars~\cite{Maurya:2022, Sokoliuk:2022}. However, analyzing the structure of the theory itself proved to be rather difficult. In particular the question how many degrees of freedom are present in $f(\bbQ)$ gravity could only be gradually answered~\cite{Hu:2022, Tomonari:2023, DAmbrosio:2023, Heisenberg:2023b} and the final answer involved very complicated perturbative techniques. 

Here, we provide an alternative derivation of the number of propagating degrees of freedom based on the Cartan-Kuranishi algorithm and formula~\eqref{eq:FinalFormula}. We begin by introducing the necessary terminology.

Let $(\M, g_{\mu\nu}, \Gamma\ud{\alpha}{\mu\nu})$ be a \textbf{metric-affine geometry} consisting of a $4$-dimensional, connected manifold $\M$, a metric $g_{\mu\nu}$, and an \textbf{affine connection} $\Gamma\ud{\alpha}{\mu\nu}$. The latter is used to define the covariant derivative $\nabla_\mu$. Its action on vectors $V^\mu$ and $1$-forms $\omega_\mu$ is
\begin{align}
    \nabla_\mu V^\nu &= \partial_\mu V^\nu + \Gamma\ud{\nu}{\mu\lambda} V^\lambda \notag\\
    \nabla_\mu \omega_\nu &= \partial_\mu \omega_\nu - \Gamma\ud{\lambda}{\mu\nu} \omega_\lambda\,.
\end{align}
Furthermore, we can construct the following three tensors, which characterize any metric-affine geometry~\cite{BeltranJimenez:2017, BeltranJimenez:2018, BeltranJimenez:2019, Heisenberg:2023a}:
\begin{align}
    &\text{Curvature tensor:} & R\ud{\alpha}{\beta\mu\nu} &\ce 2\partial_{[\mu}\Gamma\ud{\alpha}{\nu]\beta} + 2 \Gamma\ud{\alpha}{[\mu|\lambda} \Gamma\ud{\lambda}{\nu]\beta} \notag\\
    &\text{Torsion tensor:} & T\ud{\alpha}{\mu\nu} &\ce 2 \Gamma\ud{\alpha}{[\mu\nu]} \notag\\
    &\text{Non-metricity tensor:} & Q_{\alpha\mu\nu} &\ce \nabla_\alpha g_{\mu\nu} = \partial_\alpha g_{\mu\nu} - 2\Gamma\ud{\lambda}{\alpha(\mu} g_{\nu)\lambda}\,.
\end{align}
It is then postulated that curvature and torsion both vanish: 
\begin{align}
    R\ud{\alpha}{\beta\mu\nu} &\overset{!}{=} 0 &&\text{and} & T\ud{\alpha}{\mu\nu} & \overset{!}{=} 0\,.
\end{align}
These are two conditions on the connection $\Gamma\ud{\alpha}{\mu\nu}$, which imply that it has to be given by~\cite{BeltranJimenez:2017, BeltranJimenez:2018, BeltranJimenez:2019, Heisenberg:2023a}
\begin{align}
    \Gamma\ud{\alpha}{\mu\nu} &= \PD{x^\alpha}{\xi^\lambda} \partial_\mu\partial_\nu \xi^\lambda\,,
\end{align}
where $x^\alpha$ are coordinates and $\xi^\lambda$ are four arbitrary functions only subjected to the condition that $\PD{\xi^\lambda}{x^\alpha}$ is a non-degenerate matrix. Remarkably, the connection can be transformed away by a change of coordinates. In fact, when the coordinates are chosen such that $x^\mu = \xi^\mu$, one obtains $\Gamma\ud{\alpha}{\mu\nu}$, since $\partial_\mu\partial_\nu \xi^\lambda=0$. This is known as the \textbf{coincident gauge}.

Using quadratic contractions of the non-metricity tensor, it is possible to construct five distinct scalars. However, only four are needed to construct the so-called \textbf{non-metricity scalar} $\bbQ$:
\begin{align}
    \bbQ \ce -\frac14\, Q_{\alpha\beta\gamma}Q^{\alpha\beta\gamma} + \frac12\, Q_{\alpha\beta\gamma}Q^{\beta\alpha\gamma}	 + \frac14\, Q_\alpha Q^{\alpha} - \frac12 c_5\, Q_\alpha \bar{Q}^\alpha\,.
\end{align}
Here, $Q_\alpha$ and $\bar{Q}_\alpha$ refer to the two independent traces of the non-metricity tensor:
\begin{align}
    Q_\alpha &\ce Q_{\alpha\nu}{}^{\nu} &&\text{and} & \bar{Q}_\alpha &\ce Q^\nu{}_{\nu\alpha}\,.
\end{align}
The non-metricity scalar is related to the Ricci scalar $\R$ of the metric $g_{\mu\nu}$ through the identity
\begin{align}\label{eq:QRId}
    \bbQ = \R + \mathcal{D}_\mu(Q^\mu-\bar Q^\mu)\,,
\end{align}
where $\mathcal{D}_\mu$ is the covariant derivative with respect to the Christoffel symbols of the metric.  This identity is instrumental in constructing STEGR---the symmetric teleparallel equivalent of GR~\cite{BeltranJimenez:2017, BeltranJimenez:2018, BeltranJimenez:2019}. This is a reformulation of GR based on non-metricity, rather than curvature. Coincident GR, which we studied in Subsection~\ref{ssec:Examples} as an example, is a special case of STEGR, obtained by employing the coincident gauge. 

One can introduce modified non-metricity scalars, for which identity~\eqref{eq:QRId} no longer holds, to study modifications of GR~\cite{BeltranJimenez:2017, BeltranJimenez:2018, BeltranJimenez:2019, Dambrosio:2020b}. Alternatively, one can consider non-linear modifications where the Lagrangian of STEGR, given by $\sqrt{-g}\, \bbQ$, is replaced by $\sqrt{-g} f(\bbQ)$. The function $f$ is required to have a non-vanishing first derivative, denoted by $f'(\bbQ)\ce \frac{\dd f(\bbQ)}{\dd \bbQ}$, but is otherwise arbitrary (see \cite{Heisenberg:2023a} for a review). This gives rise to $f(\bbQ)$ gravity, whose vacuum field equations read
\begin{align}
    \M_{\mu\nu} \ce f'(\bbQ) G_{\mu\nu}  - \frac12 g_{\mu\nu} (f(\bbQ)-f'(\bbQ)\bbQ) + 2 f''(\bbQ) P^\alpha{}_{\mu\nu} \partial_\alpha\bbQ &= 0\notag\\
	\mathcal C_\alpha \ce \frac{1}{\sqrt{-g}}\nabla_\mu\nabla_\nu\left(\sqrt{-g}\,f'(\bbQ) P^{\mu\nu}{}_{\alpha}\right) &= 0\,.
\end{align}
We refer to $\M_{\mu\nu}$ and $\mathcal C_\alpha$ as metric field equations and connection field equations, respectively. The tensor $P\ud{\alpha}{\mu\nu}$ is called the \emph{non-metricity conjugate} and it is given by
\begin{align}
    P^\alpha{}_{\mu\nu} &\ce \frac12 \PD{\bbQ}{Q_\alpha{}^{\mu\nu}} = -\frac14 Q^\alpha{}_{\mu\nu} + \frac12 Q_{(\mu}{}^{\alpha}{}_{\nu)} +\frac14 g_{\mu\nu}Q^\alpha -\frac14 \left(g_{\mu\nu} \bar{Q}^\alpha + \delta^\alpha{}_{(\mu} Q_{\nu)}\right)\,,
\end{align}
while $G_{\mu\nu}$ is the Einstein tensor. Observe that $Q_{\alpha\mu\nu}$, and consequently also $P\ud{\alpha}{\mu\nu}$, only contains first order derivatives of the metric and no derivatives of the connection. Therefore, $\M_{\mu\nu}$ contains second order derivatives of the metric through $G_{\mu\nu}$ and $P\ud{\alpha}{\mu\nu}\partial_\alpha \bbQ$ and at most first order derivatives of the connection. The connection field equations, on the other hand, contain third order derivatives of the metric and second order derivatives of the connection. 

This is an important observation when it comes to applying the Cartan-Kuranishi algorithm. It is also important to realize that the metric and connection field equations are not independent. They are related by the contracted Bianchi identities~\cite{BeltranJimenez:2020, DAmbrosio:2023}:
\begin{align}\label{eq:ContractedBianchi}
    \mathcal{D}^\mu \M_{\mu\nu} + \mathcal{C}_\nu = 0\,.
\end{align}
To simplify the computations necessary for the execution of the Cartan-Kuranishi algorithm, we make use of the coincident gauge. This has the advantage that the connection drops out of all equations and any degrees of freedom it might have propagated are transferred to the metric. The form of the equations $\M_{\mu\nu}$ and $\mathcal{C}_\nu$ is not affected by this gauge choice. 

Next, we perform a simple count: The metric is composed of ten components $g_{\mu\nu}$, while the connection is parametrized by the four functions $\xi^\mu(x)$. This gives a total of $14$ field components which have to be solved for. To do so, we have $10+4$ equations at our disposal---ten from the metric field equations and another four from the connection field equations. However, since we fixed the coincident gauge, the four function $\xi^\mu(x)$ are already specified, leaving us with ten unknown metric components and $14$ equations. It seems that the system is overdetermined.

The resolution to this apparent tension lays in the contracted Bianchi identities~\eqref{eq:ContractedBianchi}: The ten second order metric field equations in coincident gauge depend only on the metric and nothing else. Thus, they are sufficient to determine the metric---at least in principle. Once the metric has been determined, we can assume $\M_{\mu\nu} = 0$. This implies that also $\mathcal{D}^\mu \M_{\mu\nu} = 0$. The contracted Bianchi identities then tell us that $\mathcal{C}_\nu$ is trivially equal to zero. In other words, the connection field equations do not provide any new information. They simply reduce to identities. 

This observation allows us to focus our attention exclusively on the metric field equations. Moreover, we make use of the fact that in coincident gauge we always have the freedom of choosing $\left.g_{\mu\nu}\right|_p = \eta_{\mu\nu}$, where $p$ is some arbitrary point in $\M$ and $\eta_{\mu\nu}$ denotes the Minkowski metric. 

To see why this is true, we first observe that any coordinate transformation of the form
\begin{align}
    x^\mu \mapsto \tilde{x}^\mu = M\ud{\mu}{\nu} x^\nu + x^\mu_0\,,
\end{align}
where $M\ud{\mu}{\nu}$ is a non-degenerate $4\times 4$ matrix with constant entries and $x^\mu_0$ is a constant displacement vector, \emph{preserves the coincident gauge}. Preserving the coincident gauge means that a change of coordinates leaves $\Gamma\ud{\alpha}{\mu\nu} = 0$. Indeed, we find that the above coordinate transformation has this property:
\begin{align}
    \Gamma\ud{\alpha}{\mu\nu}\quad \mapsto\quad \tilde{\Gamma}\ud{\alpha}{\mu\nu} = \PD{\tilde{x}^\alpha}{x^\beta} \PD{x^\rho}{\tilde{x}^\mu} \PD{x^\sigma}{\tilde{x}^\nu} \underbrace{\Gamma\ud{\beta}{\rho\sigma}}_{=0} + \PD{\tilde{x}^\alpha}{x^\lambda}\underbrace{\PD{^2 x^\lambda}{\tilde{x}^\mu \partial \tilde{x}^\nu }}_{=0} = 0\,.
\end{align}
The metric, on the other hand, transforms as
\begin{align}
    g_{\mu\nu}\quad \mapsto\quad \tilde{g}_{\mu\nu} = \PD{x^\alpha}{\tilde{x}^\mu} \PD{x^\beta}{\tilde{x}^\nu} g_{\alpha\beta} = \left(M^{-1}\right)\ud{\alpha}{\mu} \left(M^{-1}\right)\ud{\beta}{\nu} g_{\alpha\beta}\,,
\end{align}
where $\left(M^{-1}\right)\ud{\alpha}{\mu}$ are the components of the inverse matrix $M^{-1}$, which is guaranteed to exist due to the non-degeneracy condition on $M$.

We now demand that at $p$ the transformed metric $\tilde{g}$ equals the Minkowski metric:
\begin{align}\label{eq:ConditionOnM}
    \left(M^{-1}\right)\ud{\alpha}{\mu} \left(M^{-1}\right)\ud{\beta}{\nu} \left.g_{\alpha\beta}\right|_p = \eta_{\mu\nu}\,.
\end{align}
At any given point $p$, $\left.g_{\mu\nu}\right|_p$ is simply a symmetric matrix with constant entries. From linear algebra we know that any symmetric matrix can be diagonalized by a non-degenerate matrix constructed from the eigenvectors. We therefore conclude that we can always find a matrix $M$ which satisfies~\eqref{eq:ConditionOnM}. Our claim follows: At any point $p$, $M$ is determined by~\eqref{eq:ConditionOnM} which diagonalizes $g_{\mu\nu}$ into the Minkowski metric while preserving the coincident gauge.

Note, however, that $\partial_\alpha g_{\mu\nu}$ does not transform in a tensorial fashion and therefore 
\begin{align}
    \left.\partial_\alpha g_{\mu\nu}\right|_p \neq 0
\end{align}
in general, which in turn implies $Q_{\alpha\mu\nu} \neq 0$. From now on, we assume that we fixed the coincident gauge in this fashion. The advantage is that the field equations simplify without trivializing. In fact, if we had used Riemannian normal coordinates, as we did in the GR example~\ref{sssec:GR} or the CGR example~\ref{sssec:CGR}, then we would obtain $Q_{\alpha\mu\nu} = 0$ which implies $P\ud{\alpha}{\mu\nu} = 0$ and consequently $f''(\bbQ)P\ud{\alpha}{\mu\nu} \partial_\alpha \bbQ = 0$. In other words, we would ``turn off'' the second order terms in $\M_{\mu\nu}$ which stem from non-trivial functions $f(\bbQ)$, leaving us only with the GR part.

It should also be noted that, in general, it is not possible to fix the coincident gauge \emph{and} Riemannian normal coordinates \emph{simultaneously}. Thus, the issue of ``turning off'' the modifications $f''(\bbQ)P\ud{\alpha}{\mu\nu} \partial_\alpha \bbQ$ to the GR dynamics does not occur in practice.

With these considerations out of the way, we are now in a position for executing the Cartan-Kuranishi Algorithm~\ref{alg:CK}. The first step is to prolong the metric field equations, extract the symbols $\S_2$ and $\S_3$, and to check whether $\S_2$ is involutive. We find that this is the case. 

Next, we check for integrability conditions. There are four such integrability conditions, which prompts us to perform a prolongation with a subsequent projection. The resulting system $\R^{(1)}_2$ has still an involutive symbol and no additional integrability conditions emerge. Thus, $\R^{(1)}_2$ is involutive and the algorithm terminates.

Finally, we extract the $\beta$'s from $\S^{(1)}_2$, from which we also obtain the Cartan characters:
\begin{align}
    \beta^{(1)}_2 &= 3\,, & \beta^{(2)}_2 &= 4\,, & \beta^{(3)}_2 &= 6 \,, & \beta^{(4)}_2 &= 10 \notag\\
    \alpha^{(1)}_2 &= 37\,, & \alpha^{(2)}_2 &= 26\,, & \alpha^{(3)}_2 &= 14 \,, & \alpha^{(4)}_2 &= 0\,.
\end{align}
The associated Hilbert polynomial reads
\begin{align}
    H_2(r) = 77 + 47 r + 7 r^2\,.
\end{align}
To compute the Cartan characters, we used $m=10$, since we fixed the four functions parametrizing the connection. Also, all the $\gamma_\ell$'s vanish, since we work in a completely gauge-fixed setting. Indeed, we found that $\alpha^{(4)}_2 = 0$, which means that the compatibility condition is satisfied. From~\eqref{eq:FinalFormula} we then conclude that
\begin{align}
    \text{DOFs} = 10 - \frac12 6 = 7\,.
\end{align}
This is in agreement with the upper bound reported in~\cite{DAmbrosio:2023} and, more importantly, with the number of degrees of freedom computed in~\cite{Heisenberg:2023b} using perturbation theory.

\section{Conclusion}\label{sec:Conclusion}
At the beginning of this work, we introduced a simple, qualitative definition of degrees of freedom: the number of independent variables needed to fully specify the state of a dynamical system. In practice, determining this number is often challenging due to the presence of gauge freedom and constraints. Systematic methods have been developed to address this issue, not only enabling the counting of degrees of freedom but also uncovering hidden constraints and offering deeper insight into the underlying physics. Among the most prominent are the Dirac-Bergmann algorithm and the covariant phase space approach. However, both methods have limitations and can be difficult to apply in concrete settings. Common challenges arise in topological field theories and in models where primary constraints involve spatial derivatives of configuration space variables~\cite{SundermeyerBook, Seiler:1995, Seiler:2000, Blagojevic:2020,DAmbrosio:2023}.

In this work, we approached the problem from a different and less commonly explored perspective---one grounded in the formal theory of systems of partial differential equations. To motivate this approach, we revisited Einstein’s attempts to classify how strongly a given set of field equations constrains the fields of a theory (cf.~Subsection~\ref{ssec:EinsteinsMethod}). The notions of compatibility and strength introduced there are closely connected to the modern understanding of degrees of freedom in theories that respect classical determinism.

On its own, Einstein’s method is insufficient for analyzing broad classes of field theories or for determining the number of degrees of freedom they propagate. Its computational complexity grows rapidly, making it impractical in general settings. Nonetheless, the core idea underlies a powerful and rigorous framework for studying any field theory: analyzing the solution space of the field equations via formal power series expansions.

This is precisely where the formal theory of systems of partial differential equations becomes indispensable. Section~\ref{sec:JetBundleApproach} was devoted to developing the foundational concepts and results of this theory. In particular, we introduced jet bundles and the accompanying shift in perspective: instead of treating PDEs as equations to be solved by integration, we interpret them as (potentially non-linear) algebraic relations among field variables and their derivatives.

To realize this perspective, we introduced the notions of \emph{prolongations} and \emph{projections} in Subsection~\ref{ssec:ProlongationsProjections}. Prolongations generate higher-order equations by differentiating the original PDEs, while projections serve to reveal hidden constraint equations. Both concepts are central to constructing formal power series solutions. Projections are especially important for determining how the field equations constrain the higher-order Taylor coefficients in the expansion. At the same time, identifying constraints is essential for fixing the lower-order coefficients. To ensure consistency, all such constraints must be systematically uncovered and incorporated. Only then can the power series expansion be analyzed coherently, order by order.

In Subsection~\ref{ssec:Symbol}, we introduced the concept of the \emph{symbol} of a partial differential equation. The symbol captures information about the highest-order derivatives appearing in the equations and plays a central role in the initial value formulation of PDEs. It also serves as a diagnostic tool for identifying constraint and integrability conditions (see Theorem~\ref{thm:DimR1q} and Corollary~\ref{cor:CriterionForIntCond}). The proof of Theorem~\ref{thm:DimR1q} illustrates how such constraints and conditions can be explicitly constructed through a detailed analysis of the symbol.

Building on this, we were led to the concept of involutive equations. An equation is said to be involutive if it contains all its integrability conditions (i.e., it is formally integrable) and allows for the systematic prediction of all higher-order principal derivatives. Involutive equations are essential to our discussion of degrees of freedom. Their structural properties enable a consistent, order-by-order analysis and ultimately make it possible to count the degrees of freedom in a rigorous manner.

The Cartan-Kuranishi algorithm~\ref{alg:CK}, discussed in Subsection~\ref{ssec:CKAlgorithm}, asserts that any equation\footnote{Under mild technical assumptions, as outlined in Subsection~\ref{ssec:CKAlgorithm}.} can be converted into an equivalent involutive system through a finite sequence of prolongations and projections. We demonstrated this procedure explicitly using the Proca equation.

In Section~\ref{sec:FormalPSSExtnsionToGT}, we showed how this order-by-order construction of a formal power series solution, grounded in the Cartan-Kuranishi algorithm, is carried out. We also explained how gauge theories are accommodated within the formal theory of PDEs. In particular, we introduced the Hilbert polynomial and its gauge-corrected counterpart. These tools culminated in a general formula for counting degrees of freedom (cf.~formula~\eqref{eq:FinalFormula}).

To validate this framework, Subsection~\ref{ssec:Examples} presented 14 examples illustrating the full methodology. These examples also served to test formula~\eqref{eq:FinalFormula}, which successfully reproduced the expected results in every case.

The broad range of examples---including gauge theories, massive and massless models without gauge symmetry, systems with restored gauge symmetry via the St\"uckelberg trick, and topological theories---allowed us to identify recurring patterns and derive broader insights. These observations were summarized in Subsection~\ref{sssec:ObservationsAndInsights}.

Finally, in Subsection~\ref{ssec:fQGravity}, we applied the Cartan-Kuranishi algorithm and formula~\eqref{eq:FinalFormula} to $f(\mathbb{Q})$ gravity. By exploiting the coincident gauge and its properties, we obtained a much more concise degree-of-freedom count than the one presented in~\cite{Heisenberg:2023b}. The result derived here matches the upper bound established in~\cite{DAmbrosio:2023} and agrees with the value obtained in~\cite{Heisenberg:2023b}.

The method developed in this work is both powerful and robust. Nonetheless, it raises several open questions, reveals limitations to be addressed, and suggests new directions for exploration. Some of the unresolved issues discussed include: under what conditions does formula~\eqref{eq:FinalFormula} yield an integer number of degrees of freedom? Is the upper bound~\eqref{eq:UpperBoundSum} physically justified? Under which assumptions can one prove that the gauge-corrected Cartan characters are positive integers? Notably, all these questions center around the interplay between gauge symmetry and constraint equations.

The greatest practical challenge in applying the techniques discussed in this work is the extraction of the $\beta$ coefficients from the symbol. To obtain these crucial numbers, one must first bring the symbol into solved form. This involves applying the Gauss algorithm in conjunction with a careful analysis of the dynamical and constraint equations. For example, during Gaussian elimination, one must ensure that division by expressions which vanish on the constraint surface is avoided. As the order of differentiation or the number of fields increases, the symbol grows rapidly in complexity, rendering these computations time-consuming and technically demanding.

One potential strategy to mitigate this difficulty is to reduce the differentiation order of the equation $\R_q$. For instance, instead of formulating Maxwell's equations as a second-order system governing the vector field $A^\mu$, one can equivalently describe them as a first-order system: the dynamical equation $\partial_\mu F^{\mu\nu} = 0$, accompanied by the definition $F_{\mu\nu} = \partial_\mu A_\nu - \partial_\nu A_\mu$. This reformulation alleviates computational complexity because the size of the symbol grows linearly with the number of fields but nonlinearly with the order of differentiation.

Indeed, any equation $\R_q$ can, in principle, be converted into an equivalent first-order system through the introduction of auxiliary fields. Determining under what conditions this reduction simplifies the analysis, how it impacts the performance of the Cartan-Kuranishi algorithm, and how it influences formula~\eqref{eq:FinalFormula} is a subject we leave for future work.

A related avenue for future investigation is the treatment of Hamiltonian systems. For second-order systems $\R_2$, transitioning to a Hamiltonian formulation lowers the order of differentiation by one while doubling the number of field variables. It remains unclear how this affects the efficiency of the Cartan-Kuranishi algorithm. However, in this context, we expect formula~\eqref{eq:FinalFormula} to directly yield the number of phase space degrees of freedom, rather than the configuration space degrees of freedom.

Ultimately, the formal theory of PDEs provides not just a practical tool for degree-of-freedom counting, but a conceptual lens through which the foundations of field theory can be better understood and explored.

\section{Dedication}\label{sec:dedication}
Du w\"ahltest selbst die Stunde,\\
wann dein erstes Licht die Welt ber\"uhrte.\\
Feuerrot umh\"ullte dich,\\
wie ein Schleier aus Sehnsucht und Schmerz.\\

Schon da,\\
und doch so fern,\\
als wolltest du die Tore\\
gleich wieder schliessen.\\

Ich hielt deine Hand,\\
atmete den Atem,\\
den du nicht finden wolltest,\\
liess mein Herz für zwei schlagen,\\
liess meine Gebete taumeln –\\
wild, verzweifelt,\\
zwischen Himmel und Erde.\\

„Bleib“, flehte ich,\\
„Bleib, mein Engel aus Unschuld,\\
mein kaum gesprochener Traum.“\\

Die Stadt legte sich ein Kleid aus Farben an,\\
jubelte dein Ankommen,\\
als w\"are der Morgen selbst\\
für dich allein erwacht.\\
Die B\"aume fl\"usterten in Rosaschattierungen:\\
„Verweile.“\\
Und meine Tr\"anen,\\
sie wurden zu Str\"omen,\\
bereit dich zu tragen,\\
dorthin, wo Leben wohnt.\\

Meine Lieder z\"ahlten die Sekunden,\\
wurden zum Takt der Zeit.\\
Und dann –\\
ein Atem, der dir geh\"orte,\\
ein Blick, der die Dunkelheit zerbrach.\\

Und so erhob sich die Sonne neu,\\
aus deiner Brust,\\
aus meinen Tr\"anen,\\
aus einem Augenaufschlag,\\
der die Ewigkeit\\
für einen Moment\\
festhielt.

\newpage
\bibliographystyle{JHEP}
\bibliography{Bibliography}

\providecommand{\href}[2]{#2}\begingroup\raggedright\begin{thebibliography}{10}

\bibitem{Dirac:1950}
P.~A.~M. Dirac, \emph{Generalized hamiltonian dynamics},
  \href{https://doi.org/10.4153/CJM-1950-012-1}{\emph{Canadian Journal of
  Mathematics} {\bfseries 2} (1950) 129}.

\bibitem{Bergmann:1951}
J.~L. Anderson and P.~G. Bergmann, \emph{Constraints in covariant field
  theories}, \href{https://doi.org/10.1103/PhysRev.83.1018}{\emph{Phys. Rev.}
  {\bfseries 83} (1951) 1018}.

\bibitem{DiracBook}
P.~A.~M. Dirac, \emph{{Lectures on Quantum Mechanics}}. Dover Publications,
  1964.

\bibitem{Wipf:1993}
A.~W. Wipf, \emph{{Hamilton's formalism for systems with constraints}},
  \href{https://doi.org/10.1007/3-540-58339-4\_14}{\emph{Lect. Notes Phys.}
  {\bfseries 434} (1994) 22}
  [\href{https://arxiv.org/abs/hep-th/9312078}{{\ttfamily hep-th/9312078}}].

\bibitem{Crnkovic:1987}
C.~Crnkovic and E.~Witten, \emph{Covariant description of canonical formalism
  in geometrical theories}, {\emph{Three Hundred Years of Gravitation} (1987)
  }.

\bibitem{Lee:1990}
J.~Lee and R.~Wald, \emph{Local symmetries and constraints}, {\emph{. Math.
  Phys.} {\bfseries 31} (1990) }.

\bibitem{HenneauxBook}
M.~Henneaux and C.~Teitelboim, \emph{{Quantization of gauge systems}}.
  Princeton University Press, 1992.

\bibitem{SundermeyerBook}
K.~Sundermeyer, \emph{Constrained Dynamics}, Lecture Notes in Physics.
  Springer, 1982.

\bibitem{Seiler:1995}
W.~M. Seiler and R.~W. Tucker, \emph{{Involution and Constrained Dynamics I:
  The Dirac Approach}},
  \href{https://doi.org/10.1088/0305-4470/28/15/022}{\emph{J. Phys. A}
  {\bfseries 28} (1995) 4431}
  [\href{https://arxiv.org/abs/hep-th/9506017}{{\ttfamily hep-th/9506017}}].

\bibitem{Seiler:2000}
W.~M. Seiler, \emph{{Involution and Constrained Dynamics III: Intrinsic Degrees
  of Freedom Count}}, {\emph{Technische Mechanik} {\bfseries 20} (2000) 137}.

\bibitem{Blagojevic:2020}
M.~Blagojevi\'c and J.~M. Nester, \emph{{Local symmetries and physical degrees
  of freedom in $f(T)$ gravity: a Dirac Hamiltonian constraint analysis}},
  \href{https://doi.org/10.1103/PhysRevD.102.064025}{\emph{Phys. Rev. D}
  {\bfseries 102} (2020) 064025}
  [\href{https://arxiv.org/abs/2006.15303}{{\ttfamily 2006.15303}}].

\bibitem{Hu:2022}
K.~Hu, T.~Katsuragawa and T.~Qiu, \emph{{ADM formulation and Hamiltonian
  analysis of f(Q) gravity}},
  \href{https://doi.org/10.1103/PhysRevD.106.044025}{\emph{Phys. Rev. D}
  {\bfseries 106} (2022) 044025}
  [\href{https://arxiv.org/abs/2204.12826}{{\ttfamily 2204.12826}}].

\bibitem{Tomonari:2023}
K.~Tomonari and S.~Bahamonde, \emph{{Dirac\textendash{}Bergmann analysis and
  degrees of freedom of coincident f(Q)-gravity}},
  \href{https://doi.org/10.1140/epjc/s10052-024-12677-x}{\emph{Eur. Phys. J. C}
  {\bfseries 84} (2024) } [\href{https://arxiv.org/abs/2308.06469}{{\ttfamily
  2308.06469}}].

\bibitem{DAmbrosio:2023}
F.~D'Ambrosio, L.~Heisenberg and S.~Zentarra, \emph{{Hamiltonian Analysis of
  $f(\mathbb {Q})$ Gravity and the Failure of the Dirac\textendash{}Bergmann
  Algorithm for Teleparallel Theories of Gravity}},
  \href{https://doi.org/10.1002/prop.202300185}{\emph{Fortsch. Phys.}
  {\bfseries 71} (2023) 2300185}
  [\href{https://arxiv.org/abs/2308.02250}{{\ttfamily 2308.02250}}].

\bibitem{BeltranJimenez:2019}
J.~Beltr\'an~Jim\'enez, L.~Heisenberg and T.~S. Koivisto, \emph{{The
  Geometrical Trinity of Gravity}},
  \href{https://doi.org/10.3390/universe5070173}{\emph{Universe} {\bfseries 5}
  (2019) 173} [\href{https://arxiv.org/abs/1903.06830}{{\ttfamily
  1903.06830}}].

\bibitem{Blixt:2018}
D.~Blixt, M.~Hohmann and C.~Pfeifer, \emph{{Hamiltonian and primary constraints
  of new general relativity}},
  \href{https://doi.org/10.1103/PhysRevD.99.084025}{\emph{Phys. Rev. D}
  {\bfseries 99} (2019) 084025}
  [\href{https://arxiv.org/abs/1811.11137}{{\ttfamily 1811.11137}}].

\bibitem{Blixt:2019}
D.~Blixt, M.~Hohmann, M.~Kr\v{s}\v{s}\'ak and C.~Pfeifer, \emph{{Hamiltonian
  Analysis In New General Relativity}},  in \emph{{15th Marcel Grossmann
  Meeting on Recent Developments in Theoretical and Experimental General
  Relativity, Astrophysics, and Relativistic Field Theories}}, 5, 2019,
  \href{https://arxiv.org/abs/1905.11919}{{\ttfamily 1905.11919}},
  \href{https://doi.org/10.1142/9789811258251_0038}{DOI}.

\bibitem{Blixt:2020}
D.~Blixt, M.-J. Guzm\'an, M.~Hohmann and C.~Pfeifer, \emph{{Review of the
  Hamiltonian analysis in teleparallel gravity}},
  \href{https://doi.org/10.1142/S0219887821300051}{\emph{Int. J. Geom. Meth.
  Mod. Phys.} {\bfseries 18} (2021) 2130005}
  [\href{https://arxiv.org/abs/2012.09180}{{\ttfamily 2012.09180}}].

\bibitem{DAmbrosio:2020a}
F.~D'Ambrosio, M.~Garg, L.~Heisenberg and S.~Zentarra, \emph{{ADM formulation
  and Hamiltonian analysis of Coincident General Relativity}},
  \href{https://arxiv.org/abs/2007.03261}{{\ttfamily 2007.03261}}.

\bibitem{Dambrosio:2020b}
F.~D'ambrosio and L.~Heisenberg, \emph{{Classification of primary constraints
  of quadratic non-metricity theories of gravity}},
  \href{https://doi.org/10.1007/JHEP02(2021)170}{\emph{JHEP} {\bfseries 02}
  (2021) 170} [\href{https://arxiv.org/abs/2007.05064}{{\ttfamily
  2007.05064}}].

\bibitem{Heisenberg:2023b}
L.~Heisenberg, M.~Hohmann and S.~Kuhn, \emph{{Cosmological teleparallel
  perturbations}},
  \href{https://doi.org/10.1088/1475-7516/2024/03/063}{\emph{JCAP} {\bfseries
  03} (2024) 063} [\href{https://arxiv.org/abs/2311.05495}{{\ttfamily
  2311.05495}}].

\bibitem{EinsteinBook}
A.~Einstein, \emph{The Meaning of Relativity}. Princeton University, 1955.

\bibitem{Cartan:1930}
E.~Cartan, \emph{La the\'{e}orie des groupes finis et continus et l'analysis
  situs}, {\emph{M\'{e}m. Sci. Math.} {\bfseries 42} (1930) }.

\bibitem{CartanBook}
E.~Cartan, \emph{Les Syst\`{e}mes Diff\'{e}rentielles Ext\'{e}rieurs et leurs
  Applications Geom\'{e}triques}. Hermann (Paris), 1945.

\bibitem{Kuranishi:1957}
M.~Kuranishi, \emph{{On \'E Cartan's prolongation theorem of exterior
  differential systems}}, {\emph{Amer. J. Math.} {\bfseries 79} (1957) 1}.

\bibitem{Seiler:1995b}
W.~M. Seiler, \emph{Involution and constrained dynamics. ii. the faddeev-jackiw
  approach}, {\emph{Journal of Physics A: Mathematical and General} {\bfseries
  28} (1995) 7315}.

\bibitem{SeilerBook}
W.~M. Seiler, \emph{Involution - The Formal Theory of Differential Equations
  and its Applications in Computer Algebra}, Band 24 von Algorithms and
  computation in mathematics,. Springer Berlin Heidelberg, 2010.

\bibitem{DebeverBook}
E.~Cartan and A.~Einstein, \emph{Lettres sur le Parall\'{e}lisme Absolu
  1929-1932}, edited by R. Debever. Palais des Acad\'{e}mies, 1979.

\bibitem{BaezBook}
J.~Baez and J.~P. Muniain, \emph{Gauge Fields, Knots and Gravity}. World
  Scientific, 1994.

\bibitem{MiersemannBook}
E.~Miersemann, \emph{Partial Differential Equations}. CreateSpace Independent
  Publishing Platform, 2014.

\bibitem{DAmbrosio:2022}
F.~D'Ambrosio, S.~D.~B. Fell, L.~Heisenberg, D.~Maibach, S.~Zentarra and
  J.~Zosso, \emph{{Gravitational Waves in Full, Non-Linear General
  Relativity}},  \href{https://arxiv.org/abs/2201.11634}{{\ttfamily
  2201.11634}}.

\bibitem{Heisenberg:2018}
L.~Heisenberg, \emph{{A systematic approach to generalisations of General
  Relativity and their cosmological implications}},
  \href{https://doi.org/10.1016/j.physrep.2018.11.006}{\emph{Phys. Rept.}
  {\bfseries 796} (2019) 1} [\href{https://arxiv.org/abs/1807.01725}{{\ttfamily
  1807.01725}}].

\bibitem{BeltranJimenez:2017}
J.~Beltr\'an~Jim\'enez, L.~Heisenberg and T.~Koivisto, \emph{{Coincident
  General Relativity}},
  \href{https://doi.org/10.1103/PhysRevD.98.044048}{\emph{Phys. Rev. D}
  {\bfseries 98} (2018) 044048}
  [\href{https://arxiv.org/abs/1710.03116}{{\ttfamily 1710.03116}}].

\bibitem{BeltranJimenez:2018}
J.~Beltr\'an~Jim\'enez, L.~Heisenberg and T.~S. Koivisto, \emph{{Teleparallel
  Palatini theories}},
  \href{https://doi.org/10.1088/1475-7516/2018/08/039}{\emph{JCAP} {\bfseries
  08} (2018) 039} [\href{https://arxiv.org/abs/1803.10185}{{\ttfamily
  1803.10185}}].

\bibitem{Heisenberg:2023a}
L.~Heisenberg, \emph{{Review on f($\mathbb{Q}$) gravity}},
  \href{https://doi.org/10.1016/j.physrep.2024.02.001}{\emph{Phys. Rept.}
  {\bfseries 1066} (2024) 1}
  [\href{https://arxiv.org/abs/2309.15958}{{\ttfamily 2309.15958}}].

\bibitem{deRham:2010kj}
C.~de~Rham, G.~Gabadadze and A.~J. Tolley, \emph{{Resummation of Massive
  Gravity}}, \href{https://doi.org/10.1103/PhysRevLett.106.231101}{\emph{Phys.
  Rev. Lett.} {\bfseries 106} (2011) 231101}
  [\href{https://arxiv.org/abs/1011.1232}{{\ttfamily 1011.1232}}].

\bibitem{Aoki:2022}
K.~Aoki, J.~Beltr\'an~Jim\'enez and D.~Figueruelo, \emph{{Some disquisitions on
  cosmological 2-form dualities}},
  \href{https://doi.org/10.1088/1475-7516/2023/04/059}{\emph{JCAP} {\bfseries
  04} (2023) 059} [\href{https://arxiv.org/abs/2212.12427}{{\ttfamily
  2212.12427}}].

\bibitem{Wang:2021}
W.~Wang, H.~Chen and T.~Katsuragawa, \emph{{Static and spherically symmetric
  solutions in f(Q) gravity}},
  \href{https://doi.org/10.1103/PhysRevD.105.024060}{\emph{Phys. Rev. D}
  {\bfseries 105} (2022) 024060}
  [\href{https://arxiv.org/abs/2110.13565}{{\ttfamily 2110.13565}}].

\bibitem{DAmbrosio:2021b}
F.~D'Ambrosio, S.~D.~B. Fell, L.~Heisenberg and S.~Kuhn, \emph{{Black holes in
  f(Q) gravity}},
  \href{https://doi.org/10.1103/PhysRevD.105.024042}{\emph{Phys. Rev. D}
  {\bfseries 105} (2022) 024042}
  [\href{https://arxiv.org/abs/2109.03174}{{\ttfamily 2109.03174}}].

\bibitem{Bahamonde:2022}
S.~Bahamonde and L.~J\"arv, \emph{{Coincident gauge for static spherical field
  configurations in symmetric teleparallel gravity}},
  \href{https://doi.org/10.1140/epjc/s10052-022-10922-9}{\emph{Eur. Phys. J. C}
  {\bfseries 82} (2022) 963}
  [\href{https://arxiv.org/abs/2208.01872}{{\ttfamily 2208.01872}}].

\bibitem{Javed:2023}
F.~Javed, G.~Mustafa, S.~Mumtaz and F.~Atamurotov, \emph{{Thermal analysis with
  emission energy of perturbed black hole in f(Q) gravity}},
  \href{https://doi.org/10.1016/j.nuclphysb.2023.116180}{\emph{Nucl. Phys. B}
  {\bfseries 990} (2023) 116180}.

\bibitem{Junior:2023}
J.~T. S.~S. Junior and M.~E. Rodrigues, \emph{{Coincident $f(\mathbb {Q})$
  gravity: black holes, regular black holes, and black bounces}},
  \href{https://doi.org/10.1140/epjc/s10052-023-11660-2}{\emph{Eur. Phys. J. C}
  {\bfseries 83} (2023) 475}
  [\href{https://arxiv.org/abs/2306.04661}{{\ttfamily 2306.04661}}].

\bibitem{BeltranJimenez:2019c}
J.~Beltr\'an~Jim\'enez, L.~Heisenberg, T.~S. Koivisto and S.~Pekar,
  \emph{{Cosmology in $f(Q)$ geometry}},
  \href{https://doi.org/10.1103/PhysRevD.101.103507}{\emph{Phys. Rev. D}
  {\bfseries 101} (2020) 103507}
  [\href{https://arxiv.org/abs/1906.10027}{{\ttfamily 1906.10027}}].

\bibitem{Capozziello:2022}
S.~Capozziello and R.~D'Agostino, \emph{{Model-independent reconstruction of
  f(Q) non-metric gravity}},
  \href{https://doi.org/10.1016/j.physletb.2022.137229}{\emph{Phys. Lett. B}
  {\bfseries 832} (2022) 137229}
  [\href{https://arxiv.org/abs/2204.01015}{{\ttfamily 2204.01015}}].

\bibitem{DAmbrosio:2020c}
F.~D'Ambrosio, M.~Garg and L.~Heisenberg, \emph{{Non-linear extension of
  non-metricity scalar for MOND}},
  \href{https://doi.org/10.1016/j.physletb.2020.135970}{\emph{Phys. Lett. B}
  {\bfseries 811} (2020) 135970}
  [\href{https://arxiv.org/abs/2004.00888}{{\ttfamily 2004.00888}}].

\bibitem{DAmbrosio:2021}
F.~D'Ambrosio, L.~Heisenberg and S.~Kuhn, \emph{{Revisiting cosmologies in
  teleparallelism}},
  \href{https://doi.org/10.1088/1361-6382/ac3f99}{\emph{Class. Quant. Grav.}
  {\bfseries 39} (2022) 025013}
  [\href{https://arxiv.org/abs/2109.04209}{{\ttfamily 2109.04209}}].

\bibitem{Dimakis:2022}
N.~Dimakis, A.~Paliathanasis, M.~Roumeliotis and T.~Christodoulakis,
  \emph{{FLRW solutions in f(Q) theory: The effect of using different
  connections}}, \href{https://doi.org/10.1103/PhysRevD.106.043509}{\emph{Phys.
  Rev. D} {\bfseries 106} (2022) 043509}
  [\href{https://arxiv.org/abs/2205.04680}{{\ttfamily 2205.04680}}].

\bibitem{Frusciante:2021}
N.~Frusciante, \emph{{Signatures of $f(Q)$-gravity in cosmology}},
  \href{https://doi.org/10.1103/PhysRevD.103.044021}{\emph{Phys. Rev. D}
  {\bfseries 103} (2021) 044021}
  [\href{https://arxiv.org/abs/2101.09242}{{\ttfamily 2101.09242}}].

\bibitem{Heisenberg:2022}
L.~Heisenberg, M.~Hohmann and S.~Kuhn, \emph{{Homogeneous and isotropic
  cosmology in general teleparallel gravity}},
  \href{https://doi.org/10.1140/epjc/s10052-023-11462-6}{\emph{Eur. Phys. J. C}
  {\bfseries 83} (2023) 315}
  [\href{https://arxiv.org/abs/2212.14324}{{\ttfamily 2212.14324}}].

\bibitem{Banerjee:2021}
A.~Banerjee, A.~Pradhan, T.~Tangphati and F.~Rahaman, \emph{{Wormhole
  geometries in $f(Q)$ gravity and the energy conditions}},
  \href{https://doi.org/10.1140/epjc/s10052-021-09854-7}{\emph{Eur. Phys. J. C}
  {\bfseries 81} (2021) 1031}
  [\href{https://arxiv.org/abs/2109.15105}{{\ttfamily 2109.15105}}].

\bibitem{Parsaei:2022}
F.~Parsaei, S.~Rastgoo and P.~K. Sahoo, \emph{{Wormhole in $f(Q)$ gravity}},
  \href{https://arxiv.org/abs/2203.06374}{{\ttfamily 2203.06374}}.

\bibitem{Mustafa:2023}
G.~Mustafa, S.~K. Maurya and S.~Ray, \emph{{Relativistic wormhole surrounded by
  dark matter halos in symmetric teleparallel gravity}},
  \href{https://doi.org/10.1002/prop.202200129}{\emph{Fortsch. Phys.}
  {\bfseries 71} (2023) 2200129}.

\bibitem{Maurya:2022}
S.~K. Maurya, K.~Newton~Singh, S.~V. Lohakare and B.~Mishra, \emph{{Anisotropic
  Strange Star Model Beyond Standard Maximum Mass Limit by Gravitational
  Decoupling in $f(Q)$ Gravity}},
  \href{https://arxiv.org/abs/2208.04735}{{\ttfamily 2208.04735}}.

\bibitem{Sokoliuk:2022}
O.~Sokoliuk, S.~Pradhan, P.~K. Sahoo and A.~Baransky, \emph{{Buchdahl quark
  stars within $f(Q)$ theory}},
  \href{https://doi.org/10.1140/epjp/s13360-022-03273-7}{\emph{Eur. Phys. J.
  Plus} {\bfseries 137} (2022) 1077}
  [\href{https://arxiv.org/abs/2209.11590}{{\ttfamily 2209.11590}}].

\bibitem{BeltranJimenez:2020}
J.~Beltr\'an~Jim\'enez, L.~Heisenberg and T.~Koivisto, \emph{{The coupling of
  matter and spacetime geometry}},
  \href{https://doi.org/10.1088/1361-6382/aba31b}{\emph{Class. Quant. Grav.}
  {\bfseries 37} (2020) 195013}
  [\href{https://arxiv.org/abs/2004.04606}{{\ttfamily 2004.04606}}].

\end{thebibliography}\endgroup

\end{document}